\newif\ifabstract
\abstracttrue
 \abstractfalse 
\newif\iffull
\ifabstract \fullfalse \else \fulltrue \fi

\documentclass[11pt]{article}
\usepackage{amsfonts}
\usepackage{amssymb}
\usepackage{amstext}
\usepackage{amsmath}
\usepackage{xspace}
\usepackage{theorem}
\usepackage{graphicx}
\usepackage{url}
\usepackage{graphics}
\usepackage{colordvi}
\usepackage{colordvi}
\usepackage{subfigure}

\advance \oddsidemargin by -0.8in
\newcommand{\myparskip}{3pt}
\parskip \myparskip

\newenvironment{proofof}[1]{\noindent{\bf Proof of #1.}}%
        {\hspace*{\fill}$\Box$\par\vspace{4mm}}

\newcommand{\tk}{\tilde k}

\newcommand{\dnot}{D^{\circ}}

\newcommand{\NDP}{{\sf NDP}\xspace}
\newcommand{\NDPdisc}{{\sf NDP-Disc}\xspace}
\newcommand{\NDPcyl}{{\sf NDP-Cylinder}\xspace}
\newcommand{\NDPgrid}{{\sf NDP-Grid}\xspace}
\newcommand{\NDPplanar}{{\sf NDP-Planar}\xspace}
\newcommand{\NDPwC}{{\sf NDPwC}\xspace}
\newcommand{\EDPwC}{{\sf EDPwC}\xspace}

\newcommand{\EDP}{{\sf EDP}\xspace}

\newcommand{\DPSP}{{\sf DPSP}\xspace}

\newcommand{\algsc}{\ensuremath{{\mathcal{A}}_{\mbox{\textup{\footnotesize{AKR}}}}}\xspace}
\newcommand{\alphasc}{\ensuremath{\alpha_{\mbox{\textup{\footnotesize{\tiny{AKR}}}}}}}
\newcommand{\dface}{d_{\mbox{\textup{\footnotesize{GNC}}}}}
\newcommand{\tgamma}{\tilde{\gamma}}

\newcommand{\ceil}[1]{\ensuremath{\left\lceil#1\right\rceil}}
\newcommand{\floor}[1]{\ensuremath{\left\lfloor#1\right\rfloor}}

\newcommand{\NP}{\mbox{\sf NP}}

\newcommand{\polylog}[1]{\mathrm{polylog(#1)}}

\newcommand{\ZPTIME}{\mbox{\sf ZPTIME}}

\newcommand{\opt}{\mathsf{OPT}}

\newcommand{\alphawl}{\ensuremath{\alpha_{\mbox{\tiny{\sc WL}}}}}
\newcommand{\alphaWL}{\ensuremath{\alpha_{\mbox{\tiny{\sc WL}}}}}

\newcommand{\set}[1]{\left\{ #1 \right\}}

\newcommand{\tpset}{\tilde{\mathcal P}}
\newcommand{\ts}{\tilde s}
\newcommand{\dset}{{\mathcal D}}
\newcommand{\lface}{L_{\tiny{\mbox{face}}}}
\newcommand{\tZ}{\tilde Z}
\newcommand{\trset}{\tilde{\mathcal R}}
\newcommand{\tset}{{\mathcal T}}
\newcommand{\ttset}{\tilde{\mathcal T}}
\newcommand{\tmset}{\tilde{\mathcal M}}
\newcommand{\tS}{\tilde S}
\newcommand{\tD}{\tilde D}
\newcommand{\tT}{\tilde T}

\newcommand{\hmset}{\hat{\mathcal{ M}}}

\newcommand{\iset}{{\mathcal{I}}}
\newcommand{\pset}{{\mathcal{P}}}
\newcommand{\oset}{{\mathcal{O}}}
\newcommand{\qset}{{\mathcal{Q}}}
\newcommand{\lset}{{\mathcal{L}}}
\newcommand{\bset}{{\mathcal{B}}}
\newcommand{\aset}{{\mathcal{A}}}
\newcommand{\cset}{{\mathcal{C}}}
\newcommand{\fset}{{\mathcal{F}}}
\newcommand{\mset}{{\mathcal M}}

\newcommand{\xset}{{\mathcal{X}}}

\newcommand{\yset}{{\mathcal{Y}}}
\newcommand{\rset}{{\mathcal{R}}}
\newcommand{\kset}{{\mathcal K}}
\newcommand{\tkset}{\tilde{\mathcal{K}}}

\newcommand{\tsigma}{\tilde{\sigma}}

\newcommand{\nset}{{\mathcal N}}

\newcommand{\hset}{{\mathcal{H}}}
\newcommand{\sset}{{\mathcal{S}}}
\newcommand{\zset}{{\mathcal{Z}}}

\newcommand{\be}{\begin{enumerate}}
\newcommand{\ee}{\end{enumerate}}
\newcommand{\bd}{\begin{description}}
\newcommand{\ed}{\end{description}}
\newcommand{\bi}{\begin{itemize}}
\newcommand{\ei}{\end{itemize}}

\newtheorem{theorem}{Theorem}[section]
\newtheorem{lemma}[theorem]{Lemma}
\newtheorem{observation}[theorem]{Observation}
\newtheorem{corollary}[theorem]{Corollary}
\newtheorem{claim}[theorem]{Claim}

\newtheorem{definition}{Definition}[section]
\newtheorem{remark}{Remark}[section]
\newenvironment{proof}{\par \smallskip{\bf Proof:}}{\hfill\stopproof}
\def\stopproof{\square}
\def\square{\vbox{\hrule height.2pt\hbox{\vrule width.2pt height5pt \kern5pt
\vrule width.2pt} \hrule height.2pt}}

\newcommand{\tw}{\mathrm{tw}}

\renewcommand{\phi}{\varphi}
\newcommand{\eps}{\epsilon}

\newcommand{\half}{\ensuremath{\frac{1}{2}}}

\newcommand{\poly}{\operatorname{poly}}

\newenvironment{properties}[2][0]
{
\begin{enumerate} \setcounter{enumi}{#1}}{\end{enumerate}}

\mathchardef\hyphen="2D
\newcommand{\mincycle}{\mathrm{min\hyphen cycle}}

\newcommand{\tcclass}[2]{\tilde{\mathcal{S}}_{#1}^{(#2)}}

\begin{document}

\title{Improved Approximation for Node-Disjoint Paths in Planar Graphs\footnote{An extended abstract is to appear in STOC 2016}}
\author{Julia Chuzhoy\thanks{Toyota Technological Institute at Chicago. Email: {\tt cjulia@ttic.edu}. Supported in part by NSF grant CCF-1318242.}\and David H. K. Kim\thanks{Computer Science Department, University of Chicago. Email: {\tt hongk@cs.uchicago.edu}. Supported in part by NSF grant CCF-1318242.} \and Shi Li\thanks{Department of Computer Science and Engineering, University at Buffalo. Part of the work done while the author was at the Toyota Technological Institute at Chicago. Email: {\tt shil@buffalo.edu}.}}

\begin{titlepage}
\maketitle

\thispagestyle{empty}

\begin{abstract}
We study the classical Node-Disjoint Paths (\NDP) problem: given an $n$-vertex graph $G$ and a collection $\mset=\set{(s_1,t_1),\ldots,(s_k,t_k)}$ of pairs of vertices of $G$ called \emph{demand pairs}, find a maximum-cardinality set of node-disjoint paths connecting the demand pairs. \NDP is one of the most basic routing problems, that has been studied extensively. Despite this, there are still wide gaps in our understanding of its approximability: the best currently known upper bound of $O(\sqrt n)$ on its approximation ratio is achieved via a simple greedy algorithm, while the best current negative result shows that the problem does not have a better than $\Omega(\log^{1/2-\delta}n)$-approximation for any constant $\delta$, under standard complexity assumptions. Even for planar graphs no better approximation algorithms are known, and to the best of our knowledge, the best negative bound is APX-hardness. Perhaps the biggest obstacle to obtaining better approximation algorithms for \NDP is that most currently known approximation algorithms for this type of problems rely on the standard multicommodity flow relaxation, whose integrality gap is $\Omega(\sqrt n)$ for \NDP, even in planar graphs. In this paper, we break the barrier of $O(\sqrt n)$ on the approximability of the \NDP problem in planar graphs and obtain an $\tilde O(n^{9/19})$-approximation. We introduce a new linear programming relaxation of the problem, and a number of new techniques, that we hope will be helpful in designing more powerful algorithms for this and related problems.
\end{abstract}

\end{titlepage}

\label{--------------------------------------------sec: intro---------------------------------------------------}
\section{Introduction}\label{sec: intro}

In the Node-Disjoint Paths (NDP) problem, we are given an $n$-vertex graph $G$, and a collection $\mset=\set{(s_1,t_1),\ldots,(s_k,t_k)}$ of pairs of vertices of $G$, called \emph{source-destination}, or \emph{demand}, pairs. The goal is to route as many of the demand pairs as possible, by connecting each routed pair with a path, so that the resulting paths are node-disjoint. We denote by \NDPplanar the special case of the problem where the input graph $G$ is planar, and by \NDPgrid the special case where $G$ is the $(\sqrt n\times\sqrt n)$-grid. 
\NDP is one of the most basic problems in the area of graph routing, and it was initially introduced to the area in the context of VLSI design. In addition to being extensively studied in the area of approximation algorithms, this problem has played a central role in Robertson and Seymour's Graph Minor series. 
When the number of the demand pairs $k$ is bounded by a constant,  Robertson and Seymour~\cite{RobertsonS,flat-wall-RS} have shown an efficient algorithm for the problem, as part of the series. When $k$ is a part of input, the problem becomes \NP-hard~\cite{Karp,EDP-hardness}, even on planar graphs~\cite{npc_planar}, and even on grid graphs~\cite{npc_grid}. Despite the importance of this problem, and many efforts, its approximability is still poorly understood. The following simple greedy algorithm achieves an $O(\sqrt n)$-approximation~\cite{KolliopoulosS}: while $G$ contains any path connecting any demand pair, choose the shortest such path $P$, add $P$ to the solution, and delete all vertices of $P$ from $G$.
Surprisingly, this elementary algorithm is the best currently known approximation algorithm for \NDP, even for planar graphs. Until recently, this was also the best approximation algorithm for \NDPgrid. On the negative side, it is known that there is no $O(\log^{1/2-\delta}n)$-approximation algorithm for \NDP for any constant $\delta$, unless $\NP \subseteq \ZPTIME(n^{\poly \log n})$~\cite{AZ-undir-EDP,ACGKTZ}. To the best of our knowledge, the best negative result for \NDPplanar and for \NDPgrid is APX-hardness~\cite{NDP-grids}. Perhaps the biggest obstacle to obtaining better upper bounds on the approximability of \NDP is that the common approach to designing approximation algorithms for this type of problems is to use the multicommodity flow relaxation, where instead of connecting the demand pairs with paths, we send a (possibly fractional) multicommodity flow between them. The integrality gap of this relaxation is known to be $\Omega(\sqrt n)$, even for planar graphs, and even for grid graphs. In a recent work, Chuzhoy and Kim~\cite{NDP-grids} showed an $\tilde O(n^{1/4})$-approximation algorithm for \NDPgrid, thus bypassing the integrality gap obstacle for this restricted family of graphs. The main result of this paper is an $\tilde O(n^{9/19})$-approximation algorithm for \NDPplanar. We also show that, if the value of the optimal solution to the \NDPplanar instance is $\opt$, then we can efficiently route $\Omega\left(\frac{\opt^{1/19}}{\poly\log n}\right )$ demand pairs. Our algorithm is motivated by the work of~\cite{NDP-grids} on \NDPgrid, and it relies on approximation algorithms for the \NDP problem on a disc and on a cylinder, that we discuss next.

We start with the \NDP problem on a disc, that we denote by \NDPdisc. In this problem, we are given a planar graph $G$, together with a set $\mset$ of demand pairs as before, but we now assume that $G$ can be drawn in a disc, so that all vertices participating in the demand pairs lie on its boundary. The \NDP problem on a cylinder, \NDPcyl, is defined similarly, except that now we assume that we are given a cylinder $\Sigma$, obtained from the sphere, by removing two disjoint open discs (caps) from it. We denote the boundaries of the discs by $\Gamma_1$ and $\Gamma_2$ respectively, and we call them the \emph{cuffs} of $\Sigma$. We assume that $G$ can be drawn on $\Sigma$, so that all source vertices participating in the demand pairs in $\mset$ lie on $\Gamma_1$, and all destination vertices lie on $\Gamma_2$. Robertson and Seymour~\cite{RS-disc} showed an algorithm, that, given an instance of the \NDPdisc or the \NDPcyl problem, decides whether all demand pairs in $\mset$ can be routed simultaneously via node-disjoint paths, and if so, finds the routing efficiently. Moreover, for each of the two problems, they give an exact characterization of instances for which all pairs in $\mset$ can be routed in $G$. Several other very efficient algorithms are known for both problems~\cite{linear_two_face,planar_steiner_forest}. However, for our purposes, we need to consider the optimization version of both problems, where we are no longer guaranteed that all demand pairs in $\mset$ can be routed, and would like to route the largest possible subset of the demand pairs. We are not aware of any results for these two special cases of the \NDP problem. In this paper, we provide $O(\log k)$-approximation algorithms for both problems.

\paragraph*{Other Related Work.}
A problem closely related to \NDP is the Edge-Disjoint Paths (\EDP) problem. It is defined similarly, except that now the paths chosen to the solution are allowed to share vertices, and are only required to be edge-disjoint. 
It is easy to show, by using a line graph of the \EDP instance, that \NDP is more general than \EDP (though this transformation inflates the number of the graph vertices, so it may not preserve approximation factors that depend on $n$).  
This relationship breaks down in planar graphs, since the resulting \NDP instance may no longer be planar. 
The approximability status of \EDP is very similar to that of \NDP: there is an $O(\sqrt n)$-approximation algorithm~\cite{EDP-alg}, and it is known that there is no $O(\log^{1/2-\delta}n)$-approximation algorithm for any constant $\delta$, unless $\NP \subseteq \ZPTIME(n^{\poly \log n})$~\cite{AZ-undir-EDP,ACGKTZ}. We do not know whether our techniques can be used to obtain improved approximation algorithms for \EDP on planar graphs.
As in the \NDP problem, we can use the standard multicommodity flow LP-relaxation of the problem, in order to obtain an $O(\sqrt n)$-approximation algorithm, and the integrality gap of the LP-relaxation is $\Omega(\sqrt n)$ even in planar graphs. For several special cases of the problem better algorithms are known: Kleinberg~\cite{Kleinberg-planar}, building on the work of Chekuri, Khanna and Shepherd~\cite{CKS,CKS-planar1}, has shown an $O(\log^2n)$-approximation LP-rounding algorithm for even-degree planar graphs. Aumann and Rabani~\cite{grids1} showed an $O(\log^2 n)$-approximation algorithm for \EDP on grid graphs, and Kleinberg and Tardos~\cite{grids3,grids4} showed $O(\log n)$-approximation algorithms for broader classes of nearly-Eulerian uniformly high-diameter planar graphs, and nearly-Eulerian densely embedded graphs. Recently, Kawarabayashi and Kobayashi~\cite{KK-planar} gave an $O(\log n)$-approximation algorithm for \EDP when the input graph is either 4-edge-connected planar or Eulerian planar. It appears that the restriction of the graph $G$ to be Eulerian, or nearly-Eulerian, makes the \EDP problem significantly simpler, and in particular improves the integrality gap of the LP-relaxation.
The analogue of the grid graph for the \EDP problem is the wall graph: the integrality gap of the standard LP-relaxation for \EDP on wall graphs is $\Omega(\sqrt n)$, and until recently, no better than $O(\sqrt n)$-approximation algorithms for \EDP on walls were known. The work of~\cite{NDP-grids} gives an $\tilde O(n^{1/4})$-approximation algorithm for \EDP on wall graphs.

A variation of the \NDP and \EDP problems, where small congestion is allowed, has been a subject of extensive study. In the \NDP with congestion (NDPwC) problem, the input is the same as in the \NDP problem, and we are additionally given a non-negative integer $c$. The goal is to route as many of the demand pairs as possible with congestion at most $c$: that is, every vertex may participate in at most $c$ paths in the solution. \EDP with Congestion (EDPwC) is defined similarly, except that now the congestion bound is imposed on edges and not vertices. The classical randomized rounding technique of Raghavan and Thompson~\cite{RaghavanT} gives a constant-factor approximation for both problems, if the congestion $c$ is allowed to be as high as $\Theta(\log n/\log\log n)$. A recent line of work~\cite{CKS,Raecke,Andrews,RaoZhou,Chuzhoy11,  ChuzhoyL12,ChekuriE13,NDPwC2} has lead to an $O(\poly\log k)$-approximation for both \NDPwC and \EDPwC problems, with congestion $c=2$. For planar graphs, a constant-factor approximation with congestion 2 is known for \EDP~\cite{EDP-planar-c2}. All these algorithms perform LP-rounding of the standard multicommodity flow LP-relaxation of the problem and so it is unlikely that they can be extended to routing with no congestion.

\paragraph{Our Results and Techniques.}
Given an instance $(G,\mset)$ of the \NDP problem, we denote by $\opt(G,\mset)$ the value of the optimal solution to it.
Our first result is an approximation algorithm for \NDPdisc and \NDPcyl.

\begin{theorem}\label{thm: routing on a disc and cyl main}
There is an efficient $O(\log k)$-approximation algorithm for the \NDPdisc and the \NDPcyl problems, where $k$ is the number of the demand pairs in the instance.
\end{theorem}

We provide a brief high-level overview of the techniques we use in the proof of Theorem~\ref{thm: routing on a disc and cyl main}. We define a new intermediate problem, called Demand Pair Selection Problem (\DPSP). In this problem, we are given two disjoint directed paths $\sigma$ and $\sigma'$, and a set $\mset=\set{(s_1,t_1),\ldots,(s_k,t_k)}$ of pairs of vertices that we call demand pairs, such that all vertices in $\set{s_1,\ldots,s_k}$ lie on $\sigma$, and all vertices in $\set{t_1,\ldots,t_k}$ lie on $\sigma'$.
For any pair of vertices $v,v'\in V(\sigma)$, we denote $v\prec v'$ if $v$ lies before $v'$ on $\sigma$, and we denote $v\preceq v'$ if $v=v'$ or $v\prec v'$. For every pair $v,v'\in V(\sigma')$ of vertices, we define the relationships $v\prec v'$ and $v\preceq v'$ similarly.
 We say that two demand pairs $(s,t),(s',t')\in \mset$ \emph{cross}, if either (i) $\set{s,t}\cap \set{s',t'}\neq \emptyset$, or (ii) $s\prec s'$  and $t'\prec t$, or (iii) $s' \prec s$ and $t\prec t'$. We are also given a set $\kset$ of constraints that we describe below. The goal of the \DPSP problem is to find a maximum-cardinality subset $\mset'\subseteq \mset$ of demand pairs, such that no two pairs in $\mset'$ cross, and all constraints in $\kset$ are satisfied. There are four types of constraints in $\kset$, and each constraint is defined by a quadruple $(i,x,y,w)$, where $i\in \set{1,2,3,4}$ is the constraint type, $x,y\in V(\sigma)\cup V(\sigma')$ are vertices, and $w$ is an integer. In a type-1 constraint $(1,x,y,w)$, we have $x,y\in V(\sigma)$, and the constraint requires that the number of the demand pairs $(s,t)\in \mset'$ with $x\preceq s\preceq y$ is at most $w$. Similarly, in a type-2 constraint $(2,x,y,w)$, we have $x,y\in V(\sigma')$, and the constraint requires that the number of the demand pairs $(s,t)\in \mset'$ with $x\preceq t\preceq  y$ is at most $w$. If $(3,x,y,w)\in \kset$ is a type-3 constraint, then $x\in V(\sigma)$ and $y\in V(\sigma')$ must hold. The constraint requires that the number of the demand pairs $(s,t)\in \mset'$ with $s\preceq x$ and $y\preceq t$ is at most $w$. Similarly,  if $(4,x,y,w)\in \kset$ is a type-4 constraint, then $x\in V(\sigma)$ and $y\in V(\sigma')$ must hold, and the constraint requires that the number of the demand pairs $(s,t)\in \mset'$ with $x\preceq s$ and $t\preceq y$ is at most $w$.  We show that both \NDPdisc and \NDPcyl reduce
to \DPSP with an $O(\log k)$ loss in the approximation factor. The reduction from \NDPdisc to \DPSP uses the characterization of routable instances of \NDPdisc  due to Robertson and Seymour~\cite{RS-disc}. Finally, we show a factor-$8$ approximation algorithm for the \DPSP problem. 

The main result of our paper is summarized in the following two theorems.

\begin{theorem}\label{thm: main1}
There is an efficient $O(n^{9/19}\cdot \poly\log n)$-approximation algorithm for the \NDPplanar problem.
\end{theorem}

\begin{theorem}\label{thm: main2}
There is an efficient algorithm, that, given an  instance $(G,\mset)$ of \NDPplanar, computes a routing of $\Omega\left(\frac{(\opt(G,\mset))^{1/19}}{\poly\log n}\right )$ demand pairs of $\mset$ via node-disjoint paths in $G$.
\end{theorem}

Notice that when $\opt(G,\mset)$ is small, Theorem~\ref{thm: main2} gives a much better that $\tilde{O}(n^{9/19})$-approximation.

We now give a high-level intuitive overview of the proof of Theorem~\ref{thm: main1}. Given an instance $(G,\mset)$ of the \NDP problem, we denote by $\tset$ the set of vertices participating in the demand pairs in $\mset$, and we refer to them as \emph{terminals}. We start with a quick overview of the $\tilde O(n^{1/4})$-approximation algorithm of~\cite{NDP-grids} for the \NDPgrid problem, since their algorithm was the motivation for this work. 
The main observation of~\cite{NDP-grids} is that the instances of \NDPgrid, for which the multicommodity flow relaxation exhibits the $\Omega(\sqrt n)$ integrality gap, have terminals close to the grid boundary. When all terminals are at a distance of at least $\Omega(n^{1/4})$ from the boundary of the grid, one can find an $\tilde O(n^{1/4})$-approximation via LP-rounding (but unfortunately the integrality gap remains polynomial in $n$ even in this case).
When the terminals are close to the grid boundary, the integrality gap of the LP-relaxation becomes $\Omega(\sqrt n)$. However, this special case of \NDPgrid can be easily approximated via simple dynamic programming. \iffull For example, when all terminals lie on the grid boundary, the integrality gap of the LP-relaxation is $\Omega(\sqrt n)$, but a constant-factor approximation can be achieved via standard dynamic programming. More generally, when all terminals are within distance at most $O(n^{1/4})$ from the grid boundary, we can obtain an $O(n^{1/4})$-approximation via dynamic programming. \fi \ifabstract We omit the details here. \fi Overall, we partition the demand pairs of $\mset$ into two subsets, depending on whether the terminals lie close to or far from the grid boundary, and obtain an $\tilde{O}(n^{1/4})$-approximation for each of the two resulting problem instances separately, selecting the better of the two solutions as our output.

This idea is much more difficult to implement in general planar graphs. For one thing, the notion of the ``boundary'' of a planar graph is meaningless - any face in the drawing of the planar graph can be chosen as the outer face. We note that the standard multicommodity flow LP-relaxation performs poorly not only when all terminals are close to the boundary of a single face (a case somewhat similar to \NDPdisc), but also when there are two faces $F$ and $F'$, and for every demand pair $(s,t)\in \mset$, $s$ is close to the boundary of $F$ and $t$ is close to the boundary of $F'$ (this setting is somewhat similar to \NDPcyl). 
The notion of ``distance'', when deciding whether the terminals lie close to or far from a face boundary is also not well-defined, since we can subdivide edges and artificially modify the graph in various ways in order to manipulate the distances without significantly affecting routings. Intuitively, we would like to define the distances between the terminals in such a way that, on the one hand, whenever we find a set $\mset'\subseteq \mset$ of demand pairs, such that all terminals participating in the pairs in $\mset'$ are far enough from each other, then we can route a large subset of the demand pairs in $\mset'$. On the other hand, if we find a set $\mset''\subseteq \mset$ of demand pairs, with all terminals participating in the pairs in $\mset''$ being close to the boundary of some face (or a pair of faces), then we can find a good approximate solution to instance $(G,\mset'')$ (for example, by reducing the problem to \NDPdisc or \NDPcyl). Since we do not know beforehand which face (or faces) will be chosen as the ``boundary'' of the graph, we cannot partition the problem into two sub-problems and employ different techniques to solve each sub-problem as we did for \NDPgrid. Instead, we need a single framework in which both cases can be handled.

We assume that every terminal participates in exactly one demand pair, and that the degree of every terminal is $1$. This can be achieved via a standard transformation, where we create several copies of each terminal, and connect them to the original terminal. This transformation may introduce up to $O(n^2)$ new vertices. Since we are interested in obtaining an $\tilde O(n^{9/19})$-approximation for \NDPplanar, we denote by $N$ the number of the non-terminal vertices in the new graph $G$.  Abusing the notation, we denote the total number of vertices in the new problem instance by $n$.
It is now enough to obtain an  $\tilde O(N^{9/19})$-approximation for the new problem instance.

Our first step is to define a new LP-relaxation of the problem. We assume that we have guessed correctly the value $\opt$ of the optimal solution. We start with the standard multicommodity flow LP-relaxation, where we try to send $\opt$ flow units between the demand pairs, so that the maximum amount of flow through any vertex is bounded by $1$. We then add the following new set of constraints to the LP: for every subset $\mset'\subseteq \mset$ of the demand pairs, for every value $\opt(G,\mset')\leq z\leq k$, the total amount of flow routed between the demand pairs in $\mset'$ is no more than $z$. Adding this type of constraints may seem counter-intuitive. We effectively require that the LP solves the problem exactly, and naturally we cannot expect to be able to do so efficiently. Since the number of the resulting constraints is exponential in $k$, and since we do not know the values $\opt(G,\mset')$, we indeed cannot solve this LP efficiently. In fact, our algorithm does not attempt to solve the LP exactly. Instead, we employ the Ellipsoid algorithm, that in every iteration produces a potential solution to the LP-relaxation. We then show an algorithm that, given such a potential solution, either finds an integral solution routing $\tilde \Omega(\opt/N^{9/19})$ demand pairs, or it returns some subset $\mset'\subseteq \mset$ of demand pairs, whose corresponding LP-constraint is violated. Therefore, we use our approximation algorithm as the separation oracle for the Ellipsoid algorithm. We are then guaranteed that after $\poly(n)$ iterations, we will obtain a solution routing the desired number of demand pairs, as only $\poly(n)$ iterations are required for the Ellipsoid algorithm in order to find a feasible LP-solution.

The heart of the proof of Theorem~\ref{thm: main1} is then an algorithm that, given a potential (possibly infeasible) solution to the LP-relaxation, either finds an integral solution routing $\tilde \Omega(\opt/N^{9/19})$ demand pairs, or returns some subset $\mset'\subseteq \mset$ of demand pairs, whose corresponding LP-constraint is violated. We can assume without loss of generality that the fractional solution we are given satisfies all the standard multicommodity flow constraints, as this can be verified efficiently. For simplicity of exposition, we assume that every demand pair in $\mset$ sends the same amount of $w^*$ flow units to each other. 

 We assume for now that the set $\tset$ of terminals is $\alphawl$-well-linked, for $\alphawl=\Theta(w^*/\log n)$ - that is, for every pair $(\tset',\tset'')$ of disjoint equal-sized subsets of vertices of $\tset$, we can connect vertices of $\tset'$ to vertices of $\tset''$ by at least $\alphawl\cdot |\tset'|$ node-disjoint paths. We discuss this assumption in more detail below. 
We assume that we are given a drawing of $G$ on the sphere.
Our first step is to define the notion of distances between the terminals. In order to do so, we first construct \emph{enclosures} around them. Throughout the proof, we use a parameter $\Delta=\opt^{2/19}$. We say that a curve $\gamma$ on the sphere is a \emph{$G$-normal curve} iff it intersects the drawing of $G$ only at its vertices. The length of such a curve is the number of vertices of $G$ it contains. An enclosure around a terminal $t$ is a disc $D_t$ containing $t$, whose boundary, that we denote by $C_t$, is a $G$-normal curve of length exactly $\Delta$, so that at most $O(\Delta/\alphawl)$ terminals lie in $D_t$. 
We show an efficient algorithm to construct the enclosures $D_t$ around the terminals $t$, so that the following additional conditions hold: (i) if $D_t\subseteq D_{t'}$ for any pair $t,t'\in \tset$ of terminals, then $D_t=D_{t'}$; and (ii) if $D_t\cap D_{t'}= \emptyset$, then there are $\Delta$ node-disjoint paths connecting the vertices of $C_t$ to the vertices of $C_{t'}$. 
We now define the distances between pairs of terminals. For every pair $(t,t')\in \tset$, distance $d(t,t')$ is the length of the shortest $G$-normal curve, connecting a vertex of $C_t$ to a vertex of $C_{t'}$. 

Next, we show that one of the following has to happen: either there is a large collection $\tmset\subseteq\mset$ of demand pairs, such that all terminals participating in the pairs in $\tmset$ are at a distance at least $\Omega(\Delta)$ from each other; or there is a large collection $\tmset'\subseteq \mset$ of demand pairs, and two faces $F,F'$ in the drawing of $G$ (with possibly $F=F'$), such that for every demand pair in $\tmset'$, one of its terminals is within distance at most $\tilde O(\Delta)$ from the boundary of $F$, and the other is within distance at most $\tilde O(\Delta)$ from the boundary of $F'$. In the former case, we show that we can route a large subset of the demand pairs in $\tmset$ via node-disjoint paths, by constructing a special routing structure called a crossbar (this construction exploits well-linkedness of the terminals and the paths connecting the enclosures). In the latter case, we reduce the problem to \NDPdisc or \NDPcyl, depending on the distance between the faces $F$ and $F'$, and employ the approximation algorithms for these problems to route $\tilde\Omega\left (\frac{\opt(G,\tmset')}{\poly(\Delta)}\right )$ demand pairs from $\tmset'$ in $G$. If the resulting number of demand pairs routed is close enough to $\opt$, then we return this as our final solution. Otherwise, we show that the LP-constraint corresponding to the set $\tmset'$ of demand pairs is violated, or equivalently, the amount of flow sent by the LP solution between the demand pairs in $\tmset'$ is greater than $\opt(G,\tmset')$.

So far we have assumed that the terminals participating in the demand pairs in $\mset$ are $\alphawl$-well-linked. In general this may not be the case. Using standard techniques, we can perform a well-linked decomposition: that is, compute a subset $U\subseteq V(G)$ of at most $\opt/64$ vertices, such that, if we denote the set of all connected components of $G\setminus U$ by  $\set{G_1,\ldots,G_r}$, and for each $1\leq i\leq r$, we denote by $\mset_i\subseteq G_i$ the set of the demand pairs contained in $G_i$, then the terminals participating in the demand pairs in $\mset_i$ are $\alphawl$-well-linked  in $G_i$. We are then guaranteed that $\sum_{i=1}^r\opt(G_i,\mset_i)\geq \frac{63}{64}\opt$. It is then tempting to apply the algorithm described above to each of the graphs separately. Indeed, if, for each $1\leq i\leq r$, we find a set $\pset_i$ of node-disjoint paths, routing $\Omega\left (\frac{\opt(G_i,\mset_i)}{N_i^{9/19}\cdot \poly\log n}\right)$ demand pairs of $\mset_i$ in $G_i$ (where $N_i$ denotes the number of the non-terminal vertices in $G_i$), then we obtain an $O(N^{9/19}\cdot\poly\log n)$-approximate solution overall. Assume now that for some $1\leq i\leq r$, we find a subset $\mset'_i\subseteq \mset_i$ of demand pairs, such that $\opt(G_i,\mset'_i)<w^*|\mset'_i|/16$. Unfortunately, the set $\mset'_i$ of demand pairs does not necessarily define a violated LP-constraint, since it is possible that $\opt(G,\mset'_i)>>\opt(G_i,\mset'_i)$, if the optimal routing uses many vertices of $U$ (and possibly from some other graphs $G_j$). In general, the number of vertices in set $U$ is relatively small compared to $\opt$, so in the global accounting across all instances $(G_{i'},\mset_{i'})$, only a small number of paths can use the vertices of $U$. But for any specific instance $(G_i,\mset_i)$, it is possible that most paths in the optimal solution to instance $(G,\mset_i)$ use the vertices of $U$. In order to overcome this difficulty, we need to perform a careful global accounting across all resulting instances $(G_i,\mset_i)$. 

\paragraph{Organization}
We start with preliminaries in Section~\ref{sec: prelims}. Section~\ref{sec: disc and cylinder} is devoted to the proof of Theorem~\ref{thm: routing on a disc and cyl main}. Since this is not our main result, and the proof is somewhat long (though not very difficult), most of the proof appears in Section~\ref{sec: proofs for disc and cylinder} of the Appendix. Sections~\ref{sec: alg overview}--\ref{sec: case 2}  are devoted to the proof of Theorem~\ref{thm: main1}: Section~\ref{sec: alg overview} provides an overview of the algorithm and some initial steps; Section~\ref{sec: enclosures, shells, subsets} introduces the main technical tools that we use: enclosures, shells, and a partition of the terminals into subsets;  and Sections~\ref{sec: case 1} and~\ref{sec: case 2} deal with Case 1 (when many terminals are far from each other) and Case 2 (when many terminals are close to the boundaries of at most two faces), respectively. We prove Theorem~\ref{thm: main2} in Section~\ref{sec: proof of 2nd main thm}, and provide conclusions in Section~\ref{sec: conclusion}. For convenience, we include in Section~\ref{sec: appendix-params-table} of the Appendix a table of the main parameters used in the proof of Theorem~\ref{thm: main1}.

\label{----------------------------------------sec: prelims-------------------------------------}
\section{Preliminaries}\label{sec: prelims}

Given a graph $G$ and a subset $U$ of its vertices, we denote by $N(U)$ the set of all neighbors of $U$, that is, all vertices $v\in V(G)\setminus U$, such that there is an edge $(u,v)\in E(G)$ for some $u\in U$. We say that two paths $P$ and $P'$ are \emph{internally disjoint} iff  for every vertex $v\in P\cap P'$, $v$ is an endpoint of both $P$ and $P'$. Given a path $P$ and a subset $U$ of vertices of $G$, we say that $P$ is internally disjoint from $U$ iff every vertex in $P\cap U$ is an endpoint of $P$. Similarly, $P$ is internally disjoint from a subgraph $G'$ of $G$ iff $P$ is internally disjoint from $V(G')$.
Given a graph $G$ and a set $\mset$ of demand pairs in $G$, for every subset $\mset'\subseteq\mset$ of the demand pairs, we denote by $\tset(\mset')$ the set of all vertices participating in the demand pairs in $\mset'$. For a subset $\mset'\subseteq\mset$ of the demand pairs, and a sub-graph $H\subseteq G$, let $\opt(H,\mset')$ denote the value of the optimal solution to instance $(H,\mset')$.

Given a drawing of any planar graph $H$ in the plane, and given any cycle $C$ in $H$, we denote by $D(C)$ the unique disc in the plane whose boundary is $C$. Similarly, if $C$ is a closed simple curve in the plane, $D(C)$ is the unique disc whose boundary is $C$.  When the graph $H$ is drawn on the sphere, there are two discs whose boundaries are $C$. In such cases we will explicitly specify which of the two discs we refer to. Given any disc $D$ (in the plane or on the sphere), we use $\dnot$ to denote the disc $D$ without its boundary. We say that a vertex of $H$ belongs to disc $D$, and denote $v\in D$, if $v$ is drawn inside $D$ or on its boundary. 

Given a planar graph $G$, drawn on a surface $\Sigma$, we say that a curve $C$ in $\Sigma$ is \emph{$G$-normal}, iff it intersects the drawing of $G$ at vertices only.  The set of vertices of $G$ lying on $C$ is denoted by $V(C)$, and the length of $C$ is $\ell(C)=|V(C)|$. 
For any disc $D$, whose boundary is a $G$-normal curve, we denote by $V(D)$ the set of all vertices of $G$ lying inside $D$ or on its boundary.


\begin{definition}
Let $\gamma,\gamma'$ be two curves in the plane or on the sphere. We say that $\gamma$ and $\gamma'$ \emph{cross},  iff there is a disc $D$, whose boundary is a simple closed curve that we denote by $\beta$, such that: 

\begin{itemize}
\item  $\gamma\cap D$ is a simple open curve, whose endpoints we denote by $a$ and $b$; 
\item $\gamma'\cap D$ is a simple open curve, whose endpoints we denote by $a'$ and $b'$; and

\item $a,a',b,b'\in \beta$, and they appear on $\beta$ in this circular order.
\end{itemize}
\end{definition}

Given a graph $G$ embedded in the plane or on the sphere, we say that two paths $P,P'$ in $G$ cross iff their images cross. Similarly, we say that a path $P$ crosses a curve $\gamma$ iff the image of $P$ crosses $\gamma$.

\paragraph{Sparsest Cut.}
In this paper we use the node version of the sparsest cut problem, defined as follows. Suppose we are given a graph $G=(V,E)$ with a subset $\tset\subseteq V$ of its vertices called terminals. A vertex cut is a tri-partition $(A,C,B)$ of $V$, such that there are no edges in $G$ with one endpoint in $A$ and another in $B$. If $(A\cup C)\cap \tset,(B\cup C)\cap \tset\neq \emptyset$, then the sparsity of the cut $(A,C,B)$ is $\frac{|C|}{\min\set{|A\cap \tset|,|B\cap \tset|}+|C\cap \tset|}$. The sparsest cut in $G$ with respect to the set $\tset$ of terminals is a vertex cut $(A,C,B)$ with $(A\cup C)\cap \tset,(B\cup C)\cap \tset\neq \emptyset$, whose sparsity is the smallest among all such cuts. Amir, Krauthgamer and Rao~\cite{vertex-sparsest-planar} showed an efficient algorithm, that, given any planar graph $G$ with a set $\tset\subseteq V(G)$ of terminal vertices, computes a vertex cut $(A,C,B)$ in $G$, whose sparsity with respect to $\tset$ is within a constant factor of the optimal one. We denote this algorithm by $\algsc$, and the approximation factor it achieves by $\alphasc$, so $\alphasc$ is a universal constant.

In the special case of the sparsest cut problem that we consider in our paper, all terminals have degree $1$, and no edge of $G$ connects any pair of terminals. We show that in this case we can compute a near-optimal solution $(A,C,B)$  to the sparsest cut problem with $C\cap \tset=\emptyset$.
The proof of the following observation uses standard techniques and is deferred to the Appendix.

\begin{observation}\label{obs: sparsest cut}
Let $G$ be a planar graph and let $\tset\subseteq V(G)$ be a subset of its vertices called terminals, with $|\tset|\geq 3$. Assume that the degree of every terminal is $1$, and no edge of $G$ connects any pair of terminals. Then there is an efficient algorithm to compute a vertex cut $(A,C,B)$ in $G$, whose sparsity is within a factor $\alphasc$ of the optimal one, and $C\cap \tset=\emptyset$.
\end{observation}

\paragraph{Nested Segments.}
Suppose we are given a graph $G$, a cycle $C$ in $G$, and a collection $\Sigma$ of (not necessarily disjoint) segments of $C$, where each segment is either $C$ itself, or a sub-path of $C$. We say that $\Sigma$ is a \emph{nested set of segments of $C$} iff for all $\sigma,\sigma'\in \Sigma$, either $\sigma\subseteq \sigma'$, or $\sigma'\subseteq\sigma$, or $\sigma$ and $\sigma'$ are internally disjoint - that is, every vertex in $\sigma\cap \sigma'$ is an endpoint of both segments. We define a set of nested segments of a closed curve $C$, and of a path $P$ in $G$ similarly.

\paragraph{Decomposition of Forests.} A directed forest $F$ is a disjoint union of arborescences $\tau_1,\ldots,\tau_r$ for some $r\geq 1$, where in each arborescence $\tau_i$, all edges are directed towards the root. 
We use the following simple claim about partitioning directed forests into collections of paths.
 Similar decompositions were used in previous work, see e.g. Lemma 3.5 in~\cite{Kleinberg-planar}. The proof is included in Appendix for completeness.

\begin{claim}\label{claim: partition the forest} There is an efficient algorithm, that, given a directed forest $F$ with $n$ vertices, computes a partition $\rset=\set{R_1,\ldots,R_{\ceil{\log n}}}$ of $V(F)$ into subsets, such that for each $1\leq j\leq \ceil{\log n}$, $F[R_j]$ is a collection of disjoint directed paths, that we denote by $\pset_j$. Moreover, for all $v,v'\in R_j$, if there is a directed path from $v$ to $v'$ in $F$, then they both lie on the same path in $\pset_j$.
\end{claim}

\paragraph{Routing on a Disc.}
Assume that we are given an instance $(G,\mset)$ of the \NDPdisc problem, where $G$ is drawn in a disc $D$ whose boundary is denoted by $C$. We need the following two definitions.

\begin{definition}
We say that two demand pairs $(s,t),(s',t')\in \mset$ cross iff either $\set{s,t}\cap\set{s',t'}\neq\emptyset$, or $(s,s',t,t')$ appear on $C$ in this circular order. We say that the set $\mset$ of demand pairs is non-crossing if no two demand pairs in $\mset$ cross.
\end{definition}

\begin{definition}
Let $C$ be a closed simple curve and $\mset$ a set of demand pairs with all vertices of $\tset(\mset)$ lying on $C$.
We say that $\mset$ is an  \emph{$r$-split collection of demand pairs} with respect to $C$, iff there is a partition $\mset_1,\ldots,\mset_r$ of the demand pairs in $\mset$, and there is a partition $\set{\sigma_1,\sigma_2,\ldots,\sigma_{2r}}$ of $C$ into disjoint segments, such that $\sigma_1,\ldots,\sigma_{2r}$ appear on $C$ in this circular order, and for each $1\leq i\leq r$, for every demand pair $(s,t)\in \mset_i$, either $s\in \sigma_{2i-1}$ and $t\in \sigma_{2i}$, or vice versa.
\end{definition}

Finally, the following lemma allows us to partition any set of demand pairs into a small collection of split sets. The proof appears in the Appendix.

\begin{lemma}\label{lem: getting r-split demand pairs on a disc}
There is an efficient algorithm, that, given a closed simple curve $C$ in the plane and a set  $\mset$ of $\kappa$ demand pairs, whose corresponding terminals lie on $C$, computes a partition $\mset^1,\ldots,\mset^{4\ceil{\log \kappa}}$ of $\mset$, such that for each $1\leq i\leq 4\ceil{\log \kappa}$, set $\mset^i$ is $r_i$-split with respect to $C$ for some integer $r_i\geq 0$.
\end{lemma}

\paragraph{Routing on a Cylinder.}
Assume that we are given an instance $(G,\mset)$ of the \NDPcyl problem, where $\Gamma_1$ and $\Gamma_2$ are the cuffs of the cylinder. 

\begin{definition}
We say that a set $\mset'\subseteq \mset$ of demand pairs is \emph{non-crossing} if there is an ordering $(s_{i_1},t_{i_1}),\ldots,(s_{i_r},t_{i_r})$ of the demand pairs in $\mset'$, such that $s_{i_1},s_{i_2},\ldots,s_{i_r}$ are all distinct and appear in this counter-clock-wise order on $\Gamma_1$, and $t_{i_1},t_{i_2},\ldots,t_{i_r}$ are all distinct and appear in this counter-clock-wise order on $\Gamma_2$.
\end{definition}

It is immediate to verify that if we are given any instance $(G,\mset)$ of \NDPcyl, and any set $\mset'\subseteq \mset$ of demand pairs that can all be routed via node-disjoint paths in $G$, then set $\mset'$ is non-crossing.

\paragraph{Tight Concentric Cycles.}


 We start with the following definition.

\begin{definition}
Given a planar graph $H$ drawn in the plane and a vertex $v \in V(H)$ that is not incident to the infinite face, $\mincycle(H, v)$ is the cycle $C$ in $H$, such that: (i) $v\in\dnot(C)$; and (ii) among all cycles satisfying (i), $C$ is the one for which $D(C)$ is minimal inclusion-wise. 
\end{definition}

It is easy to see that $\mincycle(H,v)$ is uniquely defined. Indeed, consider the graph $H \setminus v$, and the face $F$ in the drawing of $H \setminus v$ where $v$ used to reside. Then the boundary of $F$ contains exactly one cycle $C$ with $D(C)$ containing $v$, and $C=\mincycle(H, v)$.
We next define a family of tight concentric cycles.

\begin{definition}
Suppose we are given a planar graph $H$, an embedding of $H$ in the plane, a simple closed $H$-normal curve $C$, and an integral parameter $r\geq 1$. A family of $r$ tight concentric cycles around $C$ is a sequence $Z_1,Z_2,\ldots,Z_r$ of disjoint simple cycles in $H$, with the following properties:

\begin{itemize}
\item $D(C)\subsetneq D(Z_1)\subsetneq D(Z_2)\subsetneq\cdots\subsetneq D(Z_r)$;   \label{prop: disc}

\item if $H'$ is the graph obtained from $H$ by contracting all vertices lying in $D(C)$ into a super-node $a$, then $Z_1=\mincycle(H',a)$; and

\item for every $1< h\leq r$, if $H'$ is the graph obtained from $H$ by contracting all vertices lying in $D(Z_{h-1})$ into a super-node $a$, then $Z_h = \mincycle(H', a)$.
\end{itemize}
\end{definition}

We will sometimes allow $C$ to be a simple cycle in $H$. The family of tight concentric cycles around $C$ is then defined similarly.

\paragraph{Monotonicity of Paths and Cycles.}

Suppose we are given a planar graph $H$, embedded into the plane,  a simple $H$-normal curve $C$ in $H$, and a family $(Z_1,\ldots,Z_r)$ of tight concentric cycles around $C$. Assume further that we are given a set $\pset$ of $\kappa$ node-disjoint paths, originating at the vertices of $C$, and terminating at some vertices lying outside of $D(Z_r)$. We would like to re-route these paths to ensure that they are monotone with respect to the cycles, that is, for all $1\leq h\leq r$, and for all $P\in \pset$, $P\cap Z_h$ is a path. We first discuss re-routing to ensure monotonicity with respect to a single cycle, and then extend it to monotonicity with respect to a family of concentric cycles.

\begin{definition}
Given a graph $H$, a cycle $C$ and a path $P$ in $H$, we say that $P$ is monotone with respect to $C$, iff $P \cap C$ is a path. 
\end{definition}

The proof of the following lemma is deferred to the Appendix.

\begin{lemma}\label{lem:reroute-montone}
Let $H$ be a planar graph embedded into the plane, $C$ a simple cycle in $H$, and $\pset$ a collection of $\kappa$ simple internally node-disjoint paths between two vertices: vertex $s$ lying in $\dnot(C)$, and vertex $t\not \in D(C)$, that is incident on the outer face. Assume further that $H$ is the union of $C$ and the paths in $\pset$, and that $C=\mincycle(H,s)$. Then there is an efficient algorithm to compute a set $\pset'$ of $\kappa$ internally node-disjoint paths connecting $s$ to $t$ in $H$, such that every path in $\pset'$ is monotone with respect to $C$. 
\end{lemma}

We now define monotonicity with respect to a family of cycles.

\begin{definition}
Let $H$ be a graph, $\zset=(Z_1,\ldots,Z_r)$ a collection of $r$ disjoint cycles, and $\pset$ a collection of node-disjoint paths in $H$. We say that the paths in $\pset$ are \emph{monotone} with respect to $\zset$, iff for every $1\leq h\leq r$, every path in $\pset$ is monotone with respect to $Z_h$.
\end{definition}

The following theorem allows us to re-route sets of paths so they become monotone with respect to a given family of tight concentric cycles.
Its proof is a simple application of Lemma~\ref{lem:reroute-montone} and is deferred to the Appendix.

\begin{theorem}\label{thm: monotonicity for shells}
Let $H$ be a planar graph embedded in the plane, $C$ any simple closed $H$-normal curve or a simple cycle in $H$, and $\zset=(Z_1,\ldots,Z_r)$ a family of $r$ tight concentric cycles in $H$ around $C$. Let $Y\subsetneq H$ be any connected subgraph of $H$ lying completely outside of $D(Z_r)$, and let $\pset$ be a set of $\kappa$ node-disjoint paths, connecting a subset $A\subseteq V(C)$ of $\kappa$ vertices to a subset $B\subseteq V(Y)$ of $\kappa$ vertices, so that the paths of $\pset$ are internally disjoint from $V(C)\cup V(Y)$. Let $H'=\left (\bigcup_{h=1}^rV(Z_h)\right )\cup \pset$. Then there is an efficient algorithm to compute a collection $\pset'$ of $\kappa$ node-disjoint paths in $H'$, connecting the vertices of $A$ to the vertices of $B$, so that the paths in $\pset'$ are monotone with respect to $\zset$, and they are internally node-disjoint from $V(C)\cup V(Y)$.
\end{theorem}


\label{------------------------------------------sec: disc and cylinder-----------------------------------------}
\section{Routing on a Disc and on a Cylinder}\label{sec: disc and cylinder}
In this section we prove Theorem~\ref{thm: routing on a disc and cyl main}. 
 In order to do so, we define a new problem, called Demand Pair Selection Problem (\DPSP), and show an $8$-approximation algorithm for it. We then show that both \NDPdisc and \NDPcyl reduce to \DPSP.

\paragraph{Demand Pair Selection Problem}
We assume that we are given two disjoint directed paths, $\sigma$ and $\sigma'$, and a collection $\mset = \set{(s_1, t_1), \ldots, (s_k, t_k)}$ of pairs of vertices of $\sigma \cup \sigma'$ that are called demand pairs, where all vertices of $S = \set{s_1, \ldots, s_k}$ lie on $\sigma$, and all vertices of $T = \set{t_1, \ldots, t_k}$ lie on $\sigma'$ (not necessarily in this order). We refer to the vertices of $S$ and $T$ as the \emph{source} and the \emph{destination} vertices, respectively. Note that the same vertex of $\sigma$ may participate in several demand pairs, and the same is true for the vertices of $\sigma'$. Given any pair $a, a'$ of vertices of $\sigma$, with $a$ lying before $a'$ on $\sigma$, we sometimes denote by $(a, a')$ the sub-path of $\sigma$ between $a$ and $a'$ (that includes both these vertices), and we will sometimes refer to it as an interval. We define intervals of $\sigma'$ similarly. 

For every pair $v, v' \in V(\sigma)$ of vertices, we denote $v \prec v'$ if $v$ lies strictly before $v'$ on $\sigma$, and we denote $v \preceq v'$, if $v \prec v'$ or $v = v'$ hold. Similarly, for every pair $v, v' \in V(\sigma')$ of vertices, we denote $v \prec v'$ if $v$ lies strictly before $v'$ on $\sigma'$, and we denote $v \preceq v'$, if $v \prec v'$ or $v = v'$ hold. We need the following definitions.

\begin{definition} Suppose we are given two pairs $(a,b)$ and $(a',b')$ of vertices of $\sigma \cup \sigma'$, with $a, a' \in \sigma$ and $b, b' \in \sigma'$. We say that $(a,b)$ and $(a',b')$ \emph{cross} iff one of the following holds: either (i) $a = a'$; or (ii) $b = b'$; or (iii) $a \prec a'$ and $b' \prec b$; or (iv) $a' \prec a$ and $b \prec b'$.
\end{definition}

\begin{definition} We say that a subset $\mset'\subseteq\mset$ of demand pairs is \emph{non-crossing} iff for all distinct pairs $(s,t), (s',t') \in \mset'$, $(s,t)$ and $(s',t')$ do not cross.
\end{definition}

Our goal is to select the largest-cardinality non-crossing subset $\mset'\subseteq\mset$ of demand pairs, satisfying a collection $\kset$ of constraints. Set $\kset$ of constraints is given as part of the problem input, and consists of four subsets, $\kset_1, \ldots, \kset_4$, where constraints in set $\kset_i$ are called \emph{type-i constraints}. Every constraint $K \in \kset$ is specified by a quadruple $(i,a,b,w)$, where $i \in \set{1,2,3,4}$ is the constraint type, $a,b\in V(\sigma \cup \sigma')$, and $1\le w\le |\mset|$ is an integer.

For every type-1 constraint $K = (1,a,b,w) \in \kset_1$, we have $a,b \in V(\sigma)$ with $a \prec b$. The constraint is associated with the sub-path $I = (a,b)$ of $\sigma$. We say that a subset $\mset' \subseteq \mset$ of demand pairs \emph{satisfies} $K$ iff the total number of the source vertices participating in the demand pairs of $\mset'$ that lie on $I$ is at most $w$.

Similarly, for every type-2 constraint $K = (2,a,b,w) \in \kset_2$, we have $a,b \in V(\sigma')$ with $a \prec b$, and the constraint is associated with the sub-path $I = (a,b)$ of $\sigma'$. A set $\mset' \subseteq \mset$ of demand pairs satisfies $K$ iff the total number of the destination vertices participating in the demand pairs in $\mset'$ that lie on $I$ is at most $w$.

For each type-3 constraint $K = (3,a,b,w) \in \kset_3$, we have $a \in V(\sigma)$ and $b \in V(\sigma')$. The constraint is associated with the sub-path $L_a$ of $\sigma$ between the first vertex of $\sigma$ and $a$ (including both these vertices), and the sub-path $R_b$ of $\sigma'$ between $b$ and the last vertex of $\sigma'$ (including both these vertices). We say that a demand pair $(s,t) \in \mset$ \emph{crosses} $K$ iff $s \in L_a$ and $t \in R_b$. A set $\mset' \subseteq \mset$ of demand pairs satisfies $K$ iff the total number of pairs $(s,t) \in \mset'$ that cross $K$ is bounded by $w$.

Finally, for each type-4 constraint $K = (4,a,b,w) \in \kset_4$, we also have $a \in V(\sigma)$ and $b \in V(\sigma')$. The constraint is associated with the sub-path $R_a$ of $\sigma$ between $a$ and the last vertex of $\sigma$ (including both these vertices), and the sub-path $L_b$ of $\sigma'$ between the first vertex of $\sigma'$ and $b$ (including both these vertices). We say that a demand pair $(s,t) \in \mset$ crosses $K$ iff $s \in R_a$ and $t \in L_b$. A set $\mset' \subseteq \mset$ of demand pairs satisfies $K$ iff the total number of pairs $(s,t) \in \mset'$ that cross $K$ is bounded by $w$.

Given the paths $\sigma, \sigma'$, the set $\mset$ of the demand pairs, and the set $\kset$ of constraints as above, the goal in the $\DPSP$ problem is to select a maximum-cardinality non-crossing subset $\mset' \subseteq \mset$ of demand pairs, such that all constraints in $\kset$ are satisfied by $\mset'$.  The proof of the following theorem is deferred to the Appendix. 

\begin{theorem}\label{thm: approximate DPSP}
There is an efficient $8$-approximation algorithm for \DPSP.
\end{theorem}

We then use Theorem~\ref{thm: approximate DPSP} in order to design $O(\log k)$-approximation algorithms for \NDPdisc and \NDPcyl. The remainder of the proof of Theorem~\ref{thm: routing on a disc and cyl main} appears in Section~\ref{sec: proofs for disc and cylinder} of the Appendix.
The algorithm for \NDPdisc exploits the exact characterization of routable instances of the problem given by Robertson and Seymour~\cite{RS-disc}, in order to reduce the problem to \DPSP. The algorithm for \NDPcyl reduces the problem to \NDPdisc and \DPSP.

\label{----------------------------------------sec: alg overview-------------------------------------}
\section{Algorithm Setup}\label{sec: alg overview}
The rest of this paper mostly focuses on proving Theorem~\ref{thm: main1}; we prove Theorem~\ref{thm: main2} using the techniques we employ for the proof of Theorem~\ref{thm: main1} in Section~\ref{sec: proof of 2nd main thm}.

We assume without loss of generality that the input graph $G$ is connected - otherwise we can solve the problem separately on each connected component of $G$. 
Let $\tset=\tset(\mset)$. 
It is convenient for us to assume that every terminal participates in exactly one demand pair, and that the degree of every terminal is $1$. This can be achieved via a standard transformation of the input instance, where we add a new collection of terminals, connecting them to the original terminals. This transformation preserves planarity, but unfortunately it can increase the number of the graph vertices.
If the original graph $G$ contained $n$ vertices, then $|\mset|$ can be as large as $n^2$, and so the new graph may contain up to $n^2+n$ vertices, while our goal is to obtain an $\tilde O(n^{9/19})$-approximation. In order to overcome this difficulty, we denote by $N$ the number of the non-terminal vertices in the new graph $G$, so $N$ is bounded by the total number of vertices in the original graph, and by $n$ the total number of all vertices in the new graph, so $n=O(N^2)$. Our goal is then to obtain an efficient $O(N^{9/19}\cdot\poly\log n)$-approximation for the new problem instance. From now on we assume that 
every terminal participates in exactly one demand pair, and the degree of every terminal is $1$.  If $(s,t)$ is a demand pair, then we say that $s$ is the mate of $t$, and $t$ is the mate of $s$.
We denote $|\mset|=k$. Throughout the algorithm, we define a number of sub-instances of the instance $(G,\mset)$, but we always use $k$ to denote the number of the demand pairs in this initial instance. We can assume that $k>100$, as otherwise we can return a routing of a single demand pair.

We assume that we are given a drawing of $G$ on the sphere. Throughout the algorithm, we will sometimes select some face of $G$ as the outer face, and consider the resulting planar drawing of $G$. 
\subsection{LP-Relaxations}

Let us start with the standard multicommodity flow LP-relaxation of the problem. Let $G'$ be the directed graph, obtained from $G$ by bi-directing its edges. For every edge $e\in E(G')$, for each $1\leq i\leq k$, there is an LP-variable $f_i(e)$, whose value is the amount of the commodity-$i$ flow through edge $e$. We denote by $x_i$ the total amount of commodity-$i$ flow sent from $s_i$ to $t_i$. For every vertex $v$, let $\delta^+(v)$ and $\delta^-(v)$ denote the sets of its out-going and in-coming edges, respectively. We denote $[k]=\set{1,\ldots,k}$. The standard LP-relaxation of the \NDP problem is as follows.

\begin{eqnarray*}
\mbox{(LP-flow1)}&&\\
\max&\sum_{i=1}^kx_i&\\
\mbox{s.t.}&&\\
&\sum_{e\in \delta^+(s_i)}f_i(e)=x_i&\forall i\in[k]\\
&\sum_{e\in \delta^+(v)}f_i(e)=\sum_{e\in \delta^-(v)}f_i(e)&\forall i\in[k], \forall v\in V(G')\setminus\set{s_i,t_i}\quad\quad \mbox{(flow conservation)}\\
&\sum_{e\in \delta^+(v)}\sum_{i=1}^kf_i(e)\leq 1&\forall v\in V(G')\quad\quad\quad\quad\quad\quad\mbox{(vertex capacity constraints)}\\
&f_i(e)\geq 0&\forall i\in[k],\forall e\in E(G')
\end{eqnarray*}

We will make two changes to (LP-flow1). First, we will assume that we know the value $X^*$ of the optimal solution, and instead of the objective function, we will add the constraint $\sum_{i=1}^kx_i\geq X^*$. We can do so using standard methods, by repeatedly guessing the value $X^*$ and running the algorithm for each such value. It is enough to show that the algorithm routes $\Omega\left (\frac{X^*}{N^{9/19}\cdot\poly\log n}\right)$ demand pairs, when the value $X^*$ is guessed correctly.

Recall that for a subset $\mset'\subseteq\mset$ of the demand pairs, and a sub-graph $H\subseteq G$, $\opt(H,\mset')$ denotes the value of the optimal solution to instance $(H,\mset')$. 
For every subset $\mset'\subseteq \mset$ of the demand pairs, we add the constraint that the total flow between all pairs in $\mset'$ is no more than $z$, for all integers $z$ between $\opt(G,\mset')$ and $k$. We now obtain the following linear program that has no objective function,  so we are only interested in finding a feasible solution.

\label{-------------------------------------LP constraints-----------------------------------------------}
\begin{eqnarray}
\mbox{(LP-flow2)}&\nonumber\\
&\sum_{i=1}^kx_i\geq X^*& \label{LP: opt-value}\\ 
&\sum_{e\in \delta^+(s_i)}f_i(e)=x_i&\forall i\in[k] \label{LP: flow-amounts}\\
&\sum_{e\in \delta^+(v)}f_i(e)=\sum_{e\in \delta^-(v)}f_i(e)&\forall i\in[k], \forall v\in V(G')\setminus\set{s_i,t_i}\ \  \mbox{(flow conservation)}\label{LP: flow-conservation}\\
&\sum_{e\in \delta^+(v)}\sum_{i=1}^kf_i(e)\leq 1&\forall v\in V(G')\quad\quad\quad\quad\mbox{(vertex capacity constraints)}\label{LP: capacity constraints}\\
&\sum_{(s_i,t_i)\in \mset'} x_i\leq z&\forall \mset'\subseteq\mset,\forall z\in \mathbb{Z}: \opt(G,\mset')\leq z\leq k\label{LP: new constraint}\\
&f_i(e)\geq 0&\forall i\in[k],\forall e\in E(G')\label{LP: non-negativity}
\end{eqnarray}
\label{-------------------------------------end LP constraints-----------------------------------------------}

We say that a solution to (LP-flow2) is \emph{semi-feasible} iff all constraints of types~(\ref{LP: opt-value})--(\ref{LP: capacity constraints}) and (\ref{LP: non-negativity}) are satisfied.
Notice that the number of the constraints in (LP-flow2) is exponential in $k$. In order to solve it, we will use the Ellipsoid Algorithm with a separation oracle, where our approximation algorithm itself will serve as the separation oracle. This is done via the following theorem, which is our main technical result.

\begin{theorem}\label{thm: solution or sep oracle}
There is an efficient algorithm, that, given any semi-feasible solution to (LP-flow2), either computes a routing of at least $\Omega\left(\frac{X^*}{N^{9/19}\cdot \poly\log n}\right )$ demand pairs of $\mset$ via node-disjoint paths, or returns a constraint of type~(\ref{LP: new constraint}), that is violated by the current solution. 
\end{theorem}

 We can now obtain an $O\left (N^{9/19}\cdot \poly\log n\right )$-approximation algorithm for \NDPplanar via the Ellipsoid algorithm. In every iteration, we start with some semi-feasible solution to (LP-flow2), and apply the algorithm from Theorem~\ref{thm: solution or sep oracle} to it. If the outcome is a solution routing at least $\Omega\left (\frac{X^*}{N^{9/19}\cdot \poly\log n}\right)$ demand pairs in $\mset$, then we obtain the desired approximate solution to the problem, assuming that $X^*$ was guessed correctly. Otherwise, we obtain a violated constraint of type~(\ref{LP: new constraint}), and continue to the next iteration of the Ellipsoid Algorithm. Since the Ellipsoid Algorithm is guaranteed to terminate with a feasible solution after a number of iterations that is polynomial in the number of the LP-variables, this gives an algorithm that is guaranteed to return a solution of value $\Omega\left (\frac{X^*}{N^{9/19} \cdot \poly\log n}\right)$ in time $\poly(n)$.
From now on we focus on proving Theorem~\ref{thm: solution or sep oracle}.

We note that, using standard techniques, we can efficiently obtain a flow-paths decomposition of any semi-feasible solution to (LP-flow2): we can efficiently find, for every demand pair $(s_i,t_i)$, a collection $\pset_i$ of paths, connecting $s_i$ to $t_i$, and for each path $P\in \pset_i$, compute a value $f(P)$, such that:

\begin{itemize}
\item For each $i\in [k]$, $\sum_{P\in \pset_i}f(P)=x_i$; 

\item For each $i\in [k]$, $|\pset_i|\leq n$; and 

\item For each $v\in V(G)$, $\sum_{i\in[k]}\sum_{\stackrel{P\in \pset_i:}{v\in P}}f(P)\leq 1$.
\end{itemize}

It is sometimes more convenient to work with the above flow-paths decomposition version of a given semi-feasible solution to (LP-flow2).

We now assume that we are given some semi-feasible solution $(x,f)$ to (LP-flow2), and define a new fractional solution based on it, where the flow between every demand pair is either $0$ or $w^*$, for some value $w^*>0$. First, for each demand pair $(s_i,t_i)$ with $x_i\leq \frac{1}{2k}$, we set $x_i=0$ and we set the corresponding flow values $f_i(e)$ for all edges $e\in E$ to $0$. Since we can assume that $X^*\geq 1$ if the graph is connected, the total amount of flow between the demand pairs remains at least $X^*/2$. We then partition the remaining demand pairs into $q=\ceil{\log 2k}$ subsets, where for $1\leq j\leq q$, set $\mset_j$ contains all demand pairs $(s_i,t_i)$ with $\frac{1}{2^{j}}< x_i \leq \frac 1 {2^{j-1}}$. There is some index $1\leq j^*\leq q$, such that the total flow between the demand pairs in $\mset_{j^*}$ is at least $\Omega(X^*/\log k)$. Let $w^*=\frac{1}{2^{j}}$. We further modify the LP-solution, as follows. First, for every demand pair $(s_i,t_i)\not\in \mset_{j^*}$, we set $x_i=0$, and the corresponding flow values $f_i(e)$ for all edges $e\in E$ to $0$. Next, for every demand pair $(s_i,t_i)\in \mset_{j^*}$, we let $\beta_i=w^*/x_i$, so $\beta_i\leq 1$. We set $x_i=w^*$, and the new flow values $f_i(e)$ are obtained by scaling the original values by factor $\beta_i$. This gives a new solution to (LP-flow2), that we denote by $(x',f')$. The total amount of flow sent in this solution is $\Omega(X^*/\log k)$, and it is easy to verify that constraints~(\ref{LP: flow-amounts})--(\ref{LP: capacity constraints}) and (\ref{LP: non-negativity}) are satisfied. For every demand pair $(s_i,t_i)\in \mset_{j^*}$, $x'_i=w^*$, and for all other demand pairs $(s_i,t_i)$, $x'_i=0$. It is easy to see that for every demand pair $(s_i,t_i)$, $x'_i\leq x_i$. Therefore, if we find a constraint of type~(\ref{LP: new constraint}) that is violated by the new solution, then it is also violated by the old solution. Our goal now is to either find an integral solution routing $\Omega\left (\frac{X^*}{N^{9/19}\cdot \poly\log n}\right)$  demand pairs, or to find a constraint of type~(\ref{LP: new constraint}) violated by the new LP-solution.
In particular, if we find a subset $\mset'\subseteq \mset_{j^*}$ of demand pairs, with $\opt(G,\mset')\leq w^*|\mset'|/2$, then set $\mset'$ defines a violated constraint of type~(\ref{LP: new constraint}) for (LP-flow2). Since from now on we only focus on demand pairs in $\mset_{j^*}$, for simplicity we denote $\mset=\mset_{j^*}$.


\label{---------------------------subsec: wld-----------------------------}
\subsection{Well-Linked Decomposition}\label{subsec: wld}
Like many other approximation algorithms for routing problems, we decompose our input instance into a collection of sub-instances that have some useful well-linkedness properties. Since the routing is on node-disjoint paths, we need to use a slightly less standard notion of node-well-linkedness, defined below. Throughout this paper, we use a parameter $\alphaWL=\frac{w^*}{512 \cdot \alphasc \cdot \log k}$.

\begin{definition}
Given a graph $H$ and a set $\tset'$ of its vertices, we say that $\tset'$ is \emph{$\alphaWL$-well-linked} in $H$ iff for every pair $\tset_1,\tset_2$ of disjoint equal-sized subsets of $\tset$, there is a set $\pset$ of at least $\alphaWL\cdot  |\tset_1|$ node-disjoint paths in $H$, connecting vertices of $\tset_1$ to vertices of $\tset_2$.
\end{definition}


\begin{definition}
Given a sub-graph $H\subseteq G$ and a subset $\mset'\subseteq \mset$ of demand pairs with $\tset(\mset')\subseteq V(H)$, we say that $(H,\mset')$ is a \emph{well-linked instance}, iff $\tset(\mset')$ is $\alphaWL$-well-linked in $H$.
\end{definition}

The following theorem uses standard techniques, and its proof is deferred to Appendix.

\begin{theorem}\label{thm: wld}
There is an efficient algorithm to compute a collection $G_1,\ldots,G_r$ of disjoint sub-graphs of $G$, and for each $1\leq j\leq r$, a set $\mset^j\subseteq \mset$ of demand pairs with $\tset(\mset^j)\subseteq V(G_j)$, such that:

\begin{itemize}
\item For all $1\leq j\leq r$, $(G_j,\mset^j)$ is a well-linked instance;
\item For all $1\leq j\neq j'\le r$, there is no edge in $G$ with one endpoint in $G_j$ and the other in $G_{j'}$; 
\item $\sum_{j=1}^r|\mset^j|\geq 63|\mset|/64$; and
\item $\left |V(G)\setminus\left(\bigcup_{j=1}^rV(G_j)\right )\right |\leq \frac{w^*\cdot |\mset|}{64}$.
\end{itemize}
\end{theorem}

For each $1\leq j\leq r$, let $W_j=w^*|\mset^j|$ be the contribution of the demand pairs in $\mset^j$ to the current flow solution and let $W=\sum_{j=1}^r W_j=\Omega(X^*/\log k)$. Let $n_j=|V(G_j)|$, and let $N_j$ be the number of the non-terminal vertices in $G_j$.
The main tool in proving Theorem~\ref{thm: solution or sep oracle} is the following theorem.

\begin{theorem}\label{thm: main}
 There is an efficient algorithm, that computes, for every $1\leq j\leq r$,  one of the following:
\begin{enumerate}
\item Either a collection $\pset^j$ of node-disjoint paths, routing $\Omega\left(W_j^{1/19}/\poly\log n\right)$ demand pairs of $\mset^j$ in $G_j$; or

\item A collection $\tmset^j\subseteq \mset^j$ of demand pairs, with $|\tmset^j|\geq |\mset^j|/2$, such that $\opt(G_j,\tmset^j)\leq w^*|\tmset^j|/8$.
\end{enumerate}
\end{theorem}
Before we prove Theorem~\ref{thm: main}, we show that Theorem~\ref{thm: solution or sep oracle} follows from it. We apply Theorem~\ref{thm: main} to every instance $(G_j,\mset^j)$, for $1\leq j\leq r$.
We say that instance $(G_j,\mset^j)$ is a type-1 instance, if the first outcome happens for it, and we say that it is a type-2 instance otherwise.
Let $I_1=\set{j\mid \mbox{$(G_j,\mset^j)$ is a type-1 instance}}$, and similarly, $I_2=\set{j\mid \mbox{$(G_j,\mset^j)$ is a type-2 instance}}$. We consider two cases, where the first case happens when $\sum_{j\in I_1}W_j\geq W/2$.

\paragraph{Case 1: $\sum_{j\in I_1}W_j\geq W/2$.} We show that in this case, our algorithm returns a routing of $\Omega\left (\frac{X^*}{N^{9/19}\cdot \poly\log n}\right)$  demand pairs. We need the following lemma, whose proof is deferred to the Appendix. The proof uses standard techniques: namely, we show that the treewidth of each graph $G_j$ is at least $\Omega(W_j/\log k)$, and so $G_j$ must contain a large grid minor.

\begin{lemma}\label{lem: nj bound} For each $1 \le j \le r$, $N_j \ge \Omega(W_j^2 / \log^2 k)$.
\end{lemma}

The number of the demand pairs we route in each type-$1$ instance $(G_j,\mset^j)$ is then at least:

\[\begin{split}
\Omega\left(\frac{W_j^{1/19}}{\poly\log n}\right )&=\Omega\left(\frac{W_j}{W_j^{18/19} \cdot \poly\log n}\right )\\
&=\Omega\left(\frac{W_j}{\left(\sqrt{N_j}\log k\right)^{18/19}\cdot\poly\log n}\right )
\\& =\Omega\left(\frac{W_j}{N^{9/19}\cdot\poly\log n}\right ). 
\end{split}\]

Overall, since $\sum_{j\in I_1}W_j\geq W/2$, the number of the demand pairs routed is $\Omega\left(\frac{W}{N^{9/19}\cdot\poly\log n}\right )=\Omega\left(\frac{X^*}{N^{9/19}\cdot\poly\log n}\right )$.

\paragraph{Case 2: $\sum_{j\in I_2}W_j\geq W/2$.}
Let $\mset'=\bigcup_{j\in I_2}\tmset^j$. Then $|\mset'|=\sum_{j\in I_2}|\tmset^j|\geq \sum_{j\in I_2}\frac{|\mset^j|}{2}\geq \frac 1 4 \sum_{j=1}^r|\mset^j|\geq \frac{|\mset|}8$. We claim that the following inequality, that is violated by the current LP-solution, is a valid constraint of (LP-flow2):

\begin{equation}\label{eq: violated constraint}
\sum_{(s_i,t_i)\in \mset'}x'_i\leq w^*|\mset'|/2.
\end{equation}

In order to do so, it is enough to prove that $\opt(G,\mset')<w^*|\mset'|/2$. Assume otherwise, and let $\pset^*$ be the optimal solution for instance $(G,\mset')$, so $|\pset^*|\geq w^*|\mset'|/2$. We say that a path $P\in \pset^*$ is bad if it contains a vertex of $V(G)\setminus\left(\bigcup_{j=1}^rV(G_j)\right )$. The number of such bad paths is bounded by the number of such vertices - namely, at most $\frac{w^*\cdot |\mset|}{64}\leq \frac{w^*|\mset'|}{8}\leq \frac{|\pset^*|} 2$. Therefore, at least $w^*|\mset'|/4$ paths in $\pset^*$ are good. Each such path must be contained in one of the graphs $G_j$ corresponding to a type-2 instance. For each $j\in I_2$, let $\hat{\mset}^j\subseteq \tmset^j$ be the set of the demand pairs routed by good paths of $\pset^*$. Then, on the one hand, $\sum_{j\in I_2}|\hmset^j|\geq w^*|\mset'|/4=w^*\sum_{j\in I_2}|\tmset^j|/4$, while, on the other hand, since all demand pairs in $\hat{\mset}^j$ can be routed simultaneously in $G_j$, for all $j\in I_2$, $|\hat{\mset}^j|\leq w^*|\tmset^j|/8$, a contradiction. We conclude that  $\opt(G,\mset')<w^*|\mset'|/2$, and (\ref{eq: violated constraint}) is a valid constraint of (LP-flow2).

From now on, we focus on proving Theorem~\ref{thm: main}. Since from now on we only consider one instance $(G_j,\mset^j)$, for simplicity, we abuse the notation and denote $G_j$ by $G$, and $\mset^j$ by $\mset$. As before, we denote $\tset=\tset(\mset)$. We denote by $W=w^*\cdot |\mset|$ the total amount of flow sent between the demand pairs in the new set $\mset$ in the  LP solution (note that this is not necessarily a valid LP-solution for the new instance, as some of the flow-paths may use vertices lying outside of $G^j$). We use $n$ to denote the number of vertices in $G$. Value $k$ - the number of the demand pairs in the original instance - remains unchanged. Our goal is to either find a collection of node-disjoint paths routing $\Omega(W^{1/19}/ \poly(\log (nk)))$ demand pairs of $\mset$ in $G$, or to find a collection $\tmset\subseteq \mset$ of at least $|\mset|/2$ demand pairs, such that $\opt(G,\tmset)\leq w^*|\tmset|/8$. We will rely on the fact that all terminals are $\alphaWL$-well-linked in $G$, for $\alphaWL=\Theta(w^*/\log k)$.  We assume that $G$ is connected, since otherwise all terminals must be contained in a single connected component of $G$ and we can discard all other connected components.

We assume that we are given an embedding of the graph $G$ into the sphere.
For every pair $v,v'\in V(G)$ of vertices, we let $\dface(v,v')$ be the length of the shortest $G$-normal curve connecting $v$ to $v'$ in this embedding, minus $1$. It is easy to verify that $\dface$ is a metric: that is, $\dface(v,v)=0$, $\dface(v,v')=\dface(v',v)$, and the triangle inequality holds for $\dface$. The value $\dface(v,v')$ can be computed efficiently, by solving an appropriate shortest path problem instance in the graph dual to $G$. Given a vertex $v$ and a subset $U$ of vertices of $G$, we denote by $\dface(v,U)=\min_{u\in U}\set{\dface(v,u)}$. Similarly, given two subsets $U,U'$ of vertices of $G$, we denote $\dface(U,U')=\min_{u\in U,u'\in U'}\set{\dface(u,u')}$. Finally, given a $G$-normal curve $C$, and a vertex $v$ in $G$, we let $\dface(v,C)=\min_{u\in V(C)}\set{\dface(v,u)}$.

Over the course of the algorithm, we will sometimes select some face of the drawing of $G$ as the outer face and consider the resulting drawing of $G$ in the plane. The function $\dface$ remains unchanged, and it is only defined with respect to the fixed embedding of $G$ into the sphere. 

\label{---------------------------------sec: enclosures, shells, subsets-----------------------------}
\section{Enclosures, Shells, and Terminal Subsets}\label{sec: enclosures, shells, subsets}

In this section we develop some of the technical machinery that we use in our algorithm, and describe the first steps of the algorithm. We start with enclosures around the terminals.

\label{--------------------------------subsec: enclosures---------------------------}
\subsection{Constructing Enclosures}\label{subsec: enclosures}

Throughout the algorithm, we use a parameter $\Delta=\ceil{W^{2/19}}$. We assume that $W>\Omega(\Delta)$, since otherwise $W$ is bounded by a constant, and we can return the routing of a single demand pair.
The goal of this step is to construct enclosures around the terminals, that are defined below. Recall that $G$ is embedded on the sphere.

\begin{definition}
An \emph{enclosure} for terminal $t\in \tset$ is a simple disc $D_t$ containing the terminal $t$, whose boundary is denoted by $C_t$, that has the following properties. (Recall that $V(D_t)$ is the set of all vertices of $G$ contained in $D_t$.)
\begin{itemize}
\item $C_t$ is a simple closed $G$-normal curve with $\ell(C_t)= \Delta$;

\item  $|\tset\cap V(D_t)|\leq 4\Delta/\alphawl$; and
\item  $V(D_t)$ induces a connected graph in $G$.
\end{itemize}
\end{definition}

The goal of this section is to prove the following theorem.

\begin{theorem}\label{thm: build enclosures}
There is an efficient algorithm, that constructs an enclosure $D_t$ for every terminal $t\in \tset$, such that for all $t,t'\in \tset$:
\begin{itemize}
\item If $D_t\subseteq D_{t'}$, then $D_t = D_{t'}$; and \label{property:build-enclosure-no-containment}

\item If $D_t\cap D_{t'} = \emptyset$, then there are $\Delta$ node-disjoint paths between $V(C_t)$ and $V(C_{t'})$ in $G$. \label{property:build-enclosure-disjoint-paths}
\end{itemize}
\end{theorem}

Notice that since $\ell(C_t) = \Delta$, every vertex of $V(C_t)$ is an endpoint of a path connecting $V(C_t)$ to $V(C_{t'})$.
In order to prove the theorem, we need the following two simple claims.

\begin{claim}\label{clm:extend-by-one}
Let $D$ be any disc on the sphere, whose boundary $C$ is a simple $G$-normal curve, such that $1\leq |V(D)|< |V(G)|-1$, and $G[V(D)]$ is connected (we allow $D$ to consist of a single point, which must coincide with a vertex of $G$). Then we can efficiently find a disc $D'$ with $V(D)\subsetneq V(D')$ and $|V(D')|=|V(D)|+1$, such that  $G[V(D')]$ is a connected graph. Moreover, if $C'$ is the boundary of $D'$, then $C'$ is a simple $G$-normal curve with $\ell(C')=\ell(C)+1$. 
\end{claim}
\begin{proof}
If $D$ consists of a single point corresponding to a vertex $v\in V(G)$, then let $u$ be any neighbor of $v$ in $G$. It is easy to construct a disc $D'$ whose boundary only contains the vertices $v$ and $u$, and has all the required properties, with $V(D')=\set{v,u}$. We now assume that $|V(D)|>1$.

Let $u\in V(C)$ be any vertex that has a neighbor in $V(G)\setminus V(D)$: since $G$ is connected and $|V(D)|< |V(G)|$, such a vertex  exists. Let $u'$ be a vertex lying next to $u$ on $C$. Then there must be a vertex $v\in V(G)\setminus V(D)$, such that:
 (i) edge $e=(u,v)$ belongs to $G$; and (ii) there is a simple $G$-normal curve $\gamma'$ connecting $u'$ to $v$, that intersects $G$ only at its endpoints, and intersects $D$ only at $u'$.  Let $\sigma,\sigma'$ be the two segments of $C$ whose endpoints are $u$ and $u'$. 
 
 Notice that due to the edge $e=(u,v)$, there is also a $G$-normal curve $\gamma$ connecting $u$ to $v$, that intersects $G$ only at its endpoints, and intersects $D$ only at $u$. Let $C_1$ be the concatenation of $\sigma,\gamma$ and $\gamma'$, and let $C_2$ be the concatenation of $\sigma',\gamma$ and $\gamma'$.
  
Let $x\in V(G)\setminus (V(D)\cup \set{v})$ be any vertex (such a vertex exists since $|V(D)|<|V(G)|-1$), and let $F$ be any face in the drawing of $G\cup C\cup \gamma\cup \gamma'$ incident on $x$. We can view the face $F$ as the outer face of our drawing, to obtain a drawing of $G\cup C\cup \gamma\cup \gamma'$ in the plane. Using this view, curve $C_1$ defines a disc $D_1$ and curve $C_2$ defines a disc $D_2$ in the plane. Exactly one of these discs contains the disc $D$ - assume w.l.o.g. that it is $D_1$. We then set $D'=D_1$ and $C'=C_1$. It is now immediate to verify that $D'$ has all required properties.
%
%
\end{proof}

\begin{claim}\label{clm:cut-in-planar-graph}
	Let $H$ be any connected planar graph drawn on a sphere, and let $s$ and $t$ be two distinct vertices of $H$. Assume that the maximum number of internally node-disjoint paths between $s$ and $t$ in $H$ is $\kappa$. Then we can efficiently find a simple closed $H$-normal curve $C$ of length $\kappa$ on the sphere, separating $s$ from $t$, with $s,t\not\in V(C)$. Moreover, if $U$ and $U'$ denote the sets of vertices lying strictly on each side of $C$, then $H[U\cup V(C)]$ and $H[U'\cup V(C)]$ are both connected.
\end{claim}

\begin{proof}
	By Menger's theorem, we can efficiently find a set $X \subseteq V(H) \setminus \{s, t\}$ of $\kappa$ vertices such that $s$ and $t$ are separated in the graph $H \setminus X$.  
	
	Let $\pset=\set{P_1,\ldots,P_{\kappa}}$ be any set of $\kappa$ internally node-disjoint paths connecting $s$ to $t$ in $H$, and let $H'$ be the sub-graph of $H$ obtained by taking the union of the paths in $\pset$. The drawing of $H'$ on the sphere consists of $\kappa$ faces, where the boundary of each face is the union of two distinct paths in $\pset$. We assume that the faces are $F_1,\ldots,F_{\kappa}$, and we assume without loss of generality that the boundary of each face $F_i$ is $P_i\cup P_{i+1}$ (and the boundary of $F_{\kappa}$ is $P_{\kappa}\cup P_1$). Notice that for each $1\leq i\leq \kappa$, $X$ contains exactly one internal vertex of $P_i$, that we denote by $x_i$. 
Let $P^s_i,P^t_i$ be the two paths in $P_i\setminus\set{x_i}$, where $P^s_i$ is the path containing $s$, so $P^t_i$ contains $t$.

Fix some $1\leq i\leq \kappa$, and let $H_i\subseteq H$ be the graph induced by all vertices lying inside $F_i$ or on its boundary. Let $H'_i=H_i\setminus\set{x_i,x_{i+1}}$. Then there is no path in $H'_i$ connecting a vertex of $P^s_i\cup P^s_{i+1}$ to a vertex of $P^t_i\cup P^t_{i+1}$: such a path would contradict the fact that $s$ is disconnected from $t$ in $H\setminus X$. Therefore, there is a curve $\gamma_i$, connecting $x_i$ to $x_{i+1}$ inside $F_i$, that intersects $H$ only at its endpoints. 
Curve $\gamma_i$ partitions $F_i$ into two subfaces: $F_i'$ and $F_i''$, with $s$ lying on the boundary of $F_i'$ and $t$ lying on the boundary of $F_i''$. Since $H$ is a connected graph, and $\gamma_i$ only intersects $H$ at its endpoints, both subgraphs induced by the vertices lying in $F_i'$ and its boundary, and in $F_i''$ and its boundary, are connected.

We  build the curve $C$ by concatenating all curves $\gamma_i$ for $1\leq i\leq \kappa$. It is easy to verify that $C$ has all required properties.\end{proof}

\begin{proofof}{Theorem~\ref{thm: build enclosures}}
We show an efficient algorithm to construct the enclosures with the desired properties.
Throughout the algorithm,  we maintain a set $\set{D_t}_{t\in \tset}$ of enclosures, such that for every pair $t,t'\in \tset$ of terminals, if $D_t\subseteq D_{t'}$, then $D_t=D_{t'}$.

The initial set of enclosures is obtained as follows. For each terminal $t\in \tset$, let $D'_t$ be the disc containing a single point - the image of the vertex $t$. We apply Claim~\ref{clm:extend-by-one} $\Delta-1$ times to $D'_t$, to obtain a disc $D_t$ whose boundary is a simple $G$-normal curve of length $\Delta$, and $D'_t\subseteq D_t$, while $|V(D_t)|= \Delta$. By Claim~\ref{clm:extend-by-one}, $V(D_t)$ induces a connected sub-graph in $G$. Since $|V(D_t)|=\Delta$, $D_t$ is a valid enclosure for $t$. While there is a pair $t,t'$ of terminals with $D_t\subsetneq D_{t'}$, we set $D_t=D_{t'}$. This finishes the definition of the initial set of enclosures. We then perform a number of iterations. In every iteration, we consider all pairs $t,t'$ of terminals with $D_t \cap D_{t'} = \emptyset$, and check whether there are $\Delta$ node-disjoint paths connecting the vertices of $V(C_t)$ to the vertices of $V(C_{t'})$ in $G$. If so, then we say that $(t,t')$ is a good pair. If all such pairs are good, then we terminate the algorithm, and output the current set of enclosures. Otherwise, let $(t,t')$ be a bad pair.  Let $H$ be the graph constructed from $G$, as follows: we delete all vertices of $V(D_t\setminus C_t)$, and add a source vertex $a$ to the interior of $D_t$. We then connect $a$ to every vertex of $V(C_t)$ with an edge. Similarly, we delete all vertices of $V(D_{t'}\setminus C_{t'})$, and add a destination vertex $b$ to the interior of $D_{t'}$. We then connect $b$ to every vertex of $V(C_{t'})$. 

Let $\kappa\leq \Delta-1$ be the maximum number of internally node-disjoint paths connecting $a$ to $b$ in $H$. We apply Claim~\ref{clm:cut-in-planar-graph} to find an $H$-normal closed curve $C$ of length $\kappa$, separating $a$ from $b$. Then $C$ defines a $G$-normal curve of length $\kappa$, such that, if $U$ and $U'$ denote the sets of vertices of $G$ lying strictly on each side of $C$, then $V(D_t)\subseteq U\cup V(C)$; $V(D_{t'})\subseteq U'\cup V(C)$; and both $G[U\cup V(C)]$ and $G[U'\cup V(C)]$ are connected. 

Notice that either $|U\cap \tset|\leq |\tset|/2$ or $|U'\cap \tset|\leq |\tset|/2$ holds - we assume w.l.o.g. that it is the former. Let $D$ be the disc whose boundary is $C$, with $D_t\subseteq D$. Since $\ell(C)<\Delta$, from the well-linkedness of the terminals $|V(D)\cap \tset|\leq 
\Delta/\alphawl$.  We next apply Claim~\ref{clm:extend-by-one} to $D$ repeatedly to obtain a disc ${D'}$, whose boundary $C'$ is a simple $G$-normal curve of length $\Delta$, so that $D \subseteq D'$; $|V(D')| \leq |V(D)| + \Delta$, and $G[D']$ is a connected graph. It is easy to verify that $D'$ is a valid enclosure for terminal $t$. We replace $D_t$ with $D'$. If there is any terminal $t''$ with $D_{t''}\subsetneq D'$, then we replace $D_{t''}$ with $D'$ as well. This finishes the description of an iteration. Notice that $\sum_{t\in \tset}|V(D_t)|$ increases by at least $1$ in every iteration, and so the number of iterations is bounded by $|V(G)|$.
\end{proofof}

\paragraph{Distances between terminals.}
For every pair $t,t'$ of terminals, we define the distance $d(t,t')$ between $t$ and $t'$ to be the length of the shortest $G$-normal open curve, with one endpoint in $V(C_t)$ and another in $V(C_{t'})$. (Notice that if $D_t\cap D_{t'}\neq \emptyset$, then $d(t,t')=1$).
We repeatedly use the following simple observation (a weak triangle inequality):

\begin{observation}
	\label{obs: weak-triangle-inequality}
	For all $t, t', t'' \in \tset$, $d(t, t'') \leq d(t, t') + d(t', t'') + \Delta/2$.
\end{observation}
\begin{proof}
Let $\gamma$ be a $G$-normal curve of length $d(t,t')$ connecting a vertex of $C_t$ to a vertex of $C_{t'}$, and let $\gamma'$ be defined similarly for $d(t',t'')$. Let $u,u'\in V(C_{t'})$ be the vertices that serve as endpoints of $\gamma$ and $\gamma'$, respectively, and let $\sigma$ be the shorter of the two segments of $C_{t'}$ between $u$ and $u'$, so the length of $\sigma$ is at most $\Delta/2+2$. By combining $\gamma,\sigma$ and $\gamma'$, we obtain a $G$-normal curve of length at most  $d(t, t') + d(t', t'') + \Delta/2$, connecting a vertex of $C_t$ to a vertex of $C_{t''}$.
\end{proof}

\label{---------------------------------SubSec: the shells -------------------------------}

\subsection{Constructing Shells}\label{subsec: - the shells}
Suppose we are given some terminal $t\in \tset$ and an integer $r\geq 1$. In this section we show how to construct a shell of depth $r$ around $t$, and explore its properties. Shells play a central role in our algorithm. In order to construct the shell, we need to fix a plane drawing of the graph $G$, by choosing one of the faces $F_t$ of the drawing of $G$ on the sphere as the outer face. The choice of the face $F_t$ will affect the construction of the shell, but once the face $F_t$ is fixed, the shell construction is fixed as well. We require that for every vertex $v$ on the boundary of $F_t$, $\dface(v,C_t)\geq r+1$, and that $C_t$ separates all vertices on the boundary of $F_t$ from $t$. We note that when we construct shells for different terminals $t,t'$, we may choose different faces $F_t,F_{t'}$, and thus obtain different embeddings of $G$ into the plane.
We now define a shell.

\begin{definition}
Suppose we are given a terminal $t\in \tset$, a face $F_t$ in the drawing of $G$ on the sphere, and an integer $r\geq 1$, such that for every vertex $v$ on the boundary of $F_t$, $\dface(v,C_t)\geq r+1$, and $C_t$ separates $t$ from the boundary of $F_t$.

A shell $\zset^r(t)$ of depth $r$ around $t$ with respect to $F_t$ is a collection $\zset^r(t)=(Z_1(t),Z_2(t),\ldots,Z_r(t))$ of $r$  tight concentric cycles around $C_t$. In other words, all cycles $Z_h(t)$ are simple and disjoint from each other, and the following properties hold. For each $1\leq h\leq r$, let $D(Z_h(t))$ be the disc whose boundary is $Z_h(t)$ in the planar drawing of $G$ with $F_t$ as the outer face. Then:

\begin{properties}{J}
\item $D_t\subsetneq D(Z_1(t))\subsetneq D(Z_2(t))\subsetneq\cdots\subsetneq D(Z_r(t))$; and  \label{prop: disc2}

\item for every $1\leq h\leq r$, if $H$ is the graph obtained from $G$ by contracting all vertices lying in $D(Z_{h-1}(t))$ into a super-node $a$, then $Z_h(t) = \mincycle(H, a)$ (when $h=1$, we contract $D_t$ into a super-node $a$).
\end{properties}
\end{definition}

Notice that from this definition we immediately obtain the following additional properties:

\begin{properties}[2]{J}
\item For every $1\leq h\leq r$,  for every vertex $v\in V(Z_h(t))$, there is a $G$-normal curve of length $2$ connecting $v$ to some vertex of $V(Z_{h-1}(t))$ (or to a vertex of $V(C_t)$ if $h=1$).  \label{prop: curve 1}

\item For every $1\leq h\leq r$, for every vertex $v\in V(Z_h(t))$, there is a $G$-normal curve $\gamma(v)$ of length $h+1$ connecting $v$ to some vertex of $V(C_t)$, so that $\gamma(v)\subseteq D(Z_h(t))$, and it is internally disjoint from $Z_h(t)$ and $C_t$.\label{prop: short curve}
\end{properties}
 
Let $\tilde U$ be the set of all vertices $v\in V(G)\setminus V(D_t)$, such that $\dface(v,C_t)<r+1$. Clearly, $G\setminus \tilde U\neq \emptyset$, since the vertices on the boundary of $F_t$ do not lie in $\tilde U$. We let $Y_t$ be the connected component of $G\setminus \tilde U$ containing the vertices on the boundary of $F_t$. The following additional property follows from the definition of the shell:

\begin{properties}[4]{J}
\item All vertices of $Y_t$ lie outside $D(Z_r(t))$.\label{prop: Yt out}
\end{properties}

Indeed, from Property~(\ref{prop: short curve}), for each $1\leq h\leq r$, $Z_h(t)\cap Y_t=\emptyset$. Since $Y_t$ is connected,  $Z_r(t)$ must separate $V(D_t)$ from $Y_t$.

For convenience, we will always denote $Z_0(t)=C_t$ and $D(Z_0(t))=D_t$. Note that $Z_0(t)$ is not a cycle in $G$ -- it is a simple closed $G$-normal curve, and so it is not part of the shell. We build the cycles $Z_1(t),\ldots,Z_r(t)$ one-by-one, and  we maintain the invariant that for each $1\leq h\leq r$, $(Z_1(t),\ldots,Z_h(t))$ is a shell of depth $h$ around $t$ with respect to $F_t$.





Assume that we have defined $Z_0(t),\ldots,Z_{h-1}(t)$ for some $1\leq h< r$, such that the above invariant holds. In order to define $Z_h(t)$, consider the drawing of the graph $G$ in the plane with $F_t$ as the outer face, and delete all vertices lying in $D(Z_{h-1}(t))$ from it. Consider the face $F$ where the deleted vertices used to be. Let $\Gamma$ be the inner boundary of $F$. Then $\Gamma$ must contain a single simple cycle $Z$, such that $D(Z_{h-1}(t))\subseteq D(Z)$ (if no such cycle exists, then from Invariant~(\ref{prop: short curve}), for some vertex $v$ on the boundary of $F_t$, there is a $G$-normal curve of length at most $r+1$ connecting $v$ to $C_t$, which contradicts the choice of $F_t$). 
We let $Z_h(t)=Z$. It is easy to see that the invariant continues to hold.
This finishes the construction of the shell. We now study its properties.


For each $1\leq h\leq r$, we let $U_h$ be the set of all vertices of $G$ lying in $\dnot(Z_h(t))\setminus D(Z_{h-1}(t))$. Let $\rset_h$ be the set of all connected components of $G[U_h]$.
Each connected component $R\in \rset_h$ may have at most one neighbor in $V(Z_h(t))$ from the definition of the shell, and it may have a number of neighbors in $V(Z_{h-1}(t))$. We say that $R\in \rset_h$ is a \emph{type-1 component} if it has one neighbor in $V(Z_h(t))$, and at least one neighbor in $V(Z_{h-1}(t))$. We denote by $u(R)$ the unique neighbor of $R$ in $V(Z_h(t))$, and we denote by $L(R)$ the set of the neighbors of $R$ that belong to $V(Z_{h-1}(t))$.
We say that $R$ is a \emph{type-2 component}, if it has at least one neighbor in $V(Z_{h-1}(t))$ and no neighbors in $V(Z_h(t))$. In this case, we let $L(R)$ be the set of the neighbors of $R$ lying in $V(Z_{h-1}(t))$, and $u(R)$ is undefined.
Otherwise, we say that it is a \emph{type-3 component}. In this case, it has exactly one neighbor in $V(Z_h(t))$, that we denote by $u(R)$,  and no neighbors in $V(Z_{h-1}(t))$, so we set $L(R)=\emptyset$ (see Figure~\ref{fig: three types of cc's}). We sometimes refer to the vertices in set $L(R)$ as \emph{the legs of $R$}.

\begin{figure}[h]
\centering
\subfigure[The three types of components in $\rset_h$: $R_1$ is of type 1, $R_2$ is of type 2, and $R_3$ is of type 3.]{\scalebox{0.35}{\includegraphics{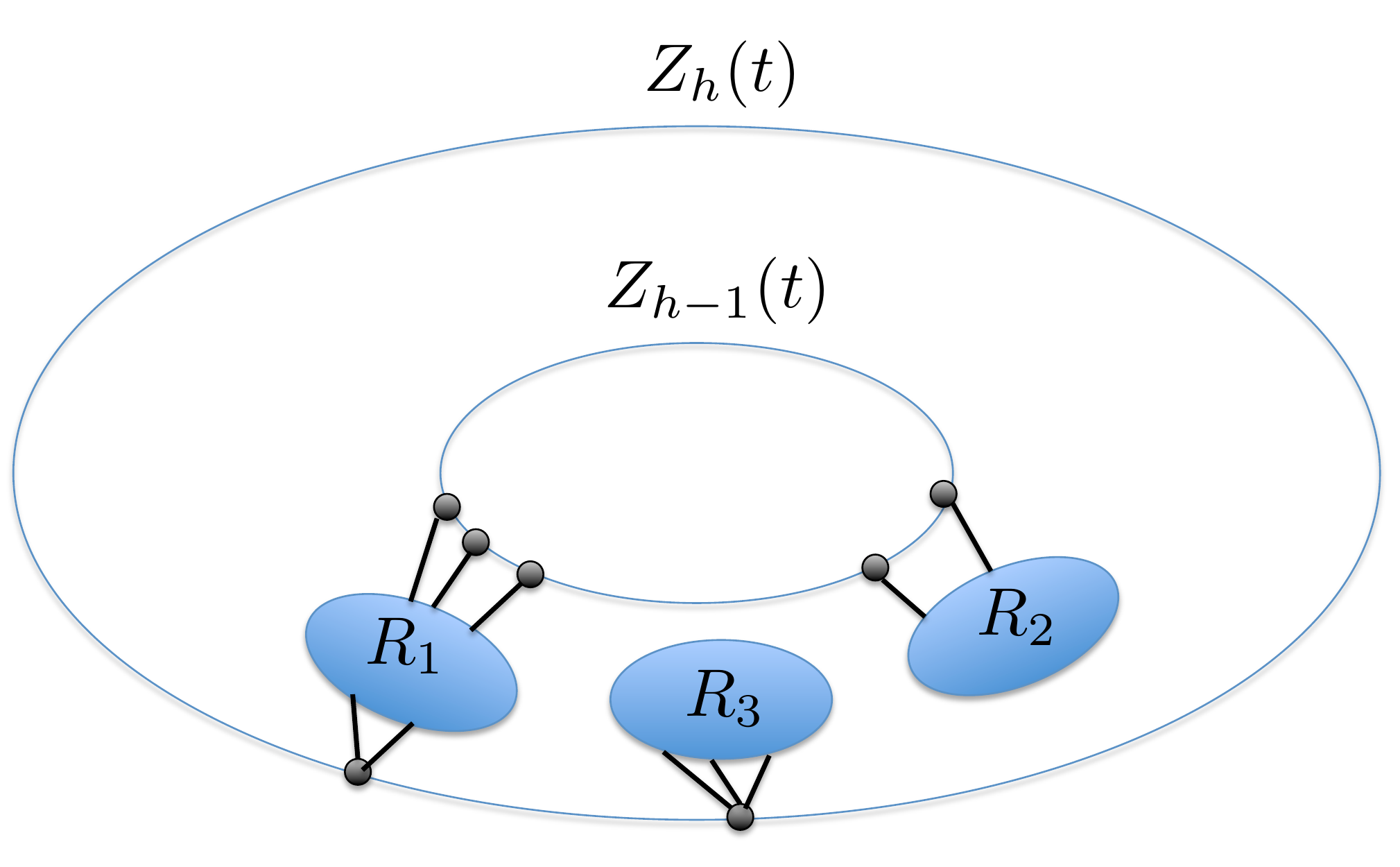}}\label{fig: three types of cc's}}
\hspace{1cm}
\subfigure[Definitions of vertices $u(R)$ and segments $\sigma(R)$]{
\scalebox{0.35}{\includegraphics{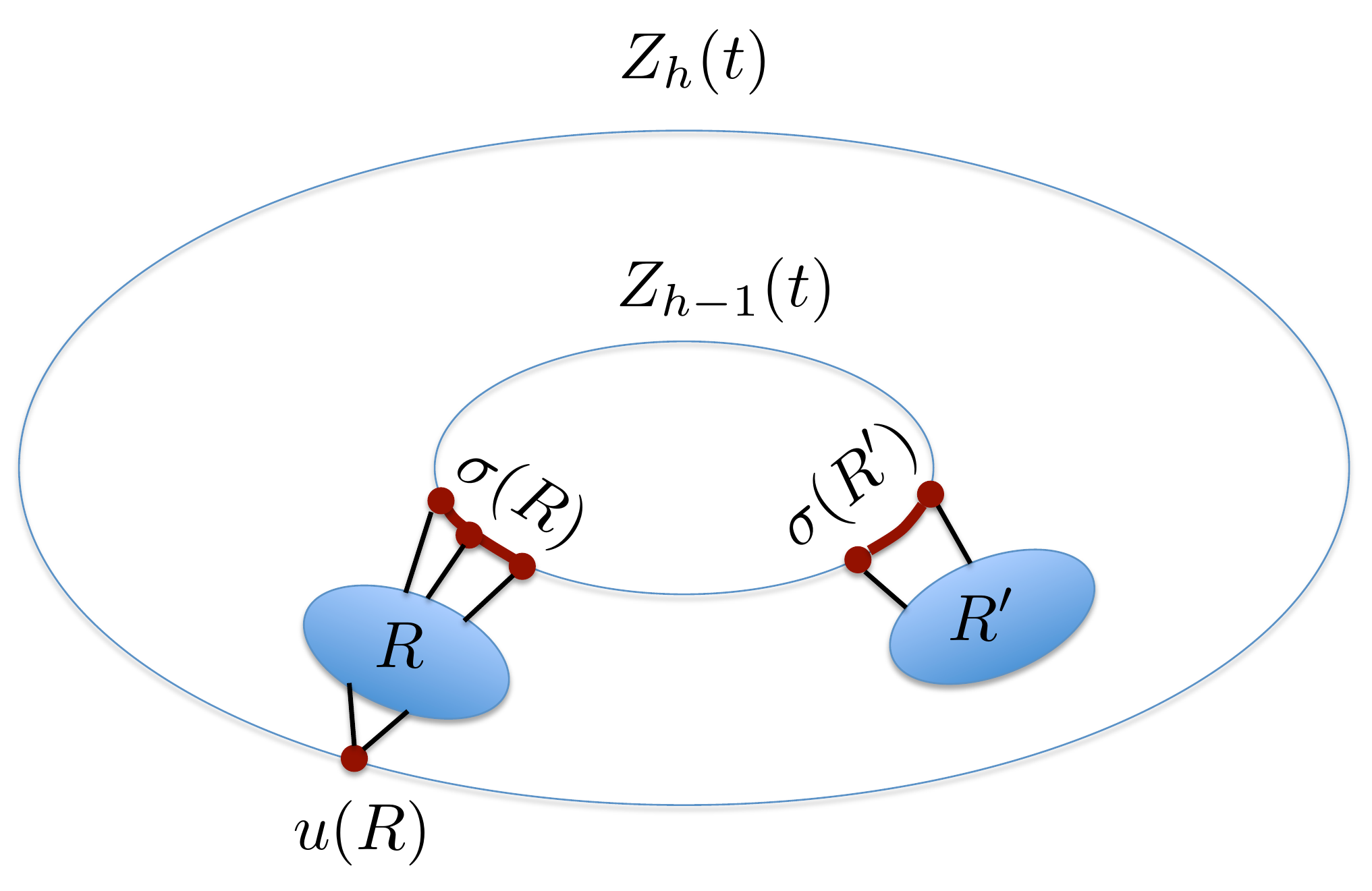}}\label{fig: disc-defs}}
\caption{Structure of the shells \label{fig: shell-structure}}
\end{figure}



Consider now any type-1 or type-2 component $R\in \rset_h$. For each such component $R$, we define a segment $\sigma(R)$ of $Z_{h-1}(R)$ containing all vertices of $L(R)$ as follows. If $|L(R)|=1$, then let $\sigma(R)$ consist of the unique vertex of $L(R)$. Otherwise, we let $\sigma(R)$ be the smallest (inclusion-wise) segment of $Z_{h-1}(t)$ containing all vertices of $L(R)$, such that, if we let $C$ be the union of $\sigma(R)$ and the outer boundary of the drawing of $G[V(R)\cup L(R)]$, then $C$ separates $R$ from $t$ (see Figure~\ref{fig: disc-defs}).
Notice that $\set{\sigma(R)\mid R\in \rset_h\mbox{ is of type 1 or 2}}$ is a nested set of segments of $Z_{h-1}(t)$. For consistency, for each type-3 component $R$, we define $\sigma(R)=\emptyset$. We denote $\rset=\bigcup_{h=1}^r\rset_h$.
We will repeatedly use the following theorem.

\begin{theorem}\label{thm: curves around CC's}
There is an efficient algorithm, that computes, for every $1\leq h\leq r$, for every component $R\in\rset_h$, a disc $\eta(R)\subseteq D(Z_h(t))$, whose boundary, denoted by $\gamma(R)$, is a simple closed $G$-normal curve of length at most $2h+1+\Delta/2$, such that:

\begin{enumerate}
\item For all $R\in \rset$, $G[V(R)\cup L(R)]\subseteq \eta(R)$, and, if $u(R)$ is defined, then $u(R)\in \eta(R)$;

\item For all $R\in \rset$, $\gamma(R)$ is disjoint from all vertices in $\bigcup_{R'\in \rset}V(R')$. In particular, for all $R'\in \rset$, either $R'\subseteq \eta(R)$, or $R'$ lies completely outside $\eta(R)$; and

\item For all $1\leq h\leq r$, for all $R,R'\in \rset_h$, either $\eta(R)\subseteq \eta(R')$, or $\eta(R')\subseteq\eta(R)$, or $\eta^{\circ}(R)\cap \eta^{\circ}(R')=\emptyset$. Moreover, if $R'\subseteq \eta(R)$, then $\sigma(R')\subseteq \sigma(R)$.

\end{enumerate} 
\end{theorem}

\begin{proof}
We use the planar drawing of $G$ where $F_t$ is the outer face.
Fix some $1\leq h\leq r$. Let $\rset_h^1,\rset_h^2$, and $\rset_h^3$ be the sets of type-1, type-2, and type-3 components of $\rset_h$, respectively.
For every type-3 component $R\in \rset_h^3$, we let $\gamma(R)$ be a simple closed $G$-normal curve containing a single vertex of $G$ - vertex $u(R)$, such that the disc $\eta(R)$, whose boundary is $\gamma(R)$, contains $R$, and it is disjoint from all other components of $\rset$.

Consider now any type-1 component $R\in \rset_h$, and let $a(R),a'(R)$ be the endpoints of $\sigma(R)$.
Let $H\subseteq G$ be obtained from $G[V(R)\cup L(R)]$, by adding all edges connecting $u(R)$ to $V(R)$ to it.
 We draw two $G$-normal curves, $\gamma_1(R)$ connecting $u(R)$ to $a(R)$, and $\gamma_1'(R)$, connecting $u(R)$ to $a'(R)$ on either side of $R$, such that the curves $\gamma_1(R),\gamma_1'(R)$ do not contain any other vertices of $G$, and they follow the boundary of the drawing of $H$ from the outside. Let $\gamma'(R)$ be the union of $\gamma_1(R)$ and $\gamma_1'(R)$, and let $\eta'(R)$ be the disc whose boundary is the union of $\gamma'(R)$ and $\sigma(R)$.

Similarly, given any type-2 component $R\in \rset_h$, we denote by  $a(R),a'(R)$ be the endpoints of $\sigma(R)$.
Let $H=G[V(R)\cup L(R)]$, and let $\gamma'(R)$ be a simple $G$-normal curve connecting $a(R)$ to $a'(R)$,  such that  $\gamma'(R)$ does not contain any other vertices of $G$, and it follows the boundary of the drawing of $H$ from the outside. Let $\eta'(R)$ be the disc whose boundary is the union of $\gamma'(R)$ and $\sigma(R)$. 

Clearly, we can draw the curves $\set{\gamma_1(R),\gamma'_1(R)\mid R\in \rset_h^1}\cup \set{\gamma'(R)\mid R\in \rset_h^2}$, so that for all $R,R'\in \rset^1_h\cup\rset^2_h$, either $\eta'(R)\subseteq \eta'(R')$, or $\eta'(R')\subseteq \eta'(R)$, or the interiors of $\eta'(R)$ and $\eta'(R')$ are disjoint. Moreover, if $\eta'(R)\subseteq \eta'(R')$, then $\sigma(R)\subseteq \sigma(R')$. Notice that the curves $\gamma'(R)$ are disjoint from all vertices in $\bigcup_{R''\in \rset}V(R'')$.

Let $A=\set{a(R),a'(R)\mid R\in \rset_h^1\cup \rset_h^2}$. Recall that from Property~(\ref{prop: short curve}), for every vertex $a\in A$, there is a $G$-normal curve $\gamma(a)$ of length at most $h$, connecting $a$ to some vertex $v\in V(C_t)$, such that $\gamma(a)\subseteq D(Z_{h-1}(t))$. In particular, $\gamma(a)$ must be disjoint from all vertices in $\bigcup_{R\in \rset}V(R)$, as it must contain one vertex from each cycle $Z_1(t),\ldots,Z_{h-1}(t)$, and a vertex of $Z_0(t)$. Let $\Gamma=\set{\gamma(a)\mid a\in A}$. We can assume without loss of generality that the curves in $\Gamma$ are non-crossing, and moreover, whenever two such curves meet, they continue together. In other words, if $\gamma(a)\cap \gamma(a')\neq \emptyset$, then this intersection is a simple $G$-normal curve that contains a vertex of $C_t$.

We say that a component $R\in \rset_h^1\cup \rset_h^2$ is good if $\gamma(a(R))$ and $\gamma(a'(R))$ intersect, and we say that it is bad otherwise. From our assumption, if $R$ is good, then $\gamma(a(R))\cap\gamma(a'(R))$ is a curve, connecting some vertex $v$ to some vertex of $C_t$. We then let $\gamma(R)$ be the union of $\gamma'(R)$, the segment of $\gamma(a(R))$ from $a(R)$ to $v$, and the segment of $\gamma(a'(R))$ from $a'(R)$ to $v$, and we let $\eta(R)$ be the disc whose boundary is $\gamma(R)$. Notice that for every component $R'\in \rset_h^1\cup \rset_h^2$, if $\eta'(R')\subseteq \eta'(R)$, then $R'$ is also a good component, and $\eta(R')\subseteq \eta(R)$.

For every bad component $R\in  \rset_h^1\cup \rset_h^2$, we let $b(R),b'(R)\in V(C_t)$ be the endpoints of $\gamma(a(R))$ and $\gamma(a'(R))$, respectively, and we let $\sigma'(R)$ be the segment of $C_t$, whose endpoints are $b(R)$ and $b'(R)$, such that the disc whose boundary is $\gamma'(R)\cup \sigma'(R)\cup \gamma(a(R))\cup \gamma(a'(R))$ does not contain $D_t$.

Notice that the segments $\sigma'(R)$ for all bad components $R$ form a nested set of intervals on $C_t$. This is since their corresponding segments $\sigma(R)$ are nested, and the curves in $\Gamma$ are non-crossing. We say that a bad component $R$ is large, iff $\sigma'(R)$ contains more than $\Delta/2+1$ vertices. Let $\rset'$ be the set of all large bad components of $\rset_h^1\cup \rset_h^2$. Then we can find an ordering $(R_1,R_2,\ldots,R_{\ell})$ of the components in $\rset'$, so that $\sigma'(R_1)\subseteq \sigma'(R_2)\subseteq\cdots\subseteq \sigma'(R_{\ell})$.

For every bad component $R\not\in \rset'$, we let $\gamma(R)$ be the union of $\gamma'(R),\sigma'(R), \gamma(a(R))$, and $ \gamma(a'(R))$, and we let $\eta(R)$ be the disc whose boundary is $\gamma(R)$. For every bad component $R\in \rset'$, we let $\sigma''(R)$ be the segment of $C_t$ with endpoints $b(R)$ and $b'(R)$, that is different from $\sigma'(R)$, so the length of $\sigma''(R)$ is at most $\Delta/2+1$. We let  $\gamma(R)$ be the union of $\gamma'(R),\sigma''(R), \gamma(a(R))$, and $ \gamma(a'(R))$, and we let $\eta(R)$ be the disc whose boundary is $\gamma(R)$. Notice that in either case, the length of $\gamma(R)$ is bounded by $2h+\Delta/2+1$.
It is immediate to verify that the resulting discs $\eta(R)$ for all $R\in \bigcup_{h'=1}^r\rset_{h'}$ have all required properties. Indeed, notice that for all type-1 and type-2 components $R\in \rset_h$, $\eta'(R)\subseteq\eta(R)$. The first property then follows from the definition of $\eta'(R)$. For the third property, $\eta(R)\subseteq \eta(R')$ for $R,R'\in \rset_h$ only if $\eta'(R)\subseteq\eta'(R')$, and from the construction of the discs $\eta'(R),\eta'(R')$, this can only happen if $\sigma(R)\subseteq \sigma(R')$. If neither of the discs $\eta'(R),\eta'(R')$ is contained in the other, then they are internally disjoint, and our construction of the discs $\eta(R),\eta(R')$ ensures that these two discs are also internally disjoint. Finally, consider two components $R,R'\in \rset$. Our construction of the disc $\eta(R)$ ensures that $\gamma(R)$ is disjoint from $V(R')$, and so either $R'\subseteq \eta(R)$, or $R'\cap \eta(R)=\emptyset$. This establishes the remaining property.

\begin{figure}[h]
 \centering
\scalebox{0.5}{\includegraphics{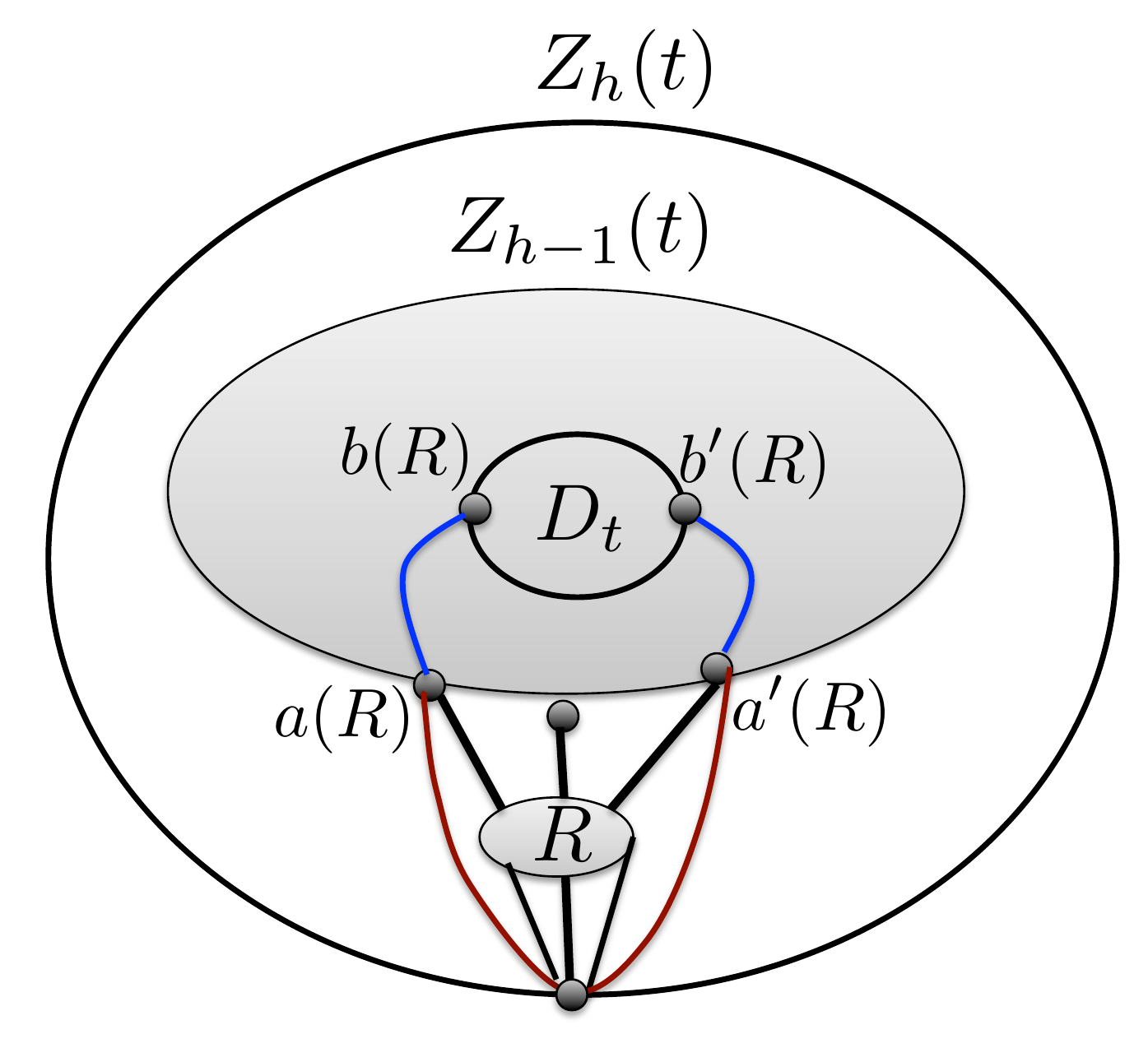}}\caption{Building $\eta(R)$. Curve $\gamma'(R)$ is shown in red, and $\gamma(a(R)),\gamma(a'(R))$ are shown in blue.\label{fig: build shell}}
\end{figure}

\end{proof}

We also need the following observation.

\begin{observation}\label{obs: connectivity of shells}
For all $1\leq h\leq r$, if $\tilde U_h$ is the set of vertices of $G$ lying in $D(Z_h(t))$, then $G[\tilde U_h]$ is connected.
\end{observation} 

\begin{proof} 
The proof is by induction on $h$. Recall that $V(D_t)$ induces a connected sub-graph in $G$, from the definition of the enclosures. Assume now that $G[\tilde U_h]$ is connected for some $0\leq h<r$. From the definition of $Z_{h+1}(t)$, and from the fact that $G$ is a connected graph, it is immediate to verify that $G[\tilde U_{h+1}]$ is also connected.
\end{proof}

Finally, we need the following observation about the interactions of different shells. 

\begin{observation}\label{obs: cycles of shells disjoint}
Let $t,t'$ be any pair of terminals, and let $r,r'>0$ be any integers, such that $d(t,t')>r+r'+1$. Let $\zset^r(t)=(Z_1(t),\ldots,Z_r(t))$ be a shell of depth $r$ around $t$ with respect to some face $F_t$, and let $\zset^{r'}(t')=(Z_1(t')\ldots,Z_{r'}(t'))$ be a shell of depth $r'$ around $t'$ with respect to some face $F_{t'}$. Then for all $1\leq h\leq r$ and $1\leq h'\leq r'$, $Z_h(t)\cap Z_{h'}(t')=\emptyset$.
\end{observation}

\begin{proof}
Assume otherwise. Then there is some vertex $v\in Z_h(t)\cap Z_{h'}(t')$. From Property~(\ref{prop: short curve}), there is a $G$-normal curve $\gamma(v)$ of length at most $h+1\leq r+1$ connecting $v$ to some vertex of $C_t$, and similarly there is a $G$-normal curve $\gamma'(v)$ of length at most $h'+1\leq r'+1$, connecting $v$ to some vertex of $C_{t'}$. Combining $\gamma(v)$ with $\gamma'(v)$, we obtain a $G$-normal curve of length at most $r+r'+1<d(t,t')$, connecting a vertex of $C_t$ to a vertex of $C_{t'}$, a contradiction.
\end{proof}

\label{--------------------------------subsec: terminal subsets---------------------------}
\subsection{Terminal Subsets}\label{subsec: terminal subsets}


Let $\Delta_0=20\left (\ceil{\log_{10/9} n}+1\right )\Delta=\Theta(\Delta \log n)$.
 Our next step is to define a family of disjoint subsets of terminals, so that the terminals within each subset are close to each other, while the terminals belonging to different subsets are far enough from each other. We will ensure that almost all terminals of $\tset$ belong to one of the 
 resulting subsets. Following is the main theorem of this section.

\begin{theorem}\label{thm: terminal clusters}
There is an efficient algorithm to compute a collection $\xset=\set{X_1,\ldots,X_q}$ of disjoint subsets of terminals of $\tset$, such that:

\begin{itemize}
\item for each $1\leq i\leq q$, for every pair $t,t'\in X_i$ of terminals, $d(t,t')\leq \Delta_0$;

\item for all $1\leq i\neq j\leq q$, for every pair $t\in X_i$, $t'\in X_j$ of terminals $d(t,t')\geq 5\Delta$; and

\item $\sum_{i=1}^q|X_i|\geq 0.9|\tset|$.
\end{itemize}
\end{theorem}

\begin{proof}
We start with $\xset=\emptyset$ and $\ttset=\tset$, and perform a number of iterations, each of which adds one subset of terminals to $\xset$, and removes some terminals from $\ttset$. The iterations are executed while $\ttset\neq \emptyset$, and each iteration is executed as follows.

Let $t^*$ be any terminal in $\ttset$. For all $1\leq i\leq \ceil{\log_{10/9}n}+1$, let $Y_i$ contain all terminals $t\in \ttset$ with $d(t^*,t)\leq 8\Delta i$. We use the following simple observation.

\begin{observation}\label{obs: small level}
There is some $2\leq i\leq \ceil{\log_{10/9} n}+1$, such that $|Y_i\setminus Y_{i-1}|\leq  |Y_{i-1}|/9$.
\end{observation}

\begin{proof}
 If the claim is false, then for every $2\leq i\leq \ceil{\log_{10/9} n}+1$, $|Y_i| \geq 10 |Y_{i-1}|/9$, and so $Y_{\ceil{\log_{10/9} n}+1} > n$, which is impossible.
\end{proof}

Fix some $2\leq i\leq \ceil{\log_{10/9} n}+1$, such that $|Y_i\setminus Y_{i-1}|\leq  |Y_{i-1}|/9$.  Notice that for every terminal $t\in Y_{i-1}$, $d(t^*,t)\leq 8\Delta \left( \ceil{\log_{10/9} n}+1\right)$, and so from Observation~\ref{obs: weak-triangle-inequality}, for any pair $t,t'\in Y_{i-1}$ of terminals, $d(t,t')\leq 16\Delta\left ( \ceil{\log_{10/9} n}+1\right)+\Delta/2\leq \Delta_0$. Moreover, if we consider any pair $t\in Y_{i-1}$, $t'\in \ttset\setminus Y_i$ of terminals, then $d(t,t')\geq 5\Delta$ must hold, since otherwise, from Observation~\ref{obs: weak-triangle-inequality}, $d(t^*,t')\leq d(t^*,t)+d(t,t')+\Delta/2\leq 8\Delta (i-1)+5\Delta+\Delta/2<8\Delta i$, contradicting the fact that $t'\not\in Y_i$. We remove all terminals of $Y_i$ from $\ttset$, add the set $X=Y_{i-1}$ to $\xset$, and continue to the next iteration.

Notice that in every iteration we discard all terminals in $Y_i\setminus Y_{i-1}$, and add a set containing $|Y_{i-1}|$ terminals to $\xset$, where $|Y_i\setminus Y_{i-1}|\leq  |Y_{i-1}|/9$. Therefore, at the end of the algorithm, $\sum_{X\in \xset}|X|\geq 0.9|\tset|$. The remaining properties of the partition are now immediate to verify.
\end{proof}

We use a parameter $\tau=W^{18/19}$. We say that a set $X\in \xset$ of terminals is \emph{heavy} if $w^*|X|\geq \tau$, and we say that it is \emph{light} otherwise. We say that a demand pair $(s,t)$ is heavy iff both $s$ and $t$ belong to heavy subsets of terminals in $\xset$. We say that it is light if at least one of the two terminals belongs to a light subset, and the other terminal belongs to some subset in $\xset$. Note that a demand pair $(s,t)$ may be neither heavy nor light, for example, if $s$ or $t$ lie in $\tset\setminus\bigcup_{X\in \xset}X$. Let $\mset_0$ be the set of all demand pairs that are neither heavy nor light. Then $|\mset_0|\leq 0.2|\mset|$. We say that Case 1 happens  if there are at least $0.1|\mset|$ light demand pairs, and we say that Case 2 happens otherwise. Notice that in Case 2, at least $0.7|\mset|$ of the demand pairs are heavy.
In the next two sections we handle Case 1 and Case 2 separately.


\label{--------------------------------------------sec: case 1-----------------------------------------------}
\section{Case 1: Light Demand Pairs}\label{sec: case 1}
Let $\mset^L\subseteq \mset$ be the set of all light demand pairs. Recall that $|\mset^L|\geq 0.1|\mset|$.
We assume w.l.o.g. that for every pair $(s,t)\in \mset^L$, $t$ belongs to a light set in $\xset$. We let $S^L,T^L\subseteq\tset$ be the sets of the source and the destination vertices of the demand pairs in $\mset^L$, respectively. Let $\lset\subseteq \xset$ be the set of all light terminal subsets. Recall that we have assumed that every terminal participates in exactly one demand pair. If $(s,t)\in \mset$, then we say that $s$ is the mate of $t$, and $t$ is the mate of $s$.
The goal of this section is to prove the following theorem.

\begin{theorem}\label{thm: case 1}
Let $p^*=\Theta\left (\frac{\alphaWL |\mset^L|}{\tau\log n}\right )$.
There is an efficient algorithm, that computes a routing of at least $\min\set{\Omega(p^*), \Omega\left(\frac{\Delta}{p^*\log n} \right )}$ demand pairs via node-disjoint paths in $G$.
\end{theorem}

We first show that Theorem~\ref{thm: case 1} concludes the proof of Theorem~\ref{thm: main} for Case 1. Indeed, since $|\mset^L|\geq 0.1|\mset|$, we get that $p^*=\Theta\left(\frac{\alphaWL |\mset|}{\tau\log n}\right )=\Theta\left(\frac{w^*|\mset|}{W^{18/19}\log n\log k}\right )=\Theta\left(\frac{W^{1/19}}{\log n\log k}\right )$. Notice also that $\Omega\left(\frac{\Delta}{p^*\log n} \right )=\Omega\left(\frac{W^{2/19}\log k}{W^{1/19}}\right)=\Omega\left(W^{1/19}\log k \right)$. Therefore, the algorithm routes $\Omega\left(\frac{W^{1/19}}{\log n\log k}\right )$ demand pairs via node-disjoint paths. The rest of this section is devoted to proving Theorem~\ref{thm: case 1}.

Our first step is to compute a large subset $\mset_0\subseteq \mset^L$ of light demand pairs, so that, if we denote by $S_0$ and $T_0$ the sets of the source and the destination vertices of the demand pairs in $\mset_0$, then there is a set $\qset$ of $|\mset_0|$ node-disjoint paths connecting the vertices of $S_0$ to a subset of vertices of $T^L$, that we denote by $T'$. Additionally, we ensure that every terminal set $X\in \lset$, $|X\cap T'|\leq 1$, and $|X\cap T_0|\leq 1$. We note that the sets $S_0$ and $T'$ do not necessarily form demand pairs. We will eventually route a subset of the pairs of $\mset_0$.

\begin{theorem}\label{thm: choosing cluster centers}
There is an efficient algorithm to compute a subset $\mset_0\subseteq \mset^L$ of $\kappa_0=\Theta\left(\frac{\alphawl|\mset^L|}{\tau}\right)$ demand pairs, and a subset $T'\subseteq T^L$ of $\kappa_0$ terminals, such that, if we denote by  $S_0$ and $T_0$ the sets of the source and the destination vertices of the demand pairs in $\mset_0$, then:

\begin{itemize}
\item There is a set $\qset$ of $\kappa_0$ node-disjoint paths connecting the vertices of $S_0$ to the vertices of $T'$; and
\item For each set $X\in \lset$, $|X\cap T'|\leq 1$, and $|X\cap T_0|\leq 1$.
\end{itemize}
\end{theorem}

\begin{proof}
Recall that $S^L$ and $T^L$ are the sets of the source and the destination vertices, respectively, of the demand pairs in $\mset^L$. 

We build the following directed flow network $\nset$. Start with graph $G$, and bi-direct all its edges. 
For every light set $X\in \lset$, add two vertices: $s_X$, connecting to every vertex $s\in S^L$, whose mate $t\in X$, and $t_X$, to which every vertex $t\in T^L\cap X$ is connected. Finally, we add a global source vertex $s_0$, that connects with directed edges to every vertex $s_X$ for $X\in \lset$, and a global destination vertex $t_0$, to which every vertex $t_X$ with $X\in \lset$ connects.  The capacities of $s_0$ and $t_0$ are infinite, the capacities of all vertices in $\set{s_X,t_X\mid X\in \lset}$ are $\tau$, and all other vertex capacities are unit. Let $|S^L|=|T^L|=\tk$. We claim that there is an $s_0$--$t_0$ flow of value at least $\alphawl\cdot \tk$ in $\nset$. Assume otherwise. Then there is a set $Y$ of vertices, whose total capacity is less than $\alphawl\cdot \tk$, separating $s_0$ from $t_0$ in $\nset$. Let $A$ denote the subset of vertices of $\set{s_X\mid X\in \lset}$ that lie in $Y$, and let $B=Y\cap \set{t_X\mid X\in \lset}$. Assume that $|A|=a$, $|B|=b$, and assume for simplicity that $a\geq b$ (the other case is symmetric). 
We next build a set $\tilde S\subseteq S^L$ of source vertices as follows: for every set $X\in \lset$ with $s_X\in A$, we add all vertices $s\in S^L$ whose mate belongs to $X$, to set $\tilde S$ (so $\tilde S$ contains all vertices $s\in S^L$, such that some edge originating at a vertex of $A$ enters $s$). Let $\tilde S'=S^L\setminus \tilde S$. Since each cluster $X\in \xset$ contains at most $\tau/w^*$ terminals, $|\tilde S'|\geq \tk-\frac{a\tau}{w^*}$.

Similarly, we let $\tilde T\subseteq T^L$ contain all terminals $t$, such that, if $t\in X$ for cluster $X\in \lset$, then $t_X\in B$. Let $\tilde T'=T^L\setminus \tilde T$. Then $|\tilde T'|\geq \tk-\frac{b\tau}{w^*}\geq \tk-\frac{a\tau}{w^*}$, since we assumed that $a\geq b$. We discard terminals from $\tilde S'$ and $\tilde T'$ until $|\tilde S'| =|\tilde T'|= \ceil{\tk-\frac{a\tau}{w^*}}$ holds. (We note that $\tk-\frac{a\tau}{w^*}>0$, since the total capacity of all vertices in $A$ is at most $a\tau<\alphawl \tk<w^*\tk$, as $\alphawl=\frac{w^*}{512 \cdot \alphasc \cdot \log k}$.)

Let $Y'=Y\setminus (A\cup B)$. Then $Y'$ is a subset of vertices of $G$, and moreover, $G\setminus Y'$ does not contain any path connecting a vertex of $\tilde S'$ to a vertex of $\tilde T'$. Indeed, if $G\setminus Y'$ contains a path $P$ connecting some vertex $s\in \tilde S'$ to some vertex $t'\in \tilde T'$, then it is easy to see that path $P$ can be extended to an $s_0$--$t_0$ path in $\nset\setminus Y$.
Notice that:

\[|Y'|<\alphawl\cdot \tk-a\tau-b\tau\leq\alphawl\cdot \tk-a\tau.\]

But from the well-linkedness of terminals, there must be a set of at least $\alphawl |\tilde S'| \geq \alphawl(\tk-a\tau/w^*)$ paths connecting the vertices of $\tilde S'$ to the vertices of $\tilde T'$. Recall that $\alphawl=\frac{w^*}{512 \cdot \alphasc \cdot \log k}$, and so:

\[\alphawl\left(\tk-\frac{a\tau}{w^*}\right)=\alphawl \tk-\frac{a\tau}{512\alphasc \log k}>\alphawl\tk-a\tau>|Y'|,\]

a contradiction. We conclude that there is an $s_0$--$t_0$ flow $F$ of value at least $\alphawl\tk$ in $\nset$.

Let $\nset'$ be a directed flow network defined exactly like $\nset$, except that now we set the capacity of every vertex in $\set{s_X,t_X\mid X\in \lset}$ to $1$. By scaling the flow $F$ down by factor $\tau$, we obtain a valid $s_0$--$t_0$ flow in $\nset'$ of value at least $\alphawl\tk/\tau$. From the integrality of flow, there is an integral $s_0$--$t_0$ flow in $\nset'$ of value $\ceil{\alphawl \tk/\tau}$. This flow defines a collection $\qset$ of $\kappa_0=\ceil{\alphawl \tk/\tau}=\Omega(\alphawl |\mset^L|/\tau)$ node-disjoint paths in graph $G$, connecting some vertices of $S^L$ to some vertices of $ T^L$. We let $S_0\subseteq  S^L$ and $T'\subseteq  T^L$ be the sets of vertices where the paths of $\qset'$ originate and terminate, respectively, and we let $T_0$ contain all mates of the source vertices in $S_0$. We then set $\mset_0$ to be the set of the demand pairs in which the vertices of $S_0$ participate. It is easy to verify that for each set $X\in \lset$, $|T'\cap X|,|T_0\cap X|\leq 1$ from the definition of the network $\nset'$.
\end{proof}

We assume that $|\mset_0|\geq 1000$, as otherwise we can route a single demand pair, and obtain a solution routing $\Omega\left (\frac{\alphawl|\mset^L|}{\tau}\right)$ demand pairs.

Recall that every set $X\in \xset$ contains at most one terminal from $T_0$. Since $|S_0|=|T_0|$, there is some set $X_0\in T_0$, that contains exactly one terminal $t_0\in T_0$, and at most one additional terminal from $S_0$. We will view $t_0$ as our ``center'' terminal, and we discard from $\mset_0$ the demand pairs in which $t_0$ and the terminal in $S_0\cap X_0$ (if it exists) participate.

The main tool in our algorithm for Case 1 is a crossbar, that we define below. Let $\Delta_1=\floor{\Delta/6}$ and $\Delta_2=\floor{\Delta_1/3}$. Given a shell $\zset(t)=(Z_1(t),\ldots,Z_{\Delta_1}(t))$ of depth $\Delta_1$ around some terminal $t$, we will always denote by $D^*(t)=D(Z_{\Delta_1}(t))$, and by $\tilde D(t)=D(Z_{\Delta_2}(t))$. We will view the cycles $Z_1(t),\ldots,Z_{\Delta_2}(t)$ as the ``inner'' part of the shell $\zset(t)$.
The crossbar is defined with respect to some large enough subset $\mset^*\subseteq \mset_0$ of demand pairs (see Figure~\ref{fig: crossbar}). 

\begin{definition}
Suppose we are given a subset $\mset^*\subseteq \mset_0$ of demand pairs and an integer $p\geq 1$. Let $S^*$ and $T^*$ be the sets of all source and all destination vertices participating in the demand pairs of $\mset^*$, respectively. A $p$-crossbar for $\mset^*$ consists of:

\begin{itemize}
\item For each terminal $t\in T^*\cup \set{t_0}$, a shell $\zset(t)$ of depth $\Delta_1$ around $t$, such that for all $t,t'\in T^*\cup \set{t_0}$, $D^*(t)\cap D^*(t')=\emptyset$, and for all $s\in S^*$ and $t'\in T^*\cup \set{t_0}$, $s\not\in D^*(t')$; and

\item For each $v\in S^*\cup T^*$, a collection $\pset(v)$ of paths, such that:

\begin{itemize}
\item For each $s\in S^*$, $\pset(s)$ contains exactly one path, connecting $s$ to a vertex of $C_{t_0}$;

\item For each $t\in T^*$, $\pset(t)$ contains exactly $p$ paths, where each path connects a vertex of $C_t$ to a vertex of $C_{t_0}$; and

\item All paths in $\pset=\bigcup_{v\in S^*\cup T^*}\pset(v)$ are node-disjoint from each other.
\end{itemize}
\end{itemize}
\end{definition}

In order to route a large subset of the demand pairs in $\mset^*$, we need a crossbar with slightly stronger properties, that we call a \emph{good crossbar}, and define below.

\begin{definition}
Given a set $\mset^*\subseteq \mset_0$ of demand pairs, where $S^*,T^*$ are the sets of the source and the destination vertices of the demand pairs in $\mset^*$ respectively, and an integer $p\geq 1$, a $p$-crossbar $\left(\set{\zset(t)}_{t\in T^*\cup \set{t_0}},\set{\pset(v)}_{v\in S^*\cup T^*}\right )$ is a \emph{good crossbar}, if the following additional properties hold:

\begin{properties}{C}
\item For all $t\in T^*$ and all $v\in (S^*\cup T^*)\setminus\set{t}$, all paths in $\pset(v)$ are disjoint from $\tD(t)$. \label{prop3: paths don't touch inner discs}

\item For all $t\in T^*$, all paths in $\pset(t)$ are monotone with respect to $(Z_1(t),\ldots,Z_{\Delta_2}(t))$. Also, for all $v\in S^*\cup T^*$, all paths in  $\pset(v)$ are monotone with respect to  $(Z_1(t_0),\ldots,Z_{\Delta_2}(t_0))$.\label{prop4: in own disc contiguous path}

\item We can partition $Z_{\Delta_2}(t_0)$ into a collection $\Sigma=\set{\sigma(v)\mid v\in S^*\cup T^*}$ of $|S^*|+|T^*|$ disjoint segments, such that for all $v,v'\in S^*\cup T^*$ with $v\neq v'$, $\sigma(v)\cap \pset(v')=\emptyset$. \label{prop-last: contiguously hit disc of t0}
\end{properties}
\end{definition}

\begin{figure}[h]
 \centering
\scalebox{0.7}{\includegraphics{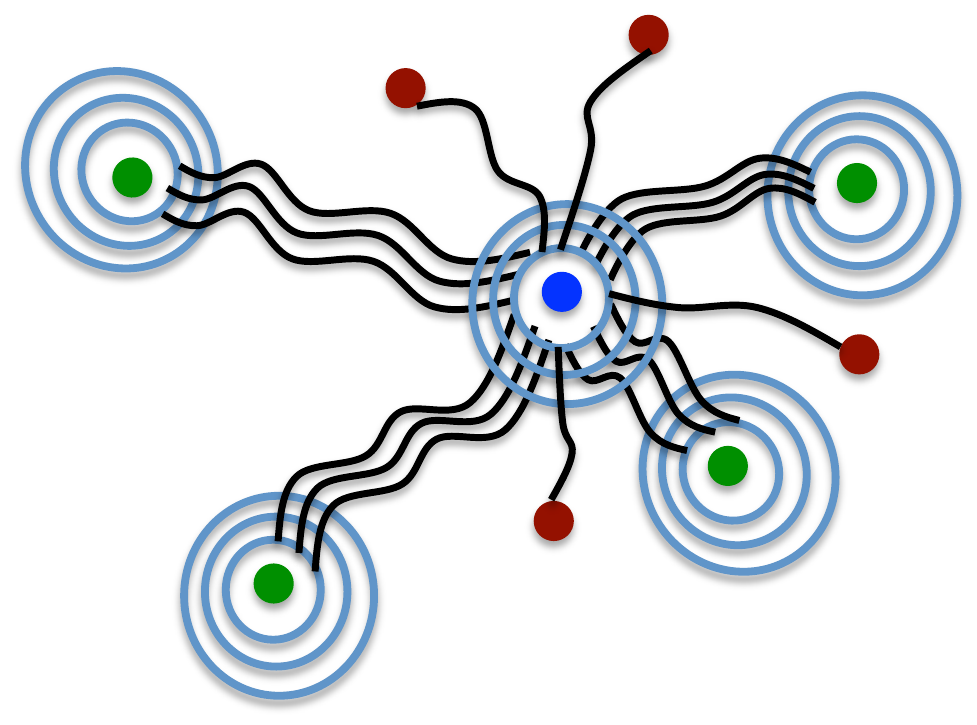}}\caption{A crossbar. The center vertex $t_0$ is shown in blue, the vertices of $S^*$ in red, and the vertices of $T^*$ in green.\label{fig: crossbar}}
\end{figure}

The following theorem shows that, given a $p$-crossbar in $G$, where $p$ is large enough, we can route many demand pairs in $\mset^*$.

\begin{theorem}\label{thm: find the routing}
Suppose we are given a subset $\mset^*\subseteq\mset_0$ of $\kappa$ demand pairs, where $S^*$ and $T^*$ are the sets of all source and all destination vertices of the demand pairs in $\mset^*$, respectively. Assume further that we are given a good $p$-crossbar $\left(\set{\zset(t)}_{t\in T^*\cup \set{t_0}},\set{\pset(v)}_{v\in S^*\cup T^*}\right )$ for $\mset^*$, and let $q=\min\set{\Delta_2, \floor{(p-1)/2},\ceil{\kappa/2}}$. Then there is an efficient algorithm that routes at least $q$ demand pairs in $\mset^*$ via node-disjoint paths in $G$.
\end{theorem} 
\begin{proof}
We can assume without loss of generality that for every terminal $t\in T^*$, for every path $P\in \pset(t)$, and for every $1\leq j\leq \Delta_2$, $P\cap Z_j(t)$ consists of a single vertex. In order to see this, recall that for each such terminal $t$, path $P$ and value $j$, $P\cap Z_j(t)$ is a path, from Property~(\ref{prop4: in own disc contiguous path}). We contract each such path $P\cap Z_j(t)$ into a single vertex. We still maintain a good $p$-crossbar for $\mset^*$ in the resulting graph, and any routing of a subset of the demand pairs in $\mset^*$ in the new graph via node-disjoint paths immediately gives a similar routing of the same subset of the demand pairs in the original graph. Using a similar reasoning, we assume without loss of generality that for every path $P\in \bigcup_{v\in S^*\cup T^*}\pset(v)$, for every $1\leq j\leq \Delta_2$, $P\cap Z_j(t_0)$ is a single vertex.

Fix an arbitrary source vertex $s\in S^*$, and consider the unique path $P(s)\in \pset(s)$. For all $1\leq j\leq \Delta_2$, let $v_j$ be the unique vertex in $Z_j(t_0)\cap P(s)$, and let $e_j$ be the edge of $Z_j(t_0)$ incident to $v_j$, as we traverse $Z_j(t_0)$ starting from $v_j$ in the clock-wise direction. Let $R_j$ be the path obtained by deleting $e_j$ from $Z_j(t_0)$. We view this path as directed in the counter-clock-wise direction along $Z_j(t_0)$, thinking of this as the left-to-right direction. Once we process every cycle $Z_j(t_0)$ for $1\leq j\leq \Delta_2$ in this fashion, we obtain a collection $R_1,\ldots,R_{\Delta_2}$ of paths. Our routing will in fact only use the paths $R_1,\ldots,R_q$.

Let $\pset=\bigcup_{v\in S^*\cup T^*}\pset(v)$.
For each $1\leq j\leq \Delta_2$, and for each path $P\in \pset$, let $u(P,j)$ be the unique vertex in $P\cap R_j$. The vertices $u(P,\Delta_2)$ define a natural left-to-right ordering $\oset$ of the paths in $\pset$: for $P,P'\in \pset$, we denote $P\prec P'$ iff $u(P,\Delta_2)$ lies to the left of $u(P',\Delta_2)$ on $R_{\Delta_2}$. Notice that, since the paths of $\pset$ are monotone with respect to $Z_1(t_0),\ldots,Z_{\Delta_2}(t_0)$,  for every pair $P,P'\in \pset$ of paths with $P\prec P'$, for every $1\leq j\leq \Delta_2$, $u(P,j)$ lies to the left of $u(P',j)$ on $R_j$.
From Property~(\ref{prop-last: contiguously hit disc of t0}) of the crossbar, for each terminal $v\in S^*\cup T^*$, all paths in $\pset(v)$ appear consecutively in the ordering $\oset$. Therefore, we obtain a natural left-to-right ordering $\oset'$ of the terminals: we say that terminal $v$ lies to the left of terminal $v'$ iff for all $P\in \pset(v)$, $P'\in \pset(v')$, $P\prec P'$.

 We say that a demand pair $(s,t)\in \mset^*$ is a left-to-right pair, if $s$ appears before $t$ in $\oset'$, and we say that it is a right-to-left pair otherwise. At least $\ceil{|\mset^*|/2}$ of the pairs belong to one of these two types, and we assume w.l.o.g. that at least $\ceil{|\mset^*|/2}$ of the pairs are of the left-to-right type (otherwise we reverse the direction of the paths $\set{R_j}_{j=1}^{\Delta_2}$, and the orderings $\oset,\oset'$). We discard from $\mset^*$ all right-to-left demand pairs, and we update the sets $S^*$ and $T^*$ accordingly. We discard additional demand pairs from $\mset^*$ as needed, until $|\mset^*|=q$ holds.

Consider the set $S^*$ of the source vertices. We assume that $S^*=(s_1,\ldots,s_q)$, with the sources indexed in the order of their appearance in $\oset'$. We define a set $\pset^*=\set{P_1^*,\ldots,P_q^*}$ of node-disjoint paths, where for each $1\leq i\leq q$, path $P_i^*$ connects $s_i$ to $t_i$. In order to do so, we first construct an initial set of paths, that we later modify.

The initial set $\pset^*$ of paths is constructed as follows. Fix some $1\leq i\leq q$. The initial path $P^*_i$ is the concatenation of the segment of the unique path $P(s_i)\in \pset(s_i)$ from $s_i$ to $u(P,i)$, and the segment of $R_i$ from $u(P,i)$ to the last (right-most) vertex of $R_i$ (see Figure~\ref{fig: initial-routing}).

Clearly, all paths in $\pset^*$ are node-disjoint, and each path $P^*_i$ originates at $s_i$, but it does not terminate at $t_i$. We now modify the paths in $\pset^*$ in order to fix this. We process the terminals $t_i\in T^*$ in their right-to-left order (the reverse order of their appearance in $\oset'$). When we process terminal $t_i$, we will ensure that path $P_i^*$ terminates at $t_i$, while all paths in $\pset^*$ remain disjoint. Intuitively, in order to do so, we send the paths $P_{i+1}^*,\ldots,P_{q}^*$ ``around'' the terminal $t_i$, using the paths of $\pset(t_i)$ and the $q$ innermost cycles of the shell $\zset(t_i)$.

Formally, throughout the algorithm, we maintain the following invariants:

\begin{itemize}
\item The paths in $\pset^*=\set{P_1^*,\ldots,P_q^*}$ always remain disjoint.
\item Once terminal $t_i$ is processed, in all subsequent iterations, path $P_i^*$ connects $s_i$ to $t_i$.
\item Let $t'$ be the last terminal processed, and for each $1\leq i\leq q$, let $v^i$ be the leftmost vertex of $R_i$ that lies on any path in $\pset(t')$. Let $P(s_i)$ be the unique path in $\pset(s_i)$. Then if $s_i$ lies to the left of $t'$ in $\oset'$, the current path $P^*_i$ contains the segment of $R_i$ between $u(P(s_i),i)$ and $v^i$ (possibly excluding $v^i$).
\end{itemize}

We now describe an iteration when the terminal $t_i$ is processed. Let $\set{Q_1,\ldots,Q_{2q+1}}$ be any set of $2q+1$ distinct paths in $\pset(t_i)$, indexed according to their ordering in $\oset$ (recall that from our definition of $q$, $|\pset(t_i)|=p\geq 2q+1$). 
We do not modify the paths $P_1^*,\ldots,P_{i-1}^*$ at this step. Path $P_i^*$ is modified as follows. From our invariants, vertex $u(Q_{q+1},i)$ belongs to $P_i^*$. We discard the last segment of $P_i^*$ that starts at $u(Q_{q+1},i)$, and then concatenate the remaining path with the segment of $Q_{q+1}$ from its starting point $a\in C_{t_i}$ to $u(Q_{q+1},i)$. In this way, we obtain a path originating at $s_i$ and terminating at $a$. Since the sub-graph of $G$ induced by $V(D(t_i))$ is connected, we can extend this path inside $D(t_i)$, so it now terminates at $t_i$.

\begin{figure}[h]
\centering
\subfigure[The initial paths, together with the inner shells of the terminals in $T^*$. Recall that for all $t\in T^*$ and $v\in S^*\cup T^*$ with $t\neq v$, all paths in $\pset(v)$ are disjoint from the inner shell of $t$.]{\scalebox{0.35}{\includegraphics{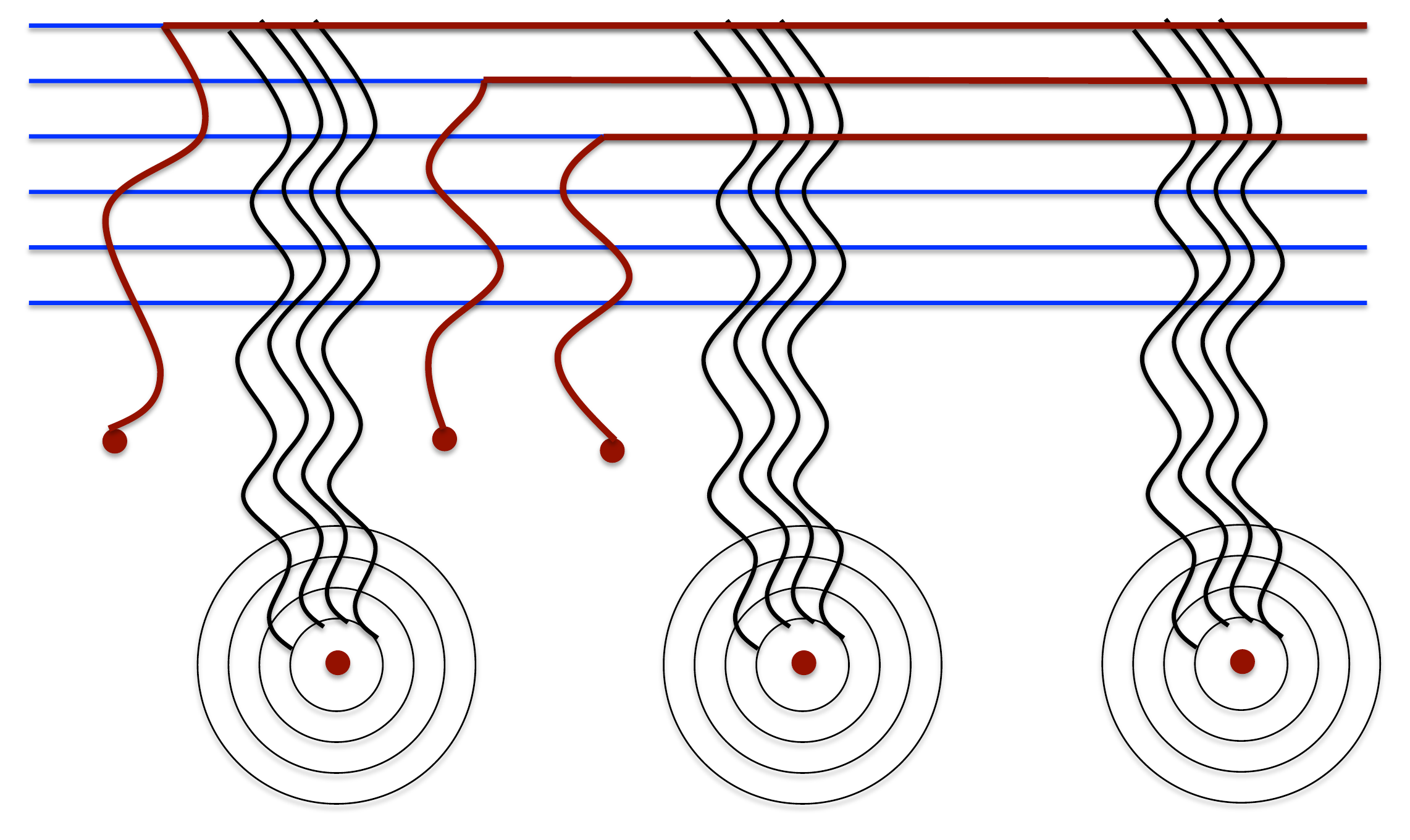}}\label{fig: initial-routing}}
\hspace{1cm}
\subfigure[Modifying the routing]{
\scalebox{0.3}{\includegraphics{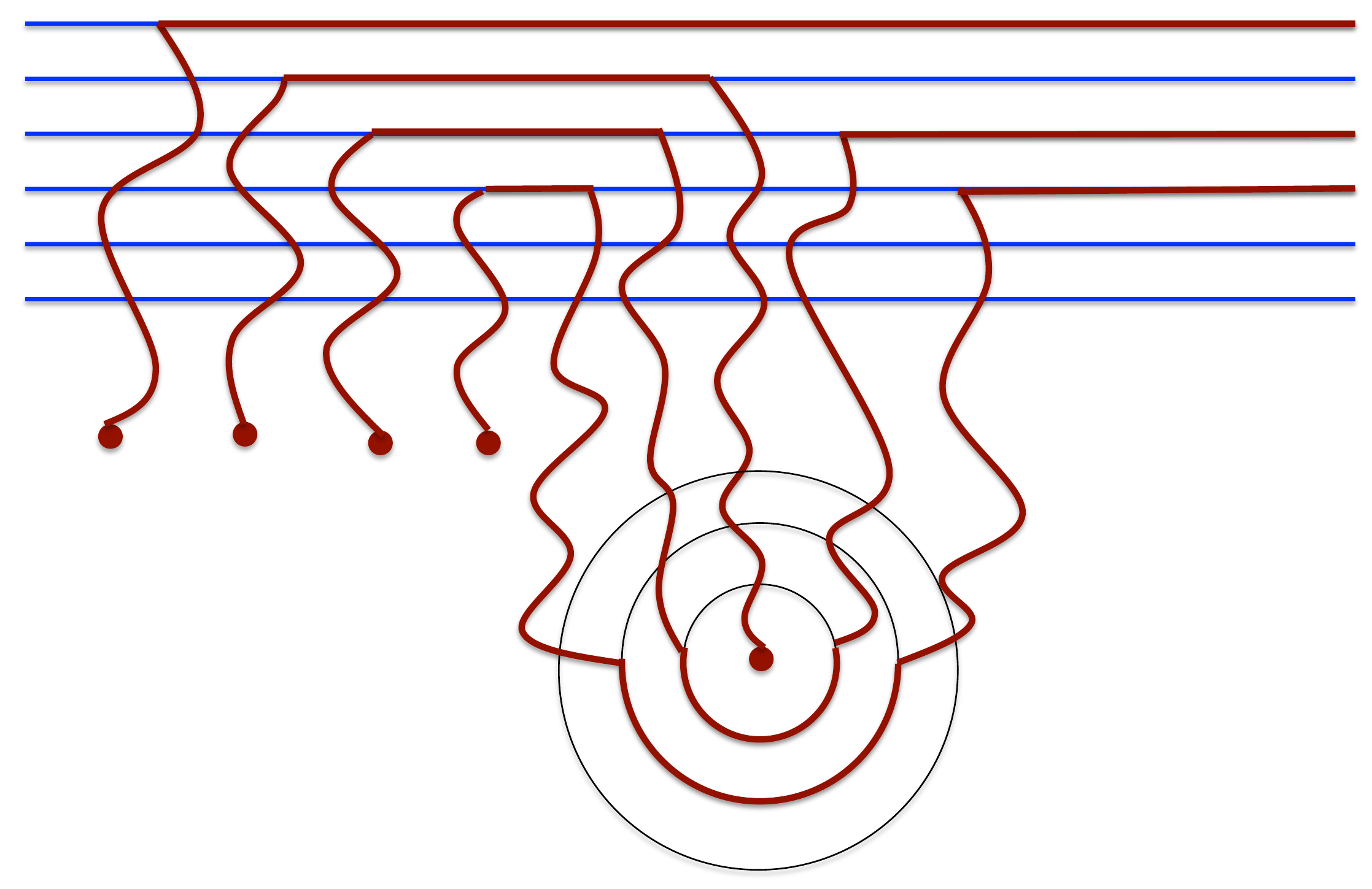}}\label{fig: routing-modified}}
\caption{Constructing the routing \label{fig: Routing}}
\end{figure}

For all $i'>i$, if $s_{i'}$ lies to the left of $t_i$, then path $P_{i'}^*$ is modified to ``go around'' $t_i$, using the paths $Q_{i'}, Q_{2q+2-i'}$ (that is, the $i'$th paths from the right and from the left), and the cycle $Z_{i'}(t_i)$, as follows. Let $Q'$ be the segment of path $Q_{i'}$, starting from the unique vertex of $Q_{i'}\cap Z_{i'}(t_i)$ (that we denote by $x$), to vertex $u(Q_{i'},i')$. Let $Q''$ be the segment of path $Q_{2q+2-i'}$, starting from the unique vertex of $Q_{2q+2-i'}\cap Z_{i'}(t_i)$ (that we denote by $y$), to vertex $u(Q_{2q+2-i'},i')$. Consider now the two segments $\sigma,\sigma'$ of $Z_{i'}(t_i)$, whose endpoints are $x$ and $y$. Exactly one of these segments (say $\sigma$) is disjoint from the new path $P_i^*$. We remove from $P_{i'}^*$ the segment of $R_{i'}$ between $u(Q_{i'},i')$ and $u(Q_{2q+2-i'},i')$, and replace it with the concatenation of $Q',\sigma$ and $Q''$ (see Figure~\ref{fig: routing-modified}). It is easy to see that all invariants continue to hold.
\end{proof}

In order to complete the proof of Theorem~\ref{thm: case 1}, we prove the following theorem that allows us to construct a large enough crossbar.

\begin{theorem}\label{thm: crossbar construction}
There is an efficient algorithm that either finds a routing of $\Omega(\kappa_0)$ demand pairs of $\mset_0$ via node-disjoint paths, or computes a subset $\mset^*\subseteq\mset_0$ of $\Omega(\kappa_0/\log n)$ demand pairs, and a good $p$-crossbar for $\mset^*$, with $p=\Omega(\Delta/\kappa_0)$.
\end{theorem}

Recall that $\kappa_0=\Theta\left(\frac{\alphawl|\mset^L|}{\tau}\right)$. 
 Letting $p^*=\frac{\kappa_0}{\log n}=\Omega\left(\frac{\alphawl|\mset^L|}{\tau\log n}\right)$,  from Theorem~\ref{thm: find the routing}, we can efficiently find a routing of at least $\min\set{\Omega(p^*),\Omega\left(\frac{\Delta}{p^*\log n}\right)}$ demand pairs, concluding the proof of Theorem~\ref{thm: case 1}. The rest of this section is dedicated to proving Theorem~\ref{thm: crossbar construction}. We do so in two steps. In the first step, we either find a routing of $\Omega(\kappa_0)$ demand pairs in $\mset_0$, or compute a large subset $\mset_3\subseteq \mset_0$ of demand pairs, and a $p$-crossbar with respect to $\mset_3$ (that may not be good). If the outcome of the first step is the former, then we terminate the algorithm and return the resulting routing. Otherwise, in the following step, we modify the set $\mset_3$ and the crossbar in order to obtain the final set $\mset^*$, and a good crossbar with respect to $\mset^*$.

\label{----------------------------------subset:  basic crossbar constr-------------------------------}
\subsection{Constructing a Basic Crossbar}\label{subsec: basic crossbar constr}
In this section we prove the following theorem.

\begin{theorem}\label{thm: basic crossbar construction}
There is an efficient algorithm that either computes a routing of $\Omega(\kappa_0)$ demand pairs in $\mset_0$ via node-disjoint paths, or finds a subset $\mset_3\subseteq\mset_0$ of $\kappa_3=\Omega(\kappa_0)$ demand pairs, and a $p_1$-crossbar for $\mset_3$, for $p_1=\Omega(\Delta/\kappa_0)$.
\end{theorem}

We construct the crossbar in two steps. Let $S_0,T_0$ be the sets of all source and all destination vertices corresponding to the demand pairs in $\mset_0$, respectively. In the first step, we construct a shell $\zset(t)$ of depth $\Delta_1$ around every terminal $t\in T_0\cup \set{t_0}$, so that their corresponding discs $D^*(t)$ are disjoint. We then select a subset $\mset_1\subseteq \mset_0$ of demand pairs, so that, if we denote by $S_1$ and $T_1$ the sets of all source and all destination vertices of the pairs in $\mset_1$, respectively, then for all $s\in S_1$ and $t'\in T_1\cup \set{t_0}$, $s\not\in D^*(t')$. In the second step, we select a subset $\mset_3\subseteq \mset_1$ of demand pairs, and complete the construction of the crossbar by computing the sets $\pset(v)$ of paths for all $v\in \tset(\mset_3)$.

\paragraph{Step 1: the Shells.}

Consider any terminal $t\in T_0\cup\set{t_0}$, and let $t'\neq t$ be any other terminal in $T_0\cup \set{t_0}$. We let $F_t$ be the face of the drawing of $G$ on the sphere containing the image of $t'$. Notice that from Theorem~\ref{thm: choosing cluster centers}, $t$ and $t'$ lie in different sets of $\xset$, and so, from Theorem~\ref{thm: terminal clusters}, $d(t,t')\geq 5\Delta$. In particular, for every vertex $v$ on the boundary of $F_t$, $\dface(v,t)>\Delta_1+1$. We can now construct a shell $\zset(t)$ of depth $\Delta_1$ around $t$, with respect to $F_t$. 

Additionally, for every terminal $t\in T'$, we construct a shell $\zset(t)$ of depth $\Delta_1$ around $t$ similarly, letting $F_t$ be the face incident on any terminal $t'\in T'\setminus\set{t}$.
This concludes the construction of the shells. We need the following lemma.

\begin{lemma}\label{lem: disc disjointness}
For all $t_1,t_2\in T_0\cup\set{t_0}$, if $t_1\neq t_2$, then $D^*(t_1)\cap D^*(t_2)=\emptyset$. 
Moreover, for all $s\in S_0$, $s\not\in D^*(t_0)$.
\end{lemma}

\begin{proof}
Let $t_1,t_2\in T_0\cup\set{t_0}$ be any two distinct terminals, and assume for contradiction that $D^*(t_1)\cap D^*(t_2)\neq\emptyset$. We need the following observation.
 
 \begin{observation}
  Either $D_{t_2}\subseteq D^*(t_1)$, or $D_{t_1}\subseteq D^*(t_2)$ must hold.
  \end{observation}
  \begin{proof}
  If $D^*(t_2)\subseteq D^*(t_1)$, or $D^*(t_1)\subseteq D^*(t_2)$, then we are done, so assume that neither of these holds.  From Observation~\ref{obs: cycles of shells disjoint}, $Z_{\Delta_1}(t_1)\cap Z_{\Delta_1}(t_2)=\emptyset$. Since $D^*(t_1)\cap D^*(t_2)\neq\emptyset$, the only other possibility is that the union of the discs $D^*(t_1)$, $D^*(t_2)$ is the entire sphere, and the boundary of each disc is contained in the other disc. In particular, $Z_{\Delta_1}(t_2)\subseteq D^*(t_1)$. 
  We claim that all vertices of $C_{t_2}$ must lie in disc $D^*(t_1)$. Indeed, assume for contradiction that some vertex $v\in V(C_{t_2})$ does not lie in $D^*(t_1)$. Since all vertices of $Z_{\Delta_1}(t_2)$ lie in $D^*(t_1)$, there is a $G$-normal curve $\gamma$ of length at most $\Delta/2+\Delta_1+2$, connecting $v$ to some vertex of $Z_{\Delta_1}(t_2)$, and this curve must contain some vertex $v'\in Z_{\Delta_1}(t_1)$. But then there is a $G$-normal curve of length at most $\Delta_1+1$ connecting $v'$ to a vertex of $C_{t_1}$, so $d(t_1,t_2)\leq 2\Delta_1+\Delta/2+2<5\Delta$, a contradiction. Therefore, all vertices of $C_{t_2}$ appear in $D^*(t_1)$. 
  
Let $t_1'$ be the terminal we have used to define the face $F_{t_1}$, so $t_1'$ lies outside $D^*(t_1)$.  From Property~(\ref{prop: Yt out}) of the shells, it is easy to see that $D_{t_1'}$ is disjoint from $D^*(t_1)$.  If $D_{t_2}\not\subseteq D^*(t_1)$, then $t_1'$ must lie inside $D_{t_2}$. Moreover, since $D_{t_1'}\cap D^*(t_1)=\emptyset$, we get that $D_{t_1'}\subsetneq D_{t_2}$, contradicting the construction of enclosures. 
 \end{proof}
 
 We assume w.l.o.g. that $D_{t_2}\subseteq D^*(t_1)$. 
Since for all $0\leq j,j'\leq \Delta_1$, $Z_j(t_1)\cap Z_{j'}(t_2)=\emptyset$, there must be some index $1\leq j\leq \Delta_1$, such that $D_{t_2}\subseteq \dnot(Z_j(t_1))\setminus D(Z_{j-1}(t_1))$. 
Let $U_j$ be the set of all vertices of $G$ lying in $\dnot(Z_j(t_1))\setminus D(Z_{j-1}(t_1))$, and let $\rset_j$ be the set of all connected components in $G[U_j]$. 
Since $V(D_{t_2})$ induces a connected graph in $G$, there is some connected component $R\in \rset_j$, such that $V(D_{t_2})\subseteq V(R)$. From Theorem~\ref{thm: curves around CC's}, there is a $G$-normal curve $\gamma(R)$ of length at most $2j+\Delta/2+1\leq 2\Delta_1+\Delta/2+1<\Delta$, such that the disc $D(\gamma(R))$ contains $R$, and $D(\gamma(R))\subseteq D^*(t_1)$. 

Let $t_1'$ be the terminal we have used to define the face $F_{t_1}$, so $t_1'$ lies outside $D^*(t_1)$ and $D_{t_1'}$ is disjoint from $D^*(t_1)$. 
Curve $\gamma(R)$ then separates $D_{t_1'}$ from $D_{t_2}$,  and clearly $D_{t_1'}\cap D_{t_2}=\emptyset$. But from our construction of enclosures, 
 there is a set of $\Delta$ node-disjoint paths connecting $V(C_{t_2})$ to $V(C_{t_1'})$, all of which must intersect $\gamma(R)$, contradicting the fact that $\gamma(R)$ contains fewer than $\Delta$ vertices.
 

 Assume now for contradiction that some vertex $s\in S_0$ lies in $D^*(t_0)$. Recall that  $s$ must lie in some set of $\xset\setminus\set{X_0}$ (where $X_0$ is the set containing $t_0$), and so $d(s,t_0)\geq 5\Delta$. Therefore, for all $1\leq j\leq \Delta_1$, $V(C_s)\cap V(Z_j(t_0))=\emptyset$ must hold. Moreover, $D^*(t_0)\not\subseteq D_s$, since otherwise, from the definition of enclosures, $D_{t_0}=D_s$ and $d(t_0,s)=1$, a contradiction. Since $D^*(t_0)\cap D_s\neq \emptyset$, we get that $D_s\subseteq D^*(t_0)$. The rest of the proof follows the same reasoning as before.\end{proof}

We say that a demand pair $(s,t)\in \mset_0$ is a type-1 pair if $s\in D^*(t)$, and we say that it is a type-2 demand pair otherwise. 
Notice that we can route all type-1 demand pairs via node-disjoint paths, where each pair $(s,t)$ is routed in $G[V(D^*(t))]$ (since from Observation~\ref{obs: connectivity of shells}, this graph is connected).
Therefore, if at least half the demand pairs in $\mset_0$ are type-1 pairs, then we obtain a routing of $\Omega(\kappa_0)$ demand pairs and terminate the algorithm. We assume from now on that at least $\kappa_0/2$ demand pairs in $\mset_0$ are type-2 demand pairs. We next build a directed conflict graph $H$, that contains a vertex $v(s,t)$ for each type-2 demand pair $(s,t)\in \mset_0$. There is a directed edge from $v(s,t)$ to $v(s',t')$ iff $s\in D^*(t')$. Since the discs $\set{D^*(t)}_{t\in T_0}$ are mutually disjoint, the out-degree of every vertex of $H$ is at most $1$, and the total average degree of every vertex in $H$ is at most $2$. Since $|V(H)|\geq \kappa_0/2$, using standard techniques, we can find an independent set $I$ of $\kappa_1=\Omega(\kappa_0)$ vertices in $H$. We let $\mset_1=\set{(s,t)\in \mset_0\mid v(s,t)\in I}$. Let $S_1$ and $T_1$ be the sets of source and destination vertices, respectively, of the demand pairs in $\mset_1$. Then for all $t,t'\in T_1$, $D^*(t)\cap D^*(t')=\emptyset$, and for all $s\in S_1,t'\in T_1$, $s\not\in D^*(t')$.

\paragraph{Step 2: the Paths.}

Recall that $\kappa_1=\Omega(\kappa_0)=\Theta\left(\frac{\alphawl |\mset|}{\tau}\right )=\Theta \left(\frac{W^{1/19}}{\log k}\right )$, while $\Delta=\Theta(W^{2/19})$. We will assume throughout that $\kappa_1<\Delta/12$, as otherwise $\kappa_0$ is bounded by some constant, and routing a single demand pair is sufficient.
The following lemma will be used in order to compute the set $\mset_3\subseteq \mset_1$ of the demand pairs.

\begin{lemma}
There is a set $\pset^S$ of $\floor{\frac{\kappa_1}2}-1$ node-disjoint paths, connecting vertices of $S_1$ to vertices of $C_{t_0}$.
\end{lemma}

\begin{proof}
Assume otherwise. Then there is a set $Y$ of at most $\floor{\frac{\kappa_1}2}-2$ vertices, such that $G\setminus Y$ contains no path from a vertex of $S_1\setminus Y$ to a vertex of $V(C_{t_0})\setminus Y$.

Let $\qset'\subseteq \qset$ be a subset of the paths, computed in Theorem~\ref{thm: choosing cluster centers}, that connect every terminal $s\in S_1$ to some terminal of $T'$, so $|\qset'|=\kappa_1$. Let $\qset''\subseteq \qset'$ be the set of paths that contain no vertices of $Y$, so $|\qset''|\geq \kappa_1/2+2$. Recall that the terminals of $T'$ all belong to distinct sets of $\xset$, and so the discs $\set{D_{t'}}_{t'\in T'}$ are mutually disjoint. Moreover, there is at most one terminal $t''$ in $T'\cap X_0$, where $X_0\in \xset$ is the set containing $t_0$. Then there must be at least one path $Q\in \qset''$, such that, if we denote by $s\in S_1$, $t'\in T'$ its two endpoints, then $t'\not\in X_0$, and $D_{t'}$ contains no vertex of $Y$. We now show that there is a path connecting $s$ to a vertex of $C_{t_0}$ in $G\setminus Y$, leading to a contradiction.

From the construction of enclosures, there is a set $\pset(t',t_0)$ of $\Delta$ node-disjoint paths, connecting the vertices of $C_{t_0}$ to the vertices of $C_{t'}$. Since we have assumed that $\kappa_1<\Delta/12$, at least one such path, say $P$, is disjoint from $Y$.  Let $v'$ be the endpoint of $P$ lying on $C_{t'}$. Assume first that $s\not\in D_{t'}$. Then $Q$ must contain a vertex of $C_{t'}$, that we denote by $v$.  Since $G[V(D_{t'})]$ is connected, there is a path $P'\subseteq G[V(D_{t'})]$ connecting $v$ to $v'$, and this path is disjoint from $Y$. The union of $Q,P$ and $P'$ then contains a path connecting $s$ to a vertex of $V(C_{t_0})$, which is disjoint from $Y$, a contradiction. 

Assume now that $s\in D_{t'}$. Then there is a path $P'\subseteq G[V(D_{t'})]$, connecting $s$ to $v'$. The union of $P$ and $P'$ then contains a path connecting  $s$ to a vertex of $V(C_{t_0})$, which is disjoint from $Y$, a contradiction. 
\end{proof}

Let $S_2\subseteq S_1$ be the set of the source vertices where the paths of $\pset^S$ originate, let $T_2\subseteq T_1$ be the set of their mates,  and let $\mset_2\subseteq \mset_1$ be the set of the demand pairs in which the vertices of $S_2$ participate. Notice that from our choice of $t_0$, there is at most one demand pair $(s,t)\in \mset_1$, where $t$ belongs to the subset $X_0\in \xset$ containing $t_0$. If such a pair belongs to $\mset_1$, then we discard it from $\mset_1$, and update $S_1$ and $T_1$ accordingly. Let $\kappa_2=|\mset_2|$, so $\kappa_2\geq \kappa_1/2-3\geq \Omega(\kappa_0)$. We now use the following theorem to compute the subset $\mset_3$ of the demand pairs and construct the corresponding paths for the crossbar.

\begin{theorem}\label{thm: paths for basic crossbar}
There is an efficient algorithm to compute a subset $\mset_3\subseteq\mset_2$ of $\floor{\kappa_2/2}$ demand pairs, and a $p$-crossbar for $\mset_3$, where $p=\Omega(\Delta/\kappa_2)$.
%
\end{theorem}

\begin{proof}
Consider any terminal $t\in T_2$, and let $X_0\in\xset$ be the set containing $t_0$. Since $t\not\in X_0$, $d(t,t_0)\geq 5\Delta$, and in particular, $D_t\cap D_{t_0}=\emptyset$. Therefore, there is a set $\tpset(t)$ of $\Delta$ node-disjoint paths, connecting the vertices of $V(C_t)$ to the vertices of $V(C_{t_0})$.

We construct a directed flow network $\nset$, by starting from $G$, and bi-directing each of its edges.
We then introduce several new vertices: a global source vertex $s^*$ of infinite capacity; an additional vertex vertex $\tilde s$ of capacity $\floor{\kappa_2/2}$, and, for each terminal $t\in T_2$ a vertex $s_t$ of capacity $\floor{\Delta/(2\kappa_2)}$. We connect $s^*$ to $\tilde s$, and to every vertex in $\set{s_t}_{t\in T_2}$. Vertex $\tilde s$ in turn connects to every vertex $s\in S_2$, and for each $t\in T_2$, vertex $s_t$  connects to every vertex in $V(C_t)$. Finally, we introduce a global destination vertex $t^*$ of infinite capacity, and connect every vertex in $V(C_{t_0})$ to it.  All vertices whose capacities have not been set so far have capacity $1$. 

Let $B$ be the total capacity of all vertices in $\set{\tilde s}\cup\set{s_t}_{t\in T_2}$. It is easy to see that there is a flow of value at least $B$ in the resulting network: we send $\half$ flow unit along each path in set $\pset^S$ (discarding one path if $|S_2|$ is odd), and $\frac{1}{2\kappa_2}$ flow units on each path in $\bigcup_{t\in T_2}\tpset(t)$ (if the flow through some vertex $s_t$ is too high due to the rounding of the capacities down, we simply lower the flow on one of the corresponding paths in $\tpset(t)$). Since the paths in $\pset^S$ are node-disjoint, the total flow on these paths causes congestion at most $\half$ on the vertices whose capacity is $1$. Since for each $t\in T_2$, the paths in $\tpset(t)$ are node-disjoint, the total congestion caused by the flow on paths in $\bigcup_{t\in T_2}\tpset(t)$ on vertices whose capacity is $1$ is at most $1/2$. Therefore, we obtain a valid $s^*$-$t^*$ flow of value $B$. From the integrality of flow, there is an integral flow of the same value.

 This flow gives a collection $\pset^1$ of $\floor{\kappa_2/2}$ paths connecting some vertices of $S_2$ to some vertices of $V(C_{t_0})$, and, for each $t\in T_2$, a collection $\pset'(t)$ of $\floor{\frac{\Delta}{2\kappa_2}}$ paths connecting the vertices of $V(C_t)$ to the vertices of $V(C_{t_0})$, such that all paths in $\pset^1\cup \set{\pset'(t)\mid t\in T_2}$ are mutually node-disjoint. We let $S_3\subseteq S_2$ be the set of the source vertices where the paths of $\pset^1$ originate. For each $s\in S_3$, let $P(s)\in \pset^1$ be the unique path originating at $s$, and let $\pset(s)=\set{P(s)}$. Let $\mset_3\subseteq \mset_2$ contain all demand pairs whose source belongs to $S_3$, and let $T_3\subseteq T_2$ be the set of the corresponding destination vertices. Then $|\mset_3|=\floor{\kappa_2/2}$, and all paths in set $\pset'=\bigcup_{v\in S_3\cup T_3}\pset(v)$ are mutually disjoint. For each terminal $t\in T_3$, set $\pset(t)$ contains $\Omega(\Delta/\kappa_2)$ paths, as required.
\end{proof}

This concludes the proof of Theorem~\ref{thm: basic crossbar construction}. Since $p_1=\Omega(\Delta/\kappa_0)$ and $\kappa_3\leq \kappa_0$, we will assume without loss of generality that $p_1\cdot \kappa_3<\Delta_2/24$ (otherwise we can lower the value of $p_1$ until the inequality holds and discard the appropriate number of paths from sets $\pset(t)$ for $t\in T_3\cup \set{t_0}$).

\label{-------------------------------------subsec: good crossbar--------------------------------------}
\subsection{Constructing a Good Crossbar}\label{subsec: good crossbar}
In this section, we complete the proof of Theorem~\ref{thm: crossbar construction}, by modifying the basic crossbar constructed in the previous section, in order to turn it into a good crossbar. Let $\left (\bigcup_{t\in T_3\cup\set{t_0}}\zset(t),\bigcup_{v\in S_3\cup T_3}\pset(v)\right )$ be the crossbar constructed in the previous section. 

For every terminal $t\in T_3$, let $A_t\subseteq V(C_t), B_t\subseteq V(C_{t_0})$ be the sets of $p_1$ vertices each, where the paths of $\pset(t)$ originate and terminate, respectively. For a source vertex $s\in S_3$, we let $A_s$ contain a single vertex $s$, and $B_s$ contain the vertex of $V(C_{t_0})$ where the unique path in $\pset(s)$ terminates. Let $A=\bigcup_{v\in S_3\cup T_3}A_v$ and $B=\bigcup_{v\in S_3\cup T_3}B_v$. Then $|A|=|B|=|\pset|=|\mset_3|(p_1+1)<\Delta_2/6$ from our assumption.

\begin{definition}
Given two equal-sized disjoint subsets $U_1,U_2$ of vertices of $G$, a $U_1$--$U_2$ linkage is a set of $|U_1|$ node-disjoint paths, connecting the vertices of $U_1$ to the vertices of $U_2$.
\end{definition}

The following observation is immediate from the definition of a crossbar.

\begin{observation}\label{obs: linkage to crossbar}
Let $\pset'$ be any $A$--$B$ linkage in $G$. Then $\left (\bigcup_{t\in T_3\cup\set{t_0}}\zset(t),\pset' \right )$ is a $p_1$-crossbar for $\mset_3$. 
\end{observation}

Our algorithm consists of three steps. In the first step, we re-route the paths in $\pset$, so that they become disjoint from the relevant inner shells, ensuring Property~(\ref{prop3: paths don't touch inner discs}). In the following step, the paths in $\pset$ are further re-routed to ensure their monotonicity with respect to the relevant inner shells, obtaining Property~(\ref{prop4: in own disc contiguous path}). The set $\mset_3$ of the demand pairs remains unchanged in these two steps. In the last step, we carefully select a final subset $\mset^*\subseteq \mset_3$ of the demand pairs to ensure Property~(\ref{prop-last: contiguously hit disc of t0}). We discard some paths from set $\pset$, but the paths themselves are not modified at this step.

\subsection*{Step 1: Disjointness of Paths from Inner Shells}
The goal of this step is to modify the paths in set $\pset$, in order to ensure Property~(\ref{prop3: paths don't touch inner discs}).
In fact, we will ensure a slightly stronger property that we will use in subsequent steps.

Suppose we are given any $A$--$B$ linkage $\pset'$. For every terminal $t\in T_3$, we let $\pset'(t)\subseteq \pset'$ be the set of paths originating from the vertices of $A_t$, and for every source vertex $s\in S_3$, we let $\pset'(s)$ contain the unique path of $\pset'$ originating at $s$.
We will always view the paths in $\pset'$ as directed from $A$ to $B$. Consider now any path $P\in \pset'$. Clearly, path $P$ has to cross $Z_{\Delta_1}(t_0)$. Let $v_P$ be the last vertex on $P$ that lies on $Z_{\Delta_1}(t_0)$, and let $e_P$ be the edge immediately preceding $v_P$ on $P$. Let $v'_P$ be the other endpoint of edge $e_P$. Then $P\setminus\set{e_P}$ consists of two disjoint paths, that we denote by $P_1$ and $P_2$, respectively, where $P_1$ starts at a vertex of $V(C_t)$ and terminates at $v'_P$, and $P_2$ starts at $v_P$ and terminates at a vertex of $V(C_{t_0})$.

We then denote $A'=\set{v'_P\mid P\in \pset'}$, $B'=\set{v_P\mid P\in \pset'}$, and $\tilde E=\set{e_P\mid P\in \pset'}$. We also denote $\pset'_1=\set{P_1\mid P\in \pset'}$ and $\pset'_2=\set{P_2\mid P\in \pset'}$. Notice that $\pset'_1$ is an $A$--$A'$ linkage, and $\pset'_2$ is a $B'$--$B$ linkage. Observe that for every terminal $t\in T_3\cup \set{t_0}$, all vertices of $A'\cup B'$ must lie outside $D(Z_{\Delta_1-1}(t))$. Also, from our definitions, every path in $\pset'_2$ is contained in $D^*(t_0)$.

\begin{remark}
We note that the definitions of the sets $A',B'$ of vertices and the set $\tilde E$ of edges depend on the  $A$--$B$ linkage $\pset'$. When we modify the linkage, these sets may change as well.
\end{remark}

We are now ready to define good $A$--$B$ linkages.

\begin{definition}
We say that an $A$--$B$ linkage $\pset'$ is a \emph{good linkage}, iff:

\begin{itemize}
\item All paths in $\pset'_1$ are disjoint from $\tD(t_0)$; and

\item For every $t\in T_3$, if some path $P\in \pset'_1\cup \pset_2'$ intersects $\tD(t)$, then $P\in\pset'_1$, and it originates at a vertex of $A_t$.
\end{itemize}
\end{definition}

Notice that if $\pset'$ is a good $A$--$B$ linkage, then $\left (\bigcup_{t\in T_3\cup\set{t_0}}\zset(t),\pset' \right )$ is a $p_1$-crossbar for $\mset_3$ that has Property~(\ref{prop3: paths don't touch inner discs}). The goal of this step is to prove the following theorem.

\begin{theorem}\label{thm: basic crossbar construction+first property}
There is an efficient algorithm that computes a good $A$--$B$ linkage $\pset'$.
\end{theorem}

The remainder of this step focuses on the proof of Theorem~\ref{thm: basic crossbar construction+first property}. In order to prove the theorem, we start with the $A$--$B$ linkage $\pset$, given by the basic crossbar that we have computed in the previous section, and iteratively modify it. Let $E_1$ be the set of all edges that belong to cycles $\bigcup_{t\in T_3\cup\set{t_0}}\bigcup_{j=\Delta_2+1}^{\Delta_1}Z_j(t)$. Given any set $\qset$ of node-disjoint paths in $G$, let $c(\qset)$ be the total number of all edges lying on the paths in $\qset$ that do not belong to $E_1$, so $c(\qset)=\sum_{P\in \qset}|E(P)\setminus E_1|$. The following lemma is key to proving Theorem~\ref{thm: basic crossbar construction+first property}.

\begin{lemma}\label{lemma: avoid inner shells}
Let $\pset$ be any $A$--$B$ linkage, and assume that it is not a good linkage. Then there is an efficient algorithm to compute an $A$--$B$ linkage $\pset'$ with $c(\pset')<c(\pset)$.
\end{lemma}

In order to complete the proof of Theorem~\ref{thm: basic crossbar construction+first property}, we start with the original $A$--$B$ linkage $\pset$. While the current linkage $\pset$ is not a good one, we apply Lemma~\ref{lemma: avoid inner shells} to it, to obtain another $A$--$B$ linkage $\pset'$ with $c(\pset')<c(\pset)$. The algorithm terminates when we obtain a good $A$--$B$ linkage. Since the values $c(\pset)$ are integers bounded by $|E(G)|$, and they decrease by at least $1$ in every iteration, the algorithm is guaranteed to terminate after $|E(G)|$ iterations. It now remains to prove Lemma~\ref{lemma: avoid inner shells}.

\begin{proofof}{Lemma~\ref{lemma: avoid inner shells}}
Let $\pset$ be any $A$--$B$ linkage, and assume that it is not a good linkage. Let $\qset=\pset_1\cup \pset_2$, so $\qset$ is an $(A\cup B')$--$(A'\cup B)$ linkage in $G\setminus \tilde E$, and $|\qset|=2|\mset_3|(p_1+1)<6\kappa_3 p_1<\Delta_2/3$ from our assumption. We will use the following simple observation.

\begin{observation}\label{obs: Q-linkage is enough}
Let $\qset'$ be any $(A\cup B')$--$(A'\cup B)$ linkage in $G\setminus \tilde E$, and let $G'$ be the graph obtained by the union of the paths in $\qset'$ and the edges of $\tilde E$. Then $G'$ contains an $A$--$B$ linkage $\pset'$. Moreover, if $c(\qset')<c(\qset)$, then $c(\pset')<c(\pset)$.
\end{observation}

\begin{proof}
We construct the following auxiliary directed graph $\tilde G$. The vertices of $\tilde G$ are $A\cup A'\cup B\cup B'$. The edges are defined as follows. First, for every path $Q\in \qset'$, that originates at a vertex $v\in A\cup B'$ and terminates at a vertex $u\in A'\cup B$, we add a directed edge from $v$ to $u$. Notice that since $\qset'$ is an $(A\cup B')$--$(A'\cup B)$ linkage, this gives a directed matching from the vertices of $A\cup B'$ to the vertices of $A'\cup B$. Finally, for every edge $e_P=(v'_P,v_P)\in \tilde E$, we add a directed edge from $v'_P$ to $v_P$ to $\tilde G$. Then in the resulting graph $\tilde G$, every vertex of $A$ has in-degree $0$ and out-degree $1$; every vertex of $B$ has out-degree $0$ and in-degree $1$, and every vertex of $A'\cup B'$ has in-degree and out-degree $1$. Therefore, graph $\tilde G$ is a collection of directed paths and cycles, and it must contain an $A$--$B$ linkage $\tilde{\pset}$. By replacing, for every path $P\in \tpset$, the edges representing the paths of $\qset'$ on $P$ with the corresponding paths, we obtain an $A$--$B$ linkage $\pset'$. Since the edges of $\tilde E$ participate in the paths in $\pset$, and the paths of $\pset'$ are contained in graph $G'$, it is immediate to verify that $c(\pset')<c(\pset)$.
\end{proof}

From Observation~\ref{obs: Q-linkage is enough}, it is now enough to construct an  $(A\cup B')$--$(A'\cup B)$ linkage $\qset'$  in $G\setminus \tilde E$, with $c(\qset')<c(\qset)$. The following technical claim is central to achieving this.

\begin{claim}\label{claim: A-B linkage}
Let $H$ be any planar graph embedded in the plane, such that $H$ is a union of a set $\zset=\set{Z_1,\ldots,Z_{h}}$ of disjoint cycles with $D(Z_1)\subsetneq D(Z_2)\subsetneq\ldots\subsetneq D(Z_{h})$, and a set $\qset=\qset_1\cup \qset_2$ of $r<h$ node-disjoint paths. Assume that for $i\in \set{1,2}$, the paths of $\qset_i$ originate at a set $U_i$ of vertices, and terminate at a set $U'_i$ of vertices. Assume further that the vertices of $U_1$ lie in $\dnot(Z_1)$, while the vertices of $U'_1\cup U_2\cup U'_2$ lie outside $D(Z_h)$. Let $H'$ be the graph obtained from $H$ by deleting, for every path $Q\in \qset_2$, every edge and vertex of $Q$ contained in $\dnot(Z_1)$. Then there is a set $\qset'$ of $r$ node-disjoint paths, connecting the vertices of $U_1\cup U_2$ to the vertices of $U'_1\cup U'_2$ in $H'$.
\end{claim}

\begin{proof}
Let $U=U_1\cup U_2$ and $U'=U'_1\cup U'_2$.
Assume otherwise. From Menger's theorem, there is a set $Y$ of at most $r-1$ vertices, such that in $H'\setminus Y$, no path connects a vertex of $U\setminus Y$ to a vertex of $U'\setminus Y$. Let $|\qset_1|=r_1$ and $|\qset_2|=r_2$.

Observe that all paths of $\qset_1$ are contained in $H'$, so for each path $P\in \qset_1$, there must be a distinct vertex $y_P\in Y$ lying on $P$. Let $Y_1\subseteq Y$ be the set of all such vertices, so $|Y_1|=r_1$.

Let $\qset^*\subseteq \qset_2$ be the set of all paths that are disjoint from $\dnot(Z_1)$, and denote $|\qset^*|=r^*$. Then every path $P\in \qset^*$ is contained in $H'$, and as before, $Y$ must contain a distinct vertex $y_P$ lying on $P$. Let $Y_2$ be the set of all such vertices $\set{y_P\mid P\in \qset^*}$. Since the paths in $\qset$ are disjoint, $Y_1\cap Y_2=\emptyset$, and $|Y_1|+|Y_2|=r_1+r^*$. Let $Y_3=Y\setminus (Y_1\cup Y_2)$, and let $\qset''=\qset_2\setminus \qset^*$. Then $|Y_3|\leq |\qset''|-1$, and due to the disjointness of the paths in $\qset$, every path in $\qset''$ is disjoint from $Y_1\cup Y_2$. Therefore, there is some path $Q\in \qset''$, such that $Q\cap Y=\emptyset$. Recall that $Q\cap \dnot(Z_1)\neq \emptyset$, and so $Q$ must intersect $Z_1$. Let $v',v''$ be the first and the last vertices of $Q$ that lie on $Z_1$. Recall that the endpoints of $Q$ lie outside $D(Z_h)$. Let $Q'$ be the segment of $Q$ between its first endpoint and $v'$, and let $Q''$ be the segment of $Q$ between $v''$ and its last endpoint. Then both $Q'$ and $Q''$ are contained in $H'$, and each of these paths intersects every cycle in $\zset$. The number of these cycles is $h>r$ from our assumption. Therefore, there is some $1\leq j\leq h$, such that $Z_j$ is disjoint from $Y$. By combining $Z_j$ with $Q'$ and $Q''$, we obtain a path connecting a vertex of $U\setminus Y$ to a vertex of $U'\setminus Y$ in $H'\setminus Y$, a contradiction.
\end{proof}

We are now ready to complete the proof of Lemma~\ref{lemma: avoid inner shells}.
Let $\pset$ be the given $A$--$B$ linkage, and assume that it is not a good linkage. 

Assume first that some path $P\in \pset_1$ has a non-empty intersection with $\tD(t_0)$. We let $\qset_1=\pset_2$, $\qset_2=\pset_1$, and  $\zset=(Z_{\Delta_2+1}(t_0),\ldots,Z_{\Delta_1-2}(t_0))$. Observe that $|\zset|=\Delta_1-\Delta_2-2\geq \Delta_2>24p_1\kappa_3>|\qset|$, the paths in $\qset_1$ originate at the vertices of $C_{t_0}$, that lie in $\dnot(Z_{\Delta_2+1})$, and terminate at the vertices of $B'$, that lie outside $D(Z_{\Delta_1-2}(t_0))$ (after we reverse them), while the paths of $\qset_2$ originate and terminate outside $D(Z_{\Delta_1-2}(t_0))$ (we also reverse them). From our assumption, at least one path in $\qset_2$ intersects $\dnot(Z_{\Delta_2+1})$. 
Let $H$ be the graph obtained by the union of the cycles in $\zset$ and the paths in $\qset$. Since the edges of $\tilde E$ cannot lie on the cycles of $\zset$, $H\subseteq G\setminus \tilde E$.
We can now apply Claim~\ref{claim: A-B linkage} to obtain a new $(A\cup B')$--$(A'\cup B)$ linkage $\qset'$ in graph $G\setminus \tilde E$.  Moreover, since we delete all edges lying on the paths of $\qset_2$ that belong to $\dnot(Z_{\Delta_2+1}(t_0))$, it is easy to verify that $c(\qset')<c(\qset)$.

Assume now that there is some path $P\in \qset$, and some terminal $t\in T_3$, such that $P$ intersects $\tD(t)$, but it does not originate at a vertex of $A_t$. Let $\qset_1\subseteq \pset_1$ be the set of all paths originating from the vertices of $A_t$, and let $\qset_2=\qset\setminus \qset_1$. It is easy to verify that we can apply Claim~\ref{claim: A-B linkage} as before, to obtain a new $(A\cup B')$--$(A'\cup B)$ linkage $\qset'$ in graph $G\setminus \tilde E$, with $c(\qset')<c(\qset)$. 

Using Observation~\ref{obs: Q-linkage is enough}, we can now obtain an $A$--$B$ linkage $\pset'$ with $c(\pset')<c(\pset)$.
\end{proofof}

\subsection*{Step 2: Monotonicity with Respect to Inner Shells}

In this step we further modify the paths in set $\pset$, in order to ensure Property~(\ref{prop4: in own disc contiguous path}). As before, given a good $A$--$B$ linkage $\pset'$, for every terminal $t\in T_3$ we denote by $\pset'(t)\subseteq \pset'$ the set of paths originating at the vertices of $A_t$, and for each $s\in S_3$, we let $\pset'(s)\subseteq \pset'$ contain the unique path originating at $s$. We define the sets $A',B'$ of vertices, the set $\tilde E$ of edges, and the sets $\pset_1',\pset_2'$ of paths with respect to $\pset'$ exactly as before.

\begin{definition}
We say that a good $A$--$B$ linkage $\pset'$ is perfect if  all paths in $\pset'$ are monotone with respect to $Z_1(t_0),\ldots,Z_{\Delta_2}(t_0)$, and for all $t\in T_3$, all paths in $\pset'(t)$ are monotone with respect to $Z_1(t),\ldots, Z_{\Delta_2}(t)$.
\end{definition}

In this step we prove the following theorem, that immediately gives a $p_1$-crossbar for $\mset_3$ with Properties~(\ref{prop3: paths don't touch inner discs}) and~(\ref{prop4: in own disc contiguous path}).

\begin{theorem}\label{thm: monotonicity}
There is an efficient algorithm, that, given a good $A$--$B$ linkage $\pset$, computes a perfect $A$--$B$ linkage $\pset'$.
\end{theorem}

\begin{proof}
Our first step is to re-route the paths in $\pset_2$, so that they become monotone with respect to $Z_1(t_0),\ldots,Z_{\Delta_2}(t_0)$. Recall that from our definition, the paths of $\pset_2$ originate from the set $B'\subseteq V(Z_{\Delta_1}(t_0))$ of vertices, terminate at the set $B\subseteq V(C_{t_0})$ of vertices, and they are internally disjoint from $V(Z_{\Delta_1}(t_0))\cup V(C_{t_0})$. 
Let $\zset=\set{Z_1(t_0),\ldots,Z_{\Delta_2}(t_0)}$, $C=C_{t_0}$, and $Y=Z_{\Delta_1}(t_0)$. Recall that from the construction of the shells, $\zset$ is a set of $\Delta_2$ tight concentric cycles around $C$. Let $H$ be the graph obtained by the union of the cycles in $\zset$ and the set $\pset_2$ of paths.
We use Theorem~\ref{thm: monotonicity for shells} to obtain a new $B'$--$B$ linkage $\pset'_2$ in $H$, such that the paths in $\pset'_2$ are monotone with respect to $Z_1(t_0),\ldots,Z_{\Delta_2}(t_0)$. Moreover, since the paths in $\pset_1$ are disjoint from the graph $H$, the paths in $\pset_2'\cup \pset_1$ remain disjoint.
Let $\tpset$ be the new $A$--$B$ linkage, obtained by concatenating the paths in $\pset_1$, the edges of $\tilde E$, and the paths in $\pset'_2$. Observe that $\tpset$ remains a good linkage.

For each terminal $t\in T_3$, let $\tpset(t)\subseteq \tpset$ be the set of paths originating from the vertices of $A_t$. For every terminal $t\in T_3$, we now re-route the paths in $\tpset(t)$, as follows.
We let $\zset=\set{Z_1(t),\ldots,Z_{\Delta_2}(t)}$ and $C=C_{t}$. As before, $\zset$ is a set of $\Delta_2$ tight concentric cycles around $C$. We let $Y$ be the sub-graph of $G$ induced by $V(D_{t_0})$. Let $H'$ be the graph obtained by the union of the cycles in $\zset$ and the set $\tpset(t)$ of paths. Let $H$ be the union of $H'$ and $G[Y]$. Then the paths in $\tpset(t)$ originate at the vertices of $A_{t}$ lying on $C_{t}$, terminate at the vertices of $Y$, and are internally disjoint from $C\cup Y$. Moreover, since $\tpset$ is a good linkage, all paths in $\tpset\setminus \tpset(t)$ are disjoint from graph $H'$. Using Theorem~\ref{thm: monotonicity for shells}, we can find a new set $\pset'(t)$ of $p_1$ disjoint paths in graph $H'$, 
connecting vertices of $V(C_t)$ to vertices of $V(C_{t_0})$. The paths in $\pset'(t)$ are guaranteed to be disjoint from all paths in $\tpset\setminus\pset'(t)$, and they remain monotone with respect to $Z_1(t_0),\ldots,Z_{\Delta_2}(t_0)$, since the paths in $\tpset(t)$ had this property, and $D^*(t)\cap D^*(t_0)=\emptyset$. They also remain disjoint from all discs $\tD(t')$ for all $t'\in T_3$ distinct from $t$. We replace the paths of $\tpset(t)$ with the paths of $\pset'(t)$ in $\tpset$, obtaining a new good $A$--$B$ linkage, and continue to the next iteration. The final set $\pset'$ of paths is obtained from $\tpset$ once all terminals of $T_3$ are processed.\end{proof}

\subsection*{Step 3: Ensuring Property~(\ref{prop-last: contiguously hit disc of t0})}

So far we have constructed a $p_1$-crossbar  $\left (\bigcup_{t\in T_3\cup\set{t_0}}\zset(t),\pset\right )$, that has properties~(\ref{prop3: paths don't touch inner discs}) and~(\ref{prop4: in own disc contiguous path}).

In this step we will discard some demand pairs from $\mset_3$, and some paths from sets $\pset(t)$ for terminals $t\in T_3$, so that the resulting set $\mset^*$ of the demand pairs, together with the shells around their destination vertices, and the resulting set $\pset^*$ of paths give a good crossbar. We do not alter the paths themselves, so Properties~(\ref{prop3: paths don't touch inner discs}) and~(\ref{prop4: in own disc contiguous path}) will continue to hold.

Consider some path $P\in \pset$, and recall that $P$ is monotone with respect to $Z_1(t_0),\ldots,Z_{\Delta_2}(t_0)$. Let $P'$ be the sub-path of $P$, starting from its first vertex (that lies in $A$), and terminating at the first vertex of $P$ that lies on $Z_{\Delta_2}(t_0)$. Let $\pset'=\set{P'\mid P\in \pset}$. As before, for every terminal $t\in T_3$, we denote by $\pset'(t)\subseteq \pset'$ the set of paths originating at the vertices of $A_t$, and for each $s\in S_3$, we denote by $\pset'(s)\subseteq \pset'$ the set containing the unique path originating from $s$.
Let $e^*$ be any edge of $Z_{\Delta_2}(t_0)$, and let $R^*$ be the path $Z_{\Delta_2}(t_0)\setminus e^*$. 
We prove the following theorem.

\begin{theorem}\label{thm: last property}
There is an efficient algorithm to compute a subset $\mset_4\subseteq \mset_3$ of $\Omega(\kappa_3/\log n)$ demand pairs, and to select, for each terminal $t\in \tset(\mset_4)\cap T_3$ a subset $\pset''(t)\subseteq \pset'(t)$ of $\Omega(p_1)$ paths, such that there is a partition $\Sigma=\set{\sigma(v)\mid v\in \tset(\mset_4)}$ of $R^*$ into disjoint segments, where for each $t\in  \tset(\mset_4)\cap T_3$, the paths in $\pset''(t)$ are disjoint from $R^*\setminus \sigma(t)$, while for each $s\in \tset(\mset_4)\cap S_3$, the unique path in $\pset'(s)$ is disjoint from $R^*\setminus \sigma(s)$.
\end{theorem}

Let $S_4$ and $T_4$ denote the sets of all source and all destination vertices of the demand pairs in $\mset_4$. Let $\pset''=\left (\bigcup_{t\in T_4}\pset''(t)\right )\cup \left(\bigcup_{s\in S_4}\pset'(s)\right )$.

In other words, Theorem~\ref{thm: last property} ensures that for each terminal $t\in T_4$, the endpoints of the paths in $\pset''(t)$ appear on $R^*$ consecutively, with respect to the endpoints of all paths in $\pset''$. 
It is now immediate to complete the construction of the good crossbar for $\mset_4$. For every path $P\in \pset(t)$ for all $t\in T_4$, we include $P$ in set $\pset^*$ only if the corresponding path $P'\in \pset'$ belongs to $\pset''$. Since the paths in $\pset$ are monotone with respect to $Z_1(t_0),\ldots,Z_{\Delta_2}(t_0)$, it is easy to see that Property~(\ref{prop-last: contiguously hit disc of t0}) is satisfied in the resulting crossbar, and we obtain a good $p^*$-crossbar for $\mset_4$, with $p^*=\Omega(p_1)=\Omega(\Delta/\kappa_0)$, and $|\mset_4|=\Omega(\kappa_0/\log n)$, as required. We now focus on the proof of Theorem~\ref{thm: last property}.

\begin{proofof}{Theorem~\ref{thm: last property}}
Let $G'$ be the graph obtained from $G$, after removing all vertices and edges lying in $\dnot_{\Delta_2}(t_0)$. Observe that all paths in $\pset'$ are still contained in $G'$. We view the face of $G'$ where $t_0$ used to reside as the outer face. Therefore, for each terminal $t\in T_3$, the paths of $\pset'(t)$ now connect the vertices of $V(C_t)$ to the vertices lying on the boundary of the outer face.

Consider any destination vertex $t\in T_3$, and let $\qset\subseteq \pset'(t)$ be any subset of its corresponding paths, with $|\qset|>2$. Let $\sigma$ be the shortest sub-path of $R^*$, containing all endpoints of the paths in $\qset$. Taking the union of $\sigma$, $\qset$ and $Z_{\Delta_2}(t)$, we obtain a new auxiliary graph $H(t,\qset)$. Let $\gamma(t,\qset)$ be the closed curve serving as the outer boundary of this graph, and let $D(t,\qset)$ be the disc whose boundary is $\gamma(t,\qset)$. 

Given two destination vertices $t,t'\in T_3$ with $t\neq t'$, and any two subsets $\qset(t)\subseteq \pset'(t)$ and $\qset(t')\subseteq \pset'(t')$ of paths, notice that the discs $D(t,\qset(t))$ and $D(t',\qset(t'))$ are either completely disjoint from each other, or one is contained in the other.
We say that there is a conflict between $(t,\qset(t))$ and $(t',\qset(t'))$ in the latter case. We also say that there is a conflict between $(t,\qset(t))$ and a source vertex $s\in S_3$, if $s\in D(t,\qset(t))$.

Let $\tmset_1=\mset_3$. Over the course of this step, we will define a series of subsets of demand pairs, $\tmset_3\subseteq \tmset_2\subseteq \tmset_1$, where $|\tmset_3|=\Omega(|\tmset_1|/\log n)$.
 For each $1\leq i\leq 3$, we will denote by $\tS_i$ and $\tT_i$ the sets of the source and the destination vertices of the pairs in $\tmset_i$, respectively. Our final set of the demand pairs will be $\mset^*=\tmset_3$. For every terminal $t\in \tset(\mset^*)$, we will define a series of subsets of paths $\qset_3(t)\subseteq \qset_2(t)\subseteq \qset_1(t)$, where we let $\qset_1(t)=\pset'(t)$, and we will ensure that $|\qset_3(t)|=\Omega(|\qset_1(t)|)=\Omega(p_1)$. Moreover, we will ensure that for all $t,t'\in \tT_3$, there is no conflict between $(t,\qset_3(t))$ and $(t',\qset_3(t'))$, and for all $s\in \tS_3$ and $t\in \tT_3$, there is no conflict between $(t,\qset_3(t))$ and $s$.

Our first step is to eliminate all conflicts between the destination vertices in $\tT_1$. We build a graph $F$, whose vertex set is $\tT_1$, and there is a directed edge $(t',t)$ iff (i) $t\neq t'$; (ii) $D(t',\qset_1(t'))\subseteq D(t,\qset_1(t))$; and (iii) there is no terminal $t''\in \tT_1\setminus\set{t,t'}$ with $D(t',\qset_1(t'))\subseteq D(t'',\qset_1(t''))\subseteq D(t,\qset_1(t))$. It is easy to see that $F$ is a directed forest. 

We use Claim~\ref{claim: partition the forest}, to obtain a partition $\set{R_1,\ldots,R_{\ceil{\log n}}}$ of $V(F)$ into subsets, such that, for each $1\leq j\leq \ceil{\log n}$, $F[R_j]$ is a collection $\yset_j$ of disjoint paths, and if vertices $v,v'\in R_j$ lie on two distinct paths in $\yset_j$, then neither is a descendant of the other in $F$.

Clearly, there is an index $1\leq j\leq \ceil{\log n}$, with $|R_j|\geq |\tT_1|/\ceil{\log n}$. We let $\tmset_2\subseteq \tmset_1$ be the set of the demand pairs whose destination vertices lie in $R_j$, and we define the sets $\tS_2$ and $\tT_2$ of source and destination vertices accordingly. For each $t\in \tT_2$, we now define a subset $\qset_2(t)\subseteq \qset_1(t)$ of paths, to ensure that there are no conflicts between the terminals in $\tT_2$.

Recall that $F[R_j]=\yset_j$ is a collection of paths, and for vertices $v,v'\in R_j$, if there is a directed path from $v$ to $v'$ in $F$, then $v,v'$ lie on the same path in $\yset_j$. Therefore, if $t,t'$ do not lie on the same path in $\yset_j$, there is no conflict between $(t,\qset_1(t))$ and $(t',\qset_1(t'))$. So we only need to resolve conflicts between terminals lying on the same path of $\yset_j$.

Let $P=(t_1,t_2,\ldots,t_r)$ be any such directed path. Then $D(t_1,\qset_1(t_1))\subseteq D(t_2,\qset_1(t_2)) \subseteq\cdots \subseteq D(t_r,\qset_1(t_r))$. Disc $D(t_{r-1},\qset_1(t_{r-1}))$ partitions the paths in $\qset_1(t_r)$ into two subsets, that go on each side of the disc. By discarding the paths in one of these two subsets (the one containing fewer paths), we can eliminate the conflict between $t_r$ and the remaining terminals on path $P$. Therefore, there is a subset $\qset_2(t_r)\subseteq \qset_1(t_r)$, containing at least $|\qset_1(t_r)|/2$ paths, such that $(t_r,\qset_2(t_r))$ has no conflict with $(t_i,\qset_1(t_i))$ for any $1\leq i< r$. We process all other terminals $t\in P$ similarly, obtaining a subset $\qset_2(t)\subseteq \qset_1(t)$ of paths, where $|\qset_2(t)|\geq |\qset_1(t)|/2$, and for all $t,t'\in P$ with $t\neq t'$, there is no conflict between $(t,\qset_2(t))$ and $(t',\qset_2(t'))$.

This completes the definition of the set $\tmset_2$ of demand pairs, and the corresponding sets $\qset_2(t)$ of paths for $t\in \tT_2$. Our next step is to eliminate conflicts between pairs $(t,\qset_2(t))$ for $t\in \tT_2$ and the sources $s\in \tS_2$. For every demand pair $(s,t)\in \tmset_2$, if there is a conflict between $(t,\qset_2(t))$ and $s$, then we can discard a subset of at most half the paths from $\qset_2(t)$ in order to eliminate this conflict. The resulting set of paths is denoted by $\qset_3(t)$. Therefore, we will assume from now on that for every demand pair $(s,t)\in \tmset_2$, there is no conflict between $(t,\qset_3(t))$ and $s$.

Finally, we build a conflict graph $H$, whose vertex set is $\set{v_{s,t}\mid (s,t)\in \tmset_2}$, and there is a directed edge $(v_{s,t},v_{s',t'})$ iff there is a conflict between $s$ and $(t',\qset_3(t'))$. Since the discs $\set{D(t,\qset_3(t))}_{t\in \tT_2}$ are all disjoint, every vertex $v_{s,t}$ has at most one out-going edge. Therefore, every vertex-induced sub-graph of $H$ has at least one vertex whose total degree (counting the incoming and the outgoing edges) is at most $2$. Using standard techniques, we can find an independent set $I$ of vertices in $H$, with $|I|=\Omega(|V(H)|)$. Our final set $\tmset_3$ of demand pairs contains all pairs $(s,t)$ with $v_{s,t}\in I$. We define the sets $\tS_3$ and $\tT_3$ of the source and the destination vertices accordingly, and the sets $\qset_3(t)$ for the destination vertices $t\in \tT_3$ remain the same.
From the above discussion, for each $t,t'\in \tT_3$ with $t\neq t'$, there is no conflict between $(t,\qset_3(t))$ and $(t',\qset_3(t'))$, and for all $s\in \tS_3$, $t'\in \tT_3$, there is no conflict between $(t',\qset_3(t'))$ and $s$.
For every source vertex $s\in \tS_3$, the set $\qset_3(s)$ contains the same path as the original set $\pset'(s)$.
\end{proofof}

\label{-------------------------------------------------sec: case 2--------------------------------------------}
\section{Case 2: Heavy Demand Pairs.}\label{sec: case 2}
In this case, we assume that at least $0.7|\mset|$ demand pairs are heavy. Let $\hset=\set{X_1,\ldots,X_q}\subseteq \xset$ be the collection of all heavy subsets of terminals, so $q\leq 2W/\tau$, and let $\mset^h$ be the set of all heavy demand pairs, so for all $(s,t)\in \mset^h$, both $s$ and $t$ lie in the sets of $\hset$. 

We partition the set $\mset^h$ of demand pairs into $q^2$ subsets, where for $1\leq i,j\leq q$, set $\mset_{i,j}$ contains all demand pairs $(s,t)$ with $s\in X_i$ and $t\in X_j$ (notice that it is possible that $i=j$). We then find an approximate solution to each resulting problem separately. The main theorem of this section is the following.

\begin{theorem}\label{thm: main for Case 2}
There is an efficient algorithm, that for each $1\leq i,j\leq q$, computes a subset $\mset'_{i,j}\subseteq \mset_{i,j}$ of at least $5|\mset_{i,j}|/6$ demand pairs, and a collection $\pset_{i,j}$ of node-disjoint paths routing a subset of the demand pairs in $\mset'_{i,j}$ in $G$, with $|\pset_{i,j}|\geq\min\set{ \frac{\opt(G,\mset_{i,j}')}{c_1\Delta_0^8\log^3n},\frac{\alphawl \cdot|\mset_{i,j}'|}{c_2\Delta_0^2}}$, for some universal constants $c_1$ and $c_2$.
\end{theorem}

Before we prove this theorem, we show that it concludes the proof of Theorem~\ref{thm: main} for Case 2. 
Let set $\pi$ contain all pairs $(i,j)$ with $1\leq i,j\leq q$, such that $|\mset_{i,j}|\geq 0.1|\mset|/q^2$, and let $\tilde \mset=\bigcup_{(i,j)\in \pi}\mset_{i,j}$. Since the total number of heavy demand pairs is at least $0.7|\mset|$, it is easy to verify that $|\tilde \mset|\geq 0.6|\mset|$.

We apply Theorem~\ref{thm: main for Case 2}, to compute, for each $(i,j)\in \pi$, the subset $\mset'_{i,j}\subseteq \mset_{i,j}$ of at least $5|\mset_{i,j}|/6$ demand pairs and the corresponding set $\pset_{i,j}$ of paths routing a subset of the demand pairs in $\mset'_{i,j}$.
Let $\tmset'\subseteq \tmset$ be the set of all demand pairs in $\bigcup_{(i,j)\in \pi}\mset'_{i,j}$. Then:

\[|\tmset'|=\sum_{(i,j)\in \pi}|\mset'_{i,j}|\geq \sum_{(i,j)\in \pi} 5|\mset_{i,j}|/6=5|\tmset|/6\geq |\mset|/2.\]

 If, for any $(i,j)\in \pi$, we obtain a solution with $|\pset_{i,j}|\geq \frac{W}{2^{13}\cdot \alphasc\cdot c_1c_2q^2\Delta_0^8 \log^3 n\cdot \log k}$, then we return the set $\pset_{i,j}$ as our final solution. 
 
 Substituting $\Delta_0=O(\Delta \log n)$, $\Delta=\ceil{W^{2/19}}$, $q=O(W/\tau)$, and $\tau=W^{18/19}$, we get that:

\[\begin{split}
|\pset_{i,j}|&\geq \Omega\left (\frac{W\tau^2}{W^2 \Delta^8 \log^{11}n\log k}\right )\\
&\geq \Omega\left (\frac{W^{36/19}}{W\cdot W^{16/19} \log^{11}n\log k}\right )\\
&=\Omega\left( \frac{W^{1/19}}{\log^{11}n\log k}\right ).
\end{split}
\]

 Otherwise,  for all $(i,j)\in \pi$, the resulting solution $|\pset_{i,j}|< \frac{W}{2^{13}\cdot \alphasc\cdot c_1c_2q^2\Delta_0^8 \log^3 n\cdot \log k}$. We then return the subset $\tilde \mset'$ of demand pairs. As observed above, $|\tmset'|\geq |\mset|/2$, so it is now enough to show that $\opt(G,\tmset')\leq w^*|\tmset'|/8$.

Assume otherwise, and let $\pset^*$ be a solution to instance $(G,\tmset')$, routing a subset $\mset^*\subseteq \tmset'$ of at least $w^*|\tmset'|/8\geq w^*|\mset|/16$ demand pairs. Then there is a pair of indexes $(i,j)\in \pi$, such that $|\mset^*\cap \mset_{i,j}'|\geq \frac{w^*|\mset|}{16q^2}$. Therefore, $\opt(G,\mset_{i,j}')\geq  \frac{w^*|\mset|}{16q^2}$. From Theorem~\ref{thm: main for Case 2}, we compute a set $\pset_{i,j}$ of paths, routing a subset of demand pairs of $\mset'_{i,j}$, with either 
$|\pset_{i,j}|\geq  \frac{\opt(G,\mset'_{i,j})}{c_1\Delta_0^8\log^3 n}$, or $|\pset_{i,j}|\geq \frac{\alphawl\cdot |\mset_{i,j}'|}{c_2\Delta_0^2}$.
In the former case,
\[|\pset_{i,j}|\geq \frac{\opt(G,\mset'_{i,j})}{c_1\Delta_0^8 \log^3 n}\geq \frac{w^*|\mset|}{16c_1q^2\Delta_0^8 \log^3 n}=\frac{W}{16c_1q^2\Delta_0^8 \log^3 n},\]

while in the latter case, observe that $|\mset_{i,j}'|\geq 5|\mset_{i,j}|/6\geq |\mset|/(12q^2)$ from the definition of $\pi$ and $\tilde{M}$. Therefore,

\[|\pset_{i,j}|\geq \frac{\alphawl \cdot |\mset_{i,j}'|}{c_2\Delta_0^2}\geq \frac{w^* \cdot|\mset|}{12\cdot 512\cdot \alphasc\cdot c_2q^2\Delta_0^2\log k}>\frac{W}{2^{13}\cdot \alphasc\cdot c_1c_2q^2\Delta_0^8 \log^2 n \log k},\]

 a contradiction.

From now on we focus on proving Theorem~\ref{thm: main for Case 2}. We fix a pair of indices $1\leq i,j\leq q$. In order to simplify the notation, we denote $\mset_{i,j}$ by $\nset$, $X_i$ by $X$ and $X_j$ by $Y$. Our goal is to compute a subset $\nset'\subseteq \nset$ of at least $5|\nset|/6$ demand pairs, together with a set $\pset$ of at least  $\min\set{\Omega\left( \frac{\opt(G,\nset')}{\Delta_0^8\log^3 n}\right ),\Omega\left(\frac{\alphawl|\nset'|}{\Delta_0^2}\right )}$ disjoint paths, routing a subset of the demand pairs in $\nset'$.
Let $x\in X$ and $y\in Y$ be any pair of terminals. Recall that for every terminal $t\in \tset(\nset)\cap X$, $d(t,x)\leq \Delta_0$, and for every terminal $t\in \tset(\nset)\cap Y$, $d(t,y)\leq \Delta_0$. We consider two subcases. The first subcase happens when $d(x,y)>5\Delta_0$, and otherwise the second subcase happens. Note that the second subcase includes the case where $X=Y$.

\label{-------------------------------------------subsec: case 2a----------------------------------------}
\subsection{Subcase 2a: $d(x,y)>5\Delta_0$}\label{subsec: Case 2a}

In this case, we set $\nset'=\nset$. We will compute a set $\pset$ of at least $\Omega\left( \frac{\opt(G,\nset)}{\Delta_0^6\log n}\right )$ node-disjoint paths, routing a subset of the demand pairs in $\nset$.
We start by defining a simpler special case of the problem, and show that we can find a good approximation algorithm for this special case. The special case is somewhat similar to routing on a cylinder, and we solve it by reducing it to this setting. We then show that the general problem in Case 2a reduces to this special case.

\subsubsection*{A Special Case}
Suppose we are given a connected planar graph $\hat G$ embedded on the sphere, and two disjoint simple cycles $Z,Z'$ in $\hat G$. Suppose also that we are given a set $\hmset$ of demand pairs, where all source vertices lie on $Z$ and all destination vertices lie on $Z'$ (we note that the same vertex may participate in a number of demand pairs). Let $D(Z),D(Z')$ be two discs with boundaries $Z$ and $Z'$, respectively, so that $D(Z)\cap D(Z')=\emptyset$. Assume additionally that we are given a closed $\hat G$-normal  curve $C$  of length at most $\Delta$, that is contained in $\dnot(Z)$, so that for every vertex $v\in Z$, there is a $\hat G$-normal curve $\gamma(v)$ of length at most $2\Delta_0$ connecting $v$ to a vertex of $C$, and $\gamma(v)$ is internally disjoint from $Z$ and $C$. Similarly, assume that we are given a closed $\hat G$-normal curve $C'$ of length at most $\Delta$, that is contained in $\dnot(Z')$, so that for every vertex $v'\in Z'$ there is a $\hat G$-normal curve $\gamma(v')$ of length at most $2\Delta_0$, connecting $v'$ to a vertex of $C'$, and $\gamma(v')$ is internally disjoint from $Z'$ and $C'$ (see Figure~\ref{fig: special case}). This finishes the definition of the special case. The following theorem shows that we can obtain a good approximation for it.

\begin{figure}[h]
 \centering
\scalebox{0.4}{\includegraphics{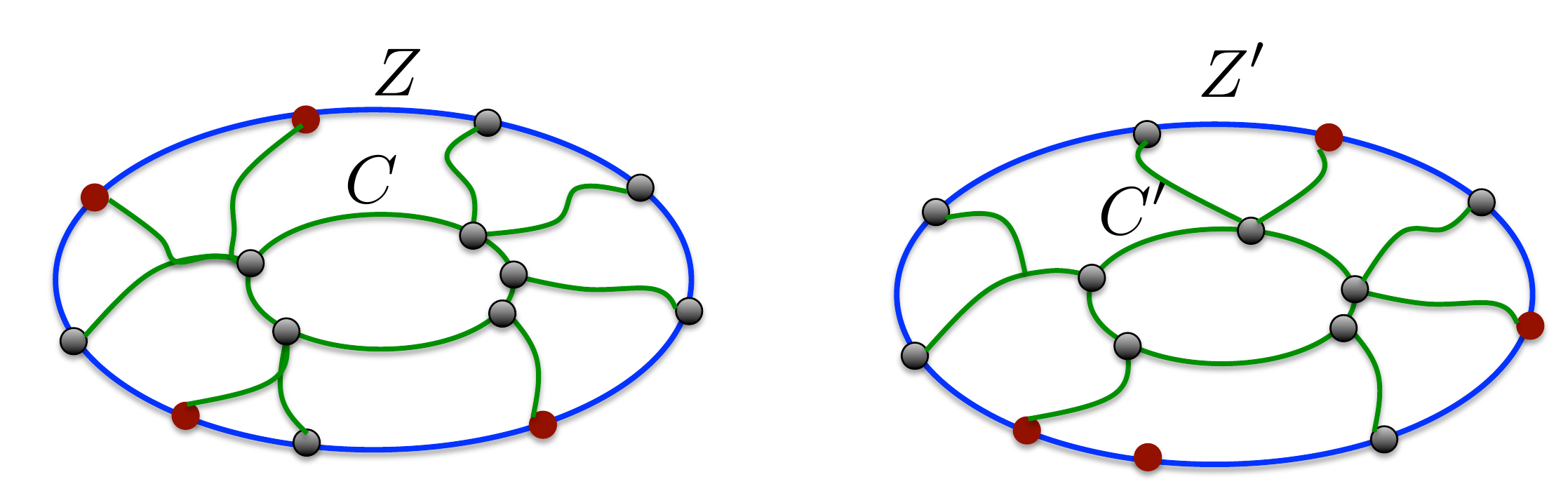}}
\caption{The special case, with the terminals shown in red.\label{fig: special case}}
\end{figure}

\begin{theorem}\label{thm: case 2a - special case}
There is an efficient algorithm, that, given any instance $(\hat G,\hat \mset)$ of the \NDP problem as above,  computes a solution of value at least $\Omega\left(\frac{OPT(\hat G,\hmset)}{\Delta_0^2\log n}\right )$.
\end{theorem}

\begin{proof}
The algorithm is very simple: we reduce the problem to routing on a cylinder, by creating two holes in the sphere. The first hole is $D(Z)$: we delete all edges and vertices that appear inside the disc, except for the edges and the vertices of $Z$. The second hole is $D(Z')$: we similarly delete all edges and vertices that lie inside the disc, except for those lying on its boundary. Let $\hat G'$ be the resulting graph. We then apply the $O(\log n)$-approximation algorithm for \NDPcyl to the resulting problem, to obtain a routing of at least $\Omega\left(\frac{\opt(\hat G',\hmset)}{\log n}\right)$ demand pairs. In order to complete the analysis of the algorithm, it is enough to prove that $\opt(\hat G',\hmset)\geq \Omega\left(\frac{\opt(\hat G,\hmset)}{\Delta_0^2}\right )$.

Notice that we can assume without loss of generality that all curves in set $\set{\gamma(v)\mid v\in V(Z)}$ are mutually non-crossing, and moreover, whenever two curves meet, they continue together. In other words, for all $v,v'\in V(Z)$, $\gamma(v)\cap \gamma(v')$ is a contiguous curve that has a vertex of $C$ as its endpoint. We make a similar assumption for curves in $\set{\gamma(v)\mid v\in V(Z')}$.

Consider the optimal solution to instance $(\hat G,\hmset)$, and let $\pset_0$ be the set of paths in this solution. We will gradually modify the set $\pset_0$ of paths, to obtain path sets $\pset_1,\pset_2,\ldots$, until we obtain a feasible solution to instance $(\hat G',\hmset)$. For every $i\geq 0$, we will denote by $\hmset_i$ the set of the demand pairs routed by $\pset_i$, and by $\kappa_i$ its cardinality.
Recall that $\kappa_0=\opt(\hat G,\hmset)$. We assume that $\kappa_0\geq 512\Delta_0^2$, as otherwise a solution routing a single demand pair gives the desired approximation, and such a solution exists in $\hat G'$, as it must be connected.

We delete from $\pset_0$ all paths that use the vertices of $C$ or $C'$. Since both curves have length at most $\Delta$, we delete at most $2\Delta$ paths in this step. Let $\pset_1$ be the resulting set of paths, and $\hmset_1$ the set of the demand pairs routed by $\pset_1$.

Our next step is to build a conflict graph $H$. Its set of vertices, $V(H)=\set{v(s,t)\mid (s,t)\in \hmset_1}$. There is a directed edge from $v(s,t)$ to $v(s',t')$, iff the path $P(s,t)\in \pset_1$ routing the pair $(s,t)$ contains a vertex of $V(\gamma(s'))\cup V(\gamma(t'))$, and we say that there is a conflict between $(s,t)$ and $(s',t')$ in this case. Since we assume that the paths in $\pset_1$ are node-disjoint, and since $|V(\gamma(s'))|,|V(\gamma(t'))|\leq 2\Delta_0$ for all $(s',t')\in \hmset_1$, the in-degree of every vertex in $H$ is at most $4\Delta_0$. Therefore, the average degree (including the incoming and the outgoing edges) of every induced sub-graph of $H$ is at most $8\Delta_0$, and there is an independent set $I\subseteq V(H)$ of cardinality at least $\frac{|\pset_1|}{8\Delta_0+1}\geq \frac{|\pset_0|}{16\Delta_0}$ in $H$.

Let $\hmset_2$ be the set of all demand pairs $(s,t)$ with $v(s,t)\in I$, and let $\pset_2\subseteq \pset_1$ be the set of paths routing the demand pairs in $\hmset_2$. Recall that the paths in $\pset_2$ are disjoint from $C\cup C'$. Moreover, if $P(s,t)\in \pset_2$ is the path routing the pair $(s,t)\in \hmset_2$, then for every demand pair $(s',t')\neq (s,t)$ in $\hmset_2$, path $P(s,t)$ is disjoint from both $\gamma(s')$ and $\gamma(t')$. It is now easy to verify that the demand pairs in $\hmset_2$ are non-crossing, that is, we can find an ordering $\hmset_2=\set{(s_1,t_1),\ldots,(s_{\kappa_2},t_{\kappa_2})}$ of the demand pairs in $\hmset_2$, so that $s_1,\ldots,s_{\kappa_2}$ appear in this counter-clock-wise order on $Z$, while $t_1,\ldots,t_{\kappa_2}$ appear in this clock-wise order on $Z'$.

Let $\kappa_3=\floor{\frac{\kappa_2}{8\Delta_0}}-1$, and let $\hmset_3=\set{(s_{8\Delta_0r},t_{8\Delta_0r})\mid 1\leq r\leq \kappa_3}$. In other words, we space the demand pairs in $\hmset_2$ out, by adding one in $8\Delta_0$ such pairs to $\hmset_3$. Let $\pset_3\subseteq \pset_2$ be the set of paths routing the demand pairs in $\mset_3$, so $|\pset_3|\geq \frac{|\pset_2|}{32\Delta_0}\geq \frac{|\pset_0|}{512\Delta_0^2}$. Our final step is to show that all demand pairs in $\hmset_3$ can be routed in graph $\hat G'$ via node-disjoint paths. In order to do so, for each such demand pair $(s_{8\Delta_0r},t_{8\Delta_0r})$, we define a segment $\mu_r$ of $Z$ containing $s_{8\Delta_0r}$, and a segment $\mu'_r$ of $Z'$ containing $t_{8\Delta_0r}$, as follows. For convenience, denote $8\Delta_0r$ by $\ell$. The first segment, $\mu_r$, is simply the segment of $Z$ from $s_{\ell-4\Delta_0}$ to $s_{\ell+4\Delta_0-1}$, as we traverse $Z$ in the counter-clock-wise order. The second segment, $\mu'_r$, is the segment of $Z$ from $t_{\ell-4\Delta_0}$ to $t_{\ell+4\Delta_0-1}$, as we traverse $Z'$ in  the clock-wise order. It is immediate to verify that all segments of $Z$ in $\set{\mu_r\mid 1\leq r\leq \kappa_3}$ are mutually disjoint, and the same holds for all segments of $Z'$ in $\set{\mu_r'\mid 1\leq r\leq \kappa_3}$.
The crux of the analysis is the following lemma.

\begin{lemma}\label{lemma: routing and segment intersections}
Let $(s_{{8\Delta_0 r}},t_{{8\Delta_0 r}})\in \mset_3$ be any demand pair, and let $P\in \pset_3$ be the path routing it. Then $P\cap Z\subseteq \mu_{{r}}$, and $P\cap Z'\subseteq \mu_{r}'$.
\end{lemma}

Before we prove this lemma, we show that we can use it to obtain a routing of the demand pairs in $\hmset_3$ in graph $\hat G'$ via node-disjoint paths. Let $(s_{\ell},t_{\ell})\in \hmset_3$ be any demand pair (where $\ell=8\Delta_0r$ for some $1\leq r\leq \kappa_3$), and let $P_{\ell}$ be the path routing $(s_{\ell},t_{\ell})$ in $\pset_3$, that we view as directed from $s_{\ell}$ towards $t_{\ell}$. Clearly, $P_{\ell}$ intersects both $\mu_{r}$ and $\mu'_{r}$. Let $v_{\ell}$ be the last vertex of $P_{\ell}$ lying on $\mu_{r}$. Then there is some other vertex appearing on $P_{\ell}$ after $v_{\ell}$ that belongs to $\mu_{r}'$. We let $v'_{\ell}$ be the first such vertex on $P_{\ell}$, and we let $P'_{\ell}$ be the segment of $P_{\ell}$ between $v_{\ell}$ and $v'_{\ell}$. Let $P^*_{\ell}$ be the path obtained as follows: we start with a segment of $\mu_{r}$ between $s_{\ell}$ and $v_{\ell}$; we then follow $P'_{\ell}$ to $v'_{\ell}$, and finally we use a segment of $\mu'_{r}$ between $v'_{\ell}$ and $t_{\ell}$. From Lemma~\ref{lemma: routing and segment intersections}, it is immediate to verify that the resulting paths in set $\set{P^*_{\ell}\mid  \ell=8\Delta_0r; 1\leq r\leq \kappa_3}$  are completely disjoint, contained in $\hat G'$, and they route all demand pairs in $\hmset_3$. It now remains to prove Lemma~\ref{lemma: routing and segment intersections}.

\begin{proof}
Fix some demand pair $(s_{{8\Delta_0 r}},t_{{8\Delta_0 r}})\in \mset_3$, and let $P\in \pset_3$ be the path routing it. We show that 
$P\cap Z\subseteq \mu_{{ r}}$. The proof that $P\cap Z'\subseteq \mu_{{r}}'$ is symmetric. For convenience, we denote $8\Delta_0r$ by $\ell$ from now on.

Assume otherwise, and let $v$ be the first vertex on $P$ that belongs to $Z\setminus  \mu_{r}$.  Let $P'$ be the sub-path of $P$ from its first vertex to $v$. Consider the planar embedding of $\hat G$, where we fix any face contained in $D(Z')$ as the outer face. In this planar embedding, denote by $Y$ the union of $D(C),\gamma(s_{\ell}),P'$, and $\gamma(v)$, and let $R$ be the outer boundary of $Y$ (notice that $P'$ may intersect $\gamma(s_{\ell})$, and that $\gamma(s_{\ell})$, $\gamma(v)$ are not necessarily disjoint).  

\begin{figure}[h]
 \centering
\scalebox{0.3}{\includegraphics{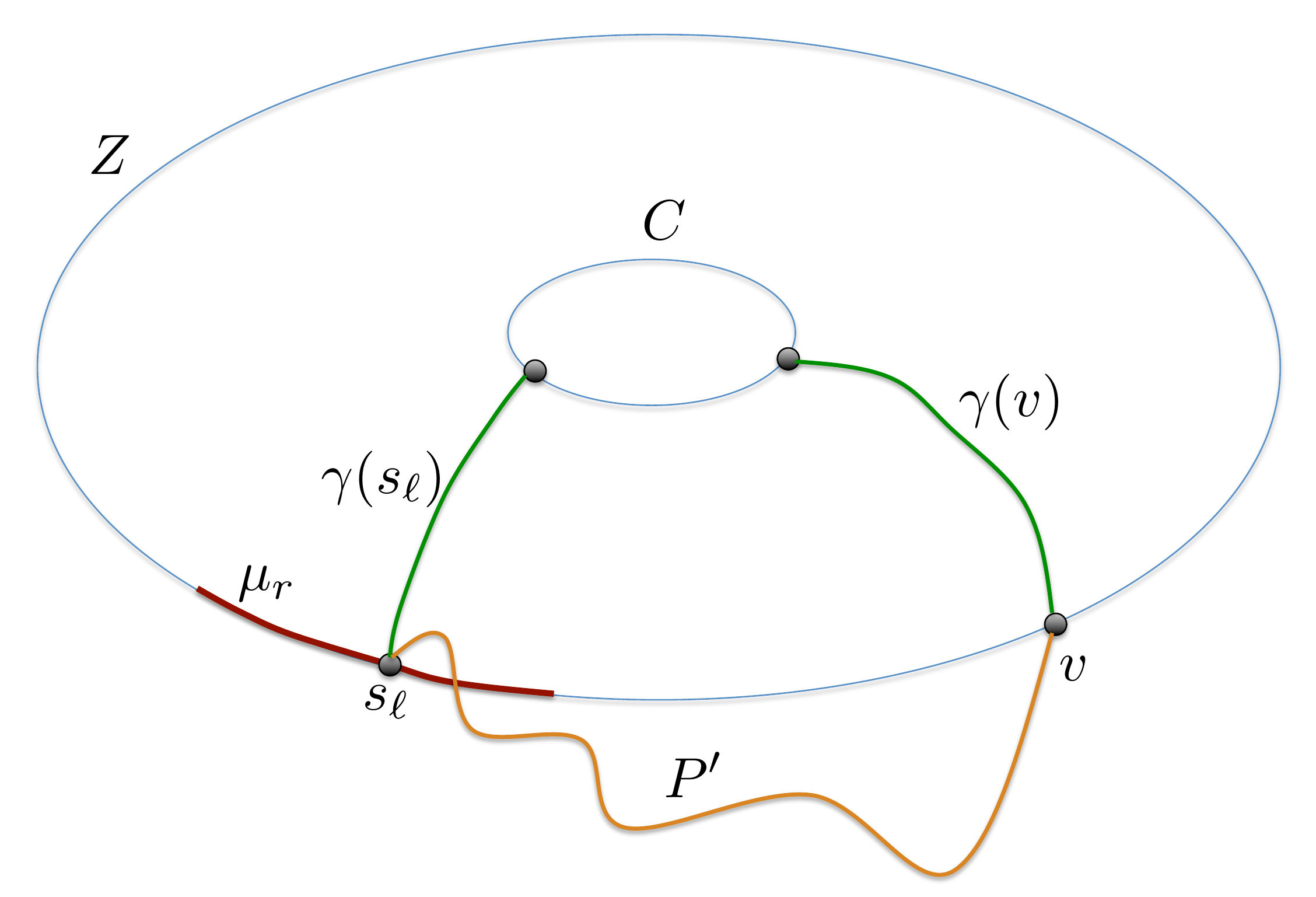}}\caption{Constructing the curve $R$\label{fig: 1a end}}
\end{figure}

Notice that, since there are no conflicts in $\pset_2$, curve $R$ does not cross any curve in $C'\cup\set{\gamma(t_h)\mid 1\leq h\leq \kappa_2}$, and so all destination vertices of the demand pairs in $\hmset_2$ lie on the outside of $R$. Let $S_1=\set{s_{\ell-4\Delta_0},\ldots,s_{\ell-1}}$, and $S_2=\set{s_{\ell+1},\ldots, s_{\ell+4\Delta_0-1}}$. 
Denote $\Gamma_1=\set{\gamma(s_i)\mid s_i\in S_1}$ and $\Gamma_2=\set{\gamma(s_i)\mid s_i\in S_2}$. Since the curves in set $\set{\gamma(u)\mid u\in V(Z)}$ are non-crossing, if $\sigma,\sigma'$ are the two segments of $Z$ whose endpoints are $s_{\ell}$ and $v$, then all vertices of $S_1\setminus\set{v}$ lie on one of the segments (say $\sigma$), while all vertices of $S_2\setminus\set{v}$ lie on the other segment. Moreover, since path $P'$ cannot cross any curve in $\Gamma_1\cup \Gamma_2$, either all sources of $S_1$, or all sources of $S_2$ lie inside the curve $R$ - let us assume that it is the former.

 All sources of $S_1$ are separated by $R$ from their destinations vertices, and yet all corresponding demand pairs are routed by $\pset_2$. Therefore, at least $4\Delta_0-2$ paths in $\pset_2$ must cross the curve $R$. Recall that none of these paths can cross $C$, $P'$, or $\gamma(s_{\ell})$. Therefore, all these paths must cross $\gamma(v)$. But since the length of $\gamma(v)$ is at most $2\Delta_0$, and the paths are node-disjoint, this is impossible.
\end{proof}

\end{proof}

\subsubsection*{Completing the Proof}
We now complete the proof of Theorem~\ref{thm: main for Case 2} for Case 2a, by reducing it to the special case defined above.

We assume that $\opt(G,\nset)>2^{13}\Delta_0^4$, since otherwise we can route a single demand pair and obtain a valid solution.
 We denote by $S$ and $T$ the sets of all source and all destination vertices of the demand pairs in $\nset$, respectively.

Our first step is to construct shells $\zset(x)=(Z_1(x),\ldots,Z_{2\Delta_0}(x))$ and $\zset(y)=(Z_1(y),\ldots,Z_{2\Delta_0}(y))$ of depth $2\Delta_0$ around $x$ and $y$, respectively. We would like to ensure that disc $D(Z_{2\Delta_0}(x))$ contains all terminals of $X$ and no terminals of $Y$, and similarly, disc $D(Z_{2\Delta_0}(y))$ contains all terminals of $Y$ and no terminals of $X$. In order to ensure this, when constructing the shell $\zset(x)$, we let the face $F_x$ (that is viewed as the outer face in the plane embedding of $G$ when constructing the shell) be the face incident on the terminal $y$, and similarly we let $F_y$ be the face incident on $x$ (recall that $d(x,y)>5\Delta_0$, so this choice of faces is consistent with the requirement that the shell depth is bounded by $\min\set{\dface(v,C_x)}-1$ over all vertices $v$ lying on the boundary of $F_x$). Let $\zset(x)=(Z_1(x),\ldots,Z_{2\Delta_0}(x))$ and $\zset(y)=(Z_1(y),\ldots,Z_{2\Delta_0}(y))$ be the resulting shells. 

\begin{claim}\label{claim: shells are good}
Disc $D(Z_{2\Delta_0}(x))$ contains all terminals of $X$ and no terminals of $Y$, and similarly, disc $D(Z_{2\Delta_0}(y))$ contains all terminals of $Y$ and no terminals of $X$. Moreover, $D(Z_{2\Delta_0}(x))\cap D(Z_{2\Delta_0}(y))=\emptyset$. 
\end{claim}

\begin{proof}
Let $t\in X$ be any terminal, and assume for contradiction that $t\not \in D(Z_{2\Delta_0}(x))$. Recall that $d(x,t)\leq \Delta_0$, and so there is a $G$-normal curve $\gamma$ of length at most $\Delta_0$, connecting some vertex $v\in C_t$ to some vertex $u\in C_x$. Then every vertex $v'\in C_t$ has a $G$-normal curve $\gamma(v')$ of length at most $\Delta_0+\Delta$ connecting it to $u$. Therefore, all vertices of $C_t$ lie in the disc  $D(Z_{\Delta_0+\Delta}(x))$ (if some vertex $v'\in C_t$ does not lie in this disc, then there are $\Delta_0+\Delta$ disjoint cycles separating $v'$ from $u$, a contradiction). We conclude that all vertices of $C_t$ lie in disc $D(Z_{2\Delta_0}(x))$, but $t$ does not lie in that disc. This can only happen if the outer face $F_x\subseteq D_t$, meaning that $y$ lies in $D_t$. But then, from our definition of enclosures, $d(t,y)=1$ must hold. However, $t\in X$, $y\in Y$, and $X\neq Y$, so $d(t,y)\geq 5\Delta$ from the definition of the family $\xset$ of sets of terminals, a contradiction. We conclude that all terminals of $X$ lie in disc $D(Z_{2\Delta_0}(x))$. Using a similar reasoning, all terminals of $Y$ lie in disc $D(Z_{2\Delta_0}(y))$.

 %

In order to complete the proof of the theorem it is now enough to show that $D(Z_{2\Delta_0}(x))\cap D(Z_{2\Delta_0}(y))=\emptyset$. Observe that from Property~(\ref{prop: Yt out}) of the shells, the vertices of $V(D_y)$ all lie outside the disc $D(Z_{2\Delta_0}(x))$, as all such vertices belong to the connected component $Y_x$. Therefore,  $y\in D(Z_{2\Delta_0}(y))\setminus D(Z_{2\Delta_0}(x))$, and similarly $x\in D(Z_{2\Delta_0}(x))\setminus D(Z_{2\Delta_0}(y))$, so neither of the two discs is contained in the other. Assume for contradiction that  the intersection of the two discs is non-empty. Observe first that the boundaries of the two discs cannot intersect. Indeed, assume otherwise, and let $v$ be any vertex lying in $Z_{2\Delta_0}(x)\cap Z_{2\Delta_0}(y)$. Then there are $G$-normal curves $\gamma$ and $\gamma'$ of length at most $2\Delta_0$ each, connecting $v$ to a vertex of $C_x$ and a vertex of $C_y$, respectively, implying that $d(x,y)\leq 4\Delta_0$, a contradiction. Therefore, all vertices of $Z_{2\Delta_0}(y)$ must lie inside the disc $D(Z_{2\Delta_0}(x))$. However, the vertices of $C_y$ lie outside $D(Z_{2\Delta_0}(x))$, and at least one such vertex $v$ can be connected to some vertex of $Z_{2\Delta_0}(y)$ with a $G$-normal curve of length at most $2\Delta_0+1$. This curve must intersect $Z_{2\Delta_0}(x)$ at some vertex that we denote by $u$, and $u$ can be in turn connected to a vertex of $C_x$ by a $G$-normal curve of length at most $2\Delta_0+1$, implying that $d(x,y)<5\Delta_0$, a contradiction.
\end{proof}

For consistency of notation, we will denote $Z_0(x)=C_x$ and $Z_0(y)=C_y$, even though both $C_x$ and $C_y$ are $G$-normal curves and not cycles.

Let $U=\bigcup_{h=0}^{2\Delta_0}V(Z_h(x))$, and let $U'=\bigcup_{h'=0}^{2\Delta_0}V(Z_{h'}(y))$.
For each $1\leq h\leq 2\Delta_0$, let $U_h$ be the set of vertices lying in $\dnot(Z_h(x))\setminus D(Z_{h-1}(x))$, and let $\rset_h$ be the set of all connected components of $G[U_h]$. For each $1\leq h'\leq 2\Delta_0$, we define $U'_{h'}$ and $\rset'_{h'}$ with respect to the shell $\zset(y)$ similarly. Let $\rset=\bigcup_{h=1}^{2\Delta_0}\rset_h$ and let $\rset'=\bigcup_{h'=1}^{2\Delta_0}\rset'_{h'}$.

Our next step is to define a mapping $\beta:S\rightarrow 2^U$ of all source vertices in $S$ to subsets of vertices of $U$, and a mapping $\beta': T\rightarrow 2^{U'}$ of all destination vertices in $T$ to subsets of vertices of $U'$. Every vertex in $S\cup T$ will be mapped to a subset of at most three vertices. We will then replace each demand pair $(s,t)$ with the set $\beta(s)\times \beta'(t)$ of demand pairs. Eventually, for every pair $0\leq h,h'\leq 2\Delta_0$ of indices, we will define a subset $\tmset_{h,h'}$ of the new demand pairs, containing all pairs whose sources lie on $Z_h(x)$ and destinations lie on $Z_{h'}(y)$, to obtain an instance of the special case, that will then be solved using Theorem~\ref{thm: case 2a - special case}.


We start by defining the mapping of the sources. For every source $s\in S$, we will also define a curve $\Gamma(s)$ that we will use later in our analysis. 
First, for every source vertex $s\in S\cap U$, we set $\beta(s)=\set{s}$, and we let $\Gamma(s)$ contain the vertex $s$ only.
Next, fix any vertex $v^*\in C_x$. For every source vertex $s\in V(D_x)\setminus V(C_x)$, we set $\beta(s)=\set{v^*}$, and $\Gamma(s)=C_x$. Finally, consider some component $R\in \rset_h$, for some $1\leq h\leq 2\Delta_0$, and let $s\in S\cap V(R)$ be any source lying in $R$.  If $|L(R)|\leq 2$, then we let $\beta(s)=L(R)\cup \set{u(R)}$ if $u(R)$ is defined, and $\beta(s)=L(R)$ otherwise. If $|L(R)|>2$, then we let $\beta(s)=\set{v}$, where $v$ is a leg of $R$, such that $v$ is not an endpoint of $\sigma(R)$ (so it is not the first and not the last leg of $R$).  We then let $\Gamma(s)$ be the outer boundary $\gamma(R)$ of the disc $\eta(R)$ given by Theorem~\ref{thm: curves around CC's}. Recall that the length of $\Gamma(s)$ is bounded by $4\Delta_0+\Delta/2+1< 5\Delta_0$. Notice that for every source $s\in S$, we now have defined a $G$-normal curve $\Gamma(s)$ of length at most $5\Delta_0$. An important property of this curve is that the disc whose boundary is $\Gamma(s)$ cannot contain any demand pair in $\nset$, as the disc itself is contained in $D(Z_{2\Delta_0}(x))$. We define the mapping $\beta': T \rightarrow 2^{U'}$, and the corresponding curves $\Gamma(t)$ for the destination vertices $t\in T$ similarly.

Let $\tmset=\bigcup_{ (s,t)\in \nset}\beta(s)\times \beta'(t)$.  In the following two theorems, we show that the problems $(G,\nset)$ and $(G,\tmset)$ are equivalent to within relatively small factors.

\begin{theorem}\label{thm: from new to old problem}
There is an efficient algorithm, that, given any solution to instance $(G,\tmset)$, that routes $\kappa$ demand pairs, finds a solution to instance $(G,\nset)$, routing at least $\frac{\kappa}{21\Delta_0}$ demand pairs.
\end{theorem}

\begin{proof}
Let $\pset_0$ be any collection of disjoint paths in graph $G$, routing a subset $\tmset_0\subseteq \tmset$ of $\kappa$ demand pairs. We assume that $\kappa\geq 21\Delta_0$, as otherwise we can return a routing of a single demand pair in $\nset$.
For every demand pair $(\tilde s,\tilde t)\in \tmset_0$, let $(s,t)$ be any corresponding demand pair in $\nset$, that is, $\tilde s\in \beta(s)$ and $\tilde t\in \beta'(t)$.

We build a conflict graph $H$, whose vertex set is $\set{v(\tilde s,\tilde t)\mid (\tilde s,\tilde t)\in \tmset_0}$, and there is a directed edge from $v(\tilde s_1,\tilde t_1)$ to $v(\tilde s_2,\tilde t_2)$ iff the unique path $P(\tilde s_1,\tilde t_1)\in \pset_0$ routing the pair $(\tilde s_1,\tilde t_1)$ intersects $\Gamma(s_2)$ or $\Gamma(t_2)$ (in which case we say that there is a conflict between $(\tilde s_1,\tilde t_1)$ and $(\tilde s_2,\tilde t_2)$). Since all paths in $\pset_0$ are node-disjoint, and all curves $\Gamma(s),\Gamma(t)$ have lengths at most $5\Delta_0$, the in-degree of every vertex in $H$ is at most $10\Delta_0$. Therefore, we can efficiently compute an independent set $I$ of size at least $\frac{\kappa}{20\Delta_0+1}\geq\frac{\kappa}{21\Delta_0}$ in $H$.

Let $\tmset_1=\set{(\tilde s,\tilde t)\mid v(\tilde s,\tilde t)\in I}$, and let $\pset_1\subseteq\pset_0$ be the set of paths routing the demand pairs in $\tmset_1$. Let  $\mset'=\set{(s,t)\mid (\tilde s,\tilde t)\in \tmset_1}$.  It is now enough to show that all demand pairs in $\mset'$ can be routed in $G$. Consider any demand pair $(\tilde s,\tilde t)\in \tmset_1$, and let $P\in \pset_1$ be the path routing $(\tilde s,\tilde t)$. We will extend the path $P$, so it connects $s$ to $t$. Notice that if $s\in U$, then $s=\tilde s$. Assume now that $s\not \in U$. If $s\in V(D_x)\setminus V(C_x)$, then $\tilde s\in V(C_x)$. Since $G[V(D_x)]$ is a connected graph, we can extend path $P$ inside the disc $D_x$, so it now originates at $s$. As $\Gamma(s)=C_x$, no other source of a demand pair in $\tmset_1$ may lie on $C_x$, and no other path in $\pset_1$ contains a vertex of $D_x$. Finally, assume that $s\in V(R)$ for some component $R\in \rset$. Since the disc whose boundary is $\Gamma(s)$ contains $R$, all vertices of $L(R)$, and $u(R)$ (if such is defined), no other path in $\pset_1$ may contain a vertex of $L(R)\cup\set{u(R)}$. Moreover, since no demand pair in $\tmset$ is contained in the disc whose boundary is $\Gamma(s)$, no other path in $\pset_1$ may intersect $R$. We extend the path $P$ inside $R$, so it now originates at $s$. We perform the same transformation to path $P$ to ensure that it terminates at $t$. It is easy to see that the resulting collection of paths is disjoint.
\end{proof}

\begin{theorem}\label{thm: from old to new problem}
$\opt(G,\tmset)\geq \frac{\opt(G,\nset)}{21\Delta_0}$.
\end{theorem}

\begin{proof}
Let $\pset_0$ be the set of paths in the optimal solution to instance $(G,\nset)$, and let $\mset_0$ be the set of the demand pairs they route. 

As before, we define a conflict graph $H$, whose vertex set is $\set{v(s,t)\mid (s,t)\in \mset_0}$, and there is a directed edge from $v(s_1,t_1)$  to $v(s_2,t_2)$ iff the unique path $P(s_1,t_1)\in \pset_0$ routing the pair $(s_1,t_1)$ intersects $\Gamma(s_2)$ or $\Gamma(t_2)$ (in which case we say that there is a conflict between $(s_1,t_1)$ and $(s_2,t_2)$). Since all paths in $\pset_0$ are node-disjoint, and all curves $\Gamma(s),\Gamma(t)$ have lengths at most $5\Delta_0$, the in-degree of every vertex in $H$ is at most $10\Delta_0$. Therefore, there is an independent set $I$ of size at least $\frac{\opt(G,\nset)}{20\Delta_0+1}\geq\frac{\opt(G,\nset)}{21\Delta_0}$ in $H$.

Let $\mset_1=\set{(s,t)\mid v(s,t)\in I}$, and let $\pset_1\subseteq \pset_0$ be the set of paths routing the demand pairs in $\mset_1$. We show that we can route $|\mset_1|$ demand pairs of $\tmset$ in $G$ via node-disjoint paths. Let $S_1$ and $T_1$ be the sets of all source and all destination vertices of the pairs in $\mset_1$, respectively.

Consider any source vertex $s\in S_1$. We say that $s$ is a \emph{good source vertex} if $s\in U$, or $s$ belongs to some component $R\in \rset$, such that $|L(R)|\leq 2$. Otherwise, $s$ is a \emph{bad source vertex}. Notice that if $s$ is a good source vertex, then the path $P\in \pset_1$ that originates at $s$ must contain a vertex $s'\in \beta(s)$: if $s\in U$, then $\beta(s)=\set{s}$; otherwise, if $s\in R$ for some component $R\in \rset$ with $|L(R)|\leq 2$, then $\beta(s)=L(R)\cup \set{u(R)}$ if $u(R)$ is defined, and $\beta(s)=L(R)$ otherwise. In either case, in order to enter $R$, path $P$ has to visit a vertex of $\beta(s)$. Therefore, if $s$ is a good source vertex, then some vertex $s'\in P$ belongs to $\beta(s)$. Similarly, we say that a destination vertex $t\in T_1$ is a \emph{good destination vertex} if $t\in U'$ or $t$ belongs to some component $R'\in \rset'$ with $|L(R')|\leq 2$. Otherwise, it is a \emph{bad destination vertex}. As before, if $t$ is a good destination vertex, then the path $P\in \pset_1$ terminating at $t$ must contain some vertex $t'\in \beta(t)$.

We transform the paths in $\pset_1$ in two steps, to ensure that they connect demand pairs in $\tmset$. In the first step, for every path $P\in \pset_1$ originating at a good source $s\in S_1$, we truncate $P$ at the first vertex $s'\in \beta(s)$, so it now originates at $s'$. Similarly, if $P$ terminates at a good destination vertex $t\in T_1$, we truncate $P$ at the last vertex $t'\in \beta(t)$, so it now terminates at $t'$. Let $\pset_1'$ be the resulting set of paths. Notice that the paths in $\pset_1'$ remain node-disjoint.

In order to complete our transformation, we need to take care of bad source and destination vertices. Let $s\in S_1$ be any bad destination vertex. If $s\in V(D_x)\setminus V(C_x)$, then let $Q_s$ be any path connecting $s$ to the unique vertex $s'\in \beta(s)$, so that $Q_s\subseteq D_x$. Such a path exists, since $G[V(D_x)]$ is connected. Otherwise, $s\in R$ for some component $R\in \rset$ with $|L(R)|\geq 3$. Recall that in this case, $\beta(s)$ contains a unique vertex, that we denote by $s'$, which is a leg of $R$, and it is not one of the endpoints of $\sigma(R)$. We then let $Q(s)$ be any path connecting $s$ to $s'$ in the sub-graph of $G$ induced by $V(R)\cup \set{s'}$. We define paths $Q(t)$ for bad destination vertices $t\in T_1$ similarly. By concatenating the paths in $\set{Q(s)}$ for all bad source vertices $s\in S_1$, $\pset_1'$, and $\set{Q(t)}$ for all bad destination vertices $t\in T_1$, we obtain the desired collection $\tilde{\pset}$ of at least $\frac{\opt(G,\nset)}{21\Delta_0}$ paths, routing demand pairs in $\tmset$. It now only remains to show that the paths in $\tpset$ are disjoint. Recall that the paths in $\pset_1'$ were node-disjoint.

 Consider some bad source vertex $s\in S_1$, and let $P\in \pset_1'$ be the path originating at $s$. We show that $Q(s)$ is disjoint from all paths in $\pset_1'\setminus\set{P}$, and it is disjoint from all other paths $Q(s_1)$ for $s_1\in S_1\setminus\set{s}$.
 
 Assume first that $s\in V(D_x)\setminus V(C_x)$. Then $\Gamma(s)=C_x$, and so no other path in $\pset_1$ (and hence in $\pset_1')$ can contain a vertex of $D_x$. It follows that $Q(s)$ is disjoint from all paths in $\pset_1'\setminus\set{P}$. It is also disjoint from all other paths  $Q(s_1)$ for $s_1\in S_1\setminus\set{s}$, since in order for $Q(s_1)$ to intersect $D_x$, vertex $s_1$ must lie on $C_x$, and this is impossible.
 
 Assume now that $s\in V(R)$ for some $R\in \rset$. Recall that the disc whose boundary is $\Gamma(s)$ contains $R\cup L(R)$. Since no other path in $\pset_1'$ may intersect $\Gamma(s)$, and no demand pair is contained in the disc whose boundary is $\Gamma(s)$, no other path in $\pset_1'$ intersects $R\cup L(R)$, and so all such paths are disjoint from $Q(s)$. Consider now some other bad source vertex $s_1\in S_1$. Note that $s_1$ cannot lie in $R$, since in this case the path of $\pset_1'$ originating at $s_1$ would have crossed $\Gamma(s)$. Therefore, $s_1$ must lie in some other component $R'\in \rset$. Then the only way for $Q(s)$ and $Q(s_1)$ to intersect is when $s'=s_1'$. In particular, $R,R'$ should both belong to some set $\rset_h$ for $1\leq h\leq 2\Delta_0$. But since the segments $\set{\sigma(R)\mid R\in \rset_h}$ are nested, due to the way we chose the mappings $\beta(s)$ and $\beta(s_1)$, this is impossible.
 
 We can similarly prove that for each bad destination vertex $t\in T_1$, if $P'\in \pset_1'$ is the path terminating at $t$, then $Q(t)$ is disjoint from all paths in $\pset_1'\setminus\set{P'}$, and from all paths $Q(t_1)$, where $t_1\in T_1\setminus\set{t}$ is a bad destination vertex. Altogether, this proves that the paths in $\tpset$ are disjoint.
\end{proof}

For each $0\leq h,h'\leq 2\Delta_0$, let $\tmset_{h,h'}\subseteq \tmset$ be the set of all demand pairs $(\tilde s,\tilde t)$ with $\tilde s\in Z_h(x)$ and $\tilde t\in Z_{h'}(y)$. 

If $h=0$ or $h'=0$, then, since $|V(C_x)|,|V(C_y)|\leq \Delta$, $\opt(G,\tmset_{h,h'})\leq \Delta$. We route any demand pair in $\tmset_{h,h'}$ to obtain a factor-$\Delta$ approximation to the problem $(G,\tmset_{h,h'})$. If  both $h,h'>0$, then we apply Theorem~\ref{thm: case 2a - special case} to obtain a collection $\pset_{h,h'}$ of at least $\Omega\left(\frac{\opt(G,\tmset_{h,h'})}{\Delta_0^2\log n}\right )$ disjoint paths, routing demand pairs in $\tmset_{h,h'}$. We then take the best among all resulting solutions.

Notice that $\set{\tmset_{h,h'}\mid 0\leq h,h'\leq 2\Delta_0}$ partition the set $\tmset$ of demand pairs, and so there is a pair $0\leq h,h'\leq 2\Delta_0$ of indices with $\opt(G,\tmset_{h,h'})\geq \frac{\opt(G,\tmset)}{(2\Delta_0+1)^2}\geq \Omega\left(\frac{\opt(G,\nset)}{\Delta_0^3}\right)$. Therefore, we obtain a routing of at least $\Omega\left(\frac{\opt(G,\nset)}{\Delta_0^6\log n}\right )$ demand pairs.

\label{-------------------------------------------subsec: case 2b----------------------------------------}
\subsection{Subcase 2b: $d(x,y)\leq 5\Delta_0$}\label{subsec: Case 2b}
We again start by defining a special case of the problem, which is similar to the problem of routing on a disc. We show an approximation algorithm for this special case that reduces it to the problem of routing on a disc, and we later use this special case in order to handle the general problem in Case 2b.

\subsubsection*{A Special Case}
Suppose we are given a connected planar graph $\hat G$ embedded on the sphere, a cycle $Z$ in $\hat G$, and a set $\hmset$ of demand pairs, such that each terminal of $\tset(\hmset)$ lies on $Z$.  Assume additionally that we are given a closed $\hat G$-normal  curve $C$  of length at most $\Delta$, that is disjoint from $Z$. Let $D(Z)$ be the disc whose boundary is $Z$, which contains $C$.  Assume further that for every vertex $v\in Z$, there is a $\hat G$-normal curve $\gamma(v)$ of length at most $16\Delta_0$ connecting $v$ to a vertex of $C$, so that $\gamma(v)$ is contained in $D(Z)$ and it is internally disjoint from $C$. This finishes the definition of the special case.

Next, we reduce this special case to routing on a disc, by creating a hole in the sphere. The hole is $\dnot(Z)$, so we delete all edges and vertices that appear inside $D(Z)$, except for the edges and the vertices of $Z$. Let $\hat G'$ be the resulting graph. We can now apply the $O(\log n)$-approximation algorithm for \NDPdisc to the resulting problem, to obtain a routing of at least $\Omega\left(\frac{\opt(\hat G',\hmset)}{\log n}\right)$ demand pairs. In order to complete the analysis of the algorithm, we prove that $\opt(\hat G',\hmset)\geq \Omega\left(\frac{\opt(\hat G,\hmset)}{\Delta_0^2\log n}\right )$.


\begin{theorem}\label{thm: case 2b - special case}
$\opt(\hat G',\hmset)\geq \Omega\left(\frac{\opt(\hat G,\hmset)}{\Delta_0^2\log n}\right )$.
\end{theorem}

\begin{proof}
As before, we can assume without loss of generality that all curves in set $\set{\gamma(v)\mid v\in V(Z)}$ are mutually non-crossing, and for all $v,v'\in V(Z)$, $\gamma(v)\cap \gamma(v')$ is a contiguous curve that has a vertex of $C$ as its endpoint. 

Consider the optimal solution to instance $(\hat G,\hmset)$, and let $\pset_0$ be the set of paths in this solution. As before, we will gradually modify the set $\pset_0$ of paths, to obtain path sets $\pset_1,\pset_2,\ldots$, until we obtain a feasible solution to instance $(\hat G',\hmset)$. For every $i\geq 0$, we will denote by $\hmset_i$ the set of the demand pairs routed by $\pset_i$, and by $\kappa_i$ its cardinality.
Recall that $\kappa_0=\opt(\hat G,\hmset)$. We assume that $\kappa_0\geq 2^{13}\Delta_0^2\log n$, as otherwise a solution routing a single demand pair gives the desired approximation, and such a solution exists in $\hat G'$, as it must be connected.

We delete from $\pset_0$ all paths that use the vertices of $C$. Since $C$ contains at most $\Delta$ vertices, we delete at most $\Delta$ paths in this step. Let $\pset_1$ be the resulting set of paths, and $\hmset_1$ the set of the demand pairs routed by $\pset_1$.

Our next step is to build a conflict graph $H$, almost exactly as before. The set of vertices is $V(H)=\set{v(s,t)\mid (s,t)\in \hmset_1}$. There is a directed edge from $v(s,t)$ to $v(s',t')$, iff the path $P(s,t)\in \pset_1$ routing the pair $(s,t)$ contains a vertex of $V(\gamma(s'))\cup V(\gamma(t'))$, and we say that there is a conflict between $(s,t)$ and $(s',t')$ in this case. Since we assume that the paths in $\pset_1$ are node-disjoint, and since $|V(\gamma(s'))|,|V(\gamma(t'))|\leq 16\Delta_0$ for all $(s',t')\in \hmset_1$, the in-degree of every vertex in $H$ is at most $32\Delta_0$. As before,  there is an independent set $I\subseteq V(H)$ of cardinality at least $\frac{|\pset_1|}{32\Delta_0+1}\geq \frac{|\pset_0|}{64\Delta_0}$ in $H$.

Let $\hmset_2$ be the set of all demand pairs $(s,t)$ with $v(s,t)\in I$, and let $\pset_2\subseteq \pset_1$ be the set of paths routing the demand pairs in $\hmset_2$. Recall that the paths in $\pset_2$ are disjoint from $C$. Moreover, if $P(s,t)\in \pset_2$ is the path routing the pair $(s,t)\in \hmset_2$, then for every demand pair $(s',t')\neq (s,t)$ in $\hmset_2$, path $P(s,t)$ is disjoint from both $\gamma(s')$ and $\gamma(t')$. It is now easy to verify that the demand pairs in $\hmset_2$ are non-crossing with respect to the cycle $Z$.

We now depart from the proof of Theorem~\ref{thm: case 2a - special case}. We use Lemma~\ref{lem: getting r-split demand pairs on a disc}, in order to compute a partition $(\nset_1,\ldots,\nset_{4\ceil{\log n}})$ of the set $\hmset_2$ of the demand pairs, so that for all $1\leq a\leq 4\ceil{\log n}$, set $\nset_a$ is $r_a$-split, for some $r_a\geq 0$. Then there must be an index $1\leq a\leq 4\ceil{\log n}$, such that $|\nset_a|\geq \frac{|\hmset_2|}{4\ceil{\log n}}\geq \Omega\left(\frac{|\pset_0|}{\Delta_0\log n}\right )$. We let $\hmset_3=\nset_a$, and $\pset_3\subseteq \pset_2$ the set of paths routing the demand pairs in $\hmset_3$. We then denote $r_a$ by $\rho$, and the partition of $\hmset_3$ associated with the definition of the $\rho$-split set of demand pairs by $\mset^1,\ldots,\mset^{\rho}$. We also denote by $\Sigma=(\sigma_1,\ldots,\sigma_{2\rho})$ the corresponding partition of $Z$ into intervals. We assume without loss of generality that for all $1\leq z\leq \rho$, all source vertices of the demand pairs in $\mset^z$ lie on $\sigma_{2z-1}$, and all corresponding destination vertices lie on $\sigma_{2z}$.

Let $I_1$ contain all indices $1\leq z\leq \rho$, with $|\mset^z|\leq 128\Delta_0$, and let $I_2$ contain all remaining indices. If $\sum_{z\in I_1}|\mset^z|\geq |\hmset_3|/2$, then we can obtain a routing of at least $\frac{|\hmset_3|}{256\Delta_0}\geq \Omega\left(\frac{|\pset_0|}{\Delta_0^2\log n}\right )$ demand pairs, as follows: for each $z \in I_1$, we route any demand pair in $\mset^z$ via the segment $\sigma_{2z-1}\cup \sigma_{2z}$ of $Z$. Therefore, we assume from now on that $\sum_{z\in I_2}|\mset^z|\geq \frac{|\hmset_3|}{2}\geq \Omega\left(\frac{|\pset_0|}{\Delta_0\log n}\right )$. We denote $\hmset_4=\bigcup_{z\in I_2}\mset^z$, and we let $\pset_4\subseteq \pset_3$ be the set of paths routing the demand pairs in $\hmset_4$. For all $z\in I_2$, we denote $|\mset^z|$ by $\kappa_4^z$.

The rest of the proof is very similar to the rest of the proof of Theorem~\ref{thm: case 2a - special case}, except that now we deal with each subset $\mset^z$ for $z\in I_2$ of the demand pairs separately. Fix some $z\in I_2$, and denote $\mset^z=\set{(s^z_1,t^z_1),\ldots,(s^z_{\kappa^z_4},t^z_{\kappa^z_4})}$. Since the demand pairs in $\mset^z$ are non-crossing, and due to the definition of the $\rho$-split instance, we can assume without loss of generality that $s_1^z,\ldots,s_{\kappa^z_4}^z,t_{\kappa^z_4}^z,\ldots,t_1^z$ appear in this counter-clock-wise order on $Z$.
 Let $\kappa^z_5=\floor{\frac{\kappa_4^z}{64\Delta_0}}-1$. Since $\kappa_4^z\geq 128\Delta_0$, $\kappa^z_5\geq \frac{\kappa_4^z}{256\Delta_0}$.
 We then let $\hmset^z$ contain all demand pairs $(s^z_{64\Delta_0r},t^z_{64\Delta_0r})$, for $1\leq r\leq \kappa^z_5$, and we let $\hmset_5=\bigcup_{z\in I_2}\hmset^z$. Notice that $\kappa_5=|\hmset_5|=\sum_{z\in I_2}\kappa^z_5\geq \sum_{z\in I_2}\frac{\kappa_4^z}{256\Delta_0}\geq \frac{\kappa_4}{256\Delta_0}\geq \Omega\left(\frac{\kappa_0}{{\Delta_0^2\log n}}\right )$. We now show that all demand pairs in $\hmset_5$ can be routed in graph $\hat G'$. This is done similarly to the proof of Theorem~\ref{thm: case 2a - special case}.
 
 Fix some $z\in I_2$, and consider the set $\hmset^z=\set{(s^z_{64\Delta_0r},t^z_{64\Delta_0r})\mid 1\leq r\leq \kappa^z_5}$ of demand pairs. For each such demand pair $(s^z_{64\Delta_0r},t^z_{64\Delta_0r})$, we define two segments $\mu^z_r\subseteq \sigma_{2z-1}$ and $\tilde{\mu}^z_r\subseteq \sigma_{2z}$ of $Z$, as follows. For brevity of notation, denote $64\Delta_0r$ by $\ell$. The first segment, $\mu^z_r$, is  the segment of $Z$ from $s^z_{\ell-32\Delta_0}$ to $s^z_{\ell+32\Delta_0-1}$, as we traverse $Z$ in the counter-clock-wise order. The second segment, $\tilde{\mu}^z_r$, is the segment of $Z$ from $t^z_{\ell-32\Delta_0}$ to $t^z_{\ell+32\Delta_0-1}$, as we traverse $Z'$ in  the clock-wise order. It is immediate to verify that all segments of $Z$ in $\set{\mu^z_r,\tilde{\mu}^z_r\mid z\in I_2, 1\leq r\leq \kappa_5^z}$ are mutually disjoint. We use the following analogue of Lemma~\ref{lemma: routing and segment intersections}.

\begin{lemma}\label{lemma: routing and segment intersections2}
For every $z\in I_2$, for every demand pair $(s^z_{{64\Delta_0 r}},t^z_{{64\Delta_0 r}})\in \hmset^z$, if $P\in \pset_5$ is the path routing this pair, then $P\cap Z\subseteq \mu^z_{{r}}\cup \tilde {\mu}^z_{r}$.
\end{lemma}

We can now use this lemma to re-route the paths in $\pset_5$, similarly to the proof of Theorem~\ref{thm: main for Case 2}. 
Fix some $z\in I_2$, and let $(s^z_{64\Delta_0r},t^z_{64\Delta_0r})\in\hmset^z$ be some demand pair, with $1\leq r\leq \kappa_5^z$. For simplicity, we denote $\ell=64\Delta_0r$. Let $P_{\ell}$ be the path routing this demand pair in $\pset_5$, that we view as directed from $s_{\ell}$ towards $t_{\ell}$. Clearly, $P_{\ell}$ intersects both $\mu^z_{r}$ and $\tilde{\mu}^z_{r}$. Let $v_{\ell}$ be the last vertex of $P_{\ell}$ lying on $\mu^z_{r}$. Then there is some other vertex appearing on $P_{\ell}$ after $v_{\ell}$ that belongs to $\tilde{\mu}^z_{r}$. We let $v'_{\ell}$ be the first such vertex on $P_{\ell}$, and we let $P'_{\ell}$ be the segment of $P_{\ell}$ between $v_{\ell}$ and $v'_{\ell}$. Let $P^*_{\ell}$ be the path obtained as follows: we start with a segment of $\mu^z_{r}$ between $s^z_{\ell}$ and $v_{\ell}$; we then follow $P'_{\ell}$ to $v'_{\ell}$, and finally we use a segment of $\tilde{\mu}^z_{r}$ between $v'_{\ell}$ and $t^z_{\ell}$. From Lemma~\ref{lemma: routing and segment intersections2}, it is immediate to verify that we obtain a set of node-disjoint paths routing all demand pairs in $\hmset_5$ in graph $\hat G'$. It now remains to prove Lemma~\ref{lemma: routing and segment intersections2}.

\begin{proof}
Fix some $z\in I_2$, and consider some demand pair $(s^z_{{64\Delta_0 r}},t^z_{{64\Delta_0 r}})\in \hmset^z$. Let $P\in \pset_5$ be the path routing this pair. We partition the cycle $Z$ into two segments: $\sigma$ containing $\sigma_{2z-1}$ and $\sigma'$ containing $\sigma_{2z}$ arbitrarily, so $\sigma$ now contains all source vertices, and $\sigma'$ now contains all destination vertices of the pairs in $\mset^z$. It is enough to show that $P\cap \sigma \subseteq \mu_r^z$ and  $P\cap \sigma'\subseteq \tilde{\mu}_r^z$.
We prove that $P\cap \sigma \subseteq \mu_r^z$. The proof that $P\cap \sigma'\subseteq \tilde{\mu}_r^z$ is symmetric.

Assume for contradiction that $P\cap \sigma \not\subseteq \mu_r^z$, and let $v$ be the first vertex on $P$ that belongs to $\sigma\setminus   \mu_{{ r}}^z$, where we view $P$ as directed from $s^z_{\ell}$ to $t^z_{\ell}$.  Let $P'$ be the sub-path of $P$ from its first vertex to $v$. Consider the planar embedding of $\hat G$, where we fix any face lying outside of $D(Z)$ as the outer face. In this planar embedding, denote by $Y$ the union of $D(C),\gamma(s^z_{\ell}),P'$, and $\gamma(v)$, and let $R$ be the outer boundary of $Y$ (see Figure~\ref{fig: 1b end}).  

\begin{figure}[h]
 \centering
\scalebox{0.3}{\includegraphics{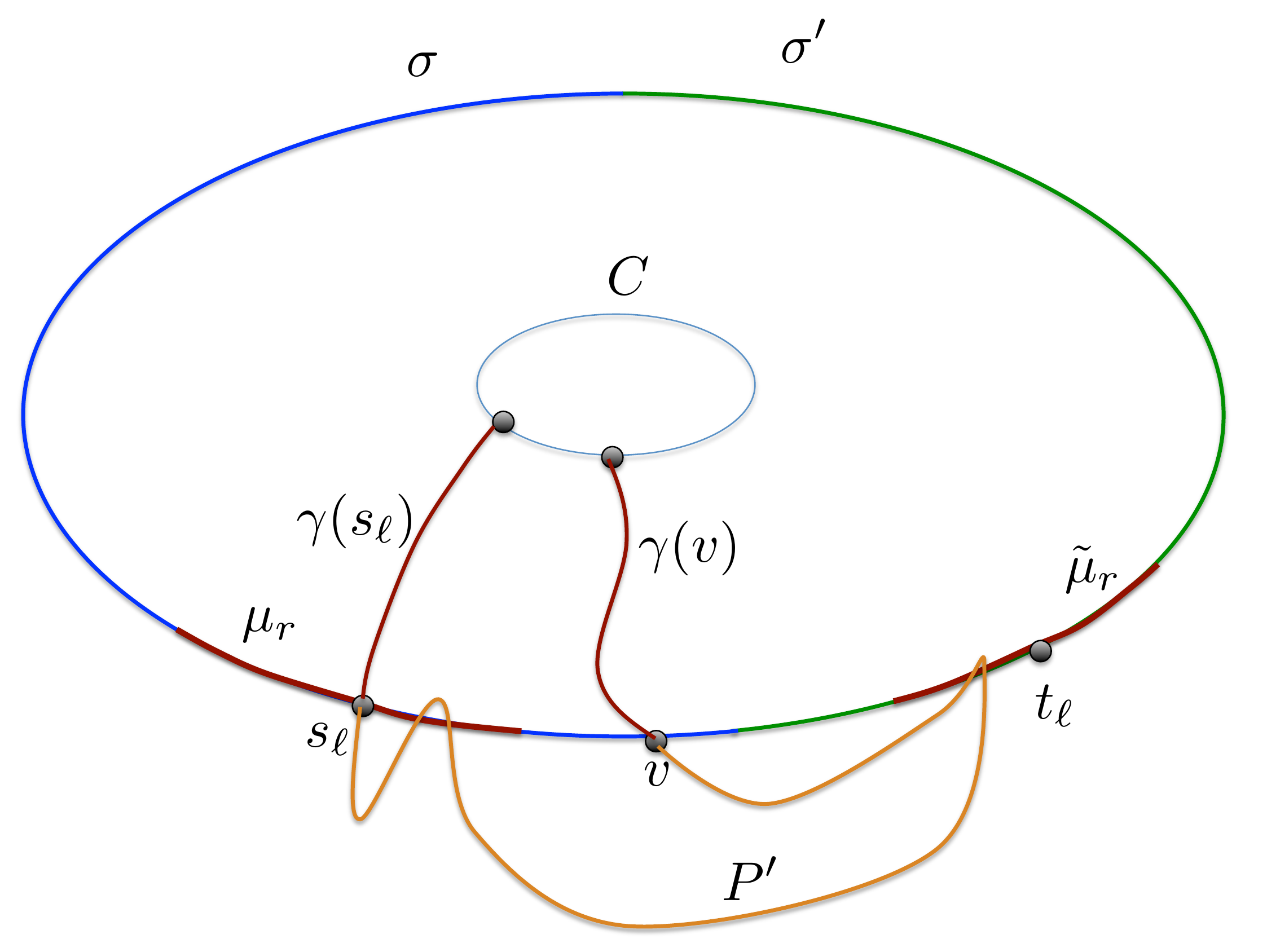}}\caption{Constructing the curve $R$\label{fig: 1b end}. The segments $\sigma$ and $\sigma'$ of $Z$ are shown in blue and green, respectively. We omit the superscript $z$ for brevity of notation.}
\end{figure}

Notice that, since there are no conflicts in $\pset_4$, curve $R$ does not cross any curve in the set $C\cup\set{\gamma(t)\mid t\in \tset(\hmset_4)\setminus\set{s^z_{\ell},t^z_{\ell}}}$, and so, since both $s^z_{\ell}$ and $v$ lie on $\sigma$, all destination vertices of the demand pairs in $\mset^z$, except for possibly $t^z_{\ell}$, lie on  one side of $R$. Let $S_1=\set{s^z_{\ell-32\Delta_0},\ldots,s^z_{\ell-1}}$, and $S_2=\set{s^z_{\ell+1},\ldots, s^z_{\ell+32\Delta_0-1}}$. 
Denote $\Gamma_1=\set{\gamma(s^z_i)\mid s^z_i\in S_1}$ and $\Gamma_2=\set{\gamma(s^z_i)\mid s^z_i\in S_2}$. Since the curves in set $\set{\gamma(u)\mid u\in V(Z)}$ are non-crossing, if $\beta,\beta'$ are the two segments of $Z$ whose endpoints are $s_{\ell}$ and $v$, then all vertices of $S_1\setminus\set{v}$ lie on one of the segments (say $\beta$), while all vertices of $S_2\setminus\set{v}$ lie on the other segment. Moreover, since path $P'$ cannot cross any curve in $\Gamma_1\cup \Gamma_2$, the vertices of $S_1\setminus\set{v}$ lie on one side of $R$, and the vertices of $S_2\setminus\set{v}$ lie on the other side of $R$. Therefore, either all vertices of $S_1\setminus\set{v}$ are separated by $R$ from all destination vertices of the demand pairs in $\mset^z$ (except for possibly $t^z_{\ell}$), or the same holds for $S_2\setminus\set{v}$ - we assume without loss of generality that it is the former.

 All sources of $S_1\setminus\set{v}$ are then separated by $R$ from their destinations vertices, and yet all corresponding demand pairs are routed by $\pset_2$. Therefore, at least $32\Delta_0-2$ paths in $\pset_2$ must cross the curve $R$. Recall that none of these paths can cross $C$, $P'$, or $\gamma(s^z_{\ell})$. Therefore, all these paths must cross $\gamma(v)$. But since the length of $\gamma(v)$ is at most $16\Delta_0$, and the paths are node-disjoint, this is impossible.
\end{proof}
\end{proof}

The rest of the proof follows the same strategy as the proof for Case 2a, but it is somewhat more involved. We break it into three steps. In the first step, we construct a single shell around the vertex $x$. In the second step, we map the terminals to the cycles of the shell. In the final step, we reduce the problem to the special case, by constructing a cycle $Z$ and mapping a subset of the terminals to the vertices of $Z$. 

We assume that $\opt(G,\nset)>2^{13}\Delta_0^4$, since otherwise we can route a single demand pair and obtain a valid solution.



\subsection*{Step 1: the Shell}

In this step we construct a shell $\zset(x)=(Z_1(x),\ldots,Z_{h}(x))$ of depth $h\leq 8\Delta_0$ around $x$. 
We first give a high-level explanation of why the construction of the shell in this case is more challenging than that in Case 2a, and motivate our construction.
Recall that given such a shell $\zset(x)$, we have defined, for all $1\leq h'\leq h$, a set $U_{h'}$ of vertices contained in $\dnot(Z_{h'}(x))\setminus D(Z_{h'-1}(x))$, and a set $\rset_{h'}$ of all connected components of $G[U_{h'}]$. For all $1\leq h'\leq h$, for every component $R\in \rset_{h'}$, we have defined the disc $\eta(R)$ (given by Theorem~\ref{thm: curves around CC's}), whose boundary $\gamma(R)$ served as the curve $\Gamma(t)$ for all terminals $t$ lying in $R$. In Case 2a, discs $\eta(R)$ had the important property that no demand pair in $\nset$ is contained in $\eta(R)$. This ensured that whenever a path routing any demand pair in $\nset$ intersects $R$, such a path must also cross $\gamma(R)$. This latter property was crucial in ensuring that, after we map all terminals to the vertices of the shell, the solution value does not change by much. Unfortunately, in Case 2b this is no longer true, and it is possible that discs $\eta(R)$, and even the components $R$ themselves, contain many demand pairs from $\nset$. We get around this problem as follows. We construct the shell around $x$ carefully, to ensure that the total number of the terminals contained in each such disc $\eta(R)$ is relatively small. We say that a pair of terminals is bad if there is a path connecting these terminals, which is completely contained in some disc $\eta(R)$, and it is good otherwise. As long as bad terminal pairs exist in our graph, we iteratively route one such pair inside the corresponding disc $\eta(R)$, and discard all other terminals lying in this disc. Since the total number of the terminals contained in each disc is relatively small, the number of the terminals we discard at this step is relatively small compared to the number of the demand pairs we route. If we manage to route a large enough number of demand pairs in this step, then we terminate the algorithm and return this solution. Otherwise, we let $\nset'$ be the subset of the demand pairs that have not been routed or discarded yet. Then $|\nset'|$ is sufficiently large relatively to $|\nset|$, and we have now achieved the property that for all components $R$, any path connecting a demand pair in $\nset'$ that intersects $R$ must also cross $\gamma(R)$. 
We now proceed to describe the shell construction.  Let $\tau^*=64\Delta_0/\alphaWL$, and let $\tset$ be the set of all terminals participating in the demand pairs in $\nset$.


For all integers $i>0$, let $S_i$ be the set of all vertices $v\not\in D_x$, with $\dface(v,V(C_x))\geq i$. We say that a connected component $R$ of $G[S_i]$ is \emph{heavy} iff $R$ contains more than $\tau^*$ terminals of $\tset$. Let $i^*$ be the largest integer, such that $G[S_{i^*}]$ contains at least one heavy connected component, and let $R^*$ be any such component. We need the following easy observation.

\begin{observation}
$i^*\leq 8\Delta_0-\Delta-4$.
\end{observation} 
\begin{proof}
Assume otherwise and let $h=8\Delta_0-\Delta-4$. Then there is some connected component $R$ of $G[S_h]$, containing more than $\tau^*$ terminals of $\tset$. Let $t\in V(R)\cap \tset$ be any such terminal. If $t\in X$, then $d(t,x)\leq \Delta_0$. If $t\in Y$, then $d(t,y)\leq \Delta_0$, and $d(x,y)\leq 5\Delta_0$, so from Observation~\ref{obs: weak-triangle-inequality}, $d(t,x)\leq 6\Delta_0+\Delta/2+1$. Therefore, some vertex $v\in C_t$ has a $G$-normal curve of length at most  $6\Delta_0+\Delta/2+1$ connecting it to a vertex of $C_x$, and so every vertex in $C_t$ has a $G$-normal curve of length at most $6\Delta_0+\Delta/2+1+\Delta/2+1< 8\Delta_0-\Delta-4$ connecting it to a vertex of $C_x$. Therefore, $V(C_t)\subseteq V(G)\setminus S_h$, while $t\in S_h$. Since $V(C_t)$ separates $V(D_t)$ from all remaining vertices of $G$, and $R$ is a connected component of $G[S_h]$, it follows that $R\subseteq V(D_t)$. But $D_t$ contains fewer than $\tau^*$ terminals from the definition of enclosures, a contradiction.
\end{proof}

Let $\kappa=|\tset\cap R^*|$, so $\kappa>\tau^*$. The following observation will be useful in order to bound the number of terminals contained in each disc $\eta(R)$. The proof follows immediately from the well-linkedness of the terminals.

\begin{observation}\label{obs: number of terminals in disc}
Let $D$ be any disc, whose boundary $\gamma$ is a $G$-normal curve of length less than $16\Delta_0$. Assume further that at least $\kappa/4$ terminals lie outside of $D$. Then $D$ contains at most $\tau^*$ terminals of $\tset$.
\end{observation}

We now provide some further intuition. Let $h'=i^*-1$, and consider a shell $\zset(x)=(Z_1(x),\ldots,Z_{h'}(x))$ of depth $h'$ around $x$, where we let $F_x$ be any face incident on a vertex of $R^*$. We can then define, for each $1\leq h''\leq h'$, the set $\rset_{h''}$ of connected components of the graph induced by all vertices lying in $D^{\circ}(Z_{h''}(x))\setminus D(Z_{h''-1}(x))$, set $\rset=\bigcup_{h''=1}^{h'}\rset_{h''}$, and use Theorem~\ref{thm: curves around CC's} in order to compute the discs $\eta(R)$ for the components $R\in \rset$. Moreover, since all vertices of $R^*$ lie outside each such disc $\eta(R)$, from Observation~\ref{obs: number of terminals in disc}, each such disc contains at most $\tau^*$ terminals. The problem is that $R^*$ may still contain many terminals, while we need to ensure that most of the terminals lie in the components of $\rset$. We get around this problem by extending the shell, and adding two outer cycles to it. First, we consider the outer boundary $\Gamma$ of the  graph $R^*$ (in the drawing where the face containing $x$ is viewed as the outer face), and carefully select some cycle $C\in \Gamma$, so that, if we add $C$ as the outer-most cycle to the shell $\zset(x)$, by setting $Z_{h'+1}(x)=C$, then for every component $R\in \rset_{h'+1}$, we still maintain the property that the corresponding disc $\eta(R)$ contains at most $\tau^*$ terminals. Finally, we take care of the terminals contained in the disc $D(C)$, that currently do not lie in any component of $\rset$. The idea is to carefully select one vertex $\tilde u\in V(C)$, and to attach a new cycle $C'$ to $\tilde u$, that lies ``inside'' cycle $C$, and is then added to the shell as the outermost cycle, so $Z_{h'+2}(x)=C'$. We then view the face whose boundary is $C'$ as the face $F_x$ in the shell construction. Therefore, once $\tilde u$ is selected and cycle $C'$ is added to the drawing of $G$, both the construction of the shell, and the construction of the discs $\eta(R)$ for $R\in \rset$ are fixed. We would like to select $\tilde u$ in such a way that each resulting disc $\eta(R)$ for $R\in \rset_{h'+2}$ contains at most $\tau^*$ terminals. We achieve this by discarding a small number of terminals and their corresponding demand pairs. We now describe the construction more formally.

Consider the graph $R^*$, and its drawing in the plane, induced by the drawing of $G$ on the sphere, where we view the face where $x$ used to be as the outer face. Let $\Gamma$ be the boundary of the outer face in this drawing of $R^*$, and let $\cset$ be the set of all simple cycles in $\Gamma$. 
Let $H$ be the block-decomposition of $\Gamma$. That is, the set $V(H)$ of vertices consists of two subsets: set $V_1$ of cut vertices of $\Gamma$, and set $V_2$, containing a vertex $v_B$ for every block $B$ (a maximal $2$-connected component) of $\Gamma$. We add an edge $(u,v_B)$ between vertices $u\in V_1$ and $v_B\in V_2$ iff $u\in V(B)$. It is easy to see that graph $H$ is a tree, and we root it at any vertex. We next define weights for the vertices of $H$. In order to do so, every terminal in $\tset\cap R^*$ will contribute a weight of $1$ to one of the vertices of $V_2$, and the weight of every vertex in $V_2$ is then the total weight contributed to it. The weights of all vertices in $V_1$ are $0$. Consider some terminal $t\in \tset\cap R^*$. If there is some simple cycle $C\subseteq \Gamma$, such that $t\in D(C)$, then $t$ contributes the weight of $1$ to the vertex $v_B$, where $B=C$ (if $t$ belongs to several such cycles, then we select one of these cycles arbitrarily). Otherwise, $t\in V(\Gamma)$, and there is some vertex $v_B\in V_2$, such that block $B$ consists of a single edge $e$, and $t$ is one of its endpoints. Among all such vertices $v_B$, we choose one arbitrarily, and contribute the weight of $t$ to $v_B$.

For every subgraph $H'\subseteq H$, the weight of $H'$ is the total weight of all vertices in $H'$.  Clearly, the weight of $H$ is $\kappa$. We need the following simple claim.

\begin{claim}\label{claim: middle vertex}
There is some vertex $u^*\in V(H)$, such that, if we root $H$ at $u^*$, then for every child $u'$ of $u^*$, the weight of the sub-tree of $H$ rooted at $u'$ is at most $\kappa/2$.
\end{claim}
\begin{proof}
 We root $H$ at any vertex $v$, and set $u=v$. We then iterate. If the current vertex $u$ has a child $u'$, such that the total weight of all vertices contained in the sub-tree of $H$ rooted at $u'$ is more than $\kappa/2$, then we move $u$ to $u'$. It is easy to see that when this procedure terminates, we will find the desired vertex $u^*$.
\end{proof}

Consider the vertex $u^*$ computed by the above claim. If $u^*\in V_1$, then let $v^*=u^*$. Otherwise, $u^*=v_B$ for some vertex $v_B\in V_2$, then let $v^*$ be any vertex of $B$. We assume  without loss of generality that $v^*$ lies on some simple cycle $C\subseteq \Gamma$: otherwise, we create an artificial cycle $C=(v^*,u_1,u_2)$, where $u_1$ and $u_2$ are two new vertices. If $v^*$ lies on several such cycles, then we let $C$ be any one of them. Notice that every connected component of $R^*\setminus V(D(C))$ contains at most $\kappa/2$ terminals, and has exactly one neighbor in $V(C)$.

Let $\tilde u$ be some vertex in $V(C)$ (that we will select later).
We then add a new cycle $C'=(v_1,v_2,v_3)$, containing all new vertices, and an edge $e=(\tilde u,v_1)$, and draw $C'$ inside $C$ (we later specify the precise location of this drawing). 
Let $G'=G\cup C'\cup\set{e}$. We let $F_x$ be the face in the drawing of $G'$ on the sphere, whose boundary is $C'$. Letting $h=i^*$, we construct a shell $\zset(x)=(Z_1(x),\ldots,Z_{h}(x))$ of depth $h=i^*\leq 8\Delta_0-\Delta-4$ around $x$, with respect to $F_x$. As before, we use $Z_0(x)$ to denote $C_x$. Notice that from our construction, $Z_h(x)=C$. We note that the choice of the vertex $\tilde u\in V(C)$ to which the cycle $C'$ is attached does not affect the construction of the shell: for any such choice, the shell will be the same. Notice also that the addition of the new cycles does not affect the routings. So abusing the notation we denote $G'$ by $G$.

For every $1\leq h'\leq h$, we let $U_{h'}$ be the set of all vertices in $\dnot(Z_{h'}(x))\setminus D(Z_{h'-1}(x))$, and  let $\rset_{h'}$ be the set of all connected components of $G[U_{h'}]$. We let $U_{h+1}$ be the set of all vertices lying outside $D(Z_h(x))$ in the planar embedding of $G$ where $F_x$ is the outer face (equivalently, $U_{h+1}$ is the set of all vertices lying in disc $D^{\circ}(C)$ in the planar embedding of $G$ where the face containing $x$ is viewed as the outer face), and denote by $\rset_{h+1}$ the set of all connected components of $G[U_{h+1}]$. We will view the components in $\rset_{h+1}$ as type-2 components with respect to the shell. For each such component $R\in \rset_{h+1}$, we let $L(R)\subseteq V(C)$ be the set of the neighbors of the vertices of $R$, and we leave $u(R)$ undefined. We need the following simple observation.

\begin{observation}\label{obs: inner components are light}
For each $R\in \rset_{h+1}$, $|V(R)\cap \tset|\leq \tau^*$.
\end{observation}
\begin{proof}
Assume otherwise, and let $R\in \rset_{h+1}$ be any component with $|R\cap \tset|> \tau^*$. We claim that all vertices of $R$ lie in $S_{i^*+1}$. Indeed, recall that all vertices of $R^*$ lie in $S_{i^*}$, and $C\subseteq R^*$ is a cycle separating $R$ from all vertices of $V(G)\setminus S_{i^*}$. Therefore, for every vertex $v\in V(R)$, $\dface(v,V(C_x))\geq i^*+1$ and $v\in S_{i^*+1}$. But then $R$ is a heavy connected component in $G[S_{i^*+1}]$, contradicting the choice of $i^*$.
\end{proof}

Notice that once $\tilde u$ and the face $F_x$ are fixed, we can define, for every component $R\in \rset_{h+1}$, the segment $\sigma(R)$ of the cycle $C=Z_h(x)$ exactly as before, and we can define the discs $\eta(R)$ for all components $R\in \bigcup_{h'=1}^{h+1}\rset_{h'}$ using Theorem~\ref{thm: curves around CC's}. (It may be convenient to think of $C'$ as the outer-most cycle of the shell; that is, we add $C'$ to the shell as $Z_{h+1}(x)$).
The main theorem summarizing the current step is the following.

\begin{theorem}\label{thm: case 2b shell}
There is an efficient algorithm to compute a vertex $\tilde u\in V(C)$, a drawing of the cycle $C'$, and a subset $\tset'\subseteq\tset$ of at most $4\tau^*$ terminals, such that, in the resulting shell $\zset(x)=(Z_1(x),\ldots,Z_{h}(x))$, for each component $R\in\bigcup_{h'=1}^{h+1}\rset_{h'}$, the resulting disc $\eta(R)$ contains at most $\tau^*$ terminals of $\tset\setminus \tset'$.
\end{theorem}

We emphasize that the shell $\zset=(Z_1(x),\ldots,Z_h(x))$ and the sets $\rset_{h'}$ of components, for $1\leq h'\leq h+1$ do not depend on our choice of the vertex $\tilde u$ or the drawing of $C'$. Similarly, the discs $\eta(R)$ for components $R$ lying in sets $\rset_{h'}$, for $1\leq h'\leq h$, given by Theorem~\ref{thm: curves around CC's}, are also independent of the choice of $\tilde u$ or the drawing of $C'$ inside $C$. The choice of $\tilde u$ only influences the discs $\eta(R)$ for $R\in \rset_{h+1}$, and so our goal is to select $\tilde u$ and $\tset'$, and to draw $C'$ inside $C$ in a way that ensures that each such disc $\eta(R)$ contains few terminals of $\tset\setminus\tset'$.
\begin{proof}
Our first step processes the components of $\rset_{h+1}$ and to select the vertex $\tilde u\in V(C)$. For this step, we will think of $G$ as being embedded on the sphere. 
Let $\kappa'$ be the number of the terminals of $\tset$ contained in the components of $\rset_{h+1}$. If $\kappa'\leq 4\tau^*$, then we add all terminals of $\tset\cap \left(\bigcup_{R\in \rset_{h+1}}V(R)\right )$ to $\tset'$, and terminate the algorithm, setting $\tilde u$ to be any vertex of $C$, and drawing $C'$ anywhere inside $C$, so the resulting drawing of $G$ is planar. (We show below that this choice satisfies the conditions of the theorem). Therefore, we assume from now on that $\kappa'\geq 4\tau^*$, and in particular $|\rset_{h+1}|\geq 1$.

Throughout the algorithm, we maintain a partition of $\rset_{h+1}$ into two subsets: $\rset'$, containing all components we have processed, and $\rset''$, containing all remaining components. At the beginning, $\rset'=\emptyset$ and $\rset''=\rset_{h+1}$.

Consider any component $R\in \rset''$. Recall that all neighbors of the vertices of $R$ must lie on $C$ from our construction, and we denoted the set of these vertices by $L(R)$. 

We say that $R$ is a good component, iff there is a segment $\mu(R)$ of $C$ containing all vertices of $L(R)$, whose endpoints, denoted by $a_1(R)$ and $a_2(R)$ belong to $L(R)$, and there is a $G$-normal curve $\gamma'(R)$, whose endpoints are $a_1(R)$ and $a_2(R)$, such that $\gamma'(R)$ is internally disjoint from $V(G)$, and the following holds: let $\eta'(R)$ be the disc, whose boundary is $\mu(R)\cup \gamma'(R)$, with $R\subseteq \eta'(R)$. Then for all $R'\in \rset''$ with $R'\neq R$, $R'$ is disjoint from $\eta'(R)$ (intuitively, these are the components that lie closest to $C$; see Figure~\ref{fig: extra discs}).

\begin{figure}[h]
 \centering
\scalebox{0.4}{\includegraphics{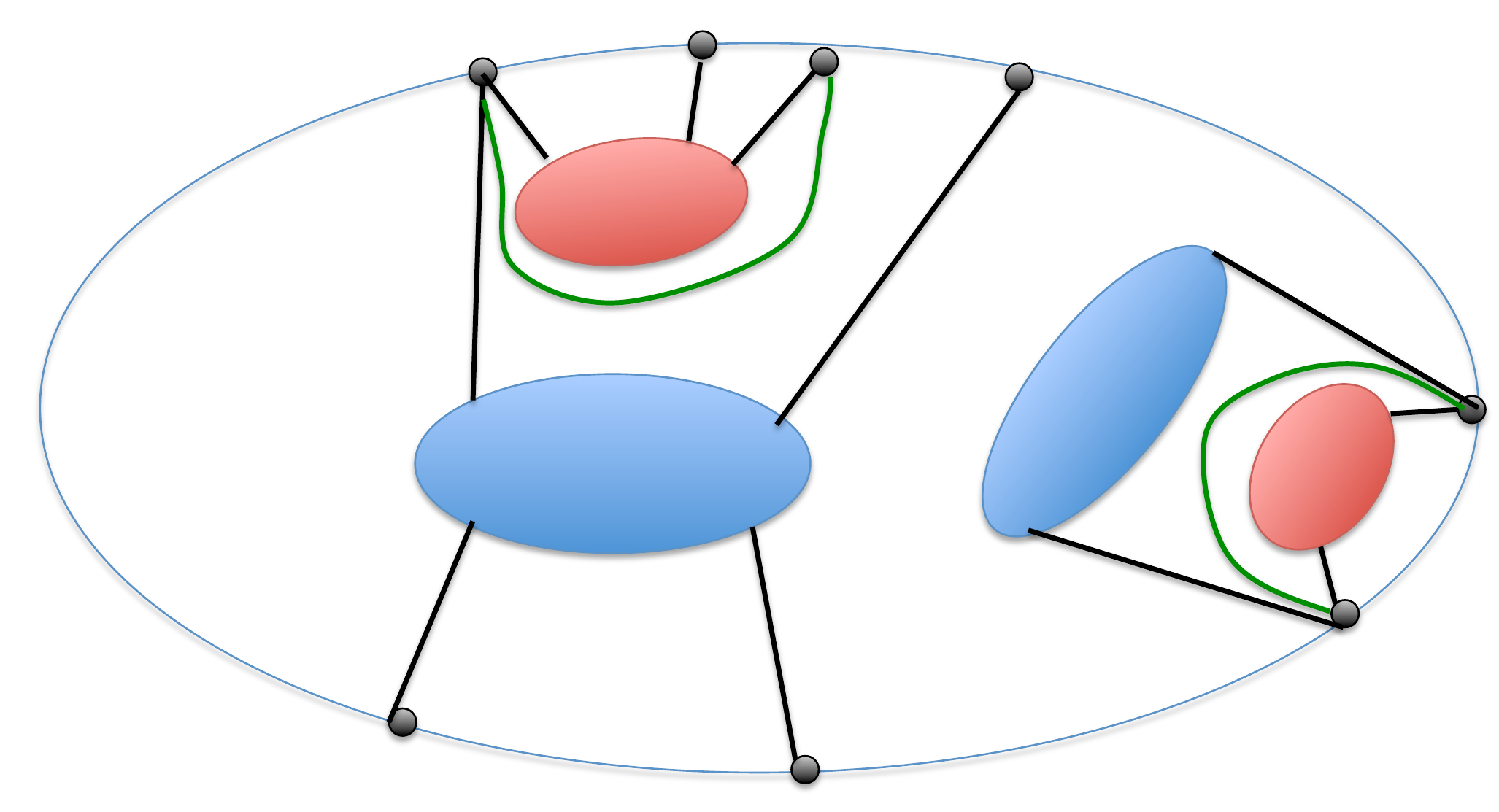}}
\caption{Good components are shown in red, and their corresponding curves $\gamma'(R)$ in green\label{fig: extra discs}}
\end{figure}

\begin{claim}\label{claim: good component}
If $|\rset''|\geq 2$, then there are at least two good components in $\rset''$.
\end{claim}

\begin{proof}
Let $\mu$ be the shortest segment of $C$, such that for some $R\in \rset''$, $L(R)\subseteq \mu$. Let $\tilde \rset\subseteq \rset''$ be the set of all components $R$ with $L(R)\subseteq \mu$. If, for any component $R\in \tilde \rset$, $|L(R)|\geq 3$, then from the definition of $\mu$ it is easy to see that $R$ is a good component. If any component $R\in \tilde \rset$ has  $|L(R)|=1$, then $\mu$ contains a single vertex - the unique vertex of $L(R)$, and it is easy to see that $R$ is a good component. Otherwise, every component in $\tilde \rset$ has $|L(R)|=2$. We then let $R$ be the component that lies closest to $\mu$. In other words, we choose the unique component $R\in \tilde \rset$, such that for some vertex $v\in V(R)$, there is a curve $\gamma$ connecting $v$ to an inner point of $\mu$, so that $\gamma$ does not contain any vertices of $G$ except for $v$, and it does not intersect the edges of $G$. It is immediate to verify that $R$ is a good component. We conclude that there is at least one good component $R\in \rset''$, with $L(R)\subseteq \mu$.

Let $a,a'$ be the endpoints of $\mu$, and let $\mu'$ be the other segment of $C$ whose endpoints are $a$ and $a'$. Since $|\rset''|\geq 2$, there is at least one component $R'\neq R$ in $\rset''$, with $L(R')\subseteq \mu'$. As before, we let $\mu''$ be the shortest (inclusion-wise) segment of $\mu'$, such that for some component $R''\in \rset''$, $L(R'')\subseteq \mu''$. Let $\tilde \rset'$ be the set of all components $R''\neq R$ in $\rset''$ with $L(R'')\subseteq \mu''$. Using the same arguments as above, we can find a second good 
component in $\rset''$.
\end{proof}

We are now ready to describe our algorithm.
We start with $\rset''=\rset_{h+1}$ and $\rset'=\emptyset$. Recall that $|\rset''|\geq 1$ must hold.
While $|\rset''|\geq 2$, let $R,R'\in \rset''$ be two distinct good components in $\rset''$. 
Notice that $\tset\cap \eta'(R)$ and $\tset\cap \eta'(R')$ are completely disjoint. Therefore, either $|\tset\cap \eta'(R)|\leq \kappa'/2$, or $|\tset\cap \eta'(R')|\leq \kappa'/2$. We assume without loss of generality it is the former. We then move $R$ from $\rset''$ to $\rset'$, and continue to the next iteration.

Notice that in every iteration of the algorithm, $|\rset''|$ decreases by $1$. The algorithm terminates when $|\rset''|=1$. Let $R$ be the remaining component in $\rset''$. We set $\tset'=\tset\cap V(R)$, and we set $\tilde u$ to be any vertex of $L(R)$. Recall that from Observation~\ref{obs: inner components are light}, $|\tset'|\leq \tau^*$.
We then add the cycle $C'$ that attaches to $\tilde u$ with an edge $e$ to $G$, and we draw $C'$ inside $C$, so it is drawn next to the edge $e$. We then
obtain a drawing of the resulting graph in the plane where the outer face $F_x$ is the face whose boundary is $C'$. We construct the shells, and discs $\eta(R)$ for $R\in \bigcup_{h'=1}^{h+1}\rset_{h'}$ as described above. From the choice of $\tilde u$, the drawing of $C'$, and the construction of the discs $\eta(R)$, for every $R\in \rset'$, $\eta'(R)\subseteq \eta(R)$, and $\eta(R)\cap (D(C')\setminus D(C))=\eta'(R)$. Therefore, for each $R\in \rset'$, at least $\kappa'/2$ terminals of $\tset'$ lie outside $\eta(R)$. The following claim will finish the proof of the theorem. 

\begin{claim}
For all $1\leq h'\leq h+1$, for all $R\in \rset_{h'}$, $\eta(R)$ contains at most $\tau^*$ terminals of $\tset\setminus \tset'$.
\end{claim}

\begin{proof}
Fix some $1\leq h'\leq h+1$ and $R\in \rset_{h'}$. From our construction, the length of the boundary of $\eta(R)$ is bounded by $2h+3+\Delta/2\leq 2(8\Delta_0-\Delta)+\Delta/2<16\Delta_0$.

Assume first that $h'<h$. Then all vertices of $\tset\cap R^*$ lie outside $D(Z_{h'}(x))$, and since $\eta(R)\subseteq D(Z_{h'}(x))$, they also lie outside $\eta(R)$. From Observation~\ref{obs: number of terminals in disc}, $\eta(R)$ contains at most $\tau^*$ terminals of $\tset$.

Assume now that $h'=h$. Notice that for every component $R'\in \rset_h$, either $R'$ is disjoint from $R^*$, or $R'\subseteq R^*$. In the latter case, $R'$ is a type-1 or a type-3 component (see Figure~\ref{fig: three types of cc's}), as all type-2 components are disconnected from $R^*$ in $G[S_{i^*}]$. Recall that $R'\in \rset_h$ is contained in $\eta(R)$ only if $\sigma(R')\subseteq\sigma(R)$, and this can only happen if $R'$ is a type-2 component. 
Let $\tilde \rset\subseteq\rset_h$ be the set of all components contained in $\eta(R)$. Then at most one of these components $R'\in \tilde\rset$ may be contained in $R^*$ - and if such a component exists, then $R'=R$. Our construction of the cycle $C$ and the choice of $F_x$ guarantee that $R'$ contains at most $\kappa/2+1$ terminals of $\tset\cap V(R^*)$, and so at least $\kappa/2-1\geq \kappa/4$ terminals lie outside $\eta(R)$. Since the length of $\gamma(R)$ is less than $16\Delta_0$, from Observation~\ref{obs: number of terminals in disc}, $\eta(R)$ contains at most $\tau^*$ terminals of $\tset$.

Finally, assume that $h'=h+1$. If $\kappa'\leq 4\tau^*$, then $V(R)\cap (\tset\setminus\tset')=\emptyset$. Therefore, we assume that $\kappa'>4\tau^*$. If $R\cap (\tset\setminus\tset')\neq \emptyset$, then $R$ was added to $\rset'$ at some iteration of the algorithm. From the above discussion, at least $\kappa'/2\geq 2\tau^*$ terminals of $\tset$ lie outside of $\eta(R)$, while the length of the boundary of $\eta(R)$ is less than $16\Delta_0$. From the well-linkedness of the terminals, $\eta(R)$ contains at most $\tau^*$ terminals of $\tset$.
\end{proof}
\end{proof}

We partition $\nset$ into two subsets: set $\mset^0$ contains all demand pairs in which the terminals of $\tset'$ participate, together with all demand pairs for which either terminal lies in $D_x$. Then $|\mset^0|\leq 4\tau^*+4\Delta/\alphaWL\leq 5\tau^*$ (we have used the definition of enclosures to bound the number of the demand pairs of the latter kind, and the fact that $\tau^*=64\Delta_0/\alphaWL$). Let $\mset^1$ contain the remaining demand pairs, and let $\rset=\bigcup_{h'=1}^{h+1}\rset_{h'}$. As we noted before, we would like to ensure that for every component $R\in \rset$ that contains a terminal of $\tset(\mset^1)$, whenever a path $P$ routing a demand pair in $\mset^1$ intersects $R$, it must cross $\gamma(R)$ - the boundary of $\eta(R)$. We achieve this property by routing a subset of the demand pairs, and discarding some additional demand pairs, in the following theorem.

\begin{theorem}\label{thm: routing close pairs}
There is an efficient algorithm to compute a partition $(\nset_0,\nset_1,\nset_2)$ of $\mset^1$, and a collection $\pset^*$ of node-disjoint paths routing at least $|\nset_1|/(h+1)$ demand pairs in $\nset_1$ in graph $G$, such that:

\begin{itemize}
\item $|\nset_0|\leq \tau^*|\nset_1|$; and
\item for every component $R\in \rset$ with $R\cap \tset(\nset_2)\neq \emptyset$, for every path $P$ routing a demand pair in $\nset_2$, $P\not\subseteq \eta^{\circ}(R)$.
\end{itemize}
\end{theorem}
\begin{proof}
We start with $\nset_0=\nset_1=\emptyset$, $\pset=\emptyset$, and $\nset_2=\mset^1$. Throughout the algorithm, we say that a component $R\in \rset$ is a live component iff $R\cap \tset(\nset_2)\neq \emptyset$.
While there is any demand pair $(s,t)\in \nset_2$, and any live component $R\in \rset$, such that some path $P$ connecting $s$ to $t$ is contained in $\eta^{\circ}(R)$, we do the following. We add $P$ to $\pset$, and we move $(s,t)$ from $\nset_2$ to $\nset_1$. We say that the component $R$ is \emph{responsible for $P$}, and we say that $P$ is a level-$h'$ path if $R\in \rset_{h'}$. Next, for every live component $R'\in \rset_{h'}$, such that $P$ intersects $\eta^{\circ}(R')$, we move all demand pairs $(s',t')\in \nset_2$ with $\set{s',t'}\cap \eta^{\circ}(R')\neq \emptyset$ from $\nset_2$ to $\nset_0$. The crux of the analysis of this algorithm is in the following claim.

\begin{claim}\label{claim: bounding number of discarded pairs}
In every iteration of the algorithm, the number of the demand pairs moved to $\nset_0$ is at most $\tau^*$.
\end{claim}
\begin{proof}
Consider some iteration of the algorithm, where a path $P$ that belongs to level $h'$ was added to $\pset$, and let $R\in \rset_{h'}$ be the component responsible for it. Recall that $P\subseteq\eta^{\circ}(R)$, and recall that Theorem~\ref{thm: curves around CC's} guarantees that all discs $\eta(R')$ for $R'\in \rset_{h'}$ are laminar: that is, for $R',R''\in \rset_{h'}$, either $\eta(R')\subseteq \eta(R'')$, or $\eta(R'')\subseteq\eta(R')$, or $\eta^{\circ}(R')\cap \eta^{\circ}(R'')=\emptyset$. Let $\rset'\subseteq \rset_{h'}$ be the set of all live components $R'$ with $P\cap \eta^{\circ}(R')\neq \emptyset$. 
Since $P\subseteq \eta^{\circ}(R)$, there is some component $R'\in \rset'$, such that $P\subseteq \eta^{\circ}(R')$, and for every component $R''\in \rset'$, $\eta^{\circ}(R'')\subseteq \eta^{\circ}(R')$. Therefore, all demand pairs moved in this step from $\nset_2$ to $\nset_0$ have at least one terminal lying in $\eta^{\circ}(R')$, and from Theorem~\ref{thm: case 2b shell}, their number is bounded by $\tau^*$.
\end{proof}

Therefore, once the algorithm terminates, $|\nset_0|\leq \tau^*|\nset_1|$ must hold. Moreover, for every component  $R\in \rset$ with $R\cap \tset(\nset_2)\neq \emptyset$, for every path $P$ routing a demand pair in $\nset_2$, $P\not\subseteq \eta^{\circ}(R)$.  Consider now the set $\pset$ of paths. Each of these paths belongs to one of $h+1$ levels, and so there is a subset $\pset^*\subseteq \pset$ of paths that belong to the same level, say $h'$, such that $|\pset^*|\geq |\pset|/(h+1)$. The paths in $\pset^*$ route a subset of the demand pairs in $\nset_1$, and it is now enough to show that they are node-disjoint. Assume for contradiction otherwise, and let $P,P'\in \pset^*$ be two distinct paths that are not disjoint. Assume that $P$ was added to $\pset$ before $P'$. Let $R'\in \rset_{h'}$ be the component responsible for $P'$, so $P'\subseteq \eta^{\circ}(R')$, and $R'$ was live when $R$ was processed. Then path $P$ must intersect $\eta^{\circ}(R')$, and so all demand pairs that have a terminal in $\eta^{\circ}(R')$, including the demand pair routed by $P'$, should have been removed from $\nset_2$ during the iteration when $P$ was added to $\pset$, a contradiction.
\end{proof}

Recall that $\tau^*=64\Delta_0/\alphawl$. We now consider two cases. First, if $|\pset^*|\geq |\nset|\cdot \frac{\alphawl}{2^{13}\Delta_0^2}$, then we return $\nset'=\nset$ and $\pset^*$ as the set of paths routing a subset of the demand pairs in $\nset'$. Therefore, we assume from now on that $|\pset^*|<|\nset|\cdot \frac{\alphawl}{2^{13}\Delta_0^2}$. Let $\nset''=\mset^0\cup \nset_0\cup \nset_1$, and let $\nset'=\nset\setminus \nset''=\nset_2$. Then:

\[\begin{split}
|\nset''|&\leq 5\tau^*+(\tau^*+1)|\nset_1|\\
&\leq 5\tau^*+(\tau^*+1)(h+1)|\pset^*|\\
&\leq 16\Delta_0\tau^*\cdot |\pset^*|\\
&\leq \frac{1024\Delta_0^2}{\alphawl}\cdot \frac{|\nset|\alphawl}{2^{13}\Delta_0^2}\\
&\leq \frac{|\nset|}{6}.\end{split}\]

Therefore, $|\nset'|\geq 5|\nset|/6$.  From now on, our goal is to find a set $\pset$ of paths, routing $\Omega\left(\frac{\opt(G,\nset')}{\Delta_0^8\log^3 n}\right )$ demand pairs in $\nset'$. 
From the above discussion, the demand pairs in $\nset'$ have the following property.

\begin{properties}{P}
\item For every component $R\in \rset$ with $R\cap \tset(\nset')\neq \emptyset$, if $P$ is any path routing any demand pair in $\nset'$, then $P\not\subseteq \eta^{\circ}(R)$.\label{prop: intersection of gammas for case 2b}
\end{properties}

\paragraph{Step 2: Mapping the Terminals.}

This step is almost identical to the similar step in Case 2a, except that now we crucially exploit Property (\ref{prop: intersection of gammas for case 2b}).
Let $\tset'=\tset(\nset')$, and let $U=\bigcup_{h'=1}^hV(Z_{h'}(x))$ be the set of all vertices lying on the cycles of the shell. Our next step is to define a mapping $\beta:\tset'\rightarrow 2^U$ of the terminals $t\in \tset'$ to subsets $\beta(t)$ of at most three vertices of $U$. For every terminal $t\in \tset'$, we also define a corresponding $G$-normal curve $\Gamma(t)$, as before. 

The mapping $\beta$ and the curves $\Gamma(t)$ are defined as follows.
First, for every terminal $t\in \tset'\cap U$, we set $\beta(t)=\set{t}$, and we let $\Gamma(t)$ contain the vertex $t$ only.
For all remaining terminals $t\in \tset'$, $t$ must lie in some component $R\in \rset$.  If $|L(R)|\leq 2$, then we let $\beta(t)=L(R)\cup \set{u(R)}$ if $u(R)$ is defined, and $\beta(t)=L(R)$ otherwise. If $|L(R)|>2$, then we let $\beta(t)=\set{v}$, where $v$ is a leg of $R$, which is not an endpoint of $\sigma(R)$ (in other words, $v$ is not the first and not the last leg of $R$).  We then let $\Gamma(t)$ be the boundary $\gamma(R)$ of the disc $\eta(R)$ given by Theorem~\ref{thm: case 2b shell}. Recall that the length of $\Gamma(t)$ is bounded by $2h+\Delta/2+3< 16\Delta_0$, as $h\leq 8\Delta_0-\Delta-4$.

Let $\tmset=\bigcup_{ (s,t)\in \nset'}\beta(s)\times \beta(t)$.  In the following two theorems, we relate the values of the solutions to problems $(G,\nset')$ and $(G,\tmset)$.

\begin{theorem}\label{thm: from new to old problem2}
There is an efficient algorithm, that, given any solution to instance $(G,\tmset)$, that routes $\kappa$ demand pairs, finds a solution to instance $(G,\nset')$, routing at least $\Omega\left(\frac{\kappa}{\Delta_0}\right )$ demand pairs.
\end{theorem}

\begin{proof}
Let $\pset_0$ be any collection of disjoint paths in graph $G$, routing a subset $\tmset_0$ of $\kappa$ demand pairs in $\tmset$. We assume that $\kappa\geq 64\Delta_0$, as otherwise we can return a routing of a single demand pair in $\nset'$.
For every demand pair $(\tilde s,\tilde t)\in \tmset_0$, let $(s,t)$ be any corresponding demand pair in $\nset'$, that is, $\tilde s\in \beta(s)$ and $\tilde t\in \beta(t)$.

We build a conflict graph $H$, whose vertex set is $\set{v(\tilde s,\tilde t)\mid (\tilde s,\tilde t)\in \tmset_0}$, and there is a directed edge from $v(\tilde s_1,\tilde t_1)$ to $v(\tilde s_2,\tilde t_2)$ iff the unique path $P(\tilde s_1,\tilde t_1)\in \pset_0$ routing the pair $(\tilde s_1,\tilde t_1)$ intersects $\Gamma(s_2)$ or $\Gamma(t_2)$ (in which case we say that there is a conflict between $(\tilde s_1,\tilde t_1)$ and $(\tilde s_2,\tilde t_2)$). Since all paths in $\pset_0$ are node-disjoint, and all curves $\Gamma(s),\Gamma(t)$ have lengths at most $16\Delta_0$ each, the in-degree of every vertex in $H$ is at most $32\Delta_0$. Therefore, we can efficiently compute an independent set $I$ of size at least $\frac{\kappa}{32\Delta_0+1}\geq\frac{\kappa}{64\Delta_0}$ in $H$.

Let $\tmset_1=\set{(\tilde s,\tilde t)\mid v(\tilde s,\tilde t)\in I}$, and let $\pset_1\subseteq\pset_0$ be the set of paths routing the demand pairs in $\tmset_1$. Let  $\mset'=\set{(s,t)\mid (\tilde s,\tilde t)\in \tmset_1}$.  It is now enough to show that all demand pairs in $\mset'$ can be routed in $G$. Consider any demand pair $(\tilde s,\tilde t)\in \tmset_1$, and let $P\in \pset_1$ be the path routing $(\tilde s,\tilde t)$. We will extend the path $P$, so it connects $s$ to $t$, by appending two paths: $Q(\tilde s)$ connecting $\tilde s$ to $s$, and $Q(\tilde t)$ connecting $\tilde t$ to $t$, to it. If $s\in U$, then $s=\tilde s$, and we define $Q(\tilde s)=\emptyset$. Assume now that $s\not \in U$, and let $R\in \rset$ be the component in which $s$ lies. Let $Q(\tilde s)$ be any path that starts from $\tilde s$, terminates at $s$, and except for its first edge, is contained in $R$. We define the path $Q(\tilde t)$ similarly. 
Let $\pset^*$ be the set of paths obtained by concatenating the paths in $\pset_1$ with the paths in $\set{Q(\tilde s),Q(\tilde t)\mid (\tilde s,\tilde t)\in \tmset_1}$. Then the paths in $\pset^*$ route all demand pairs in $\mset'$, so $|\pset^*|\geq \frac{\kappa}{64\Delta_0}$. It is now enough to prove that the paths in $\pset^*$ are node-disjoint. We do so in the following claim.

\begin{claim}
The paths in $\pset^*$ are node-disjoint.
\end{claim}

\begin{proof}
Observe that all endpoints of all paths of $\pset_1$ are distinct.
We first prove that for all terminals $\tilde t,\tilde t'\in \tset(\tmset_1)$, with $\tilde t\neq \tilde t'$, paths $Q(\tilde t)$ and $Q(\tilde t')$ are node-disjoint. Assume otherwise. Then $Q(\tilde t),Q(\tilde t')\neq \emptyset$, and $t,t'\not \in U$. Let $t,t'\in \tset(\tmset')$ be the terminals corresponding to $\tilde t$ and $\tilde t'$, respectively, so $Q(\tilde t)$ connects $\tilde t$ to $t$, and $Q(\tilde t')$ connects $\tilde t'$ to $t'$. Let $R,R'\in \rset$ be the components to which $t$ and $t'$ belong, respectively. Recall that except for its first edge, $Q(\tilde t)$ is contained in $R$, and the same holds for $Q(\tilde t')$ and $R'$. Since $\tilde t\neq \tilde t'$, but $Q(\tilde t)\cap Q(\tilde t')\neq \emptyset$, we get that $R=R'$. But then $\tilde t'\in L(R)\cup \set{u(R)}$,  and so it lies in $\eta(R)$. Let $P'$ be the path of $\pset_1$, such that $\tilde t'$ is an endpoint of $P'$, and let $\tilde s'$ be its other endpoint. Since $\Gamma(t)=\gamma(R)$, and set $\pset_1$ is conflict-free, path $P'$ cannot intersect $\gamma(R)$. Therefore, $\tilde s'$ must belong to $\eta^{\circ}(R)$. Let $s'$ be the terminal in $\tset(\mset')$ corresponding to $\tilde s'$, so path $Q(\tilde s')$ connects $\tilde s'$ to  $s'$. We claim that if $Q(\ts')\neq \emptyset$, then $Q(\tilde s')$ is also contained in $\eta^{\circ}(R)$. Indeed, if $Q(\ts')\neq \emptyset$, then $s'$ belongs to some component $R''\in \rset$, and $\tilde s'\in L(R'')\cup \set{u(R'')}$. But since $\tilde s'\in \eta^{\circ}(R)$, and since $\gamma(R)$ is a $G$-normal curve disjoint from $V(R'')$, $R''\subseteq \eta^{\circ}(R)$ must hold. We conclude that $Q(\tilde s')$ is contained in $\eta^{\circ}(R)$. By concatenating $Q(\tilde s'),P'$ and $Q(\tilde t')$, we obtain a path $P$, connecting $s'$ to $t'$, where $P$ is contained in $\eta^{\circ}(R)$, violating Property~(\ref{prop: intersection of gammas for case 2b}).
We conclude that paths $Q(\tilde t)$ and $Q(\tilde t')$ are node-disjoint. 

It is now enough to show that for every terminal $\tilde t\in \tset(\tmset_1)$, and for every path $P\in \pset_1$, such that $\tilde t$ is not an endpoint of $P$, $Q(\tilde t)$ is disjoint from $P$. Assume otherwise, and let $\tilde s',\tilde t'$ be the endpoints of the path $P$. 
Since we have assumed that $Q(\tilde t)\cap P\neq \emptyset$, $Q(\tilde t)\neq\emptyset$, and so the vertex $t\in \tset(\mset')$ serving as the other endpoint of $Q(\tilde t)$ must lie in some component $R\in \rset$.
Notice that since the paths in $\pset_1$ have no conflicts, $P$ is disjoint from $\gamma(R)$, and so it must be contained in $\eta^{\circ}(R)$. Using the same reasoning as above, we conclude that if $Q(\tilde s')\neq \emptyset$, then  it is contained in  $\eta^{\circ}(R)$, and the same holds for $Q(\tilde t')$. Therefore, there is a path $P$, obtained by concatenating $P',Q(\tilde s')$ and $Q(\tilde t')$, connecting the pair $(s',t')\in \nset'$, with $P'\subseteq\eta^{\circ}(R)$, contradicting Property~(\ref{prop: intersection of gammas for case 2b}). Since the paths in $\pset_1$ are node-disjoint, it is now immediate to see that the paths in $\pset^*$ must also be node-disjoint.
\end{proof}

\end{proof}

\begin{theorem}\label{thm: from old to new problem2}
$\opt(G,\tmset)\geq \frac{\opt(G,\nset')}{64\Delta_0}$.
\end{theorem}

\begin{proof}
Let $\pset_0$ be the set of paths in the optimal solution to instance $(G,\nset')$, and let $\mset_0$ be the set of the demand pairs they route. We can assume that $|\mset_0|\geq 64\Delta_0$, as otherwise we can route a single demand pair in $\tmset$. 

As before, we define a conflict graph $H$, whose vertex set is $\set{v(s,t)\mid (s,t)\in \mset_0}$, and there is a directed edge from $v(s_1,t_1)$  to $v(s_2,t_2)$ iff the unique path $P(s_1,t_1)\in \pset_0$ routing the pair $(s_1,t_1)$ intersects $\Gamma(s_2)$ or $\Gamma(t_2)$ (in which case we say that there is a conflict between $(s_1,t_1)$ and $(s_2,t_2)$). Since all paths in $\pset_0$ are node-disjoint, and all curves $\Gamma(s),\Gamma(t)$ have lengths at most $16\Delta_0$ each, the in-degree of every vertex in $H$ is at most $32\Delta_0$. Therefore, we can efficiently compute an independent set $I$ of size at least $\frac{\opt(G,\nset')}{32\Delta_0+1}\geq\frac{\opt(G,\nset')}{64\Delta_0}$ in $H$.

Let $\mset_1=\set{(s,t)\mid v(s,t)\in I}$, and let $\pset_1\subseteq \pset_0$ be the set of paths routing the demand pairs in $\mset_1$. We show that we can route $|\mset_1|$ demand pairs of $\tmset$ in $G$ via node-disjoint paths. Let $\tset_1$ be the sets of all terminals participating in the pairs in $\mset_1$.

Consider any terminal $t\in \tset_1$. We say that $t$ is a \emph{good terminal} if $t\in U$, or $t$ belongs to some component $R\in \rset$, such that $|L(R)|\leq 2$. Otherwise, $t$ is a \emph{bad terminal}. Notice that if $t$ is a good terminal, then the path $P\in \pset_1$ that contains $t$ as its endpoint must contain a vertex $t'\in \beta(t)$: if $t\in U$, then $\beta(t)=\set{t}$; otherwise, if $t\in R$ for some component $R\in \rset$ with $|L(R)|\leq 2$, then $\beta(t)=L(R)\cup \set{u(R)}$ if $u(R)$ is defined, and $\beta(t)=L(R)$ otherwise. In either case, in order to enter $R$, path $P$ has to visit a vertex of $\beta(t)$ (it is impossible that $P\subseteq R$ due to Property~(\ref{prop: intersection of gammas for case 2b})). Therefore, if $t$ is a good terminal, then some vertex $t'\in P$ belongs to $\beta(t)$. 

We transform the paths in $\pset_1$ in two steps, to ensure that they connect demand pairs in $\tmset$. In the first step, for every path $P\in \pset_1$ originating at a good terminal $s\in \tset_1$, we truncate $P$ at the first vertex $s'\in \beta(s)$, so it now originates at $s'$. Similarly, if $P$ terminates at a good terminal $t\in T_1$, we truncate $P$ at the last vertex $t'\in \beta(t)$, so it now terminates at $t'$. Let $\pset_1'$ be the resulting set of paths. Notice that the paths in $\pset_1'$ remain node-disjoint.

In order to complete our transformation, we need to take care of bad terminals. Let $t\in \tset_1$ be any bad terminal. Then $t\in R$ for some component $R\in \rset$ with $|L(R)|\geq 3$. Recall that in this case, $\beta(t)$ contains a unique vertex, that we denote by $t'$, which is a leg of $R$, and it is not one of the endpoints of $\sigma(R)$. We then let $Q(t)$ be any path connecting $t$ to $t'$ in the sub-graph of $G$ induced by $V(R)\cup \set{t'}$. By concatenating the paths in $\set{Q(t)}$ for all bad terminals $t\in \tset_1$, and the paths in $\pset_1'$, we obtain a collection $\tilde{\pset}$ of at least $\frac{\opt(G,\nset')}{64\Delta_0}$ paths, routing demand pairs in $\tmset$. It now only remains to show that the paths in $\tpset$ are disjoint. Recall that the paths in $\pset_1'$ were node-disjoint.

\begin{claim}
The paths in $\tpset$ are node-disjoint.
\end{claim}

\begin{proof}
Consider first some pair $t_1,t_2$ of bad terminals. We show that the paths $Q(t_1)$ and $Q(t_2)$ are disjoint. Let $P_1$ and $P_2$ be the paths in $\pset_1'$, for which $t_1$ and $t_2$ serve as endpoints, respectively. Let $P_2'\in \pset_1$ be the path corresponding to $P_2$, that is, $P_2$ is a sub-path of $P_2'$.

Recall that $t_1\in R$ for some $R\in \rset$, and recall that the disc whose boundary is $\Gamma(t_1)$ contains $R\cup L(R)$. Path $P_2'$ cannot cross $\gamma(R)$ since the paths in $\pset_1$ are are conflict-free, and it is not contained in $\eta^{\circ}(R)$, since that would violate Property~(\ref{prop: intersection of gammas for case 2b}). Therefore, path $P_2'$ lies completely outside $\eta(R)$, and so does path $P_2$.  Let $R'\in \rset$ be the component to which $t_2$ belongs. Then $R'\cap \eta(R)=\emptyset$ must also hold, since $\gamma(R)$ cannot intersect $R'$, from Theorem~\ref{thm: curves around CC's}. From our definition of $\beta(t)$ for bad terminals $t$, $\beta(t_1)\neq \beta(t_2)$, and each such set contains exactly one vertex. It is now easy to see that $Q(t_1)$ and $Q(t_2)$ are disjoint.

Consider now some bad terminal $t\in \tset_1$, and let $P\in \pset_1'$ be any path, such that $t$ is not an endpoint of $P$. We next show that $Q(t)$ is disjoint from $P$. Let $t',t''$ be the endpoints of path $P$, and let $R\in \rset$ be the component containing $t$. Let $P'\in \pset_1$ be the path corresponding to $P$, so $P\subseteq P'$. Path $P'$ cannot intersect $\gamma(R)$ since the paths in $\pset_1$ are conflict-free, and it is not contained in $\eta^{\circ}(R)$ due to Property~(\ref{prop: intersection of gammas for case 2b}). Therefore, $P'$, and hence $P$, lie completely outside $\eta(R)$. Since $Q(t)\subseteq \eta(R)$, we get that $Q(t)\cap P=\emptyset$.
Since the paths in $\pset_1'$ are node-disjoint, it is now immediate to verify that the paths in $\tpset$ are node-disjoint as well.
\end{proof}
\end{proof}

For all $1\leq h',h''\leq h$, we let $\tmset_{h',h''}$ be the set of all demand pairs $(\ts,\tilde t)\in \tmset$ with $\ts\in Z_{h'}(x)$ and $\tilde t\in Z_{h''}(x)$. 
Since $h\leq 16\Delta_0$, we obtain the following corollary.

\begin{corollary}\label{cor: opt before reducing to special case}
There are some $1\leq h',h''\leq h$, such that $\opt(G,\tmset_{h',h''})\geq \Omega\left(\frac{\opt(G,\nset')}{\Delta_0^3}\right )$, and for any solution to instance $(G,\tmset_{h',h''})$, routing $\kappa$ demand pairs, we can efficiently obtain a solution to instance $(G,\nset')$, routing $\Omega(\kappa/\Delta_0)$ demand pairs.
\end{corollary}

It is now enough to prove the following theorem.

\begin{theorem}\label{thm: before reduction}
There is an efficient algorithm, that for all $1\leq h',h''\leq h$ computes a set $\pset_{h',h''}$ of disjoint paths, routing $\Omega\left(\frac{\opt(G,\tmset_{h',h''})}{\Delta_0^4\log^3n}\right )$ demand pairs of $\tmset_{h',h''}$ in $G$.
\end{theorem}

\paragraph{Step 3: Reduction to Routing on a Disc.}
In this step, we complete the proof of Theorem~\ref{thm: main for Case 2} by proving Theorem~\ref{thm: before reduction}.
We fix some $1\leq h',h''\leq h$. If $h'=h''$, then all terminals of $\tset(\tmset_{h',h''})$ lie on $Z_{h'}(x)$, and we obtain an instance of the special case,  with $Z=Z_{h'}(x)$, $C=C_x$ and $\hmset=\tmset_{h',h''}$. We let $\hat G'$ be the graph obtained from $G$ by deleting all vertices and edges lying in $\dnot(Z_{h'})$, and apply the $O(\log n)$-approximation algorithm for \NDPdisc to the resulting instance $(\hat G', \hmset)$. From Theorem~\ref{thm: case 2b - special case},
$\opt(\hat G',\hmset)\geq \Omega\left(\frac{\opt(G,\hmset)}{\Delta_0^2\log n}\right )$, and so overall we obtain a routing of  $\Omega\left(\frac{\opt(G,\tmset_{h',h''})}{\Delta_0^2\log^2 n}\right )$ demand pairs of $\tmset_{h',h''}$.
Therefore, we assume that $h'\neq h''$ from now on.
We assume without loss of generality that $h'<h''$, and that all source vertices of $\tmset_{h',h''}$ lie on $Z_{h'}(x)$ and all destination vertices of $\tmset_{h',h''}$ lie on $Z_{h''}(x)$. 

In order to simplify the notation, we denote $\rho=h''-h'+1$, and we denote cycles $Z_{h'}(x),Z_{h'+1}(x),\ldots,Z_{h''}(x)$ by $\tZ_1,\tZ_2,\ldots, \tZ_{\rho}$, respectively. We also denote $\tmset_{h',h''}$ by $\tmset$, and the sets of all source and all destination vertices of the demand pairs in $\tmset$ by $\tS$ and $\tT$, respectively. We assume that $\opt(G,\tmset)\geq 500\Delta_0^4$, since otherwise routing a single demand pair of $\tmset$ is sufficient. We also assume that for all $1\leq i<\rho$, no edge connects a vertex of $\tZ_i$ to a vertex of $\tZ_{i+1}$, since we can subdivide each such edge with a vertex.

For convenience of notation, for each $1\leq i<\rho$, we denote the set of vertices of $G$ lying in $\dnot(\tZ_{i+1})\setminus D(\tZ_i)$ by $\tilde U_i$, and we denote by $\trset_i$ the set of all {\bf type-1} connected components of $G[\tilde U_i]$. Let $R\in \trset_i$ be any such connected component. Recall that we have defined a segment $\sigma(R)$ of $\tZ_i$ containing all vertices of $L(R)$. We view $\sigma(R)$ as directed in the counter-clock-wise direction of $\tZ_i$, and we let $u'(R)\in L(R)$ be the first vertex of $\sigma(R)$. We then define $\chi(R)$ to be any path that connects $u'(R)$ to $u(R)$, such that all inner vertices of $\chi(R)$ belong to $R$.

The idea of the proof is to construct a collection $\qset$ of special paths, connecting the vertices of $\tZ_1$ to the vertices of $\tZ_{\rho}$, that we call \emph{staircases}. We use the paths in $\qset$, in order to map all the source vertices in $\tilde S$ to some vertices of $\tZ_{\rho}$. We then reduce the problem to routing on a disc, by creating a hole $\dnot(\tZ_{\rho})$ in the sphere, so that all terminals participating in the new set $\hmset$ of demand pairs lie on the boundary $\tZ_{\rho}$ of the hole.

In order to define the staircases, it will be convenient to work with a directed graph $G'$, obtained from a sub-graph of $G$, as follows. First, we add the cycles $\tZ_1,\tZ_2,\ldots,\tZ_{\rho}$ to $G'$, and for each such cycle $\tZ_i$, we direct all its edges in the counter-clock-wise direction along $\tZ_i$. Next, for each $1\leq i<\rho$, for each component $R\in \trset_i$, we add the path $\chi(R)$ to $G'$, and direct all its edges so that the path is directed from $u'(R)$ to $u(R)$. We are now ready to define a staircase.

A staircase is simply a directed path in $G'$, connecting some vertex of $\tZ_1$ to some vertex of $\tZ_{\rho}$, which is internally disjoint from $\tZ_{\rho}$. Observe that we can decompose any such staircase $Q$ into $2\rho-2$ segments $\mu_1(Q),\chi_1(Q),\ldots,\mu_{\rho-1}(Q),\chi_{\rho-1}(Q)$, where for $1\leq i<\rho$, $\mu_i(Q)$ is a directed sub-path of $\tZ_i$ (possibly consisting of a single vertex), and $\chi_i(Q)=\chi(R)$ for some $R\in \trset_i$.

For every vertex $v\in V(\tZ_1)$, we build a special staircase $Q(v)$, as follows. Intuitively, we will try to minimize the lengths of the segments $\mu_i(Q(v))$. Denote $v=v_1$. For each $1\leq i<\rho$, we now define the segments $\mu_i(Q(v))$ and $\chi_i(Q(v))$, and the vertex $v_{i+1}\in V(\tZ_{i+1})$, which is the last vertex of $\chi_i(Q(v))$, assuming that we are given the vertex $v_i\in \tZ_i$. We let $\mu_i(Q(v))$ be the shortest directed segment of $\tZ_{i}$, starting from $v_i$, that terminates at some vertex $v_i'$, such that for some component $R\in \trset_{i}$, $u'(R)=v'_i$. If $R\in \trset_i$ is a unique component with this property, then we set $\chi_{i}(Q(v))=\chi(R)$. Otherwise, let $\rset'\subseteq \trset_{i}$ be the set of all components $R$ with $u'(R)=v'_i$. Intuitively, we would like to set $\chi_i=\chi(R^*)$, where $R^*$ is the first component of $\rset'$ in the counter-clock-wise order  (see Figure~\ref{fig: staircase}). 
In order to define $R^*$ formally, we need the following observation.

\begin{observation}\label{obs: aux}
There is some component $R\in \trset_i$, such that $R\not\in \rset'$.
\end{observation}

\begin{proof}
From our assumption, there is a set $\pset$ of at least $500\Delta_0^4$ node-disjoint paths, connecting the vertices of $\tS$ to the vertices of $\tT$ in $G$. For every path $P\in \pset$, let $v_P$ be the first vertex of $P$ lying on $\tZ_{i+1}$. Since we have assumed that no edge connects a vertex of $\tZ_i$ to a vertex of $\tZ_{i+1}$, there must be some component $R'\in \trset_i$, such that $v_P=u(R')$, and moreover, $P$ must contain some vertex of $L(R')$. We say that $R'$ is responsible for $P$. Notice that each component of $\rset_i$ may be responsible for at most one path in $\pset$.

Notice that for all components $R'\in \rset'$, except for maybe one, $L(R)=\set{v'_i}$, and so the components of $\rset'$ may be responsible for at most two paths in $\pset$. But $|\pset|\geq 500\Delta_0^4$, so $\trset_i\setminus \rset'\neq \emptyset$.
\end{proof}


We draw a closed $G$-normal curve $\gamma\subseteq \dnot(\tZ_{i+1})\setminus D(\tZ_i)$, so that for each $R'\in \trset_i$, $\gamma\cap V(R')$ is a contiguous set of vertices on $\gamma$, and all vertices that lie on $\gamma$ belong to $\bigcup_{R'\in \trset_i}V(R')$. We then let $\sigma^*$ be the shortest segment of $\gamma$, containing all vertices of $\bigcup_{R'\in \rset'}V(R')$, and no other vertices (in particular, the vertices of all components in $\rset_i\setminus \rset'$ do not lie on $\sigma^*$). Let $v^*$ be the first vertex on $\sigma^*$ in the counter-clock-wise direction, and let $R^*\in \rset'$ be the unique component containing $v^*$. Finally, we set $\chi_i(Q(v))=\chi(R^*)$.
The final staircase $Q(v)$ is a concatenation of the paths in $\set{\mu_i(Q(v)),\chi_i(Q(v))}_{i=1}^{\rho-1}$.  Each such staircase $Q(v)$ defines an undirected path in graph $G$, that we also denote by $Q(v)$, and we do not distinguish between them.

\begin{figure}[h]
 \centering
\scalebox{0.5}{\includegraphics{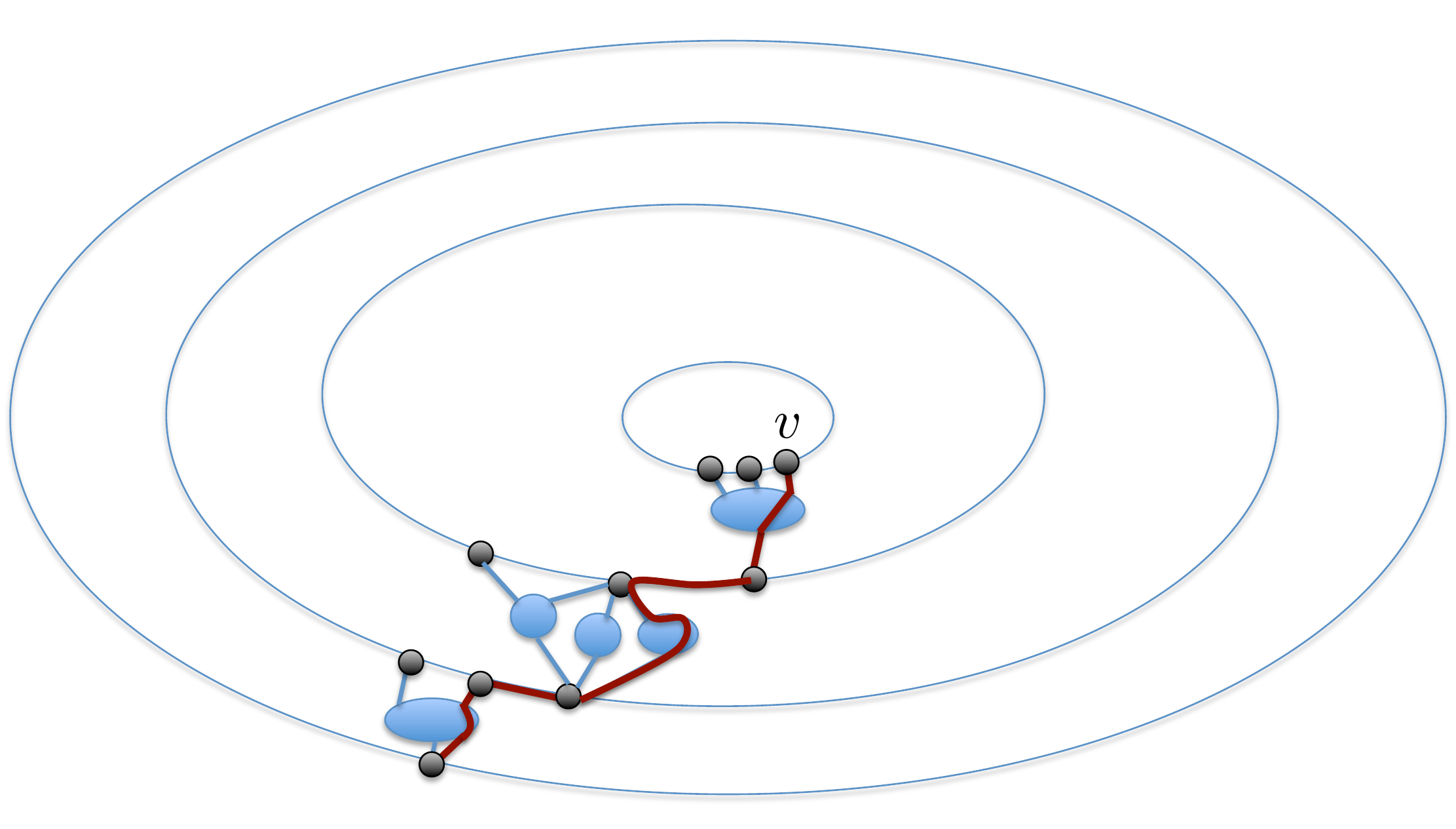}}
\caption{A staircase $Q(v)$ in graph $G$. \label{fig: staircase}}
\end{figure}

We will use the following lemma in order to bound the number of paths in the optimal solution intersecting any staircase $\qset(v)$.
\begin{lemma}\label{lem: intersection of paths and staircases}
Let $\pset^*$ be any set of node-disjoint paths, connecting a subset of the demand pairs in $\tmset$. Then for each $v\in V(\tZ_1)$, the number of paths in $\pset^*$ that intersect $Q(v)$ is at most $O(\Delta_0^2)$.
\end{lemma}

\begin{proof}
For simplicity of notation, denote $Q=Q(v)$.
We construct a collection $\dset$ of at most $3\rho\leq O(\Delta_0)$ discs, such that for each disc $\eta\in \dset$, its boundary is a $G$-normal curve of length at most $O(\Delta_0)$. We ensure that every vertex $v'\in V(Q)$ is contained in some disc $\eta\in \dset$, and no terminal of $\tT$ is contained in $\bigcup_{\eta\in \dset}\eta^{\circ}$. It is then easy to see that every path $P\in \pset^*$ that intersects $Q$ must also contain at least one vertex on the boundary of some disc in $\dset$, and, since the total length of all boundaries of the discs in $\dset$ is bounded by 
$O(\Delta_0^2)$, at most $O(\Delta_0^2)$ paths in $\pset^*$ intersect $Q$. 

For simplicity of notation, for all $1\leq i<\rho$, we denote $\mu_i(Q)$ by $\mu_i$, and $\chi_i(Q)$ by $\chi_i$. We let $v_i,v_i'$ be the first and the last vertex of $\mu_i$, respectively, and we let $R_i\in \trset_i$ be the component with $\chi(R_i)=\chi_i(Q)$, so that $v_{i+1}=u(R_i)$, and $v_i'=u'(R_i)$ (see Figure~\ref{fig: staircase2}). We start with $\dset=\emptyset$. For each $1\leq i<\rho$, we add at most three discs to $\dset$.

Fix some $1\leq i<\rho$. We first add to $\dset$ the disc $\eta(R_i)$, given by Theorem~\ref{thm: curves around CC's}. Notice that this disc contains all vertices of $\chi_i(Q)$, and its boundary is a $G$-normal curve of length $O(\Delta_0)$. Let $\rset'$ be the set of all type-1 and type-2 components $R$ of $G[\tilde U_i]$, such that $v_i$ is an inner vertex of $\sigma(R)$. Since the segments in $\set{\sigma_R\mid R\in \rset'}$ form a nested set, we can assume that $\rset'=\set{R^1,\ldots, R^q}$ with $\sigma(R^1)\subseteq \sigma(R^2)\subseteq \cdots\subseteq \sigma(R^q)$. Notice that among all components in $\rset'$, only $R^q$ may be a type-1 component, and the remaining components are type-2 components. If $\rset'\neq \emptyset$, then we let $v^*_i$ be the last vertex on $\sigma(R^q)$ in the counter-clock-wise direction, so $v^*_i\in \mu_i$, and we add the disc $\eta(R^q)$ to $\dset$. Otherwise, we let $v^*_i=v_i$. Let $\mu'_i\subseteq \mu_i$ be the segment of $\mu_i$ between $v_i^*$ and $v'_i$ (see Figure~\ref{fig: staircase2}). Notice that all vertices of $\mu_i\setminus \mu'_i$ are contained in the discs we have already added to $\dset$. Our final step is to add a disc $\eta^*_i$, containing all vertices of $\mu'_i$, that is defined as follows. From our construction of the staircases and the set $\rset'$, if $R$ is a component of $G[\tilde U_i]$, and some inner vertex of $\mu'_i$ belongs to $L(R)$, then $R$ is a type-2 component, and $\sigma(R)\subseteq \mu'(R)$. Therefore,  we can draw a $G$-normal curve $\gamma_i$ with endpoints $v'_i$ and $v^*_i$, so that $\gamma_i$ is contained in $\dnot(\tZ_{i+1})\setminus D(\tZ_i)$, and it is internally disjoint from all vertices of $G$. From Property~(\ref{prop: short curve}) of shells, we can construct curves $\gamma,\gamma'$ that connect $v'_i$ and $v^*_i$, respectively, to some vertices $a,a'\in C_x$, so that $\gamma,\gamma'\subseteq D(\tZ_i)$, and the length of each curve is bounded by $i\leq 8\Delta_0$. Combining $\gamma_i,\gamma,\gamma'$, and one of the segments of $C_x$ with endpoints $a$ and $a'$, we obtain a closed curve $\gamma^*_i$ of length $O(\Delta_0)$, such that the disc $\eta^*_i$, whose boundary is $\gamma^*_i$, contains all vertices of $\mu'_i$. We then add $\eta^*_i$ to $\dset$. 

\begin{figure}[h]
 \centering
\scalebox{0.5}{\includegraphics{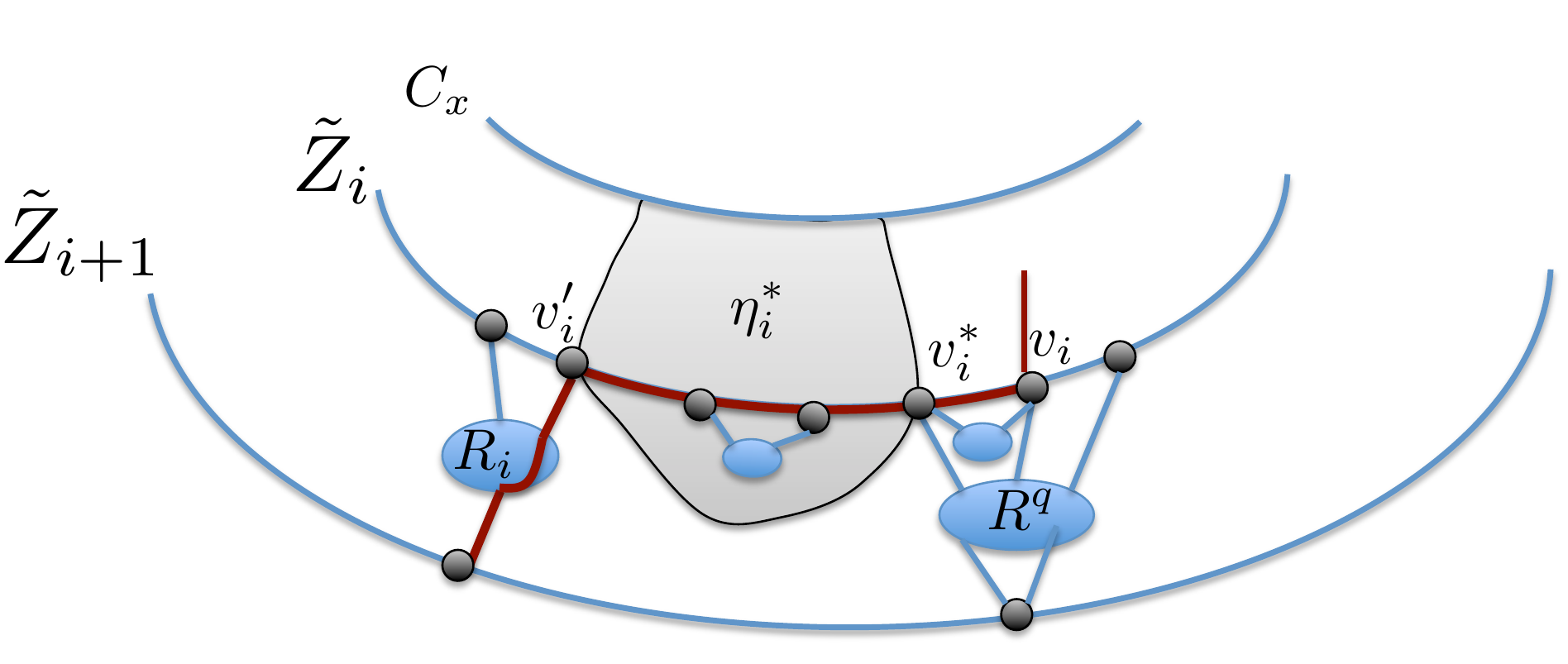}}
\caption{Constructing the disc $\eta^*_i$. The staircase is shown in red.\label{fig: staircase2}}
\end{figure}

It is easy to verify that all vertices of $\mu_i\cup \chi_i$ lie in the discs $\eta(R_i),\eta(R^q),\eta^*_i$, and for each such disc, its boundary is a $G$-normal curve of length at most $O(\Delta_0)$. Moreover, from our construction, none of these discs contains a vertex of $Z_{\rho}$, except as part of its boundary. We conclude that every vertex of $Q$ belongs to some disc $\eta\in \dset$, the boundary of each such disc is a $G$-normal curve of length at most $O(\Delta_0)$. Every path $P\in \pset^*$ intersecting $Q$ must intersect the boundary of at least one disc in $\dset$, and since $|\dset|=O(\Delta_0)$, the number of such paths is bounded by $O(\Delta_0^2)$.
\end{proof}

Consider now the set $\qset=\set{Q(v)\mid v\in V(\tZ_1)}$ of staircases. Notice that once a pair of staircases meet, they always continue together. Therefore, we can partition the set $\qset$ of staircases into equivalence classes $\qset_1,\ldots,\qset_{\lambda}$, where two staircases belong to the same class iff they terminate at the same vertex. From the above discussion, staircases belonging to distinct equivalence classes must be disjoint. For each such set $\qset_j$, let $V_j\subseteq V(\tZ_1)$ be the set of all vertices where the staircases of $\qset_j$ originate. 
We also need the following lemma.

\begin{lemma}\label{lem: few path from each}
Let $\pset^*$ be any set of node-disjoint paths, connecting a subset of the demand pairs in $\tmset$. Then for each $1\leq j\leq \lambda$,   the number of paths in $\pset^*$ that originate at the vertices of $V_j$ is at most $O(\Delta_0^2)$.
\end{lemma}

\begin{proof}
Fix some $1\leq j\leq \lambda$, and let $\pset'\subseteq \pset^*$ be the set of paths originating at the vertices of $V_j$. We assume that $\qset_j=\set{Q_1,\ldots,Q_r}$, and the endpoints of the paths $Q_i$ that belong to $V_j$ appear consecutively in this order on $\tZ_1$. Let $v_1,v_r\in V_j$ be the vertices where the paths $Q_1$ and $Q_r$ originate. From Property~(\ref{prop: short curve}) of shells, we can construct $G$-normal curves $\gamma,\gamma'$, connecting $v_1$ and $v_r$, respectively, to some vertices $a,a'$ of $C_x$, so that the length of each such curve is bounded by $h'\leq 8\Delta_0$, and both curves are contained in $D(\tZ_1)$. We claim that every path in $\pset'$ must contain a vertex of $V(\gamma)\cup V(\gamma')\cup V(Q_1)\cup V(Q_r)\cup V(C_x)$. Indeed, consider the curve $\gamma^*$ obtained by the union of the images of $Q_1$ and $Q_r$, the curves $\gamma$ and $\gamma'$, and one of the segments of $C_x$ with endpoints $a$ and $a'$. Then for every path of $\pset'$, its source lies inside or on the curve $\gamma^*$, and its destination lies either on  $\gamma^*$ (when the destination is the common endpoint of all paths in $\qset_j$), or outside the curve $\gamma^*$. Therefore, every path in $\pset'$ must contain a vertex of $V(\gamma)\cup V(\gamma')\cup V(Q_1)\cup V(Q_r)\cup V(C_x)$. Since $\gamma,\gamma'$ and $C_x$ contain $O(\Delta_0)$ vertices each, and since from Lemma~\ref{lem: intersection of paths and staircases} at most $O(\Delta_0^2)$ paths in $\pset^*$ may intersect $Q_1$ and $Q_r$, we get that $|\pset'|\leq O(\Delta_0)^2$.
\end{proof}

We are now ready to complete our reduction. For every source vertex $s\in \tS$, let $s'$ be the vertex of $\tZ_{\rho}$ that serves as the other endpoint of path $Q(s)$. We define a new set of demand pairs $\hmset=\set{(s',t)\mid (s,t)\in \tmset}$. Notice that $(G,\hmset)$ is now an instance of the special case, where we use $Z=\tZ_{\rho}$ and $C=C_x$.  We construct a graph $\hat G$ from $G$, by creating a hole $\dnot(\tZ_{\rho})$ in the sphere, and removing from $G$ all edges and vertices that lie in the interior of $D(\tZ_{\rho})$. We then apply the $O(\log n)$-approximation algorithm for \NDPdisc to the resulting problem, to obtain a routing of at least $\Omega\left(\frac{\opt(\hat G,\hmset)}{\log n} \right )$ demand pairs. 
From Theorem~\ref{thm: case 2b - special case}, $\opt(\hat G,\hmset)\geq \Omega\left(\frac{\opt(G,\hmset)}{\Delta_0^2\log n}\right )$, and so we obtain a set $\pset$ of disjoint paths, routing of $\Omega\left(\frac{\opt(G,\hmset)}{\Delta_0^2\log^2 n}\right )$ demand pairs of $\hmset$ in graph $\hat G$. Notice that for all $1\leq j\leq \lambda$, if we denote by $b_j$ the unique vertex where all staircases of $\qset_j$ terminate, then at most one demand pair in which $b_j$ participates is routed by $\pset$. 
We construct a set $\qset'\subseteq \qset$ of staircases as follows. For every demand pair $(s',t)\in \hmset$ routed by $\pset$, we select one source vertex $s\in \tS$ with $(s,t)\in\tmset$, such that the staircase $Q(s)$ terminates at $s'$, and we add $Q(s)$ to $\qset'$. Notice that all staircases in $\qset'$ are disjoint from each other, since all of them belong to different sets $\qset_j$.

Since the paths in $\pset$ do not use any vertices in the interior of $D(\tZ_{\rho})$, we can combine them with the paths in $\qset'$, to obtain a routing of  $\Omega\left(\frac{\opt(G,\hmset)}{\Delta_0^2\log^2 n}\right )$ demand pairs of $\tmset$ in graph $G$. In order to complete the proof of Theorem~\ref{thm: before reduction}, it is now enough to show that $\opt(G,\hmset)\geq \Omega\left(\frac{\opt(G,\tmset)}{\Delta_0^2}\right )$. We do so in the following claim.

\begin{claim}\label{claim 2b - reduction}
$\opt(G,\hmset)\geq \Omega\left(\frac{\opt(G,\tmset)}{\Delta_0^2}\right )$.
\end{claim}

\begin{proof}
Let $\pset_0$ be a set of paths routing $\kappa_0=\opt(G,\tmset)$ demand pairs of $\tmset$ in $G$. We show that we can route 
$\Omega(\kappa_0/\Delta_0^2)$ demand pairs of $\hmset$ in $G$. Let $\tmset_0\subseteq\tmset$ be the set of the demand pairs routed by set $\pset_0$. For each demand pair $(s,t)\in \tmset_0$, let $P(s,t)\in \pset_0$ be the path routing this demand pair.
We build a conflict graph $H$, whose vertex set is $\set{v(s,t)\mid (s,t)\in \tmset_0}$, and there is a directed edge from $v(s_1,t_1)$ to $v(s_2,t_2)$ iff one of the following happens: either (i) $s_1$ and $s_2$ both belong to the same set $V_j$, for $1\leq j\leq \lambda$; or (ii) path $P(s_2,t_2)$ intersects $Q(s_1)$. From Lemmas~\ref{lem: intersection of paths and staircases} and \ref{lem: few path from each}, the in-degree of every vertex in $H$ is at most $O(\Delta_0^2)$. Therefore, there is an independent set $I$ of $\Omega(\kappa_0/\Delta_0^2)$ vertices in $H$. We let $\tmset_1=\set{(s,t)\mid v(s,t)\in I}$, and we let $\pset_1\subseteq\pset_0$ be the set of paths routing the demand pairs in $\tmset_1$.

Let $\hmset'\subseteq \hmset$ contain, for every demand pair $(s,t)\in \tmset_1$, the pair $(s',t)$. It is now enough to show that all demand pairs in $\hmset'$ can be routed in $G$. Let $\qset'$ contain all staircases $Q(s)$, where $s$ participates in some demand pair in $\tmset_1$. Then all staircases in $\qset'$ are disjoint from each other, since they all belong to different sets $\qset_j$, for $1\leq j\leq \lambda$. Moreover, for each staircase $Q(s)\in \qset'$, all paths in $\pset_1\setminus\set{P(s,t)}$ are disjoint from $Q(s)$. By concatenating the paths in $\pset_1$ and the staircases in $\qset'$, we obtain a collection of node-disjoint paths routing all demand pairs in $\hmset'$
\end{proof}



\label{-----------------------------------sec: proof of 2nd main thm------------------------------------------}
\section{Proof of Theorem~\ref{thm: main2}}\label{sec: proof of 2nd main thm}

We perform a transformation to instance $(G,\mset)$ as before, to ensure that every terminal participates in at most one demand pair, and the degree of every terminal is $1$. The number of vertices in the new instance is bounded by $2n^2$, and abusing the notation we denote this number by $n$.
 We use the following analogue of Theorem~\ref{thm: solution or sep oracle}.

\begin{theorem} \label{thm: separation solution opt epsilon} There is an efficient algorithm, that, given any semi-feasible solution to (LP-flow2), either computes a routing of at least $\Omega \left( \frac{{(X^*)}^{1/19}}{\poly \log n}\right)$ demand pairs in $\mset$ via node-disjoint paths, or returns a constraint of type (5), that is violated by the current solution. 
\end{theorem}

We show that the above theorem implies Theorem~\ref{thm: main2}. The Ellipsoid Algorithm, in every iteration, applies the above theorem to the current semi-feasible solution $(x,f)$ to (LP-flow2). If the outcome is a solution routing at least $\Omega \left( \frac{{(X^*)}^{1/19}}{\poly \log n}\right)$ demand pairs, then we obtain the desired routing, assuming that $X^* = \opt$ is guessed correctly. Otherwise, we obtain a violated constraint of type (5), and continue to the next iteration of the Ellipsoid Algorithm. The algorithm is guaranteed to terminate with a feasible solution after a number of iterations that is polynomial in the number of the LP-variables, so we obtain an efficient algorithm, that returns a solution routing $\Omega\left(\frac{(\opt(G,\mset))^{1/19}}{\poly\log n}\right )$ demand pairs. We now focus on proving Theorem~\ref{thm: separation solution opt epsilon}.

We again process the fractional solution $(x,f)$ to obtain a new fractional solution $(x', f')$, where every demand pair sends either $0$ or $w^* $ flow units, in the same way as described in Section~\ref{sec: alg overview}. We let $\mset' \subseteq \mset$ denote the set of the demand pairs $(s_i,t_i)$ with non-zero flow value $x'_i$ in this new solution. As before, the total flow between the demand pairs in $\mset'$ is at least $\Omega(X^*/\log k)$ in the new solution, and, if we find a subset $\mset'' \subseteq \mset'$ of demand pairs with $\opt(G, \mset'') \le w^* |\mset''|/2$, then set $\mset''$ defines a violated constraint of type (5) for (LP-flow2). Therefore, we focus on set $\mset'$ and for simplicity denote $\mset = \mset'$.

We decompose the input instance $(G, \mset)$ into a collection of well-linked instances $\set{(G_j, \mset^j)}_{j=1}^r$ using Theorem~\ref{thm: wld}. For each $1 \le j \le r$, let $W_j = w^*|\mset^j|$ be the contribution of the demand pairs in $\mset^j$ to the current flow solution and let $W = \sum_{j=1}^r W_j = \Omega(X^*/\log k)$.

Theorem~\ref{thm: main} guarantees that for each $1 \le j \le r$, we can obtain one of the following:
\begin{enumerate}
\item Either a collection $\pset^j$ of node-disjoint paths, routing $\Omega(W_j^{1/19}/\poly\log n)$ demand pairs of $\mset^j$ in $G_j$; or

\item A collection $\tmset^j\subseteq \mset^j$ of demand pairs, with $|\tmset^j|\geq |\mset^j|/2$, such that $\opt(G_j,\tmset^j)\leq w^*|\tmset^j|/8$.
\end{enumerate}

We say that instance $(G_j, \mset^j)$ is a type-1 instance, if the first outcome happens for it, and we say that it is a type-2 instance otherwise. Let $I_1=\set{j\mid \mbox{$(G_j,\mset^j)$ is a type-1 instance}}$, and similarly, $I_2=\set{j\mid \mbox{$(G_j,\mset^j)$ is a type-2 instance}}$. We consider two cases, where the first case happens when $\sum_{j\in I_1}W_j\geq W/2$, and the second case when $\sum_{j\in I_2}W_j\geq W/2$. In the second case, we let $\mset' = \bigcup_{j \in I_2} \tmset^j$, and by the same reasoning as in Section~\ref{sec: alg overview}, the following inequality, that is violated by the current LP-solution, is a valid constraint of (LP-flow2): 

\[\sum_{(s_i, t_i) \in \mset'}x'_i \le w^* |\mset'|/2.\]

We now focus on Case 1, where the number of paths routed for each instance $(G_j, \mset^j)$ with $j \in I_1$ is at least $|\pset^j| = \Omega\left(\frac{W_j^{1/19}}{\poly\log n}\right)$. Since $\sum_{j\in I_1}W_j\geq W/2 = \Omega\left(X^*/\log k\right)$, the total number of paths routed is:
\[\begin{split}\sum_{j \in I_1} |\pset^j| 
&\ge\sum_{j\in I_1} \Omega\left(\frac{W_j^{1/19}}{\poly\log n}\right)\\
&\geq\sum_{j\in I_1} \Omega\left(\frac{W_j}{W^{18/19}\cdot \poly\log n}\right)\\
&=
\Omega\left(\frac{W^{1/19}}{\poly\log n}\right)\\
&= \Omega\left(\frac{{(X^*)}^{1/19}}{\poly\log n}\right).\end{split}\]

\label{-----------------------------------sec: conclusion------------------------------------------}
\section{Conclusion and Open Problems}\label{sec: conclusion}


In this paper we showed the first approximation algorithm for the \NDPplanar problem, whose approximation factor breaks the $\Omega(n^{1/2})$ barrier of the multicommodity flow LP-relaxation. We introduce a number of new techniques, that we hope will be helpful in obtaining better approximation algorithms for this problem. We note that our initial motivation came from the improved approximation algorithm for \NDPgrid of~\cite{NDP-grids}. Even though adapting their main idea to the more general setting of planar graphs is technically challenging, we believe that the work of~\cite{NDP-grids} on the much simpler and better structured grid graphs helped crystallize the main conceptual idea that 
eventually lead to this result. Therefore, we believe that studying the \NDPgrid problem can be very helpful in understanding the more general \NDPplanar problem. The best current approximation algorithm for \NDPgrid achieves an $\tilde O(n^{1/4})$-approximation, and it seems likely that this approximation ratio can be improved. We leave open the question of whether the techniques introduced in this paper can help improve the $O(n^{1/2})$-approximation factor of~\cite{EDP-alg} for \EDP on planar graphs. Finally, we remark that the complexity of the \NDPdisc and the \NDPcyl problems is still not well-understood: we provide an $O(\log k)$-approximation algorithm for both problems, and we are not aware of any results that prove that the optimization versions of \NDPdisc or \NDPcyl are \NP-hard. We note that the \EDP problem is known to be \NP-hard for both these settings~\cite{EDP-disc-hardness}, but we are not aware of any approximation algorithms for it.

\bibliography{NDP-planar-graphs.v3}
\bibliographystyle{alpha}

\label{------------------------------------------------Appendix-----------------------------------------------------}
\label{------------------------------------------------Appendix-------------------------------------------------------}
\appendix

\section{Proofs Omitted from Section~\ref{sec: prelims}}\label{sec: proofs from prelims}
\subsection{Proof of Observation~\ref{obs: sparsest cut}}

For every terminal $t\in \tset$, let $v_t\in V(G)$ be the unique neighbor of $t$ in $G$. Let $(A,C,B)$ be the cut computed by algorithm $\algsc$, whose sparsity $\beta=\frac{|C|}{\min\set{|A\cap \tset|,|B\cap \tset|}+|C\cap \tset|}$ is within factor $\alphasc$ of the optimal one. If $\beta> 1$, then we replace cut $(A,C,B)$ with the cut $(\emptyset,\tset,V(G)\setminus \tset)$, obtaining a cut of sparsity $1$, and continue with our algorithm, so we assume that $\beta\leq 1$ from now on. While $C\cap \tset\neq \emptyset$, let $t\in C\cap \tset$ be any terminal in $C$. 

If $|A\cap \tset|<|B\cap \tset|$, then we move $t$ from $C$ to $A$, and if $v_t\not \in C$, we add $v_t$ to $C$. Let $(A',C',B')$ be the resulting tri-partition of $V(G)$. It is easy to see that $(A',C',B')$ is a valid vertex cut, and $|C'|\leq |C|$, while $\min\set{|A'\cap \tset|,|B'\cap \tset|}+|C'\cap \tset|= \min\set{|A\cap \tset|,|B\cap \tset|}+|C\cap \tset|$, so the sparsity of the cut does not increase.

The case where $|B\cap \tset|<|A\cap \tset|$ is dealt with similarly.

Finally, assume that $|A\cap \tset|=|B\cap \tset|$. 
If $v_t\in A$, then we move $t$ from $C$ to $A$, and otherwise we move $t$ from $C$ to $B$. Let $(A',C',B')$ be the resulting tri-partition of $V(G)$. It is easy to see that $(A',C',B')$ is a valid vertex cut, and $|C'|=|C|-1$. 
It is also easy to verify that $\min\set{|A'\cap \tset|,|B'\cap \tset|}+|C'\cap \tset|\geq \min\set{|A\cap \tset|,|B\cap \tset|}+|C\cap \tset|-1$, and $\min\set{|A'\cap \tset|,|B'\cap \tset|}+|C'\cap \tset|>0$  hold. Since $\beta\leq 1$, the sparsity of the cut does not increase. Once we process all terminals in set $C$ in this fashion, we obtain a final vertex cut $(A,C,B)$ whose sparsity is at most $\beta$, and $C\cap \tset=\emptyset$. 

\subsection{Proof of Claim~\ref{claim: partition the forest}}

We compute a partition $\rset(\tau)$ for every tree $\tau\in F$ separately. The partition is computed in iterations, where in the $j$th iteration we compute the set $R_j(\tau)\subseteq V(\tau)$ of vertices, together with the corresponding collection $\pset_j(\tau)$ of paths. For the first iteration, if $\tau$ contains a single vertex $v$, then we add this vertex to $\rset_1(\tau)$ and terminate the algorithm. Otherwise, for every leaf $v$ of $\tau$, let $P(v)$ be the longest directed path of $\tau$, starting at $v$, that only contains degree-1 and degree-2 vertices, and does not contain the root of $\tau$. We then add the vertices of $P(v)$ to $R_1(\tau)$, and the path $P(v)$ to $\pset_1(\tau)$. Once we process all leaf vertices of $\tau$, the first iteration terminates. It is easy to see that all resulting vertices in $R_1(\tau)$ induce a collection $\pset_1(\tau)$ of disjoint paths in $\tau$, and moreover if $v,v'\in R_1(\tau)$, and there is a path from $v$ to $v'$ in $\tau$, then $v,v'$ lie on the same path in $\pset_1(\tau)$. We then delete all vertices of $R_1(\tau)$ from $\tau$.

The subsequent iterations are executed similarly, except that the tree $\tau$ becomes smaller, since we delete all vertices that have been added to the sets $R_j(\tau)$ from the tree.

It is now enough to show that this process terminates after $\ceil{\log n}$ iterations. In order to do so, we can describe each iteration slightly differently. Before each iteration starts, we contract every edge $e$ of the current tree, such that at least one endpoint of $e$ has degree $2$ in the tree, and $e$ is not incident on the root of $\tau$. We then obtain a tree in which every inner vertex (except possibly the root) has degree at least $3$, and delete all leaves from this tree. The number of vertices remaining in the contracted tree after each such iteration therefore decreases by at least factor $2$. It is easy to see that the number of iteration in this procedure is the same as the number of iterations in our algorithm, and is bounded by $\ceil{\log n}$.
For each $1\leq j\leq \ceil{\log n}$, we then let $R_j=\bigcup_{\tset\in F}R_j(\tset)$.

\subsection{Proof of Lemma~\ref{lem: getting r-split demand pairs on a disc}}

We denote by $\tset$ the set of all vertices participating in the demand pairs in $\mset$, and we refer to them as terminals.
Consider any demand pair $(s,t)\in \mset$, and let $\sigma(s,t),\sigma'(s,t)$ be the two segments of $C$ whose endpoints are $s$ and $t$. We assume without loss of generality that $|\sigma(s,t)\cap \tset|\leq |\sigma'(s,t)\cap \tset|$, and we denote $\delta(s,t)=|\sigma(s,t)\cap \tset|-1$. By possibly renaming the terminals $s$ and $t$, we assume that $s$ appears before $t$ on $\sigma(s,t)$ as we traverse it in counter-clock-wise direction along $C$. Our first step is to partition the demand pairs in $\mset$ into $\ceil{\log \kappa}$ subsets $\nset_1,\ldots,\nset_{\ceil{\log \kappa}}$, as follows. For each $1\leq i\leq \ceil{\log\kappa}$, $\nset_i$ contains all demand pairs $(s,t)$ with $2^{i-1}\leq \delta(s,t)<2^i$. In order to complete the proof of the lemma, it is enough to show that for each $1\leq i\leq \ceil{\log\kappa}$, we can partition $\nset_i$ into four sets of demand pairs, each of which is $r$-split, for some integer $r$.

Fix some $1\leq i\leq \ceil{\log \kappa}$, and assume that $\tset=\set{v_1,\ldots,v_{2\kappa}}$, where the vertices are indexed in the circular order of their appearance on $C$. We let $A$ contain all vertices $v_j$, where $j=1$ modulo $2^{i-1}$. For convenience, we denote $A=\set{a_1,\ldots,a_z}$, and we assume that the vertices of $A$ appear in this circular order on $C$, with $a_1=v_1$. If $|A|=1$, then $|\tset|\leq 2^{i-1}$, and so $\delta(s,t)\leq 2^{i-2} $ for all $(s,t)\in \mset$, and $\nset_i=\emptyset$. If $|A|=2$, then let $\beta,\beta'$ be the two segments of $C$ between $a_1$ and $a_2$, where $\beta$ contains $a_1$ but not $a_2$, and $\beta'$ contains $a_2$ but not $a_1$. Then for every pair $(s,t)\in \nset_i$, one of the two terminals lies on $\beta$ and the other on $\beta'$, and so $\nset_i$ is $1$-split. We assume from now on that $z\geq 3$.

For each $1\leq j\leq z-1$, let $\beta_j$ be the segment of $C$ from $a_j$ and $a_{j+1}$, as we traverse $C$ in the counter-clock-wise order, so that $\beta_j$ includes $a_j$ but excludes $a_{j+1}$. We let $\beta_{z}$ be the segment of $C$ from $a_{z}$ to $a_1$, as we traverse $C$ in the counter-clock-wise order, so that $\beta_{z}$ includes $a_{z}$, but not $a_1$. 
Then for every segment $\beta_j$ with $1\leq j<z$, $|\beta_j\cap \tset|=2^{i-1}$, while $|\beta_z\cap \tset| \leq 2^{i-1}$.
Notice that for every demand pair $(s,t)\in \nset_i$, one of the following must happen: either  (i) $s\in \beta_j$, $t\in \beta_{j+1}$ for some $1\leq j\leq z$, where we treat $z+1$ as $1$; or (ii) $s\in \beta_z$, $t\in \beta_2$; or (iii) $s\in \beta_{z-1}$, $t\in \beta_1$.

We are now ready to partition $\nset_i$ into four subsets. The first subset, $\nset_i^1$, contains all pairs $(s,t)\in \nset_i$ with $s\in \beta_{z-1}$, $t\in \beta_1$. The second set, $\nset_i^2$, contains all pairs $(s,t)\in \nset_i$, with  $s\in \beta_z$, $t\in \beta_1\cup \beta_2$. It is immediate to verify that each of these two sets is $1$-split. The third set, $\nset_i^3$, contains all pairs $(s,t)\in \nset_i\setminus(\nset_i^1\cup \nset_i^2)$, where $s\in \beta_j$, for an odd index $1\leq j<z$. This set is $\floor{z/2}$-split, with the segments $\beta_1,\ldots,\beta_{z}$ giving the splitting. The last set, $\nset_i^4$, contains all pairs $(s,t)\in \nset_i\setminus(\nset_i^1\cup \nset_i^2\cup \nset^3_i)$. This set is similarly $\floor{z/2}$-split, since for all $(s,t)\in \nset_i^4$, $s\in \beta_j$ for an even index $1\leq j<z$.
Overall, we obtain a partition of $\mset$ into $4\ceil{\log \kappa}$ sets, each of which is $r$-split for some integer $r$.

\subsection{Proof of Lemma~\ref{lem:reroute-montone}}

Consider any path $P\in \pset$.  A sub-path $Q$ of $P$ is called a bump, if the two endpoints of $Q$ lie on $C$, and all the intermediate vertices of $Q$ do not lie on $C$ (notice that $Q$ may be simply an edge of $C$). Since path $P$ is simple, a bump $Q$ cannot contain $s$. 

We now define a shadow of a bump $Q$. If $Q$ is an edge of $C$, then the shadow of $Q$ is $Q$ itself.  Assume now that $Q$ is not an edge of $C$. Let $u, v \in V(C)$ be the two endpoints of $Q$, and let $\sigma,\sigma'$ be the two segments of $C$ with endpoints $u$ and $v$. 
Let $C_1$ be the union of $\sigma$ and $Q$, and $C_2$ the union of $\sigma'$ and $Q$. Then one of the two corresponding discs, $D(C_1)$ and $D(C_2)$ contains $s$ - we assume that it is the latter disc. We then let $\sigma$ be the \emph{shadow} of $Q$ on $C$ (see Figure~\ref{fig: shadow}).

\begin{figure}[h]
 \centering
\scalebox{0.5}{\includegraphics{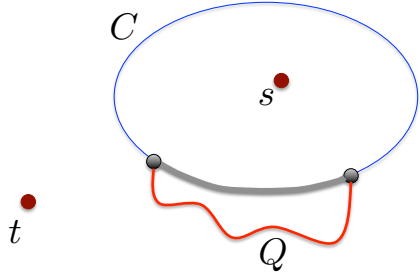}}
\caption{A bump $Q$ and its shadow\label{fig: shadow}}
\end{figure}

We note that all inner points on the image of $Q$ must lie outside $\dnot(C)$, as otherwise, we can find a cycle $C'$ in $H$ with $D(C')$ containing $s$ and $D(C')\subsetneq D(C)$, contradicting the fact that $C=\mincycle(H,s)$. We obtain the following two simple observations.

\begin{observation}\label{obs: escape C}
Let $P\in \pset$ be any path, and let $R$ be the longest segment of $P$, starting from $s$ and terminating at a vertex of $C$, such that $R$ does not contain any vertex of $C$ as an inner vertex. Let $v\in V(C)$ be the endpoint of $R$, and let $R'$ be the sub-path of $P$ from $v$ to $t$. Then  every point $p$ on the image of $R'$ lies outside $\dnot(C)$.
\end{observation}
\begin{proof}
If any such point $p$ lies in $\dnot(C)$, then for some bump $Q$ of $P$, some inner point on the image of $Q$ lies in $\dnot(C)$, and this is impossible, as observed above.
\end{proof}

\begin{observation}\label{obs:shadow-node-disjoint}
Let $P,P'\in \pset$ be two distinct paths, let $Q$ be a bump on $P$, and let $Q'$ be a bump on $P'$. Then the shadows of $Q$ and $Q'$ are node-disjoint. 
\end{observation}
\begin{proof}
Let $\sigma$ be the shadow of $Q$, and let $a,b$ be its endpoints. Let $\sigma'$ be the shadow of $Q'$, and let $a',b'$ be its endpoints. Since the inner points of the images of both  $Q$ and $Q'$ lie outside $\dnot(C)$, and their endpoints are all distinct, if $\sigma$ and $\sigma'$ are not disjoint, then either $\sigma\subsetneq \sigma'$, or $\sigma'\subsetneq \sigma$. Assume without loss of generality that the latter is true. Let $A$ be the disc whose boundary is $\sigma\cup Q$. Then $t\not\in A$, but some vertex of $P'$ lies in $A$. Let $R'$ be the longest sub-path of $P'$ starting from $s$ and terminating at a vertex of $C$, such that $R'$ is internally disjoint from $C$. Then there is a vertex $v'\in V(Q')$, that does not lie on $R'$. Let $R''$ be the sub-path of $P'$ from $v'$ to $t$. Then path $R''$ originates in disc $A$ and terminates outside $A$. But it can only leave $A$ by traveling inside $\dnot(C)$, which is impossible from Observation~\ref{obs: escape C}.
\end{proof}

We now re-route each path $P\in \pset$, as follows. Let $u$ and $v$ be the first and the last vertex of $P$ that belong to $C$, respectively. Let $\tilde P$ be  the union of the shadows of all bumps of $P$. Then $\tilde P$ is a path, contained in $C$, that contains $u$ and $v$. Let $\tilde P'$ be a simple sub-path of $\tilde P$ connecting $u$ to $v$, and let $P'$ be obtained by concatenating the segment of $P$ from $s$ to $u$, $\tilde P'$, and the segment of $P$ from $v$ to $t$. Notice that path $P'$ is monotone with respect to $C$, and all paths in the resulting set $\pset'=\set{P'\mid P\in \pset}$ are internally node-disjoint by Observation~\ref{obs:shadow-node-disjoint}. 

\subsection{Proof of Theorem~\ref{thm: monotonicity for shells}}

We perform $r$ iterations. At the beginning of the $h$th iteration, for $1\leq h\leq r$, we assume that we are given a set $\pset_{h-1}$ of $\kappa$ node-disjoint paths in $H'$, connecting the vertices of $A$ to the vertices of $B$, so that the paths in $\pset_{h-1}$ are internally disjoint from $V(C)\cup V(Y)$, and they are monotone with respect to $Z_1,\ldots,Z_{h-1}$. The output of iteration $h$ is a set $\pset_h$ of $\kappa$ node-disjoint paths in $H'$,  connecting the vertices of $A$ to the vertices of $B$, so that the paths in $\pset_{h}$ are internally disjoint from $V(C)\cup V(Y)$, and they are monotone with respect to $Z_1,\ldots,Z_h$. The output of the algorithm is the set $\pset_r$ of paths computed in the last iteration. For simplicity of notation, we denote $Z_0=C$, even though $C$ is not a cycle of $H$.

We start with $\pset_0=\pset$. It is immediate to see that this is a valid input to iteration $1$. Assume now that we are given a set $\pset_{h-1}$ of paths, which is a valid input to iteration $h$. The iteration is then executed as follows. Let $H''$ be the graph obtained by starting with $H''=\left (\bigcup_{h=0}^rV(Z_h)\right )\cup \pset_{h-1}$, and then contracting all vertices lying in $D(Z_{h-1})$ into a source vertex $s'$, and all vertices of $B$ into a destination vertex $t'$. Since $B\subseteq V(Y)$, where $Y$ is a connected sub-graph of $G$, $Y\cap D(Z_r)=\emptyset$, and the paths in $\pset_{h-1}$ are internally disjoint from $V(Y)$, graph $H''$ is a planar graph. Moreover, $Z_h=\mincycle(H'',s')$ from the definition of tight concentric cycles. We consider the drawing of $H''$ in the plane where $t'$ is incident on the outer face. It is easy to see that $Z_h=\mincycle(H'',s')$ still holds with respect to this new drawing of $H''$.

From Lemma~\ref{lem:reroute-montone}, there is a set $\qset$ of $\kappa$ internally node-disjoint paths in $H''$, connecting $s'$ to $t'$, that are monotone with respect to $Z_h$.
In order to construct the set $\pset_h$ of paths, let $P\in \pset_{h-1}$ be any path, and let $v_P$ be the last vertex of $P$ lying on $Z_{h-1}$. Let $P'$ be the sub-path of $P$ starting from $v_P$ and terminating at the endpoint of $P$ lying in $B$. 
Notice that from the monotonicity of the paths in $\pset_{h-1}$ with respect to $Z_1,\ldots,Z_{h-1}$, there are exactly $\kappa$ edges leaving the vertex $s'$ in $H''$, each edge lying on a distinct path $P\in \pset_{h-1}$. If edge $e$ is leaving $s'$ in $H''$, and $e$ lies on $P$, then it is incident on $v_P$, and so exactly one path of $\qset$ originates at $v_P$. We denote this path by $Q_P$. For each path $P\in \pset_{h-1}$, we let $P^*$ be the path obtained from $P$ by replacing $P'$ with $Q_P$. Notice that, since the paths in $\pset_{h-1}$ are internally disjoint from $B$, exactly $\kappa$ edges are incident on $t'$ in graph $H'$, each of which is incident on a distinct vertex of $B$ in $H$. It is now easy to verify that the set $\pset_h$ contains $\kappa$ node-disjoint paths, connecting the vertices of $A$ to the vertices of $B$, and the paths in $\pset_h$ are internally disjoint from $V(C)\cup V(Y)$ and monotone with respect to $Z_1,\ldots,Z_h$.
We return $\pset_r$ as the output of the algorithm.

\section{Proofs Omitted from Section~\ref{sec: disc and cylinder}}\label{sec: proofs for disc and cylinder}
\subsection{Proof of Theorem~\ref{thm: approximate DPSP}}


The following definition and observation allow us to slightly relax the problem. 

\begin{definition} Let $c\geq 1$ be an integer, and let $\mset'\subseteq \mset$ be a subset of the demand pairs. Given a constraint $K=(i,a,b,w)$ of type $1$ or $2$, we say that $\mset'$ violates $K$ by a factor of at most $c$, iff the number of  the demand pairs $(s,t)\in \mset'$ with either $s$ or $t$ lying in $(a,b)$ is at most $cw$. Likewise, given a constraint $K=(i,a,b,w)$ of type $3$ or $4$, we say that $\mset'$ violates $K$ by a factor of at most $c$ iff the number of the demand pairs in $\mset'$ crossing $K$ is at most $c w$.
\end{definition}


\begin{observation}\label{obs: violated to satisfied} There is an efficient algorithm, that, given a \DPSP instance $(\sigma,\sigma',\mset,\kset)$ and a non-crossing set $\mset'\subseteq\mset$ of demand pairs that violates every constraint in $\kset$ by a factor of at most $c$, for any integer $c>1$, computes a subset $\mset''\subseteq\mset'$ of at least $|\mset'|/c$ demand pairs, satisfying all constraints in $\kset$.
\end{observation}

\begin{proof} Assume that  $\mset'=\set{(s_{j_1},t_{j_1}),\ldots,(s_{j_z},t_{j_z})}$, where $s_{j_1}\prec\cdots\prec s_{j_z}$ and $t_{j_1}\prec\cdots\prec t_{j_z}$. Let $\mset''=\set{(s_{j_{\ell}},t_{j_{\ell}})\mid 1\leq \ell\leq z\mbox{ and } \ell\equiv 1\mod c}$. Clearly, $\mset''\subseteq \mset'$ and $|\mset''|\geq |\mset'|/c$. We now claim that all constraints in $\kset$ are satisfied by $\mset''$.

Indeed, consider any constraint $K=(i,a,b,w)\in \kset$. Assume first that $K$ is a type-1 constraint. Then at most $cw$ demand pairs in $\mset'$ have a source vertex in the interval $(a,b)$. It is easy to see that the number of the demand pairs of $\mset''$ that have a source vertex in the interval $(a,b)$ is at most $w$. If $K$ is a type-2 constraint, the argument is similar. Assume now that $K$ is a type-3 constraint (the case where it is a type-4 constraint is symmetric). Let $\mset_K\subseteq \mset'$ be the set of all demand pairs of $\mset'$ crossing $K$. The key observation is that the demand pairs of $\mset_K$ appear consecutively in the ordered set $\mset'$. Since $|\mset_K|\leq cw$, it is easy to see that $|\mset''\cap \mset_K|\leq w$, and so set $\mset''$ satisfies the constraint $K$. 
\end{proof}


Let $r=\lceil\log |\mset|\rceil$, and for $1\leq j\leq r+1$, set $W_j=2^{j}$. We partition the constraints of $\kset$ into $r$ levels, where for $1\le j\le r$, the $j$th level contains all constraints $(i,a,b,w)$ with $W_{j-1}\le w<2\cdot W_{j-1}=W_j$. For all $1\le i\le 4$ and $1\le j\le r$, we denote by $\sset_j^{(i)}$ the set of all type-$i$ constraints that belong to level $j$. Let $\sset_j=\bigcup_{i=1}^4\sset_j^{(i)}$ be the set of all level-$j$ constraints.

Our algorithm employs dynamic programming. It is convenient to view the algorithm as constructing $r$ dynamic programming tables - one for each level. Consider some level $1\leq j\leq r$. Let $I=(x,y)\subseteq \sigma$, $I'=(x',y')\subseteq \sigma'$ be a pair of intervals. We say that it is a \emph{good level-$j$} pair if the following conditions hold for every level-$j$ constraint $K=(i,a,b,w)\in \sset_j$:

\begin{properties}{C}
\item if $K$ is a type-1 constraint, then $I$ is not contained in $(a,b)$;
\item if $K$ is a type-2 constraint, then $I'$ is not contained in $(a,b)$;
\item if $K$ is a type-3 constraint, then either $I$ is not contained in $L_a$, or $I'$ is not contained in $R_{b}$; and
\item if $K$ is a type-4 constraint, then either $I$ is not contained in $R_a$, or $I'$ is not contained in $L_{b}$.
\end{properties}

For each $1\le j\le r$, let $\pset_j$ denote the set of all good level-$j$ pairs of intervals. The level-$j$ dynamic programming table, $\Pi_j$ contains an entry $\Pi_j[I,I']$ for every good level-$j$ pair $(I,I')\in \pset_j$ of intervals. The entry will either remain empty, or it will contain a collection of non-crossing demand pairs from $\mset$  of cardinality exactly $W_j$, whose sources lie in $I$ and destinations lie in $I'$.

We now describe an efficient algorithm that computes the entries of the dynamic programming tables. We start with $j=1$.  For every pair $(I,I')\in\pset_1$ of good level-$1$ intervals, if there are two distinct non-crossing demand pairs $(s,t),(s',t')\in\mset$ with $s,s'\in I$ and $t,t'\in I'$, then we set $\Pi_1[I,I']=\set{(s,t),(s',t')}$. Otherwise, we set $\Pi_1[I,I']=\emptyset$. 

Assume now that we have constructed the tables for levels $1,\ldots,j-1$, and consider the level-$j$ table, for some $1< j\leq r$, and its entry $\Pi_j(I,I')$ for some good level-$j$ pair $(I, I')\in\pset_j$ of intervals, where $I=(x,y)$ and $I'=(x',y')$. If there exist two pairs of vertices $u,v\in I$ and $u',v'\in I'$ such that $u\prec v$, $u'\prec v'$, and both $\Pi_{j-1}[(x,u),(x',u')]$ and $\Pi_{j-1}[(v,y),(v',y')]$ are non-empty, then we set $\Pi_j[I,I']=\Pi_{j-1}[(x,u),(x',u')]\cup \Pi_{j-1}[(v,y),(v',y')]$. Otherwise, we set $\Pi_j[I,I']=\emptyset$. This completes the description of the algorithm that computes the entries of the dynamic programming tables. We now proceed to analyze it, starting with the following easy observations.

\begin{observation}\label{obs: dp table cardinality} For all $1\le j\le r$ and $(I,I')\in\pset_j$, either $|\Pi_j[I,I']|=0$ or $|\Pi_j[I,I']|=W_j$.
\end{observation}
\begin{proof}
The proof is by induction on $j$.
Clearly, the claim holds for $j=1$ and all $(I,I')\in\pset_1$ by our construction. Consider now some $j>1$, and assume that the claim holds for all values $j'<j$. Consider some entry $\Pi_j(I,I')$ of the level-$j$ dynamic programming table. Our algorithm either sets $\Pi_j[I,I']=\emptyset$ or $\Pi_j[I,I']=\Pi_{j-1}[(x,u),(x',u')]\cup \Pi_{j-1}[(v,y),(v',y')]$. The latter only happens when both $\Pi_{j-1}[(x,u),(x',u')], \Pi_{j-1}[(v,y),(v',y')]$ are non-empty, and so by the induction hypothesis, $|\Pi_{j-1}[(x,u),(x',u')]|,|\Pi_{j-1}[(v,y),(v',y')]|=W_{j-1}$, giving us $|\Pi_j[I,I']|=2\cdot W_{j-1}=W_j$.
\end{proof}

\begin{observation}\label{obs: dp table non-crossing} For all $1\le j\le r$ and $(I,I')\in\pset_j$, $\Pi_j[I,I']$ is a non-crossing subset of the demand pairs in $\mset$, where every $(s,t)\in T_j[I,I']$ has $s\in I$ and $t\in I'$.
\end{observation}
\begin{proof}
The proof is again by induction on $j$. 
Clearly, the claim holds for $j=1$ and all $(I,I')\in\pset_1$ from our construction. Consider now some level $j>1$ and assume that the claim holds for levels $1,\ldots,(j-1)$. Let $(I,I')\in\pset_j$ be a good level-$j$ pair of intervals, with $I=(x,y)$ and $I'=(x',y')$. The claim trivially holds when $\Pi_j[I,I']=\emptyset$, so we assume that $\Pi_j[I,I']=\Pi_{j-1}[(x,u),(x',u')]\cup \Pi_{j-1}[(v,y),(v',y')]$, where $x\preceq u\prec v\preceq y$ and $x'\preceq u'\prec v'\preceq y'$. From the induction hypothesis, no two demand pairs from $\Pi_{j-1}[(x,u),(x',u')]$ can cross, and the same holds for the demand pairs of $\Pi_{j-1}[(v,y),(v',y')]$. Since every demand pair $(s,t)\in \Pi_{j-1}[(x,u),(x',u')]$ has $s\in (x,u), t\in (x',u')$, and every demand pair $(s',t')\in \Pi_{j-1}[(v,y),(v',y')]$ has $s'\in (v,y),t'\in (v',y')$,  it is immediate to verify that no pair of demands $(s,t)\in \Pi_{j-1}[(x,u),(x',u')]$ and $(s',t')\in \Pi_{j-1}[(v,y),(v',y')]$ can cross, and for every demand pair $(s,t)\in T_j[I,I']$, $s\in I$ and $t\in I'$ holds.
\end{proof}

\begin{observation}\label{obs: dp table constraint violation} For all $1\le j\le r$ and $(I,I')\in\pset_j$, the set $\Pi_j[I,I']$ of demand pairs violates every constraint in $\kset$ by at most factor $4$.
\end{observation}

\begin{proof}
Fix some $1\leq j\leq r$ and $(I,I')\in \pset_j$, and consider the corresponding table entry $\Pi_j[I,I']$. The claim holds trivially when $|\Pi_j[I,I']|=0$, so we assume that $|\Pi_j[I,I']|=W_j$. Consider some constraint $K=(i,a,b,w)\in\sset_{j'}^{(i)}$, for some level $1\leq j'\leq r$. if $j'\geq j$, then $w\geq W_{j-1}$ must hold, while $|\Pi_j[I,I']|=W_j$, so the constraint is violated by the factor of at most $4$.

Consider now the case where $j'<j$. Assume first that $i=1$. Note that $\Pi_j[I,I']$ is the union of exactly $2^{j-j'}$ non-empty level-$j'$ table entries, each of which contains a set of demand pairs of cardinality exactly $W_{j'}$. Let $\rset$ be the set of all these level-$j'$ table entries. We claim that  there are at most two table entries $\Pi_{j'}(\hat I,\hat I')\in \rset$ with $\hat I\cap (a,b)\neq\emptyset$. Indeed, assume for contradiction that there are  $3$ such distinct entries in $\rset$, say $\Pi_{j'}[I_1,I'_1],\Pi_{j'}[I_2,I'_2]$, and $\Pi_{j'}[I_3,I'_3]$, with $I_1\cap (a,b),I_2\cap (a,b),I_3\cap (a,b)\neq\emptyset$. Then at least one of the intervals $I_1,I_2,I_3$ must be contained in $(a,b)$, violating our condition for good level-$j'$ pairs of intervals. We conclude that the number of the demand pairs in $\Pi_j[I,I']$ with a source vertex in $(a,b)$ is bounded by $2\cdot W_{j'}=4\cdot W_{j'-1}\le 4\cdot w$. The proof for the case where $i=2$ is analogous.

Assume now that $K$ is a type-$3$ constraint. As before, $\Pi_j[I,I']$ is the union of exactly $2^{j-j'}$ non-empty level-$j'$ table entries, each of which stores a set of demand pairs of cardinality exactly $W_{j'}$. Let $\rset$ be the set of these level-$j'$ table entries. We claim that there are at most two entries $\Pi_{j'}(\hat I,\hat I')\in \rset$, with $\hat I\cap L_a\neq \emptyset$ and $\hat I'\cap R_b\neq \emptyset$.
Indeed, assume otherwise, and let $\Pi_{j'}[I_1,I'_1],\Pi_{j'}[I_2,I'_2],\Pi_{j'}[I_3,I'_3]$ be three distinct entries in $\rset$, such that for each $1\leq \ell\leq 3$, $I_{\ell}\cap L_a\neq \emptyset$ and $I'_{\ell}\cap R_b\neq \emptyset$. Assume that $I_1,I_2,I_3$ appear on $\sigma$ in this order, and recall that from our construction they are disjoint. Then it is easy to see that $I_2\subseteq L_a$ and $I'_2\subseteq R_b$ must hold, contradicting our definition of good level-$j'$ pairs of terminals. Therefore, there are at most two entries $\Pi_{j'}(\hat I,\hat I')\in \rset$, with $\hat I\cap L_a\neq \emptyset$ and $\hat I'\cap R_b\neq \emptyset$. Demand pairs participating in solutions corresponding to other entries in $\rset$ cannot cross $K$, and so the number of the demand paris crossing $K$ is at most $2\cdot W_{j'}=4\cdot W_{j'-1}\le 4\cdot w$. The case where $K$ is a type-$4$ constraint is treated similarly.
\end{proof}

From the discussion so far, every entry of every dynamic programming table contains a non-crossing set of demand pairs, that violates every constraint of $\kset$ by at most factor $4$. Let $\tilde {\mset}$ be the largest-cardinality set of demand pairs stored in any entry any of the tables, and let $\opt$ be the value of the optimal solution to the \DPSP problem instance. The following theorem is central to our analysis.

\begin{theorem} \label{thm: opt analysis} If $\opt\geq 2$, then $|\tilde {\mset}|\geq \opt/2$. \end{theorem}\label{thm: relating to opt}

We can now apply Observation~\ref{obs: violated to satisfied} to compute a subset $\tilde{\mset'}\subseteq \tilde{\mset}$ of at least $|\tilde {\mset}|/4\geq \opt/8$ non-crossing demand pairs satisfying all constraints in $\kset$ (if $|\tmset|=0$, then $\opt\leq 1$ must hold, and finding an optimal solution is trivial). In order to complete the proof of Theorem~\ref{thm: approximate DPSP}, it now remains to prove Theorem~\ref{thm: opt analysis}.

\begin{proofof}{Theorem~\ref{thm: opt analysis}}
Denote $\kappa=\opt$, and let $\mset^*=\set{(s_1,t_1),\ldots,(s_\kappa,t_\kappa)}$ be the optimal solution to the \DPSP instance, where $s_1\prec\ldots\prec s_\kappa$ and $t_1\prec\ldots\prec t_\kappa$. Let $r'$ be the largest value for which $\kappa/W_{r'}\ge 1$ (this is well-defined since we have assumed that $\opt\geq 2$).
Let $\mset^{**}\subseteq\mset^*$ contain the first $W_{r'}$ demand pairs of $\mset^*$, so $|\mset^{**}|\geq |\mset^*|/2$. We will show that some entry of the level-$r'$ dynamic programming table stores a solution whose cardinality is at least $W_{r'}$.

For every level $1\le j\le r'$, we define a partition $\sset_j$ of $\mset^{**}$ into $2^{r'-j}$ subsets, each containing exactly $W_j$ consecutive demand pairs from $\mset^{**}$. For every set $S\in \sset_j$ of the partition, we define a pair $(I(S),I'(S))$ of intervals, with $I(S)\subseteq \sigma$ and $I'(S)\subseteq \sigma'$, as follows. Let $(s,t),(s',t')$ be the first and the last demand pairs of $S$, respectively (this is well-defined since the demand pairs are non-crossing). Then $I(S)=(s,s')$ and $I'(S)=(t,t')$. We denote by $\qset_j$ the resulting collection of $2^{r'-j}$ pairs of intervals. Note that for every pair $(I_1,I'_1),(I_2,I_2')\in \qset_j$, $I_1\cap I_2=\emptyset$ and $I_1'\cap I_2'=\emptyset$. We need the following claim.

\begin{claim} For every $1\leq j\leq r'$, every pair $(I,I')\in\qset_j$ of intervals is a good level-$j$ pair.
\end{claim}

\begin{proof} Fix some $1\le j\le r'$ and some pair $(I,I')\in \qset_j$ of intervals. From our definition of the pairs of intervals in $\qset_j$, there are exactly $W_j$ demand pairs $(s,t)$ in $\mset^*$ with $s\in I$ and $t\in I'$. Consider any level-$j$ constraint $K=(i,a,b,w)\in \sset_j$, and recall that $W_{j-1}\leq w<W_j$ must hold.

Assume first that $K$ is a type-$1$ constraint. Then $I$ cannot be contained in $(a,b)$, since then $\mset^*$ would have $W_j>w$ demand pairs whose sources lie in $I$, and hence in $(a,b)$. The analysis for type-$2$ constraints is similar.

Assume now that $K$ is a type-$3$ constraint, and assume for contradiction that $I\subseteq L_a$ and $I'\subseteq R_b$. Then $\mset^*$ has $W_j>w$ demand pairs $(s,t)$ with $s\in I$ and $t\in I'$, each of which crosses the constraint $K$, a contradiction. The case where $K$ is a type-$4$ constraint is proved similarly.
\end{proof}

From the above claim, for all $1\leq j\leq r'$ and $(I,I')\in \qset_j$, there is an entry $\Pi_j[I,I']$ in the level-$j$ dynamic programming table. It is now enough to show that each such entry contains a solution of cardinality $W_j$.

\begin{claim} For each $1\leq j\leq r'$, for every pair $(I(S),I'(S))\in \qset_j$ of intervals, entry $\Pi_j[I(S),I'(S)]$ contains a solution of value $W_j$.
\end{claim}

\begin{proof} 
The proof is by induction on $j$. The claim clearly holds for $j=1$, since there are two distinct non-crossing demand pairs $(s,t),(s',t')$ with $s,s'\in I(S)$, $t,t'\in I'(S)$ - the two demand pair lying in $S$. Assume now that the claim holds for levels $1,\ldots,j-1$, for some $1<j\leq r'$, and we would like to prove it for $j$.

Our definition of the partitions $\sset_{\ell}$ of $\mset^{**}$ ensures that there are exactly two distinct sets $S_1,S_2\in \sset_{j-1}$, with $S_1,S_2\subseteq S$ (and in fact $S_1\cup S_2=S$). From the induction hypothesis, the entries of $\Pi_{j-1}$ corresponding to pairs $(I(S_1),I'(S_1))$ and $(I(S_2),I'(S_2))$ each contain $W_{j-1}$ demand pairs, and so $\Pi_j[I(S),I'(S)]$ must contain $W_j$ demand pairs.
\end{proof}

Recall that $\sset_{r'}$ contains a single set of demand pairs - the set $\mset^{**}$. We conclude that the corresponding entry of the level-$r'$ dynamic programming table contains $W_{r'}\geq \opt/2$ demand pairs.
\end{proofof}

\label{-----------------------------------------------sec: routing on a disc-------------------------------------------}
\subsection{Approximation Algorithm for \NDPdisc}\label{subsec: routing on a disc}
In this section we prove Theorem~\ref{thm: routing on a disc and cyl main} for \NDPdisc.
The proof builds on the work of Robertson and Seymour~\cite{RS-disc}, who gave a precise characterization of the instances \NDPdisc, where all demand pairs can be routed simultaneously. 
Many of the definitions below are from~\cite{RS-disc}.
Let $\Sigma$ be a disc, whose boundary is denoted by $\Gamma$, and let $G$ be any graph drawn on $\Sigma$. Suppose we are given a set $\mset=\set{(s_1,t_1),\ldots,(s_k,t_k)}$ of pairs of vertices of $G$, called demand pairs, and let $\tset$ be the set of all vertices participating in the demand pairs in $\mset$, that we refer to as terminals. We identify the graph $G$ with its drawing. A \emph{region} of $G$ is a connected component of $\Sigma\setminus G$. We say that the drawing of $G$ is \emph{semi-proper} if no inner point on an image of an edge of $G$ lies on $\Gamma$, and we say that it is \emph{proper} with respect to $\mset$, if additionally $V(G)\cap \Gamma=\tset$. 

Suppose we are given a planar graph $G$, together with a set $\mset$ of demand pairs, such that $G$ is drawn properly (with respect to $\mset)$ on $\Sigma$.
Let $W$ be a set of points on $\Gamma$, constructed as follows. First, we add to $W$ all points corresponding to the vertices of $\tset$. Next, for every segment $\beta$ of $\Gamma\setminus \tset$, we add one arbitrary point $p\in \beta$ to $W$. 
A vertex $v$ of $G$ is \emph{peripheral} if $v\in \Gamma$, and a region of $G$ is peripheral, if it contains a segment of $\Gamma\setminus \tset$. Given two points $x,y\in W$, we denote by $\Delta_{\mset}(x,y)$ the total number of the demand pairs $(s,t)\in \mset$, where either $\set{s,t}\cap \set{x,y}\neq \emptyset$, or $s$ and $t$ belong to different segments of $\Gamma\setminus\set{x,y}$.
We need the following definition.

\begin{definition} Given $G$, $\mset$ and $W$ as above, for any $x,y\in W$, an $(x,y)$-chain is a sequence $A_1,A_2,\ldots,A_r$, such that:

\begin{itemize}
\item for all $1\leq i<r$, one of $A_i,A_{i+1}$ is a vertex of $G$, the other is a region, and they are incident;
\item if $A_1$ is a vertex then $A_1=x$; if $A_1$ is a region then $x\in A_1$;
\item similarly, if $A_r$ is a vertex then $A_r=y$, and if $A_r$ is a region then $y\in A_r$; and

\item for all $1\leq i\leq r$, $A_i$ is peripheral if and only if $i=1$ or $i=r$.
\end{itemize}

The length of an $(x,y)$-chain is the number of its terms that are vertices. The redundancy of an $(x,y)$-chain is its length minus $\Delta_{\mset}(x,y)$.
\end{definition}

The following theorem, proved by Roberston and Seymour~\cite{RS-disc} characterizes routable sets of demand pairs for the \NDPdisc problem.

\begin{theorem}[Theorem 3.6 in~\cite{RS-disc}]\label{thm: routing on a disc - RS}
Let $G$ be a planar graph properly drawn on a disc $\Sigma$, with respect to a set $\mset$ of demand pairs, and let $W$ be defined as above. Then there is a set of node-disjoint paths, routing all demand pairs in $\mset$ if and only if: (i) $\mset$ is a non-crossing set of demand pairs; and (ii) for all $x,y\in W$, every $(x,y)$-chain has a non-negative redundancy.
Moreover, there is an efficient algorithm to determine whether these conditions hold, and if so, to find a routing of the demand pairs in $\mset$.
\end{theorem}

Notice that if $G$ is a graph drawn on a disc $\Sigma$, such that all terminals appear on the boundary $\Gamma$ of the disc, then, by slightly altering $\Gamma$, we can ensure that the drawing of $G$ is proper with respect to $\mset$. Therefore, we assume from now on that we are given a proper drawing $\phi$ of $G$ on $\Sigma$ with respect to $\mset$.
We prove Theorem~\ref{thm: routing on a disc and cyl main} for \NDPdisc in three steps. In the first step, we prove a stronger version of the theorem for the case where $G$ is connected and the set $\mset$ of terminals is $1$-split. In the second step, we prove the theorem for the case where $G$ is $2$-connected, and we make no assumptions on the demand pairs. In the third step we prove the theorem without any additional assumptions.
\subsubsection{Special Case: $G$ is Connected and $\mset$ is $1$-Split}\label{subsec: NDP on discs 2-connected 1-split}

In this section, we prove the following theorem.

\begin{theorem}\label{thm: routing on a disc 1-split 2-connected}
There is an efficient algorithm, that, given a connected planar graph $G$ with a set $\mset$ of demand pairs, and a proper drawing of $G$ on the disc $\Sigma$ with respect to $\mset$, where $\mset$ is $1$-split with respect to the disc boundary, computes a routing of $\opt(G,\mset)/8$ demand pairs of $\mset$ via node-disjoint paths in $G$.
\end{theorem}

Recall that if $\mset$ is $1$-split with respect to the boundary $\Gamma$ of the disc $\Sigma$, then we can partition $\Gamma$ into two disjoint segments, $\sigma$ and $\sigma'$, such that for every demand pair $(s,t)\in\mset$, one of the vertices $s,t$ lies on $\sigma$, and the other on $\sigma'$. We assume without loss of generality that all source vertices of the demand pairs in $\mset$ lie on $\sigma$, and all destination vertices lie on $\sigma'$. Let $\tset$ be the set of all vertices participating in the demand pairs in $\mset$, that we refer to as terminals. We denote by $S$ and $T$ the sets of the source and the destination vertices of the demand pairs in $\mset$, respectively. As before, we construct a set $W$ of points: start with $W=\tset$; then for every segment of $\Gamma\setminus\tset$, add one arbitrary point of that segment to $W$.

The idea is to reduce this problem to the \DPSP problem.  We let $\tilde{\sigma}$ be a directed path, whose vertices correspond to the points of $W\cap \sigma$, ordered in their counter-clock-wise order on $\Gamma$. We let $\tilde{\sigma}'$ be a directed path, whose vertices correspond to the points of $W\cap \sigma'$, ordered in their clock-wise order on $\Gamma$. Therefore, all vertices of $S$ lie on $\tilde{\sigma}$, and all vertices of $T$ lie on $\tilde{\sigma}'$. Moreover, $\mset'$ is a non-crossing set of demand pairs with respect to $\tilde{\sigma}$ and $\tilde{\sigma}'$ iff it is non-crossing with respect to $\Gamma$.

Suppose we are given any pair $(x,y)$ of vertices of $W$. Let $\gamma(x,y)$ denote the shortest $G$-normal curve connecting $x$ to $y$ in the drawing of $G$ on $\Sigma$, so that $\gamma(x,y)$ is contained in the disc $\Sigma$. Curve $\gamma(x,y)$ can be found efficiently by considering the graph $G'$ dual to $G$, deleting the vertex corresponding to the outer face of $G$ from it, and computing shortest paths between appropriately chosen vertices of $G'$ (that correspond to the faces of $G$ incident on $x$ and $y$). Let $\ell(x,y)$ be the length of $\gamma(x,y)$. Notice that, since $G$ is connected, if $\Delta_{\mset}(x,y)\geq 1$, then $\ell(x,y)\geq 1$.
 We now construct a collection $\kset$ of constraints of the \DPSP probem, as follows.

For every pair $x,y$ of vertices of $W\cap \tilde{\sigma}$, with $\Delta_{\mset}(x,y)\geq 1$,  we add a type-$1$ constraint $(1,x,y,\ell(x,y))$ to $\kset$. Similarly, for every pair $x,y$ of vertices of $W\cap \tilde{\sigma}'$, with $\Delta_{\mset}(x,y)\geq 1$, we add a type-$2$ constraint $(2,x,y,\ell(x,y))$ to $\kset$.
For every pair $x\in W\cap \tilde {\sigma}$, $y\in W\cap \tilde{\sigma'}$ of vertices,  with $\Delta_{\mset}(x,y)\geq 1$, we add a type-3 constraint $(3,x,y,\ell(x,y))$ and a type-4 constraint $(4,x,y,\ell(x,y))$ to $\kset$. This finishes the definition of the \DPSP problem instance. The following observation is immediate.

\begin{observation}\label{obs: routing to DPSP}
Let $\mset'\subseteq \mset$ be any set of demand pairs that can be routed via disjoint paths in $G$. Then $\mset'$ is a valid solution to the \DPSP problem instance.
\end{observation}


We apply the $8$-approximation algorithm for the \DPSP problem, to obtain a set $\mset'$ of non-crossing demand pairs (with respect to $\tsigma$ and $\tsigma'$) satisfying all constraints in $\kset$, with $|\mset'|\geq |\opt(G,\mset)|/8$. As observed above, the pairs in $\mset'$ are non-crossing with respect to $\Gamma$.
It now only remains to show that all demand pairs in $\mset'$ can be routed in $G$. The algorithm from Theorem~\ref{thm: routing on a disc - RS} can then be used to find the routing. The following theorem will finish the proof of Theorem~\ref{thm: routing on a disc 1-split 2-connected}.

\begin{theorem}\label{thm: DPSP to routing}
Let $\mset'\subseteq\mset$ be a set of non-crossing demand pairs (with respect to $\tsigma$ and $\tsigma'$), that satisfy all constraints in $\kset$. Then all demand pairs in $\mset'$ can be routed in $G$.
\end{theorem}

\begin{proof}
Let $\tset'=\tset(\mset')$, and let $\phi$ be the current proper drawing of $G$ with respect to $\mset$. Notice that $\phi$ is not necessarily a proper drawing of $G$ with respect to $\mset'$, since the vertices of $\tset\setminus\tset'$ may lie on $\Gamma$. We can obtain a proper drawing $\phi'$ of $G$ with respect to $\mset'$ by moving all such terminals inside the disc, so they no longer lie on $\Gamma$. It is immediate to verify that the demand pairs in $\mset'$ remain non-crossing with respect to $\Gamma$ in the new drawing.

As before, we construct a set $W'$ of points of $\Gamma$, by first adding all points corresponding to the terminals of $\tset'$ to $W'$. Additionally, for every segment $\beta$ of $\Gamma\setminus\tset'$, we add an arbitrary point of $\beta$ to $W'$. It now remains to show that for every pair $x,y$ of points in $W'$, for every $(x,y)$-chain in $\phi'$, the redundancy of the chain (with respect to $\mset'$) is non-negative.

Assume otherwise. Let $x,y\in W'$ be any pair of points, and let $\aset=(A_1,\ldots,A_r)$ be an $x$-$y$ chain in the drawing $\phi'$, such that the redundancy of $\aset$ is negative. We modify the chain $\aset$, by replacing $A_1$ and $A_r$ with elements $A_1'$ and $A_r'$, as follows.

If $A_1$ is a vertex of $\tset'$, then we let $A_1'=A_1$, and we let $x'=A_1'$. Otherwise, $A_1$ is a region of $G$ in the drawing $\phi'$. Let $v=A_2$, so $v\in V(G)$. Then $v$ must lie on a boundary of some peripheral region $R$ of $G$ in the drawing $\phi$. Let $A_1'=R$, let $\beta(R)$ be the segment of $\Gamma$ that serves as part of the boundary of the region $R$, and let $x'\in W$ be the point lying on the interior of $\beta(R)$. Similarly, if $A_r$ is a vertex of $\tset'$, then we let $A_r'=A_r$, and we let $y'=A_r'$.  Otherwise, $A_r$ is a region of $G$ in the drawing $\phi'$. Let $v'=A_{r-1}$, so $v'\in V(G)$. Then $v'$ must lie on a boundary of some peripheral region $R'$ of $G$ of the drawing $\phi$. Let $A'_r=R'$, let $\beta(R')$ be the segment of $\Gamma$ that serves as part of the boundary of the region $R'$, and let $y'\in W$ be the point lying on the interior of $\beta(R)$. We then obtain a sequence $\aset'=(A_1',A_2,\ldots,A_{r-1},A_r')$ (it may not be a valid chain for the drawing $\phi$, since some of the elements $A_i$, for $1<i<r'$, may be peripheral with respect to $\phi$). We can then construct a $G$-normal curve $\gamma'$ in the original drawing $\phi$ of $G$, connecting $x'$ to $y'$, such that $\gamma'\cap V(G)$ only contains the vertices that participate in $\aset'$, and $\gamma'$ is contained in $\Sigma$. Let $\ell$ denote the length of the chain $\aset$, and let $\Delta=|\Delta_{\mset'}(x,y)|$. We can assume that $\Delta\geq 1$, since otherwise it is immediate to verify that $\Delta\leq \ell$. Then $\ell(x',y')\leq \ell$, from the definition of $\ell(x',y')$. Notice that one of the segments of $\Gamma\setminus\set{x,x'}$ contains no terminals of $\tset'$, and the same holds for one of the segments of $\Gamma\setminus\set{y,y'}$. We now consider three cases.
 
Assume first that both $x'$ and $y'$ lie on $\tilde{\sigma}$. Then pair $(s,t)\in \Delta_{\mset'}(x,y)$ iff $s$ belongs to the sub-path $(x',y')$ of $\tilde{\sigma}$. Since we assumed that $\Delta\geq 1$, we get that $\Delta_{\mset}(x',y')\geq 1$,  so constraint $(1,x',y',\ell(x',y'))$ belongs to $\kset$, and we get that $\Delta\leq \ell(x',y')\leq \ell$, a contradiction.

The case where $x',y'\in \tilde{\sigma}'$ is dealt with similarly.

Assume now that $x'\in \tilde{\sigma}$ and $y'\in \tilde{\sigma}'$. As before, since we assumed that $\Delta\geq 1$, we get that $\Delta_{\mset}(x',y')\geq 1$. Consider the two corresponding type-3 and type-4 constraints $K=(3,x',y',\ell(x',y')), K''=(4,x',y',\ell(x',y'))\in \kset$. Let $L_{x'},R_{x'}$ be the segments of $\tilde{\sigma}$ from its first endpoint to $x'$, and from $x'$ to its last endpoint, respectively, where both segments include $x'$. Define segments $L_{y'},R_{y'}$ of $\tilde{\sigma}'$ similarly. Recall that a demand pair $(s,t)$ crosses $K$ iff $s\in L_{x'}$ and $t\in R_{y'}$, and it crosses $K'$ iff $s\in R_{x'}$ and $t\in L_{y'}$. Let $\mset_1\subseteq\mset'$ be the set of the demand pairs crossing $\kset$, and let $\mset_2\subseteq \mset'$ be the set of the demand pairs crossing $\kset'$. Since $\mset'$ is a non-crossing set of demand pairs, it is easy to verify that either $\mset_1\setminus\set{(x',y')}=\emptyset$, or $\mset_2\setminus\set{(x',y')}=\emptyset$. We assume without loss of generality that it is the latter. Notice that if $(x',y')\in \mset'$, then it belongs to both $\mset_1$ and $\mset_2$. Constraint $(3,x',y',\ell(x',y'))$ then ensures that $|\mset_1\cup \mset_2|\leq |\mset_1|\leq \ell(x',y')\leq \ell$. It is easy to verify that $\Delta_{\mset'}(x,y)\subseteq\mset_1\cup \mset_2$, and so $\Delta\leq \ell$, a contradiction.
\end{proof}

\subsubsection{Special Case: $G$ is $2$-connected}\label{subsec: NDP on disc 2-connected}
The goal of this section is to prove the following theorem.

\begin{theorem}\label{thm: 1-split to general 2-connected} There is an efficient algorithm, that, given any $2$-connected planar graph $G$, and a set $\mset$ of $k$ demand pairs for $G$, such that $G$ is properly drawn on a disc with respect to $\mset$, returns an $O(\log k)$-approximate solution to problem $(G,\mset)$. 
\end{theorem}

\begin{proof}
Let $|\mset|=k$, and let $z=4\ceil{\log k}$. We use Lemma~\ref{lem: getting r-split demand pairs on a disc} to compute a partition $\mset^1,\ldots,\mset^z$ of $\mset$ into $z$ disjoint subsets, so that for each $1\leq i\leq z$, $\mset^i$ is $r_i$-split for some integer $r_i$. We prove the following lemma.

\begin{lemma}\label{lem: solving split instance}
There is an efficient algorithm, that, given an index $1\leq i\leq z$, computes a solution to instance $(G,\mset^i)$ routing $\Omega(\opt(G,\mset^i))$ demand pairs.
\end{lemma}

Theorem~\ref{thm: 1-split to general 2-connected}  then easily follows from Lemma~\ref{lem: solving split instance}, since there is some index $1\leq i\leq z$, for which $\opt(G,\mset^i)=\Omega(\opt(G,\mset)/\log k)$. In the rest of this section, we focus on proving Lemma~\ref{lem: solving split instance}.

For simplicity, we denote $\mset^i$ by $\mset$ and $r_i$ by $r$. We assume that we are given a partition $\set{\mset_1,\ldots,\mset_r}$ of $\mset$, and a collection $\sigma_1,\ldots,\sigma_{2r}$ of disjoint segments of $\Gamma$, such that for all $1\leq j\leq r$, for every demand pair $(s,t)\in \mset_j$, $s\in \sigma_{2j-1}$ and $t\in \sigma_{2j}$.

Clearly, for each $1\leq j\leq r$, set $\mset_j$ is $1$-split. Therefore, we can use the algorithm from Section~\ref{subsec: NDP on discs 2-connected 1-split} in order to compute a set $\pset_j$ of disjoint paths, routing a subset $\mset'_j\subseteq \mset_j$ of demand pairs, with $|\mset'_j|\geq \Omega(\opt(G,\mset_j))$. Since $\sum_{j=1}^r\opt(G,\mset_j)\geq \opt(G,\mset)$, we get that $\sum_{j=1}^r|\mset_j'|\geq \Omega(\opt(G,\mset))$.

For every $1\leq j\leq r$, we compute a subset $\mset''_j\subseteq \mset'_j$ of demand pairs, as follows. If $|\mset'_j|\leq 3$, then we let $\mset''_j$ contain any demand pair from $\mset'_j$. Otherwise, we assume that $\mset'_j=\left((s^j_1,t^j_1),\ldots,(s^j_{q_j},t^j_{q_j})\right )$, where the vertices $s^j_1,s^j_2,\ldots,s^j_{q_j}$ appear in this counter-clock-wise order on $\sigma_{2j-1}$. We then add to $\mset''_j$ all demand pairs $(s^j_a,t^j_a)$, where $a=1$ modulo $3$. Notice that $\sum_{j=1}^r|\mset_j''|\geq \sum_{j=1}^r\Omega(|\mset_j'|)\geq \Omega(\opt(G,\mset))$. It is now enough to prove that all demand pairs in $\mset''=\bigcup_{j=1}^r\mset''_j$ can be routed in $G$. It is easy to verify that the demand pairs in $\mset''$ are non-crossing, since for each $1\leq j\leq r$, the demand pairs in $\mset'_j$ are non-crossing, and pairs that belong to different sets $\mset_j$ cannot cross. Let $\tset''=\tset(\mset'')$, and let $\phi''$ is a proper drawing of $G$ in $\Sigma$ with respect to $\mset''$, obtained from the original drawing $\phi$ by moving all terminals of $\tset\setminus\tset''$ to the interior of the disc. 
We define a set $W$ of points on $\Gamma$ as before: it contains all terminals of $\tset''$, and for every segment of $\Gamma\setminus\tset''$, $W$ contains an arbitrary point in the interior of the segment.

Let $(x,y)\in W$ be any pair of points, and let $\aset=(A_1,A_2,\ldots,A_p)$ be any $(x,y)$-chain. Let $\ell$ denote the length of $\aset$. It is now enough to prove that $|\Delta_{\mset''}(x,y)|\leq \ell$.
Assume first that for some $1\leq j\leq z$, $x,y\in \sigma_{2j-1}\cup \sigma_{2j}$. Then $\Delta_{\mset''}(x,y)$ only contains demand pairs from $\mset''_j$. Since all demand pairs in $\mset''_j$ can be routed in $G$ via node-disjoint paths, $|\Delta_{\mset''}(x,y)|\leq \ell$.

Assume now that $x\in \sigma_{2j-1}\cup \sigma_{2j}$ and $y\in \sigma_{2j'-1}\cup \sigma_{2j'}$ for some $j\neq j'$. Then $\Delta_{\mset''}(x,y)$ only contains demand pairs from $\mset''_j\cup\mset''_{j'}$. Let $\nset_1=\Delta_{\mset''}(x,y)\cap \mset''_j$ and $\nset_2=\Delta_{\mset''}(x,y)\cap \mset''_{j'}$.  Let $\nset'_1=\Delta_{\mset'}(x,y)\cap \mset'_j$, and let $\nset'_2=\Delta_{\mset'}(x,y)\cap \mset'_{j'}$. Since both $\mset'_j$ and $\mset'_{j'}$ can be routed in $G$ via node-disjoint paths, $|\nset'_1|,|\nset'_2|\leq \ell$. From our selection of the sets $\mset''_j$, $\mset''_{j'}$ of demand pairs, $|\nset_1|\leq \max\set{|\nset'_1|/2,1}$, and $|\nset_2|\leq \max\set{|\nset'_2|/2,1}$. Since graph $G$ is $2$-vertex connected, we can assume that either $|\Delta_{\mset''}(x,y)|\leq 1$, or $\ell\geq 2$. In the former case, $\ell\geq |\Delta_{\mset''}(x,y)|$ since the graph is connected.  In the latter case, we are now guaranteed that $|\Delta_{\mset''}(x,y)|=|\nset'_1|+|\nset'_2|\leq \ell$.
\end{proof}
\subsubsection{The General Case}\label{NDP on grids general}
In this section, we complete the proof of Theorem~\ref{thm: routing on a disc and cyl main} for \NDPdisc. We assume that the input graph $G$ is connected, since otherwise we can solve the problem separately on each connected component of $G$.

%

We use a block-decomposition of $G$. Recall that a block of $G$ is a maximal $2$-node-connected component of $G$. A block-decomposition of $G$ is a tree $\tau$, whose vertex set consists of two subsets, $V(\tau)=V_1\cup V_2$, where $V_1$ is the set of all cut-vertices of $G$, and $V_2$ contains a vertex $v_B$ for every block $B$ of $G$. There is an edge between a vertex $u\in V_1$ and $v_B\in V_2$ iff $u\in v_B$. We choose an arbitrary vertex $r\in V_1$ as the root of $\tau$. We assume that the value of the optimal solution is at least $10$, as otherwise we can route a single demand pair. We then discard from $\mset$ all demand pairs in which $r$ participates -- this changes the value of the optimal solution by at most $1$.

Over the course of the algorithm, we will gradually change the tree $\tau$, by deleting vertices from it. Given any vertex $u$ in the current tree $\tau$, we let $\tau_u$ denote the subtree of $\tau$ rooted at $u$, and we let $G_u$ be the subgraph of $G$ induced by the union of all blocks $B$ with $v_B\in V(\tau_u)$.  For every block $B$, we denote by $p(B)$ the unique vertex of $B$ that serves as the parent of the vertex $v_B$ in $\tau$. As the tree $\tau$ changes, so do the trees $\tau_u$ and the graphs $G_u$.

If any block $B$ of $G$ contains all the terminals, then we can apply the algorithm from Theorem~\ref{thm: 1-split to general 2-connected} to instance $(B,\mset)$ of \NDPdisc to obtain an $O(\log k)$-approximate solution. Otherwise, for every block $B$, if we denote by $\Gamma(B)$ the set of all cut vertices of $G$ that belong to $B$, and by $\tset(B)$ the set of all terminals lying in $B$, then we can draw $B$ inside a disc, so that the vertices of $\Gamma(B)\cup \tset(B)$ lie on its boundary.

We gradually construct a solution $\pset$ to the \NDPdisc instance $(G,\mset)$, starting from $\pset=\emptyset$. Throughout the algorithm, we ensure that all paths in $\pset$ are disjoint from the vertices of $G_r$, where $r$ is the root vertex of $\tau$, and $G_r$ is computed with respect to the current tree $\tau$. Clearly, the invariant holds at the beginning of the algorithm.

We maintain a collection $\bset$ of blocks, that is initialized to $\bset=\emptyset$. For every block $B\in \bset$, we will define two subsets, $\mset_B,\mset'_B\subseteq \mset$ of demand pairs, so that eventually, $\set{\mset_B,\mset'_B\mid B\in \bset}$ is a partition of $\mset$. 

 In every iteration, we consider all vertices $v_B\in V_2$ that belong to the current tree $\tau$, such that $G_{v_B}\setminus p(B)$ contains at least one demand pair $(s,t)\in \mset$. Among all such vertices, we choose one that is furthest from the root of the tree, breaking ties arbitrarily. Let $v_B$ be the selected vertex.
We add $B$ to $\bset$, and we define two new subsets of demand pairs $\mset_B,\mset'_B$, as follows. Set $\mset_B$ contains all demand pairs $(s,t)$ with $s,t\in G_{v_B}\setminus p(B)$. Set $\mset_B'$ contains all demand pairs $(s,t)$ with exactly one of $s,t$ lying in $G_{v_B}\setminus p(B)$, while the other terminal must belong to $G_r$. Notice that any path connecting any demand pair in $\mset'_B$ must use the vertex $p(B)$.

Next, we construct a new instance of the \NDPdisc problem, on the graph $B$. The corresponding set of demand pairs, that we denote by $\nset_B$, is defined as follows. Consider any demand pair $(s,t)\in \mset_B$. We define a new demand pair $(s',t')$ representing $(s,t)$, with $s',t'\in V(B)$, and define two paths: $Q(s)$ connecting $s$ to $s'$, and $Q(t)$ connecting $t$ to $t'$. If $s\in V(B)$, then we let $s'=s$, and $Q(s)=\set{s}$. Otherwise, we let $s'\in V_1$ be the unique child of the vertex $v_B$ in $\tau$, such that $s\in G_{s'}$. Notice that $s'$ must be a vertex of $B$. We then let $Q(s)$ be any simple path connecting $s$ to $s'$ in graph $G_{s'}$. We define the vertex $t'\in B$, and a path $Q(t)$ connecting $t$ to $t'$ similarly. Notice that it is possible that $s'=t'$. Let $\nset_B=\set{(s',t')\mid (s,t)\in \mset_B}$ be the resulting set of demand pairs. All vertices participating in the demand pairs in $\nset_B$ belong to $B$.  Consider the \NDP instance $(B,\nset)$. It is immediate to verify that we can draw $B$ in a disc, so that all vertices participating in the demand pairs in $\mset(B)$ lie on the boundary of the disc, and clearly $B$ is $2$-connected. 
We need the following immediate observation.

\begin{observation}\label{obs: old to new pairs}
$\opt(B,\nset_B)\geq \opt(G,\mset_B)$.
\end{observation}

\begin{proof}
Let $\tpset$ be the optimal solution to instance $\opt(G,\mset_B)$. We can assume that all paths in $\tpset$ are simple. Let $P\in \tpset$ be any such path, and assume that $P$ connects some demand pair $(s,t)\in \mset_B$. Then it is easy to see that $s',t'\in P$, and moreover, the segment of $P$ between $s'$ and $t'$ is contained in $B$. By appropriately truncating every path in $\tpset$, we can obtain a solution to instance $(B,\nset_B)$ of the \NDP problem of the same value.
\end{proof}

Since $B$ is $2$-vertex connected, and can be drawn in a disc with all vertices participating in the demand pairs in $\nset_B$ lying on the disc boundary, we can apply the algorithm from Theorem~\ref{thm: 1-split to general 2-connected} to instance $(B,\nset_B)$, to compute a set $\pset(B)$ of node-disjoint paths, routing a subset of at least $\Omega(\opt(B,\nset_B)/\log k)$ demand pairs of $\mset_B$ in $B$. We can assume that $|\pset(B)|\geq 1$, since otherwise we can route any demand pair in $\nset_B$.
We would like to ensure that all paths in $\pset(B)$ avoid the vertex $p(B)$. If $|\pset(B)|=1$, then, since $B$ is $2$-vertex connected, we can re-route the unique path in $\pset(B)$ inside $B$, so that its endpoints remain the same, but it avoids the vertex $p(B)$ (since $G_{v_B}\setminus p(B)$ contains at least one demand pair, we can ensure that the endpoints of the path are distinct from $p(B)$). Otherwise, we discard from $\pset(B)$ the path that uses vertex $p(B)$, if such exists. By concatenating the paths in $\pset(B)$ with the paths in $\set{Q(s),Q(t)\mid (s,t)\in \mset(B)}$, we obtain a collection $\pset'(B)$ of at least $\Omega(\opt(B,\nset)/\log k)$ node-disjoint paths, connecting demand pairs in $\mset(B)$. We add the paths in $\pset'(B)$ to $\pset$, and delete from $\tau$ all vertices of $\tau_{v_B}$. Since we have ensured that the paths in $\pset'(B)$ are disjoint from $p(B)$, the invariant that the paths in $\pset$ are disjoint from the new graph $G_r$ continues to hold. The algorithm terminates when no demand pair $(s,t)\in \mset$ is contained in $G_r$. We claim that the resulting collection $\set{\mset(B),\mset'(B)\mid B\in \bset}$ of sets of demand pairs partitions $\mset$. Indeed, consider any demand pair $(s,t)\in \mset$, and consider the last iteration $i$ when both $s,t\in G_r$. Let $v_B$ be the vertex that was processed in the following iteration. If both $s,t\in G_{v_B}\setminus p(B)$, then $(s,t)$ was added to $\mset_B$. Otherwise, exactly one of $s,t$ belongs to $G_{v_B}\setminus p(B)$, while the other belongs to $G_r$, so $(s,t)$ was added to $\mset'_B$. We now obtain a set $\pset$ of disjoint paths, routing a subset of vertices in $\mset$. We  show that $|\pset|\geq \Omega(\opt(G,\mset)/\log k)$. Let $\pset^*$ be the optimal solution to instance $\opt(G,\mset)$, and let $\mset^*$ be the set of the demand pairs routed by $\pset^*$. For every block $B\in \bset$, let $\tmset(B)=\mset^*\cap \mset_B$, and let $\tmset'_B=\mset^*\cap \mset'_B$. 

From Observation~\ref{obs: old to new pairs}, set $\pset'(B)$ of paths routes at least $\Omega(\opt(B,\nset_B)/\log k)\geq \Omega(\opt(G,\mset_B)/\log k)\geq \Omega(|\tmset_B|/\log k)$ demand pairs. Therefore, $|\pset|=\sum_{B\in \bset}|\pset'(B)|\geq \sum_{B\in \bset}\Omega(|\tmset_B|/\log k)$. On the other hand, as observed above, for every block $B\in \bset$, $|\tmset'_B|\leq 1$ (since all paths routing the pairs in $\tmset'_B$ contain vertex $p(B)$), while $|\pset'(B)|\geq 1$. Therefore, $|\pset|\geq \sum_{B\in \bset}|\tmset'_B|$. Overall, $|\pset|\geq \sum_{B\in \bset}\Omega((|\tmset_B|+|\tmset_B'|)/\log k)=\Omega(|\mset^*|/\log k)$.
\label{-------------------------------------------subsec: on a cylinder-------------------------------------------------------}
\subsection{Approximation Algorithm for \NDPcyl}
In this section we prove Theorem~\ref{thm: routing on a disc and cyl main} for \NDPcyl.
Recall that in the \NDPcyl problem, we are given a cylinder $\Sigma$, obtained from the sphere, by removing two open discs from it. We denote the boundaries of the two discs by $\Gamma_1$ and $\Gamma_2$, respectively. We assume that we are given a graph $G$, drawn on $\Sigma$, and a set $\mset$ of demand pairs. We denote by $S$ and $T$ the sets of all source and all destination vertices participating in the demand pairs in $\mset$. We say that a drawing of $G$ is proper with respect to $S$ and $T$ iff the vertices of $S$ lie on $\Gamma_1$, the vertices of $T$ lie on $\Gamma_2$, and no other edges or vertices of $G$ intersect $\Gamma_1$ or $\Gamma_2$. We can assume without loss of generality that we are given a proper drawing of the input graph $G$ on $\Sigma$ with respect to $S$ and $T$. We also assume that the graph $G$ is connected, as otherwise we can solve the problem for each connected component of $G$ separately. We assume that we know the value $\opt$ of the optimal solution to instance $(G,\mset)$, and a demand pair $(s^*,t^*)\in \mset$ that is routed by some optimal solution to instance $(G,\mset)$. We can make these assumptions by solving the problem for every possible value of $\opt$ between $1$ and $|\mset|$, and every choice of $(s^*,t^*)\in \mset$. It is enough to show that the algorithm returns the desired solution when the value $\opt$ and the pair $(s^*,t^*)$ are guessed correctly. We can also assume that $\opt>10$, since otherwise routing a single demand pair gives a desired solution. 

We define a set $W_1$ of points on $\Gamma_1$, as follows. First, we add to $\Gamma_1$ all points corresponding to the vertices of $S$. Next, for every segment of $\Gamma_1\setminus W_1$, we add an arbitrary point on the segment to $W_1$. We define a set $W_2$ of points on $\Gamma_2$ similarly, using $T$ instead of $S$. Our first step is to compute the shortest $G$-normal curve $\gamma^*\subseteq \Sigma$, connecting a point of $W_1$ to a point of $W_2$. We consider two cases.

Assume first that the length of $\gamma^*$ is less than $\opt/2$. Then we can cut the cylinder $\Sigma$ along the curve $\gamma^*$, deleting from $G$ all vertices lying on $\gamma^*$, to obtain a disc $\Sigma'$, and a drawing of $G$ on $\Sigma'$, where all terminals of $S\cup T$ lie on the boundary of the disc. It is easy to see that the value of the optimal solution of the resulting problem instance is at least $\opt/2$. We can now apply the $O(\log k)$-approximation algorithm for \NDPdisc from Section~\ref{subsec: routing on a disc} to obtain an $O(\log k)$-approximate solution to the new \NDPdisc instance, which in turn gives an $O(\log k)$-approximate solution to the original instance of \NDPcyl.

We assume from now on that the length of $\gamma^*$ is at least $\opt/2$. 
We reduce the problem to \DPSP. Let $\sigma$ be cycle, whose vertices are $W_1$, and they are connected in the order of their appearance on $\Gamma_1$. We delete the edge of $\sigma$ incident on $s^*$, that appears after $s^*$ in the counter-clock-wise traversal of $\Gamma_1$, and direct all edges of the resulting path away from $s^*$. We define a path $\sigma'$ similarly - start with the cycle, whose vertices are $W_2$, and they are connected in the order of their appearance on $\Gamma_2$. Delete the edge incident on $t^*$, that appears after $t^*$ in the counter-clock-wise traversal of $\Gamma_2$, and direct all edges of the resulting path away from $t^*$. Our next step is to define a set $\kset$ of constraints for the \DPSP problem instance. The instance we construct will only contain type-1 and type-2 constraints.

Let $a^*$ be the last vertex of $\sigma$. The first constraint that we add to $\kset$ is $(1,s^*,a^*,\opt/2)$. This constraint ensures that overall we will not attempt to route more than $\opt/2$ demand pairs.

Consider now any pair $x,y\in W_1$ of points. 
Let $\beta_1(x,y)$ and $\beta_2(x,y)$ be the two segments of $\Gamma_1$ whose endpoints are $x$ and $y$. For $i\in \set{1,2}$, we let $\ell_i(x,y)$ be the smallest number of vertices that need to be removed from $G$, in order to disconnect all vertices of $\beta_i(x,y)\cap S$ from the vertices of $T$ - this value can be computed efficiently using standard minimum cut algorithms. 
%
%
Let $w_i=\ell_i(x,y)$. We assume w.l.o.g. that $x$ lies before $y$ on $\sigma$. If $s^*$ is not an inner vertex on $\beta_i(x,y)$, then we add the constraint $(1,x,y,w_i)$ to $\kset$. Otherwise, we add two constraints: $(1,s^*,x,w_i)$ and $(1,y,a^*,w_i)$ to $\kset$. For every pair of points $(x,y)\in W_1$, we therefore add at most three type-1 constraints to $\kset$.

We process all pairs of points $(x,y)\in W_2$, and add corresponding type-2 constraints to $\kset$ similarly, except that we use $t^*$ instead of $s^*$. This finishes the description of the \DPSP instance.
We start with the following easy observation.

\begin{observation}\label{obs: from cylinder optimal to DPSP}
Let $\pset^*$ be the optimal solution to the \NDP instance $(G,\mset)$, such that $(s^*,t^*)$ is routed by $\pset^*$, and let $\mset^*$ be the set of the demand pairs routed by $\pset^*$. Let $\mset^{**}\subseteq \mset^*$ be any subset of $\floor{|\mset^*|/2}$ demand pairs. Then $\mset^{**}$ is a feasible solution to the \DPSP instance $(\sigma,\sigma',\mset,\kset)$. \end{observation}

\begin{proof}
Since we assume that the demand pair $(s^*,t^*)$ is routed by $\pset^*$, and since the demand pairs in $\mset^*$ must be non-crossing with respect to $\Gamma_1$ and $\Gamma_2$, due to the way in which we have defined the paths $\sigma$ and $\sigma'$, set $\mset^{**}$ must be non-crossing with respect to $\sigma$ and $\sigma'$.

Recall that we have added the constraint $(1,s^*,a^*,|\opt|/2)$ to $\kset$, where $s^*$ and $a^*$ are the endpoints of $\sigma$. Since $|\mset^{**}|\leq |\mset^*|/2=|\opt|/2$, set $\mset^{**}$ satisfies this constraint.

Consider now any pair $(x,y)$ of points in $W_1$, and fix some $i\in \set{1,2}$. Since set $\mset^*$ of demand pairs is routable in $G$, the number of the source vertices of the demand pairs in $\mset^*$ that lie on $\beta_i(x,y)$ is at most $\ell_i$, as the value of the minimum cut separating the vertices of $S\cap \beta_i(x,y)$ from the vertices of $T$ is $\ell_i$. It is now easy to verify that all type-1 constraints in $\kset$ corresponding to the pair $(x,y)$ are satisfied by $\mset^*$, and hence by $\mset^{**}$. Type-2 constraints are dealt with similarly.
\end{proof} 

Our next step is to use the algorithm from Theorem~\ref{thm: approximate DPSP}, in order to compute a set $\mset'\subseteq \mset$ of non-crossing (with respect to $\sigma$ and $\sigma'$) demand pairs, satisfying all constraints in $\kset$, with $|\mset'|\geq \Omega(\opt(G,\mset)/\log k)$. We assume that $\mset'=\set{(s_1,t_1),\ldots,(s_r,t_r)}$, where $s_1,\ldots,s_r$ appear in this circular order on $\Gamma_1$, and if $(s^*,t^*)\in\mset'$, then $s_1=s^*$. If $|\mset'|\leq 10$, then a routing of a single demand pair gives a feasible solution to the \NDP problem instance and achieves the desired approximation ratio. Therefore, we assume that $|\mset'|>10$. We let $\mset''$ contain all demand pairs $(s_j,t_j)$, where $j=0$ modulo $8$, and $1\leq j\leq r$. Notice that $\mset''$ excludes the pair $(s^*,t^*)$. Let $S''$ and $T''$ be the sets of the source and the destination vertices of the demand pairs in $\mset''$. We need the following theorem.

\begin{theorem}\label{thm: max flow}
There is a set $\pset$ of $|S''|$ node-disjoint paths in $G$, connecting the vertices of $S''$ to the vertices of $T''$.
\end{theorem}

We prove Theorem~\ref{thm: max flow} below, after we complete the proof of Theorem~\ref{thm: routing on a disc and cyl main} using it.
Denote $|\mset''|=\kappa^*$, and recall that from our constraints, $\kappa^*\leq |\opt|/4$. Our first step is to construct a collection $\zset=(Z_1,\ldots,Z_{\kappa^*})$ of $\kappa^*$ tight concentric cycles around $\Gamma_1$, where we consider a planar drawing of $G$, whose outer face contains $\Gamma_2$. In order to do so, we denote $\Gamma_1=Z_0$, and perform $\kappa^*$ iteration, where in the $i$th iteration we construct the cycle $Z_i$. In order to execute the $i$th iteration, for $1\leq i\leq \kappa^*$, we contract $D(Z_{i-1})$ into a single vertex $s$, to obtain a new graph $H_i$. We view the face of $H_i$ containing $\Gamma_2$ as the outer face in the planar drawing of $H_i$, and we then let $Z_i=\mincycle(H,s)$. Since the length of the shortest $G$-normal curve connecting a point of $\Gamma_1$ to a point of $\Gamma_2$ is at least $|\opt|/2>\kappa^*$, it is easy to verify that we can successfully complete the construction of the set $\zset$ of $\kappa^*$ cycles, so that all cycles are disjoint from the vertices lying on $\Gamma_2$.

Our next step is to re-route the paths in $\pset$, so that they become monotone with respect to $\zset$. In order to do so, we construct a graph $H$, as follows. We start with the union of the paths in $\pset$ and the cycles in $\zset$. We then add a new cycle $Y$ connecting the vertices of $T''$ in the order in which they appear on $\Gamma_2$. We can now use Theorem~\ref{thm: monotonicity for shells} to find a collection $\pset'$ of $\kappa^*$ node-disjoint paths in $H$, connecting vertices of $S''$ to vertices of $T''$, that are monotone with respect to $\zset$. It is easy to see that the paths of $\pset'$ are contained in $G$.

We assume that $\pset'=\set{P_1,P_2,\ldots,P_{\kappa^*}}$, and for each $1\leq i\leq \kappa^*$, we denote by $a_i\in S''$ and $b_i\in T''$ the endpoints of $P_i$. We assume that $a_1,a_2,\ldots,a_{\kappa^*}$ appear in this circular order on $\Gamma_1$. Consider the source vertex $a_1\in S''$, and let $(a_1,b_{1+z})\in \mset''$ be the demand pair in which $a_1$ participates. We can assume without loss of generality that $z\leq \kappa^*/2$, since if this is not the case, we can re-order the vertices $a_1,\ldots,a_{\kappa^*}$ in the opposite direction around $\Gamma_1$. Observe that for all $1\leq j\leq \floor{\kappa^*/2}$, pair $(a_j,b_{j+z})\in \mset''$, since the demand pairs in $\mset''$ are non-crossing. We now show how to route all demand pairs in $\set{(a_j,b_{j+z})}_{1\leq z\leq \floor{\kappa^*/2}}$. Fix some $1\leq j\leq \floor{\kappa^*/2}$. We view the paths in $\pset'$ as directed from $S''$ to $T''$. Let $P_j'$ be the sub-path of $P_j$ from $a_j$ to the first vertex $v_j$ of $P_j$ lying on $Z_{\kappa^*-j+1}$. Let $P''_{j+z}$ be the sub-path of $P_{j+z}$, from the last vertex $v'_{j+z}$ of $P_{j+z}$ lying on $Z_{\kappa^*-j+1}$ to $b_{j+z}$. Finally, let $Q_j$ be the segment of $Z_{\kappa^*-j+1}$ between $v_j$ to $v'_{j+z}$, that intersects the paths $P_j,P_{j+1},\ldots,P_{j+z}$, but no other paths of $\pset'$. By combining $P_j'$, $P''_{j+z}$ and $Q_j$, we obtain a path $P^*_j$, connecting $a_j$ to $b_{j+z}$. We then set $\pset^*=\set{P^*_j\mid 1\leq j\leq z}$. 
It is immediate to verify that the paths in $\pset^*$ are node-disjoint. Therefore, we obtain a solution routing $\Omega(\opt(G,\mset))$ demand pairs in $\mset$. 
%
%
It now only remains to prove Theorem~\ref{thm: max flow}.

\begin{proofof}{Theorem~\ref{thm: max flow}}
Assume for contradiction that there is no such set of paths. Denote $|S''|=\kappa$ and recall that $\kappa<\opt/2$. Then there is a set $Y$ of at most $\kappa-1$ vertices, so that $G\setminus Y$ contains no path connecting a vertex of $S$ to a vertex of $T$. Consider the drawing of $G$ on the sphere $\Sigma''$, obtained from the drawing of $G$ on the cylinder $\Sigma$, by adding back the two caps with the boundaries $\Gamma_1$ and $\Gamma_2$. We can then construct a simple closed $G$-normal curve $\gamma$ of length at most $\kappa-1$ in $\Sigma''$, so that all vertices of $S''$ lie in one of the discs of $\Sigma''$ with boundary $\gamma$, and all vertices of $T''$ lie on the other disc (but the vertices of $S''$ and $T''$ may lie on $\gamma$). Notice that $\gamma$ has to cross $\Gamma_1$ or $\Gamma_2$. Indeed, otherwise, since there are $\opt$ node-disjoint paths connecting the vertices of $\Gamma_1$ to the vertices of $\Gamma_2$, all such paths would have to cross $\gamma$, and so the length of $\gamma$ should be at least $\opt>\kappa$, a contradiction. Moreover, since the length of the shortest $G$-normal curve connecting a point of $\Gamma_1$ to a point of $\Gamma_2$ is at least $\opt/2>\kappa$, curve $\gamma$ may not intersect both $\Gamma_1$ and $\Gamma_2$. We assume without loss of generality that $\gamma$ crosses $\Gamma_1$, and not $\Gamma_2$.

Let $\rset$ be the set of segments of $\gamma$, obtained by deleting all points of $\gamma$ that lie outside the cylinder $\Sigma$ (that is, the points that lie in the interior of the cap whose boundary is $\Gamma_1$). All curves in $\rset$ are mutually disjoint. For each curve $\gamma'\in \rset$, let $S(\gamma')\subseteq S''$ be the set of the source vertices that $\gamma'$ separates from $\Gamma_2$ in the cylinder $\Sigma$. Then $\bigcup_{\gamma'\in \rset}S(\gamma')=S''$ must hold, and so there must be some curve $\gamma^*\in \rset$, such that the length of $\gamma^*$ is less than $|S(\gamma^*)|$. Let $\ell^*$ denote the length of $\gamma^*$. 

Let $x',y'$ be the endpoints of the curve $\gamma^*$, so $x',y'\in \Gamma_1$. If $x'\in W_1$, then we let $x=x'$. Otherwise, we let $x$ be the closest to $x'$ point of $W_1\setminus S''$ on $\Gamma_1$. We define point $y'$ for $y$ similarly. Consider the two segments $\beta_1(x,y)$ and $\beta_2(x,y)$ of $\Gamma_1$, whose endpoints are $x$ and $y$. One of the segments, say $\beta_1(x,y)$ must contain all points of $S(\gamma^*)$. Since the vertices lying on $\gamma^*$ separate all vertices of $S(\gamma^*)$ from the vertices of $T$, $\ell_1(x,y)\leq \ell^*$. Assume without loss of generality that $x$ lies before $y$ on $\sigma$.

Assume first that $s^*$ does not lie on $\beta_1(x,y)$, and consider the corresponding constraint $K=(1,x,y,w_1)\in \kset$. As observed above, $w_1\leq \ell^*$. Due to the way we have selected the subset $\mset''\subseteq \mset'$ of the demand pairs, we are guaranteed that $|S(\gamma^*)|\leq w_1/2\leq \ell^*$, a contradiction.

Assume now that $s^*$ lies on $\beta_1(x,y)$. Using the same reasoning as above, $w_1\leq \ell^*$. Let $\beta_1',\beta_1''$ be the segments of $\beta_1(x,y)$ between $x$ and $s^*$, and between $s^*$ and $y$, respectively, where the last segment excludes $s^*$. Let $\Delta_1,\Delta_1'$ be the number of the source vertices lying on $\beta_1$ and $\beta_1''$ respectively, that participate in the demand pairs in $\mset'$. Since the constraints $(1,s^*,x,w_1)$ and $(1,y,a^*,w_1)$ belong to $\kset$, $\Delta_1+\Delta_1'\leq 2\ell^*$. Due to the way we have selected the subset $\mset''\subseteq \mset'$ of the demand pairs, we are guaranteed that $|S(\gamma^*)|\leq \max\set{w_1,1}\leq \ell^*$, a contradiction.
\end{proofof}

\section{Proofs Omitted from Section~\ref{sec: alg overview}}\label{sec: proofs alg overview}
\subsection{Proof of Theorem~\ref{thm: wld}}

Let $\tau = \frac{w^*}{512 \cdot \log k}$. The algorithm is iterative and maintains a set $U$ of vertices. We start with $U=\emptyset$, and in every iteration we add vertices to $U$. The algorithm terminates when no vertices have been added to $U$ in an iteration. Each iteration is executed as follows:

Let $\hset$ denote the set of all connected components of $G\setminus U$, and for each $H \in \hset$, let $\tset_H \subseteq \tset(\mset)$ denote the set of all terminals contained in $V(H)$. For each $H \in \hset$ with $|V(H)\cap \tset_H|>3$, we use  Observation~\ref{obs: sparsest cut} to compute a vertex cut $(A,C,B)$ in $H$ whose sparsity $\phi$ with respect to $\tset_H$ is within a factor $\alphasc$ from the value of the sparsest cut, and $C\cap \tset_H=\emptyset$. If $\phi < \tau$, then we add the vertices of $C$ to $U$. This finishes the description of the iteration execution.

Consider the set $\set{G_1, \ldots, G_r}$ of all components of $G \setminus U$ once the algorithm terminates. For all $1\leq j\leq r$, let $\mset^j = \set{(s_i,t_i) \in \mset \mid s_i,t_i \in V(G_j)}$. The algorithm output is $\set{(G_j, \mset^j)}_{j=1}^r$.

Clearly, the algorithm is efficient and terminates after at most $n+1$ iterations, since the size of $U$ increases after each iteration, except for the last one. It is easy to see that no edge connects a vertex of $G_j$ to a vertex of $G_{j'}$ for any $1 \le j \neq j' \le r$, from the definition of the graphs $G_j$.

We now verify that for all $1\leq j\leq r$, $(G_j, \mset^j)$ is a well-linked instance. Fix some $1 \le j \le r$, and let $\tset', \tset'' \subseteq \tset(\mset^j)$ be two disjoint equal-sized subsets of $\tset(\mset^j)$. Assume for contradiction that there are fewer than $\alphaWL\cdot |\tset'|$ node-disjoint paths in $\mset^j$ connecting the vertices of $\tset'$ to the vertices of $\tset''$. Then by Menger's Theorem, there exists a set $Z \subseteq V(G_j)$ of fewer than $\alphaWL\cdot |\tset'|$ vertices in $G^j$, such that there is no path from $\tset'$ to $\tset''$ in $G_j\setminus Z$. Note that we may assume that $Z \cap \tset(\mset^j) = \emptyset$. Otherwise, since the terminals have degree $1$ and form an independent set in $G$, we may simply replace each terminal in $Z$ with its unique neighbor. Consider a vertex cut $(A', C', B')$ of $G_j$, defined as follows: $C' = Z$, $A'$ is the union of the vertices of all components of $G_j \setminus Z$ intersecting $\tset'$, and $B' = V(G_j) \setminus (A' \cup C')$. This is a valid vertex cut, with $\tset'\subseteq A'$ and $\tset''\subseteq B'$. The sparsity of cut $(A', C', B')$ with respect to $\tset(\mset^j)$ is at most $\frac{|Z|}{\min\set{|\tset_H\cap A'|,|\tset_H\cap B'|}} < \frac{\alphaWL |\tset'|}{\min\set{|\tset'|,|\tset''|}}=\alphaWL$. Therefore, the algorithm from Observation~\ref{obs: sparsest cut} should have returned a cut of sparsity less than $\alphasc\cdot \alphaWL = \tau$, a contradiction.

We now show that $|U| = \left |V(G)\setminus\left(\bigcup_{j=1}^rV(G_j)\right )\right |\leq \frac{w^*\cdot |\mset|}{64}$. Consider a single iteration of the algorithm. Let $H$ be any component of $G\setminus U$, and let $\tset_H$ be the set of all terminals contained in $V(H)$. Suppose the algorithm computes a vertex cut $(A,C,B)$ in $H$ with respect to $\tset_H$ of sparsity $\phi < \tau$, and adds the vertices of $C$ to $U$. Assume without loss of generality that $|A \cap \tset_H| \le |B \cap \tset_H|$. Since we have assumed that $C\cap \tset_H=\emptyset$, the sparsity of the cut $\phi = \frac{|C|}{|A \cap \tset_H|}$. Moreover, since $\phi < \tau$, $|C| < \tau \cdot |A \cap \tset_H|$ must hold. We  charge the value of $\tau$ to every terminal in $A \cap \tset_H$, so that the total amount charged to the terminals of $A\cap \tset_H$ is $\tau \cdot |A \cap \tset_H|\geq |C|$. This charging scheme is repeated whenever a set of vertices is added to $U$ throughout the different iterations and components considered by the algorithm. Note that a terminal may be charged during multiple iterations, but at most once per iteration. Clearly, the sum of the total charges to all of the terminals is at least $|U|$. Also, each terminal $t \in \tset(\mset)$ can be charged at most $\lfloor \log 2k \rfloor$ times, since whenever $t$ is charged and $U$ is updated, the number of terminals in the component containing $t$ in $G\setminus U$ falls by at least a factor $2$. Therefore, the total charge to all terminals is at most $\tau\cdot |\tset(\mset)|\cdot \log 2k  \le (\frac{w^*}{512 \cdot \log k}) \cdot ( \log 2k) \cdot |\tset(\mset)| \leq  \frac{w^* \cdot |\mset|}{64}$, and $|U|$ is also bounded by this amount.

Finally, we verify that $\sum_{j=1}^r|\mset^j|\geq 63|\mset|/64$. Recall that for each demand pair $(s_i, t_i) \in \mset$, the current LP-solution sends $w^*$ flow units between $s_i$ and $t_i$. Let $\tmset=\mset\setminus\left(\bigcup_{j=1}^r\mset^j\right )$. If $(s_i, t_i) \in \tmset$, then all of the $w^*$ flow units between $s_i$ and $t_i$ must pass through $U$. Since $|U| \le \frac{w^* \cdot |\mset|}{64}$, $|\tmset|\leq \frac{|\mset|}{64}$, and the theorem follows.

\subsection{Proof of Lemma~\ref{lem: nj bound}}


We use the notion of treewidth, which is usually defined via tree decompositions. A tree decomposition of a graph $H $ consists of a tree $\tau$ and a collection $\set{\beta_v \subseteq V(H)}_{v \in V(\tau)}$ of vertex subsets, called bags, such that the following two properties are satisfied: (i) for each edge $(a,b) \in E(H)$, there is some node $v \in V(\tau)$ with both $a,b \in \beta_v$ and (ii) for each vertex $a \in V$, the set of all nodes of $\tau$ whose bags contain $a$ form a non-empty connected subtree of $\tau$. The \emph{width} of a given tree decomposition is $\max_{v \in V(\tau)}\set{ |\beta_v|} - 1$, and the treewidth of a graph $H$, denoted by $\tw(H)$, is the width of a minimum-width tree decomposition for $H$.

\begin{claim}\label{claim: treewidth}
For each $1\leq j\leq r$, $\tw(G_j) \ge \frac{W_j}{2^{12} \cdot \alphasc \cdot \log k}$.
\end{claim}

We first prove the lemma assuming Claim~\ref{claim: treewidth}.
It is well known that any planar graph of large treewidth contains a large grid as a minor~\cite{RST94,DH05}. We use the following theorem.

\begin{theorem}[Theorem 1.2 in~\cite{DH05}]
 For any fixed graph $H$, every $H$-minor-free graph of treewidth $w$ has an $(\Omega(w) \times \Omega(w))$ grid as a minor. 
 \end{theorem}
 
 Therefore, in particular, every planar graph of treewidth $w$ contains an $\left(\Omega(w) \times \Omega(w)\right)$ grid as a minor. So $G_j$ must contain a grid minor of size $\left(\Omega(W_j/\log k)\times \Omega(W_j/\log k)\right )$. Since all terminals have degree $1$ in $G$, the number of the non-terminal vertices, $N_j\geq  \Omega(W_j^2 / \log^2 k)$. It now remains to prove Claim~\ref{claim: treewidth}.

\begin{proofof}{Claim~\ref{claim: treewidth}}
For convenience, we let $\kappa = |\tset(\mset^j)|$. Assume for contradiction that $\tw(G_j)< \frac{W_j}{2^{12} \cdot \alphasc \cdot \log k}$ and consider a tree decomposition $\tau$ of width less than $\frac{W_j}{2^{12} \cdot \alphasc\cdot\log k}$. Note that $\tau$ cannot be a singleton vertex, as $\kappa > \frac{W_j}{2^{12} \cdot \alphasc\cdot\log k}$, since $W_j = w^* \cdot \kappa/2$ and $w^*<1$. For any given subtree $\tau'$ of $\tau$, we let $\beta(\tau') = \bigcup_{u \in V(\tau')} \beta_u$. We say that a vertex $v \in V(\tau)$ is good iff every component of $G \setminus \beta_v$ contains at most $\kappa/2$ terminals.

\begin{claim} \label{claim: good vertex} There is a good vertex in $\tau$.
\end{claim}

\begin{proof}
Note that for a vertex $v \in V(\tau)$, if $\tau_1, \ldots, \tau_\ell$ are the connected subgraphs of $\tau \setminus \set{v}$, then every connected component $C$ of $G_j\setminus \beta_v$ must have $V(C) \subseteq \beta(\tau_p) \setminus \beta_v$ for some $1 \le p \le \ell$. 
Also note that the sets $\set{\beta(\tau_1) \setminus \beta_v, \ldots, \beta(\tau_\ell) \setminus \beta_v}$ are pairwise vertex disjoint.

Root the tree $\tau$ at any vertex $v_0$, and start with $v=v_0$. While the current vertex $v$ has a child $v_i$, such that the sub-tree $\tau_i$ of $\tau$ rooted at $v_i$ has $|(\beta(\tau_i)\setminus \beta_v)\cap \tset(\mset^j)|>|\tset(\mset^j)|/2$, we set $v=v_i$, and continue to the next iteration. It is immediate to verify that when the algorithm terminates, the final vertex $v$ is good.
\end{proof}

Let $v \in V(\tau)$ be a good vertex, and let $C_1, \ldots, C_a$ denote the connected components of $G_j\setminus \beta_v$. For all $1\leq p\leq a$,  let $\kappa_p = |\tset(\mset^j) \cap C_p|$. Note that $|\beta_v| \le \frac{W_j}{2^{12} \cdot \alphasc\cdot\log k} = \frac{w^* \cdot \kappa}{2^{13} \cdot \alphasc \cdot \log k} \le \kappa/4$, and hence $|\tset(\mset^j) \setminus \beta_v| = \sum_{p=1}^a \kappa_p \ge 3\kappa/4$. We claim that there exist two disjoint subsets $\tset', \tset'' \subseteq \tset(\mset^j) \setminus \beta_v$ such that $|\tset'| = |\tset''|=\ceil{\kappa/8}$, while $\tset'$ and $\tset''$ are separated by $\beta_v$ in $G_j$. Let $1 \le b < a$ be the smallest index for which $\sum_{p=1}^b \kappa_p \ge \kappa/8$. Then $\sum_{p=1}^b \kappa_p \le (1/8 + 1/2)\kappa = 5\kappa/8$, and so $\sum_{p=b+1}^a \kappa_p \ge (3/4 - 5/8)\kappa = \kappa/8$. We then let $\tset' \subseteq \left( \bigcup_{p=1}^b C_p \cap \tset(\mset^j) \right)$ and $\tset'' \subseteq \left( \bigcup_{p=b+1}^a C_p \cap \tset(\mset^j) \right)$, respectively, be subsets of size $\ceil{\kappa/8}$, so $\tset'\cap \tset''=\emptyset$, and $\beta_v$ separates $\tset'$ from $\tset''$ in $G_j$.

Since the terminals are $\alphaWL$-well-linked in $G_j$, there is a set of at least $\alphaWL\cdot |\tset'|= \frac{w^* \cdot |\tset'|}{512 \cdot \alphasc \cdot \log k} \ge \frac{w^* \cdot \kappa}{2^{12} \cdot \alphasc \cdot \log k} = \frac{W_j}{2^{12} \cdot \alphasc \cdot \log k}$ node-disjoint paths from $\tset'$ to $\tset''$ in $G_j$. Since $\tset'$ and $\tset''$ are separated by $\beta_v$ in $G \setminus \beta_v$, each path must intersect at least one distinct vertex of $\beta_v$. However, we have assumed that $|\beta_v| < \frac{W_j}{2^{12} \cdot \alphasc\cdot\log k}$, a contradiction.
\end{proofof}



\newpage
\section{Table of Parameters}\label{sec: appendix-params-table}
\renewcommand{\arraystretch}{1.4}
\begin{tabular}{|l|l|p{10cm}|} \hline
$\alphawl$&$\frac{w^*}{512 \cdot \alphasc \cdot \log k}=\Theta(w^*/\log k)$& well-linkedness parameter, where $k$ is the number\\
&& of the demand pairs in the original instance.\\ \hline
$\Delta$&$\ceil{W^{2/19}}$&Minimum distance between terminals in distinct terminal sets of $\xset$.\\ \hline
$\Delta_0$&$\Theta(\Delta\log n)$&Maximum distance between terminals in each terminal set of $\xset$.\\ \hline
$\tau$&$W^{18/19}$&Threshold for light and heavy clusters in $\xset$\\ \hline
$\Delta_1$&$\floor{\Delta/6}$& Depth of shells in Case 1\\ \hline
$\Delta_2$&$\floor{\Delta_1/3}$&Depth of inner shells in Case 1\\ \hline
\end{tabular}

\newpage




\label{------------------------------------------------END-------------------------------}

\end{document}